\documentclass[11pt]{article}
\usepackage{styling}
\usepackage{amsmath,amsthm,amsfonts}
\usepackage[margin=1in]{geometry}
\usepackage{mathrsfs}
\usepackage{graphicx}
\usepackage{xcolor}
\usepackage{framed}
\usepackage{enumitem}
\usepackage{fancybox, pdfsync}
\usepackage{mathtools}
\usepackage{capt-of}
\usepackage{comment}
\usepackage{algorithmic}
\usepackage{multirow}
\usepackage{nicematrix}
\usepackage[font=small,skip=0pt]{caption}

\newcommand{\mai}[1]{\mathbf{A}^{(#1)}}
\newcommand{\vxhat}{\widehat{\mathbf{x}}}
\newcommand{\vzhat}{\widehat{\mathbf{z}}}
\newcommand\Otil{\widetilde{O}}
\newcommand{\mm}{\mathbf{M}}
\newcommand{\vbi}[1]{\mathbf{b}^{(#1)}}
\newcommand\levover{\widehat{\tau}}
\newcommand{\matil}{\widetilde{\mathbf{A}}}
\newcommand{\mstil}{\widetilde{\mathbf{S}}}
\newcommand{\vxtil}{\widetilde{\mathbf{x}}}

\newcommand{\vp}{\mathbf{p}}
\newcommand{\twopartdef}[4]
{
\left\{
\begin{array}{ll}
#1 & \mbox{} #2 \\
#3 & \mbox{} #4
\end{array}
\right.
}

\newcommand\lev{\tau}
\newcommand{\vs}{\mathbf{s}}

\usepackage{tikz}

\newcommand*\circled[1]{\tikz[baseline=(char.base)]{%
            \node[shape=circle,fill=white,draw,inner sep=.6pt] (char) {#1};}}
\newcommand{\ms}{\mathbf{S}}

\newcommand{\x}{\mathbf{x}}
\renewcommand{\b}{\mathbf{b}}

\newcommand{\mb}{\mathbf{B}}
\newcommand{\mi}{\mathbf{I}}
\newcommand{\mg}{\mathbf{G}}
\newcommand{\mghat}{\widehat{\mathbf{G}}}
\newcommand{\mj}{\mathbf{J}}
\newcommand{\vecv}{\mathbf{v}}
\global\long\def\vu{\mathbf{u}}%
\global\long\def\vutil{\widetilde{\mathbf{u}}}%
\global\long\def\vv{\mathbf{v}}%
\newcommand{\veps}{\widehat{\epsilon}}
\newcommand{\mbtil}{\widetilde{\mathbf{B}}}
\newcommand{\nullsp}{\mathcal{N}}
\newcommand{\mshat}{\widehat{\mathbf{S}}}
\newcommand{\msbar}{\overline{\mathbf{S}}}
\newcommand{\mxbar}{\overline{\mathbf{X}}}
\newcommand{\mxhat}{\widehat{\mathbf{X}}}
\newcommand{\vxstar}{\mathbf{x}^\ast}
\newcommand{\mh}{\mathbf{H}}
\newcommand{\mx}{\mathbf{X}}

\newcommand{\mv}{\mathbf{V}}
\newcommand{\mw}{\mathbf{W}}
\newcommand{\mM}{\mathbf{M}}
\newcommand{\mX}{\mathbf{X}}

\newcommand{\blev}{\mathcal{L}}

\renewcommand{\norm}[2]{\|#1\|_{#2}}

\newcommand{\pr}{\Pr}

\newcommand{\tildeO}{\widetilde{O}}
\newcommand{\widetildeO}{\widetilde{O}}

\newcommand{\matv}{\mathbf{V}}
\newcommand{\mU}{\mathbf{U}}
\newcommand{\sketchfactor}{c}
\newcommand{\lewisWtApx}{\mathbf{q}_i}
\newcounter{casenum}
\newenvironment{caseof}{\setcounter{casenum}{1}}{\vskip.5\baselineskip}
\newcommand{\case}[2]{\vskip.5\baselineskip\par\noindent {\bfseries Case \arabic{casenum}:} #1\\#2\addtocounter{casenum}{1}}

\newcommand{\mahat}{\widehat{\mathbf{A}}}

\newcommand{\mutil}{\widetilde{\boldsymbol{\mu}}}
\newcommand{\dmutil}{\delta_{\widetilde{\mu}}}
\newcommand{\fro}{\mathrm{F}}

\DeclareMathOperator{\Tr}{Tr}

\DeclareMathOperator{\rank}{\textrm{rank}}
\DeclareMathOperator{\logdet}{logdet}

\newcommand{\uij}{u_i^{(j)}}
\newcommand{\vij}{v_i^{(j)}}

\newcommand{\mt}{\mathbf{T}}

\newcommand{\vg}{\mathbf{g}}
\newcommand{\Ai}{{\mathbf{A}^{(i)}}}
\newcommand{\Aone}{\mathbf{A}^{(1)}}

\newcommand{\tildeBj}{\widetilde{\mathbf{B}}^{(j)}}

\newcommand{\Aj}{\mathbf{A}^{(j)}}

\newcommand{\ws}[1]{\textcolor{blue}{#1 --WS}}
\newcommand{\pswt}[1]{{\color{blue} [#1 -- SP]}}
\newcommand{\snote}[1]{\footnote{\textbf{SP: }#1}}
\newcommand{\Guanghao}[1]{\textcolor{blue}{#1 --GY}}
\newcommand{\td}[1]{\todo[inline]{#1\,--WS}}
\newcommand{\mehrdad}[1]{\textcolor{orange}{Mehrdad: #1}}

\renewcommand{\epsilon}{\varepsilon}

\DeclareUnicodeCharacter{0301}{*************************************}
\definecolor{mygreen}{rgb}{.13, .55, .13}
\definecolor{myblue}{rgb}{.1, .43, .95}
\definecolor{myred}{rgb}{.7, .133, .133}

\global\long\def\inprod#1#2{\langle#1,#2\rangle}%

\newcommand{\vvi}[1]{\mathbf{v}_{#1}}

\newcommand{\dx}{\mathbf{\boldsymbol{\delta}_x}}
\newcommand{\ds}{\mathbf{\boldsymbol{\delta}_s}}

\newcommand\old{\text{old}}

\newcommand\appr{\text{apr}}

\newcommand{\va}{\mathbf{a}}

\newcommand{\mk}{\mathbf{K}}
\newcommand{\md}{\mathbf{D}}

\newcommand{\mproj}{\mathbf{P}}

\newcommand{\vgtil}{\widetilde{\mathbf{g}}}

\newcommand{\vd}{\mathbf{d}}

\global\long\def\barrini{\psi_{\textrm{in,}i}}%
\global\long\def\barrouti{\psi_{\textrm{out,}i}}%
\global\long\def\barrout{\psi_{\textrm{out}}}%
\global\long\def\barr{\psi}%
\global\long\def\barrUniK{\psi_{\mathcal{K}}}%
\global\long\def\barrUniKNew{\psi_{\mathcal{K},\textrm{new}}}%

\global\long\def\vxos{\vx_{\textrm{out}}^{\star}}%
\global\long\def\vxosi{\vx_{\textrm{out,}i}^{\star}}%
\global\long\def\vxi{\vx_{\textrm{i}}}%

\global\long\def\vv{\mathbf{v}}%
\global\long\def\vw{\mathbf{w}}%
\global\long\def\vz{\mathbf{z}}%
\global\long\def\vzi{\vz_{i}}%
\global\long\def\vy{\mathbf{y}}%
\global\long\def\vxt{\vx_{t}}%
\global\long\def\vu{\mathbf{u}}%
\global\long\def\vutil{\widetilde{\mathbf{u}}}%
\global\long\def\ve{\mathbf{e}}%
\global\long\def\deltaX{\delta_{\vx}}%
\global\long\def\deltaXi{\delta_{\vx,i}}%

\global\long\def\ki{\mathcal{K}_{i}}%
\global\long\def\kini{\mathcal{K}_{\textrm{in,}i}}%
\global\long\def\kin{\mathcal{K}_{\textrm{in}}}%

\global\long\def\kout{\mathcal{K}_{\textrm{out}}}%
\global\long\def\kouti{\mathcal{K}_{\mathrm{out,}i}}%
\global\long\def\kcal{\mathcal{K}}%

\global\long\def\oi{\mathcal{O}_{i}}%
\global\long\def\hi{\mathcal{H}_{i}}%
\global\long\def\conv{\mathrm{conv}}%
\global\long\def\h{\mathcal{H}}%
\global\long\def\vol{\mathrm{vol}}%

\global\long\def\kouthat{\widehat{\mathcal{K}}_{\textrm{out}}}%

\global\long\def\tinit{t_{\textrm{init}}}
\global\long\def\tend{t_{\textrm{end}}}

\usepackage{tikz}

\def\defeq{\stackrel{\mathrm{def}}{=}}

\def\ceil#1{\left\lceil #1 \right\rceil}

\def\abs#1{|#1  |}

\makeatletter

\newcommand\ptope{\mathcal{P}}
\newcommand\dualptope{\mathcal{D}}

\newcommand{\vshat}{\widehat{\mathbf{s}}}
\newcommand{\vxbar}{\overline{\mathbf{x}}}
\newcommand{\vvbar}{\overline{\mathbf{v}}}
\newcommand{\vwbar}{\overline{\mathbf{w}}}
\newcommand{\mwbar}{\overline{\mathbf{W}}}
\newcommand{\vyhat}{\widehat{\mathbf{y}}}
\newcommand{\tauhatv}{\widehat{\boldsymbol{\tau}}}

\newcommand{\vybar}{\overline{\mathbf{y}}}

\newcommand{\mabar}{\overline{\mathbf{A}}}
\newcommand{\vbbar}{\overline{\mathbf{b}}}
\newcommand{\vcbar}{\overline{\mathbf{c}}}
\newcommand{\vctil}{\widetilde{\mathbf{c}}}
\newcommand{\vchat}{\widehat{\mathbf{c}}}
\newcommand{\vsbar}{\overline{\mathbf{s}}}
\newcommand\ipm{\textsc{IPM}}
\newcommand\final{\text{end}}

\newcommand{\vvhat}{\widehat{\mathbf{v}}}
\newcommand{\vwhat}{\widehat{\mathbf{w}}}
\newcommand{\mwhat}{\widehat{\mathbf{W}}}

\usepackage[ruled,linesnumbered]{algorithm2e}

\usepackage[
backend=biber,
style=alphabetic,
maxbibnames=10,
maxalphanames=6, 
backref=true
]{biblatex}

\usepackage{fontawesome}

\usepackage[capitalize, nameinlink]{cleveref}
\newcommand\numberthis{\addtocounter{equation}{1}\tag{\theequation}} %
\numberwithin{equation}{section}

\crefname{prob}{Problem}{Problems}
\crefname{problem}{Problem}{Problems} 
\crefname{alg}{Algorithm}{Algorithms}
\crefname{conj}{Conjecture}{Conjectures}
\crefname{sec}{Section}{Sections}
\crefname{assumption}{Assumption}{Assumptions}
\crefformat{none}{(#2#1#3)}
\crefname{equation}{Equation}{Equations}
\crefname{lem}{Lemma}{Lemmas}
\crefname{rem}{Remark}{Remarks}
\crefname{claim}{Claim}{Claims}
\crefname{figure}{Figure}{Figures}

\crefname{lemma}{Lemma}{Lemmas}
\crefname{prop}{Proposition}{Propositions}
\crefname{defn}{Definition}{Definitions}
\crefname{cor}{Corollary}{Corollaries}
\crefname{ineq}{Inequality}{Inequalities}
\crefname{fact}{Fact}{Facts}
\crefname{item}{Step}{Steps}
\crefname{assumption}{Assumption}{Assumptions}
\creflabelformat{ineq}{#2{\upshape(#1)}#3} 
\crefalias{thm}{theorem}
\crefalias{AlgoLine}{line}

\crefname{enumi}{Step}{Steps} 
\let\cref\cref

\definecolor{mygreen}{rgb}{.13, .55, .23}
\definecolor{myblue}{rgb}{.1, .43, .95}
\definecolor{myred}{rgb}{.7, .133, .133}
\definecolor{mylightblue}{RGB}{0,128,255}
\definecolor{otherlightblue}{RGB}{0, 100, 200}
\definecolor{otherblue}{RGB}{0, 50, 100}
\definecolor{othergreen}{RGB}{60, 120, 0}

\hypersetup{
	colorlinks=true,
	pdfpagemode=UseNone,
	citecolor=othergreen,
	linkcolor=otherlightblue,
	urlcolor=magenta
}

\newlist{itemizec}{itemize}{2}
\setlist[itemizec,1]{label=\faCaretRight ,wide, parsep= 0.05pt, left = 15pt}

\begin{document}

\allowdisplaybreaks

\title{Improving
 the Bit Complexity of Communication\\  for Distributed Convex Optimization\footnotetext{\llap{\textsuperscript{1}}The author ordering is alphabetical.  A preliminary version of this manuscript is published in STOC 2024.} }

\author{
Mehrdad Ghadiri\thanks{Massachusetts Institute of Technology. \texttt{mehrdadg@mit.edu}} 
\and Yin Tat Lee\thanks{University of Washington, Seattle and Microsoft Research.  \texttt{yintat@uw.edu}} 
\and Swati Padmanabhan \thanks{Massachusetts Institute of Technology. \texttt{pswt@mit.edu}. Work done in part at the University of Washington, Seattle.} \\
\and William Swartworth \thanks{Carnegie Mellon University. 
\texttt{wswartwo@andrew.cmu.edu}} 
\and David P. Woodruff\thanks{Carnegie Mellon University. 
\texttt{dwoodruf@cs.cmu.edu}}
\and Guanghao Ye \thanks{Massachusetts Institute of Technology. \texttt{ghye@mit.edu}
}}

\date{}
\pagenumbering{roman}

\maketitle  

\thispagestyle{empty}
\begin{abstract}
We consider the communication complexity of some fundamental convex optimization problems in the point-to-point (coordinator) and blackboard communication models. We strengthen known bounds for approximately solving linear regression, $p$-norm regression (for $1\leq p\leq 2$), linear programming, minimizing the sum of finitely many convex nonsmooth functions with varying supports, and low rank approximation; for a number of these fundamental problems our bounds are optimal, as proven by our lower bounds.

For example, for solving least squares regression in the coordinator model with $s$ servers, $n$ examples, $d$ dimensions, and coefficients specified using at most $L$ bits, we improve the prior communication bound of Vempala, Wang, and Woodruff (SODA, 2020) from $\tilde{O}(sd^2L)$ to $\tilde{O}(sdL + {d^2}{\epsilon^{-1}}L)$, which is optimal up to logarithmic factors. We also study the problem of solving least squares regression in the coordinator model to \textit{high accuracy}, for which we provide an algorithm with a communication complexity of $\Otil(sd (L+\log\kappa) \log(\epsilon^{-1}) + d^2L)$, matching our improved lower bound for well-conditioned matrices up to a $\log(\epsilon^{-1})$ factor.
Among our techniques, we use the notion of \textit{block leverage scores}, which have been relatively unexplored in this context, as well as dropping all but the ``middle" bits in Richardson-style algorithms. We also introduce a new communication problem for accurately approximating inner products and establish a lower bound using the spherical Radon transform. Our lower bound can be used to show the first separation of linear programming and linear systems in the distributed model when the number of constraints is polynomial, addressing an open question in prior work. 

We also give an improved algorithm for high-accuracy linear programming in the coordinator model that computes an approximate solution on well-conditioned inputs using  $\tilde{O}(sd^{1.5} L + d^2 L)$ communication.  This  improves over the previous bound of $sd^2L$.
Finally, we give an improved algorithm, in the blackboard model of communication, for the problem $\min_{\theta\in \R^d} \sum_{i=1}^s f_i(\theta)$ where each $f_i$ is convex, Lipschitz, and supported on $d_i\leq d$ (potentially overlapping) coordinates of $\theta$ using $\tilde{O}(\sum_{i=1}^s d_i^2 L)$ communication. Our techniques yield improved rates for decomposable submodular function minimization in the non-distributed setting as well.
\end{abstract}

\newpage
\pagenumbering{arabic}
\tableofcontents

\newpage
\section{Introduction}
\looseness=-1 
The scale of modern optimization problems often necessitates working with datasets that are distributed across multiple machines, which then communicate with each other to solve the optimization problem at hand. A crucial performance metric for algorithms  in such distributed settings is the {communication complexity}. Traditionally, this has referred to the \emph{number of rounds} of communication   needed between the machines to  solve the problem, and there has been a long line of work (which we shortly describe) optimizing this metric. However, as was highlighted in ~\cite{vww20, peng2021solving, ghadiri2023bit}, in many core algorithmic primitives underlying recent advances in continuous optimization, the claimed (theoretical) runtimes are predicated on the assumption of exact computations with infinite precision. When analyzed under the finite-precision  model, the true runtimes can be substantially higher. As a consequence, inferring the true cost of distributed optimization algorithms built with these components requires a careful analysis. To address this need, our focus in this paper is on designing, for some fundamental optimization problems, distributed algorithms that are efficient in the \emph{total number of bits communicated}. 

\looseness=-1Before describing our setup and results, we first provide a brief overview of prior work in the related area of distributed optimization, a mature field encompassing  problems spanning engineering, control theory, signal processing, and machine learning.  
For instance, multi-agent coordination, distributed tracking
and localization, estimation problems in sensor networks, opinion dynamics, and
packet routing are all naturally cast as distributed convex
minimization~\cite{lesser2003distributed, sayed2014adaptation, bertsekas2015parallel}. 
Classically, the primary goal in these problems was to design a communication strategy between the computational agents so that they eventually arrive at the optimal objective value~\cite{tsitsiklis1984problems}.
A considerable  body of work~\cite{jadbabaie2003coordination, xiao2007distributed, nedic2009distributed, sundhar2010distributed} has therefore been devoted to obtaining  asymptotic convergence guarantees for these problem classes.
Going beyond asymptotic analysis, recent years have witnessed extensive progress in obtaining \emph{non-asymptotic rates} (typically in terms of the number of rounds of communication) for 
problems in distributed machine learning such as distributed PAC learning~\cite{balcan2012distributed}, distributed online
prediction~\cite{d12b}, distributed estimation~\cite{nedic2009distributed, johansson2010randomized, d12}, and distributed delayed stochastic optimization~\cite{nedic2001distributed, agarwal2011distributed}.

\looseness=-1A related paradigm that has recently emerged  in distributed computing is that of \emph{federated learning}~\cite{kairouz2021advances}. In this paradigm, the processes of data acquisition, processing, and model training are largely carried out on a network's edge nodes such as smartphones~\cite{bonawitz2019towards}, wearables~\cite{huang2020loadaboost}, location-based services~\cite{samarakoon2019distributed}, and IoT sensors~\cite{mcmahan2017communication, konevcny2016federated}, under the orchestration of a central coordinator. 
Similar to the recent works on distributed machine learning mentioned in the preceding paragraph, for the works in this setting as well, it is the number of rounds of communication that is typically used as a proxy for total communication cost. Additional important concerns for works in federated learning include user privacy and robustness to distribution shifts in users' samples~\cite{reisizadeh2020robust} and to heterogeneity in the computational capabilities of the nodes~\cite{reisizadeh2022straggler}. 
\looseness=-1Finally, while our focus in this paper is the theory, we note that advances in the practice of distributed computing have been tremendously spurred by the development of programming models like MapReduce~\cite{dean2008mapreduce}, which enable parallelizing the computation, distributing the
data, and handling failures across thousands of machines.

\looseness=-1\paragraph{Our setup.} As mentioned earlier,  only recently has there been a surge of interest in studying the bit complexity of optimization algorithms~\cite{vww20, peng2021solving, ghadiri2023bit}. In this paper, we hope to continue pushing efforts in this direction and study the number of bits communicated to solve various \textit{distributed} convex optimization problems under two models of communication, defined next. Our goal is to compute approximate solutions 
with efficient communication complexity.  

\begin{definition}[Coordinator Model]\label{def:coordinatormodel} There are $s$ machines (servers) and a central coordinator. Each machine can send information to and receive information from the coordinator. Any bit sent or received to the coordinator counts toward the communication complexity of the algorithm.
\end{definition}

\begin{definition}[Blackboard Model]\label{def:blackboardmodel} There are $s$ machines and a coordinator (blackboard). Each machine can send information to and receive information from the coordinator. Only bits \textit{sent} to the coordinator count toward the communication complexity of the algorithm.
\end{definition}

\looseness=-1The coordinator model is equivalent, up to a factor of two and an additive $\log s$ bits per message, to the \emph{point-to-point model of computation}, in which  machines directly interact with each other. The blackboard model may be viewed as having a shared memory between the machines, since it costs the machines only to write on to  the blackboard, while reading from the blackboard is free. 

\looseness=-1We consider several fundamental optimization problems that have been studied extensively outside the distributed setting: least squares regression, low rank approximation, linear programming, and optimizing a sum of convex nonsmooth functions.
We provide improved communication upper and lower bounds for these problems in the aforementioned distributed settings.
While we obtain nearly tight upper and lower bounds for several of these problems in the ``worst-case" settings, e.g., when matrices are arbitrarily poorly conditioned, another important component of our work is in improving bounds for well-behaved inputs, e.g., well-conditioned matrices or decomposable functions.

\subsection{Our Contributions}\label{sec:contributions}
\looseness=-1In this paper, we address
the communication complexity of least squares regression, low-rank approximation, and linear programming in the coordinator model, and finite-sum minimization of Lipschitz functions in the blackboard model. Our central technical novelty lies in developing efficient --- in terms of bit complexity ---  methods for leverage score sampling, inverse maintenance, cutting-plane methods, and the use of block leverage scores
in the distributed setting and in finite arithmetic.  We summarize some of our results  in \cref{tab:coordinator_model_regression}, with all formal statements in this section. 

\begin{table}[t]

 \centering
  \resizebox{\linewidth}{!}{%
  \begin{NiceTabular}{*{4}{c}}
  \CodeBefore
  \rowcolors{1}{}{gray!10}[respect-blocks]
  \Body
    \toprule
        Problem & Communication Model & Authors & Total Communication \\ 
    \midrule
       \Block{2-1}{$\ell_2$ Regression} & \Block{2-1}{Coordinator} & \cite{vww20} &  $O(sd^2L)$\\
         & & Ours (\cref{thm:main-lowAcc-lin-reg})
         &  $\tilde{O}(sdL + d^2L)$\\
        \Block{2-1}{$\ell_1$ Regression} & \Block{2-1}{Coordinator} & \cite{vww20} &  $O(sd^2L)$\\
         & & Ours (\cref{thm:ell1_subspace_embedding}) &  $\tilde{O}(sdL + d^2L)$ \\
         \Block{2-1}{Low Rank Approximation} & \Block{2-1}{Coordinator} & \cite{kannan2014principal, boutsidis2016optimal} & $O(skdL + sk^2L)$ \\ 
         & &  Ours (\cref{thm:main-low-rank-matrix-approx}) & $\tildeO(kdL + skL)$  \\ 
         \Block{2-1}{Linear Programming}& \Block{2-1}{Coordinator} & \cite{vww20} & $\widetilde{O}(sd^3L + d^4L)$ \\ 
         & & Ours (\cref{thm:ipm}) & $\widetilde{O}(sd^{1.5}L + d^2L)$\\
         \Block{2-1}{Decomposable Function Minimization}& \Block{2-1}{Blackboard} & \cite{dong2022decomposable} & $\widetilde{O}(\max_{i\in [s]} d_i \cdot\sum_{i=1}^s d_i L)$  \\ 
         & & Ours & $\widetilde{O}(\sum_{i=1}^s d_i^2 L)$ \\ 
\bottomrule
    \end{NiceTabular}
    }
    \caption{Our communication complexity results for least squares regression, $\ell_p$ regression for $1\leq p < 2$, low-rank approximation for constant $\eps$, linear programming with polynomial condition number and polynomial $R/r$ (all in the coordinator model), and decomposable function minimization (in the blackboard model). Please see \cref{rem:blackboardL1L2}, \cref{rem:blackboardLHighAccL2}, and \cref{rem:bbLP} for remarks on prior work in the \emph{blackboard model} on $\ell_1$ and $\ell_2$ regression and linear programming. 
    } 
    \label{tab:coordinator_model_regression}
\end{table}

\paragraph{General Setup.} In all problems, we consider a matrix that is divided among $s$ servers according to the $\textit{row-partition }$ model.
This is in contrast to the arbitrary partition model, in which each server holds a matrix $\ma^{(i)}$, with $\ma$ being the sum of the servers' matrices, i.e.,  $\ma  = \sum_{i\in [s]} \ma^{(i)}$. In our model, the $i^\mathrm{th}$ machine stores a matrix $\mai{i}\in \R^{n_i \times d}$, and our problem matrix $\ma\in \R^{n\times d}$, with $n = \sum_{i=1}^s n_i$,  is formed by vertically stacking all the $\mai{i}$ matrices, i.e., $\ma = [\ma^{(i)}]$. 

\looseness=-1For least squares regression and linear programming, each server additionally holds a vector $\vbi{i}\in\R^{n_i}$ whose vertical concatenation we denote by $\R^n \ni \vb = [\vb^{(i)}]$, with $n = \sum_{i=1}^s n_i$. When considering linear programming and finite-sum minimization, the vector $\vc$ (where $\vc$ is the vector that appears in the objective obtained by reducing the original finite-sum minimization using an epigraph trick)
is also shared between the machines (or can be shared with $O(sd)$ communication). We explicitly describe the setup for each problem in its corresponding section.  

\looseness=-1We assume that the entries of  $\mai{i}$ and  $\vbi{i}$ can be represented with $L$ bits.  We often model this by assuming that all entries are integers in $\{-2^{L}+1,  \ldots, 2^L\}.$ Sometimes it will be more convenient to work with normalized vectors and matrices, in which case we allow entries to be of the form $c\, 2^{-L}$ with $c\in\{-2^{L}+1,  \ldots, 2^L\}$. We say that such numbers are expressed to $L$ bits of precision.

\subsubsection{Least Squares Regression, $\ell_p$ Regression, and Low-Rank Approximation}\label{sec:contributions_regression_low_accu}

\looseness=-1In many large-scale machine learning applications,
one is faced with a large, potentially noisy regression problem for which a constant factor approximation  is acceptable. Specifically, we are interested in computing an approximate solution $\widehat{\vx}$ satisfying, for a given constant $\eps$, the bound 
\[\norm{\ma \widehat{\vx} - \vb}{2} \leq (1+\eps)\min_{\vx}\norm{\ma \vx - \vb}{2}.\numberthis\label[ineq]{ineq:l2regression_desired_accuracy}
\] We formalize our  setup below.

\begin{problem}[Setup in the Coordinator Model]
\label[prob]{def:lin-reg-setting}
Suppose there is a coordinator and $s$ machines that communicate with each other as per the coordinator model of communication (\cref{def:coordinatormodel}) with shared randomness. Suppose each machine $i\in[s]$ holds a matrix $\mai{i}\in \R^{n_i\times d}$ and a vector $\vb^{(i)}\in \R^{n_i}$. Denote $\ma = [\mai{i}]\in \R^{n\times d}$  and $\vb = [\vb^{(i)}]\in \R^{n}$, both represented with $L$ bits in fixed-point arithmetic. Moreover, suppose the condition number of $\ma$ is bounded by $\kappa$\footnote{Note that not all of our bounds depend on $\kappa$.}. 
\end{problem}

For least squares regression in this model, \cite{vww20} gave  upper and lower bounds of $\widetilde{O}(sd^2L)$ and  $\Omega(sd + d^2L)$, respectively. Their upper bound comes from sending $(\ma^{(i)})^\top \ma^{(i)}$'s and $(\ma^{(i)})^\top \vb^{(i)}$'s to the coordinator which  then computes the exact solution by the normal equations. On the other hand, they show that consistent\footnote{The system $(\ma, \vb)$ is consistent if for some $\vx$, we have $\ma \vx = \b.$} linear systems can be solved exactly using only $\tildeO(sd + d^2L)$ communication. Furthermore, for consistent systems, the optimal regression error is zero, and so a regression algorithm must output the precise solution.  Least squares regression is, therefore, certainly as hard as solving consistent linear systems. This motivates  the following question:
{\it Is solving least squares regression to constant accuracy harder than solving a consistent linear system?}

\looseness=-1
Our key (and surprising) takeaway message for this setting is that for constant $L$, \emph{regression is no harder than solving linear systems}.
Specifically, we give a protocol which, for any constant $\epsilon > 0$, achieves $\widetilde{O}(sdL + d^2 L)$ bits of communication for least squares regression, thus improving upon \cite{vww20}'s $\widetilde{O}(sd^2L)$ upper bound and matching its lower bound of $\Omega(sd + d^2L)$ for constant $L$. 
Our upper bound also gives the first separation for least squares regression between the row-partition model and the arbitrary partition model, for which \cite{li2023ell_p} showed an $\Omega(sd^2)$ lower bound. 

\begin{restatable}[$\ell_2$ Regression in the Coordinator Model]{theorem}{thmLowAccuracyLinearRegression}\label{thm:main-lowAcc-lin-reg}
Given $\epsilon>0$ and a least squares regression problem in the setup of \cref{def:lin-reg-setting} with input matrix $\ma=[\ma^{(i)}]\in \R^{n\times d}$ and vector $\vb=[\vb^{(i)}]\in\R^n$, there is a randomized protocol 
that allows the coordinator to solve the least squares regression problem with constant probability
and relative error $(1\pm \eps)$ using
\[\tildeO\left(sd L + d^2\eps^{-1} L\right) \text{ bits of communication}.\]  Additionally, if $\kappa$ is a known upper bound on the condition number of $\ma$ then there is a protocol using $\tildeO(sd \log\kappa +  d^2\eps^{-1} L)$ communication. 
\end{restatable}

\looseness=-1If $L$ is not constant, there still remains a gap  between the above bound and \cite{vww20}'s lower bound of $\widetilde{\Omega}(sd + d^2L)$. By proving an improved $\widetilde{\Omega}(sdL)$ lower bound for $\ell_2$ regression under a mild restriction on the number of rounds of the protocol, we close this gap (cf. \Cref{sec:introduction_lower_bounds} and \Cref{sec:lower_bounds_full_section}).

\looseness=-1Our upper bound from \cref{thm:main-lowAcc-lin-reg} extends to $\ell_p$ regression for $1\leq p < 2$, as captured by \cref{thm:ell1_subspace_embedding}. Notably, our  protocols for regression have a small $\widetilde{O}(1)$ number of rounds of communication, with no dependence on the condition number of $\ma$.

\begin{restatable}[$\ell_p$ Regression for $1\leq p< 2$ in the Coordinator Model]{theorem}{thmEllpRegressionLessThanTwo}\label{thm:ell1_subspace_embedding}
For the setup described in \cref{def:lin-reg-setting}, there exists a randomized protocol that, with a probability of at least $1-\delta$, allows the coordinator to produce an $\eps$-distortion $\ell_p$ subspace embedding for the column span of $\ma$ using only \[\tildeO\left((sdL + d^2{\eps^{-4}} L)\log({\delta^{-1}})\right) \text{ bits of communication}.\] As a result, the coordinator can solve $\ell_p$ regression (for $1\leq p< 2$) with the same communication. 
\end{restatable}

While the focus of our work for regression has been on the coordinator model (\Cref{thm:main-lowAcc-lin-reg} and \Cref{thm:ell1_subspace_embedding}), we note that \cite{vww20} already provide optimal communication cost algorithms for constant-accuracy regression in the \emph{blackboard model}, as remarked below. 

\begin{remark}\label{rem:blackboardL1L2}
For constant-accuracy $\ell_1$ and $\ell_2$ regression in the blackboard model, \cite{vww20} provides optimal algorithms with communication cost $\widetilde{O}(s + d^2 L)$. 
\end{remark}

\paragraph{Low Rank Approximation.} 
As an application of our aforementioned least squares regression techniques, we obtain improved bounds for low-rank approximation in the distributed setting, a problem several prior works \cite{kannan2014principal, boutsidis2016optimal, bhojanapalli2014tighter,FSS20} have considered.
Notably, \cite{boutsidis2016optimal} studied the variant of the problem wherein the rows\footnote{Their matrices are transposed relative to ours, so their columns are partitioned among servers.} of $\ma$ are partitioned among $s$ servers, and all servers must learn a projection $\Pi$ that yields an approximately optimal Frobenius-norm error: 
\[
\norm{\ma\Pi - \ma}{\fro} \leq (1+\eps) \norm{\ma_k - \ma}{\fro},\numberthis\label[ineq]{ineq:desired-frob-norm-ineq-low-rank-approx}
\]  
where $\ma_k$ is the best rank-$k$ approximation of $\ma$.
In this setting, \cite{boutsidis2016optimal} provide an upper bound of $O(skdL)$ for constant $\eps$, along with a nearly matching lower bound of $\Omega(skd).$  However, their lower bound crucially requires  $\textit{all}$ servers to learn the projection. A natural question we answer is if relaxing this constraint could yield a better communication complexity. In other words: 
\emph{Is it possible to do better when only the coordinator needs to learn the projection?} 

\begin{restatable}[Low-Rank Approximation in the Coordinator Model]{theorem}{thmLowRankMatrixApproximation}\label{thm:main-low-rank-matrix-approx}
For the setup described in \cref{def:lin-reg-setting}, suppose that the $s$ servers have shared randomness.
Then there is a randomized protocol using \[\widetilde{O}\left(kL\cdot(d\eps^{-2} + s\eps^{-1})\right) \text{ bits of communication},\] that with constant probability 
lets the coordinator produce a rank-$k$ orthogonal projection $\Pi \in \R^{d\times k}$ (where $k\leq d$) satisfying \cref{ineq:desired-frob-norm-ineq-low-rank-approx}. 
\end{restatable}

\subsubsection{High-Accuracy Least Squares Regression}\label{sec:high-accu-lin-reg-contributions}
\looseness=-1Complementing our constant-factor regression results in the previous paragraphs, we study \textit{regression solved to machine precision}. We show that when the matrix $\ma\in\R^{n\times d}$ has a small condition number (i.e., $\poly(d)$), we obtain an $\widetilde{O}(sdL+d^2L\log({\eps}^{-1}))$ communication complexity of solving  least squares regression  to high accuracy. Specifically, we obtain the following result.

\begin{restatable}[High-Accuracy $\ell_2$ Regression in the Coordinator Model]{theorem}{thmHighAccuracyLinearRegression}\label{thm:main_lin_reg}
Given $\epsilon>0$ and the least squares regression setting of \cref{def:lin-reg-setting} 
with input matrix $\ma=[\ma^{(i)}]\in \R^{n\times d}$ and vector $\vb=[\vb^{(i)}]\in\R^n$, there is a randomized algorithm that, with high probability, outputs a vector $\vxhat$ such that
\[
\norm{\ma \vxhat - \vb}{2} \leq \epsilon \cdot \norm{\ma (\ma^\top \ma)^{\dagger} \ma^\top \vb}{2} + \min_{\vx \in \R^d} \norm{\ma \vx - \vb}{2}. \numberthis\label[ineq]{eq:linearRegressionValueGuarantee}
\] Let $\kappa$ be the condition number of $\ma$. Then the algorithm uses \[\Otil(sd (L+\log\kappa) \log(\epsilon^{-1}) + d^2L ) \text{ bits of communication.}\]   
 Moreover, the vector $\vxhat$ is available on all the machines at the end of the algorithm.
\end{restatable}

\looseness=-1This result improves upon the $\widetilde{O}(sd^2L)$ bound of \cite{vww20} (which also gives an associated lower bound of $\widetilde{\Omega}(sd + d^2L)$). The error guarantee of \cref{thm:main_lin_reg} is different than that of \cref{thm:main-lowAcc-lin-reg} in two ways. First, the error is additive in \cref{thm:main_lin_reg} instead of multiplicative. The main reason is that the solution produced by the algorithm of \cref{thm:main_lin_reg} is available to all the machines instead of only being available only to the coordinator. This is needed when we use this result for each iteration of linear programming (\cref{thm:ipm}). We note that the solution produced by the algorithm of \cref{thm:main-lowAcc-lin-reg} can be shared among all the machines, but we would need to use the rational number representation to share it, and the communication cost would increase significantly in this case.
The second difference is that the dependence of the running time on the error parameter $\epsilon$ is logarithmic in \cref{thm:main_lin_reg}. This allows us to achieve high-accuracy solutions, which are again needed for linear programming results to deal with adaptive adversary issues.

Our improvement is achieved by a  novel rounding procedure for Richardson's iteration with preconditioning and has consequences outside the distributed setting as well. In particular, it implies an improvement for the bit complexity of solving a least squares regression problem (with an input that has constant bit complexity) from 
$
\Otil((d^\omega + (\nnz(\ma) + d^2)\cdot \log^2(\epsilon^{-1}))\cdot \log \kappa)
$ \cite{ghadiri2023bit} to
$\Otil((d^\omega + (\nnz(\ma) + d^2)\cdot \log(\epsilon^{-1}))\cdot \log \kappa)$,
where $\nnz(\ma)$ is the number of nonzero entries of $\ma$. 

\begin{remark}\label{rem:blackboardLHighAccL2} While our result in \Cref{thm:main_lin_reg} operates only in the coordinator model,  \cite{vww20} studies this problem in the blackboard model as well. In particular, for  $\ell_2$ regression in the blackboard model with general accuracy parameter $\varepsilon$, \cite{vww20}  provides an algorithm with communication cost $\widetilde{O}(s + d^2 L\varepsilon^{-1})$, with an associated lower bound of $\widetilde{\Omega}(s+d\varepsilon^{-1/2} + d^2 L)$ for $s\geq \Omega(\varepsilon^{-1/2})$.
\end{remark}

\subsubsection{High-Accuracy Linear Programming}
A core technical component in achieving the results of \cref{sec:high-accu-lin-reg-contributions} is the communication-efficient computation of a spectral approximation of a matrix via its intimate connection to its approximate leverage scores. We utilize this idea to develop communication-efficient high-accuracy linear programming too, as we describe next. 

\looseness=-1The work of \cite{vww20} studied this problem and gave an upper bound of $\widetilde{O}(sd^3L + d^4L)$ by implementing  Clarkson's algorithm~\cite{clarkson1995vegas} in the coordinator model. To obtain this bound, \cite{vww20} first note that following the analysis of the original algorithm in \cite{clarkson1995vegas}, the total number of rounds of communication is $O(d\log d)$. In each round, the coordinator sends to all the $s$ servers a vector $\vx_R$, which is an optimal solution to the linear program  $\ma\vx\leq \vb$. By polyhedral theory, there exists a non-singular  
subsystem $\mathbf{B}\vx\leq \mathbf{c}$, such that $\vx_R$ is the unique solution of $\mathbf{B}\vx=\mathbf{c}$. By Cramer's rule, each of $d$ entries of $\vx$ is a ratio of integers between $-d!2^{dL}$ and $d!2^{dL}$ and can therefore be represented in $\widetilde{O}(dL)$ bits. Multiplying all these quantities yields the claimed communication complexity.

\looseness=-1We take a different approach and improve upon \cite{vww20}'s above bound of $\widetilde{O}(sd^3L + d^4L)$ to $\widetilde{O}(sd^{1.5}L + d^2L)$. Our improvement is achieved by essentially adapting to the distributed setting recent advances in interior point methods for solving linear programs~\cite{ls14,DBLP:conf/focs/LeeS15,van2020solving}, with the associated toolkit of a weighted central path approach, efficient inverse maintenance, and data structures for efficient matrix-vector operations. Our rate holds
 for linear programs that have a  small outer radius and a well-conditioned constraint matrix $\ma$, as we formalize next.  
\begin{restatable}[Linear Programming in the Coordinator Model]{theorem}{ipmthm}
\label{thm:ipm}
Given $\epsilon>0$, input matrix $\ma=[\ma^{(i)}]\in \R^{n\times d}$, and vectors $\vc=[\vc^{(i)}]\in\R^n$ and $\vb\in\R^d$ in the setup of \cref{def:lin-reg-setting}, there is a randomized algorithm that, with high probability, outputs a vector $\vxhat\in\R^n$ such that
\[
\norm{\ma^\top \vxhat - \vb}{2}\leq \epsilon \cdot (\norm{\ma}{\fro}\cdot R +\norm{\vb}{2}) ~~ \text{and} ~~ \vc^\top \vxhat \leq \min_{\vx:\ma^\top \vx = \vb,\vx\geq 0} \vc^\top \vx + \epsilon \cdot \norm{\vc}{2} \cdot R,\numberthis\label[ineq]{eq:lpSolutionConditions}
\] 
where $R$ is the linear program's outer radius, i.e., $\norm{\vx}{2} \leq R$ for all feasible $\vx$.
The algorithm uses 
\[\Otil((sd^{1.5} (L+\log (\kappa Rr^{-1} \epsilon^{-1})) +d^2 L \log(\epsilon^{-1}))\cdot \log(\epsilon^{-1})) \text{ bits of communication},\] where $\kappa$ is the condition number of $\ma$, and $r$ is the linear program's inner radius, i.e., there exists a feasible $\vx$ with $\vx_i\geq r$ for all $i\in[n]$.  
 Moreover, the vector $\vxhat$ is available on all the machines at the end of the algorithm.
\end{restatable}

As a special case of our \cref{thm:ipm},  when our linear program has parameters $\kappa,R$, and $\epsilon^{-1}$ of the scale $\poly(d)$, we obtain a communication complexity of $\Otil(sd^{1.5} L + d^2 L)$, which is also an improvement over \cite{vww20}'s previous bound. 

\begin{remark}\label{rem:bbLP} While we study linear programming only in the coordinator model, \cite{vww20} studies this in the blackboard model as well. Specifically, in constant dimensions, \cite{vww20} provides a randomized communication complexity of $\widetilde{\Omega}(s + L)$ for linear programming in the blackboard model. 
\end{remark}

\subsubsection{Finite-Sum Minimization with Varying Supports}
\looseness=-1Another problem class naturally amenable to study in the distributed setting is that of finite-sum minimization. We consider, in the blackboard model, the problem $\min_{\vx} \sum_{i=1}^s f_i(\vx)$ where each $f_i:\R^{d}\mapsto\R$ is $\mu$-Lipschitz, convex,  nonsmooth, and supported on (potentially overlapping) $d_i$ coordinates. We call this problem ``decomposable nonsmooth convex optimization''. The assumption of varying supports appears prominently in decomposable submodular function minimization \cite{axiotis2021decomposable,rafiey2022sparsification} and was recently studied by \cite{dong2022decomposable}. This problem, without this assumption, has seen extensive progress in variants of stochastic gradient descent (cf. \cref{sec:decomposable_techniques}). We formalize below the problem setup in the blackboard model. 

\begin{problem}[Decomposable Nonsmooth Convex Optimization Setup]
\label[prob]{def:fin-sum-setting}
Suppose there is a blackboard/coordinator and $s$ machines that communicate with each other as per the blackboard model of communication (\cref{def:blackboardmodel}). Suppose each machine $i\in[s]$ holds an oracle $\oi$ that returns a subgradient (represented with $L$ bits in fixed-point arithmetic) of the function $f_i:\R^{d}\mapsto\R$.
\end{problem}

\looseness=-1Directly adapting the algorithm of \cite{dong2022decomposable} to the above model yields a communication cost of $\widetilde{O}(\max_{j\in[s]}d_j L \sum_{i=1}^s d_i)$.
In this work, we improve this cost to $\widetilde{O}(\sum_{i=1}^s d_i^2 L)$, as formalized next. 

\begin{restatable}[Distributed Decomposable Nonsmooth Convex Optimization]{theorem}{thmmainFinSumMain} \label{thm:mainFinSumMin} 
Given $\varepsilon>0$ and the setup of \cref{def:fin-sum-setting}, consider the problem $\min_{\theta} \sum_{i=1}^s f_i(\theta)$, where each $f_i:\R^d\mapsto\R$ is convex, $\mu$-Lipschitz, and dependent on $d_i$ coordinates of $\theta$. Define $\theta^\star := \arg\min_{\theta\in\R^d} \sum_{i=1}^s f_i(\theta)$. Suppose further that we know an initial $\theta^{(0)}\in \R^d$ such that
$\|\theta^\star - \theta^{(0)}\|_2\leq D$. 
Then, 
there is an algorithm that outputs a vector $\vtheta\in \R^d$ such that 
\[
\sum_{i=1}^s f_i(\vtheta) \leq \sum_{i=1}^s f_i(\theta^\star)  + \epsilon \cdot \mu D.
\]
Our algorithm uses \[O\left(\sum_{i=1}^s d_i^2 \log (sd\epsilon^{-1})\cdot L\right) \text{ bits of communication, }\] where $L=O(\log d)$ is the word length. At the end of our algorithm, all servers hold this solution. 
 \end{restatable}

\looseness=-1Our technical novelty --- modifying the analysis and slightly modifying the algorithm of \cite{dong2022decomposable} --- yields an improvement in  not just the distributed setting but also in the (non-distributed) setting \cite{dong2022decomposable}
 studied this problem in. Specifically, as a corollary (\cref{cor:finSumMinMain}), we improve the total oracle cost of decomposable nonsmooth convex optimization from $\widetilde{O}(\mathcal{O}_{\max}\cdot\sum_{i=1}^s d_i)$ to $\widetilde{O}(\sum_{i=1}^s \mathcal{O}_i \cdot d_i)$, where $\mathcal{O}_i$ is the cost of invoking the $i^\mathrm{th}$ separation oracle, and $\mathcal{O}_{\max}$ is the maximum of all $\mathcal{O}_i$. 

\begin{restatable}[Solving \cref{def:fin-sum-setting}.]{theorem}{corFinSumMinDist}\label{cor:finSumMinMain}
\looseness=-1Consider $\min_{\theta\in\R^d}\sum_{i=1}^s f_i(\theta)$ with each
 $f_i:\R^d \mapsto \R$ convex, $\mu$-Lipschitz, possibly non-smooth functions, depending on $d_i$ coordinates of $\vtheta$, and accessible via a (sub-)gradient oracle.
Define $\theta^\star := \arg\min_{\theta\in\R^d} \sum_{i=1}^s f_i(\theta)$. Suppose we are given a vector $\theta^{(0)}\in \R^d$ such that
$\|\theta^\star - \theta^{(0)}\|_2\leq D$. Then, given a weight vector $\vw\in \R^{s}_{\geq 1}$ with which we define $m\defeq \sum_{i\in[s]}  w_i d_i$, there is an algorithm that, in time $\operatorname{poly}(m \log(\eps^{-1}))$, outputs a vector $\vtheta\in \R^d$ such that 
\[
\sum_{i=1}^s f_i(\vtheta) \leq \sum_{i=1}^s f_i(\theta^\star)  + \epsilon \cdot \mu D.
\]
Moreover, let $n_i$ be the number of subgradient oracle calls to $f_i$. Then, the algorithm's total oracle cost is 
\[
\sum_{i=1}^s w_i \cdot n_i = O(m \log (m / \epsilon )).
\]
\end{restatable}

\looseness=-1As alluded to earlier, an important special case of decomposable nonsmooth convex optimization is decomposable submodular function minimization, which in turn has witnessed a long history of research~\cite{jegelka2013reflection, nishihara2014convergence, ene2017decomposable, karri2019fast, axiotis2021decomposable}. 
Therefore, outside of distributed optimization, an immediate application of \cref{cor:finSumMinMain} is an improved cost of decomposable submodular function minimization, as we describe in \cref{cor:SFM_finSumMin}.

\begin{restatable}[Faster Submodular Function Minimization]{corollary}{corSFMusingFinSumMinImproved}\label{cor:SFM_finSumMin}
    Let $V = \{1, 2, \dots, m\}$, and
$F: 2^V \mapsto [-1,1]$ be given by $F(S) = \sum_{i = 1}^n F_i(S \cap V_i)$, where each $F_i: 2^{V_i} \mapsto \mathbb{R}$ is a submodular function on $V_i \subseteq V$. 
We can find an $\epsilon$-additive approximate minimizer of $F$ in \[O\left(\sum_{i=1}^n |V_i|^2\log(n\epsilon^{-1})\right) \text{ evaluation oracle calls.}\] 
\end{restatable}

\looseness=-1To contextualize our above result for decomposable SFM, \cite{dong2022decomposable} improved upon the cost of $O\left( V_{{\max}}^2 \sum_{i=1}^n |V_i|^4 \log(n\eps^{-1}) \right)$ by \cite{axiotis2021decomposable} to get a cost of $O\left( V_{{\max}} \sum_{i=1}^n |V_i| \log(n\eps^{-1}) \right)$. In cases where the $|V_i|$ are highly non-uniform, our result of $O\left(\sum_{i=1}^n |V_i|^2\log(n\epsilon^{-1})\right)$ evaluation oracle calls is therefore an improvement upon what is, to the best of our knowledge, the previous fastest result.

\subsubsection{Lower Bounds}

Finally, we complement our upper bounds results from the previous sections with lower bounds. \cite{vww20} asked the following question: from the perspective of communication complexity, is solving a linear program harder than (exactly) solving a linear system? 
Towards answering this question, they showed that, in constant dimensions, checking feasibility of a linear program requires $\widetilde{\Omega}(sL)$ communication in the coordinator model, while feasibility for linear systems requires only $\tildeO(s + L)$ communication, thereby demonstrating an exponential separation between the two problems. 
However their lower bound for linear programs was based on a hard instance with $2^{\Omega(L)}$ constraints. 
So for linear program feasibility problems with $n$ constraints they leave open the possibility of a protocol with communication cost $O(s\log n) + o(sL)$. 
This is an important limitation of their lower bound, since they show for example, that a modified Clarkson's Algorithm \cite{vww20, clarkson1995vegas}, indeed gives a protocol with $\log n$ dependence. This, therefore, motivates the following question: 
\textit{is there an exponential separation between checking feasibility of linear programs and solving linear systems when there are only $\poly(s+d)$ constraints?} 

We answer this question in the affirmative, showing that such a separation does in fact hold, even for linear feasibility problems with $O(s+d)$ constraints.

\begin{restatable}{theorem}{thmLinFeasLowerBound}
\label{thm:linear_feasibility_lower_bound}
Any protocol solving Linear Feasibility (\cref{prob:linear_feasibility}) in the coordinator model requires at least $\Omega(sdL)$ communication for protocols that exchange at most $c L/\log L$ rounds of messages with each server and with $\log d\leq 5L$.  This bound holds even when the number of constraints is promised to be at most $O(s+d).$ For constant $d$ the $\Omega(sL)$ lower bound holds with no assumption on the number of rounds.
\end{restatable}

In addition to linear programming, one could also ask to get tight lower bounds for relative error least squares regression as discussed above.  \cite{vww20} gave a lower bound of $\Omega(sd + d^2 L)$ for constant $\eps$, however our algorithm requires $\tildeO(sdL + d^2 L)$ bits.  We close this gap by showing that the $sdL$ term is unavoidable.  Perhaps surprisingly, our regression lower bound follows from the same techniques that we we use to derive our linear programming lower bound.

\subsection{Technical Overview}\label{sec:technicalOverview}
Before providing the details of our algorithms and analyses for each of the aforementioned results, we give high-level overviews of the techniques we use for each of them. 

\subsubsection{Least Squares Regression and Subspace Embeddings} 
\looseness=-1We give two protocols for the regression problem instance of $\min_{\vx} \|\ma\vx-\vb\|_2$. 
 The first is based on sketching what we refer to as the {\it block leverage scores}, which for us is simply the sum of the leverage scores of the rows in that block (leverage scores computed with respect to $\ma$). Our second protocol is based on \emph{non-adaptive adaptive sketches} \cite{mahabadi2020non} from the in the data streaming literature.
 \looseness=-1Both of our protocols for regression operate by constructing a subspace embedding\footnote{Recall that $\ms$ is an $\eps$-distortion $\ell_p$ subspace embedding for $\ma$ if $\norm{\ms\ma\vx}{p} = (1\pm \eps)\norm{\ma\vx}{p}$ for all $\vx$.} matrix $\ms$ for the span of the columns of $\ma$ and $\vb.$  This is a stronger guarantee than solving the regression problem, as the coordinator may compute $\ms \ma \vx - \ms \vb = \ms(\ma \vx - \vb )$ and output the solution to the sketched regression problem \cite{w14}. The subspace embedding construction proves useful in contexts other than regression too.  Indeed, we require the subspace embedding construction to get improved communication for low-rank approximation. We now describe our two approaches below.

 \paragraph{Block Leverage Scores.}\looseness=-1While block leverage scores have previously been considered in various forms~\cite{oswal2019block, kyng2016sparsified, perelli2021regularization, xu2016sub, manoj2023changeofmeasure} as far as we are aware they have (naturally) been used only in the context of sampling entire blocks at a time. In our setting, we are ultimately interested only in sampling rows, but find that approximating the block leverage score of each server is a useful subroutine. Specifically we show that for small $k$, sampling a $k\times d$ row-sketch of each block is almost sufficient to estimate all the block leverage scores. The catch is that we fail to accurately estimate block leverage scores that are larger than $k.$ Intuitively, this is because such blocks could have more than $k$ ``important" rows.  So our approach is to attempt to estimate the leverage score of all blocks via sketching using a small value of $k$.  We might find that a small number of blocks have leverage scores that are too big for the estimates of their leverage scores to be accurate. To fix this, we focus on those blocks and sample a larger row-sketch from them in order to get a better estimate of their block leverage scores.  Taking a larger sketch requires more communication per block. Crucially, however, the number of servers with leverage score greater than $k$ is at most $d/k.$  Thus we may proceed in a series of rounds, where in round $r$ we focus on servers with leverage score at least $2^r$.  There are at most $d/2^r$ such servers, and for each server we take a sketch of total size roughly $2^r d,$ so each round after (of which there are only $O(\log d)$) uses roughly $d^2$ communication.  We note that the first round requires a roughly $1\times d$ sized sketch from all servers, which yields an $sd$ dependence.
Once we have estimates of the block leverage scores, we observe that sampling sketched rows from the blocks proportional to the block leverage scores suffices to obtain a subspace embedding for $\ma.$

\paragraph{Non-adaptive Adaptive Sketching. }\looseness=-1 When $p=2$, our protocol  runs the recursive leverage score sampling procedure of \cite{cohen2015uniform} adapted to the distributed setting.  
One potential approach is to run this algorithm by sketching the inverse spectral approximations and broadcasting them to the servers.  Unfortunately, when $\ma$ is nearly singular, these sketches can have a high bit complexity. To avoid this, we instead use a version of an \textit{$\ell_2$ sampling sketch} which can be applied on the servers' sides and sent to the coordinator, which allows the coordinator to sample from the appropriate (relative) leverage score distribution.  An issue arises if some relative scores are  much larger than one, as we need to truncate them to roughly one before using them as sampling probabilities (up to scaling).  To fix this, we first give a subroutine to identify this subset of outlying rows.

\paragraph{$\ell_p$ Regression beyond $p=2$.}

Our ``non-adaptive adaptive" protocol above extends to give optimal guarantees for $\ell_1$ regression and $\ell_p$ regression for $1\leq p \leq 2$ essentially by using the more general recursive Lewis weight sampling protocol of \cite{cohen2015lp}.

For $2 < p < 4$, the recursive Lewis weight sampling algorithm of \cite{cohen2015lp} can also be run exactly to construct an $\ell_p$ subspace embedding, simply by broadcasting the approximate Lewis quadratic form to all servers on each round. Since the quadratic form is a $d\times d$ matrix, this broadcasting incurs an $\tildeO(sd^2L)$ cost per round and hence an $\tildeO(sd^2 L)$ cost for computing approximate Lewis weights for all rows.  The coordinator then must sample $d^{p/2}$ rows, resulting in a cost of $\tildeO(sd^2L + d^{\max(p/2,1) + 1}).$

If one is interested in sampling a coreset of rows to obtain an $\ell_p$ subspace embedding, then the $d^{p/2 + 1}$ term is unavoidable as we need to sample at least $d^{p/2}$ rows  \cite{li2021tight}. However for $1\leq p \leq 2$ our approach for $\ell_2$ regression shows that the $sd^2L$ term can be improved to $sdL.$  Whether the $sd^2L$ term can be improved for all $p$ is an interesting question that we leave to future work.

\subsubsection{High-Accuracy Least Squares Regression}\label{sec:technicalOverviewHighAccRegression} 
\looseness=-1 While constant-factor approximations often suffice, in certain settings it is important to ask for a solution that is optimal to within machine precision, e.g., if such solutions are used in iterative methods for solving a larger optimization problem.  In this setting we consider the problem instance $\min_{\vx} \|\ma\vx-\vb\|_2$, where, given a row-partitioned system $\ma$, $\vb$, the goal of the coordinator is to output an $\vx$ for which $\|\ma\vx-\vb\|_2 \leq \min_{\vx}\|\ma\vx-\vb\|_2 + \epsilon \cdot \|\ma (\ma^\top \ma)^{\dagger} \ma^\top \vb\|_{2}$. 
A direct application of gradient descent requires about $\kappa$ iterations, where $\kappa$ is the condition number of matrix $\ma$. 

\looseness=-1Our protocol for solving this problem in the distributed setting to a \emph{high accuracy} is Richardson's iteration with preconditioning, coupled with careful rounding. We precondition using a constant factor spectral approximation of the matrix to reduce the number of iterations to only $\log(\eps^{-1})$.  This version of Richardson's iteration is equivalent to performing Newton's method with an approximate Hessian since the preconditioner spectrally approximates the  Hessian inverse $(\ma^\top \ma)^{-1}$. Computing a spectral approximation to $\ma$ is equivalent to computing a subspace embedding, so to compute the preconditioner we  employ our regression protocol from above. 
We note that   the refinement sampling procedure described in 
\cref{subsec:refinement-sampling} (and displayed in the algorithm format in
\cref{alg:levscoresRefinementSampling}) is very similar to \cref{alg:recursiveSamplingCohenPeng}, since both are based on the same algorithm from \cite{cohen2015uniform}. However, we believe there are sufficient differences in the specifics to merit writing out the latter in full. 
Alternatively, it is possible to employ a somewhat simpler protocol (which also has the advantage of computing approximations to \textit{all} leverage scores) since in the high-precision setting we allow for a condition-number dependence.

 The key novelty in the implementation of our Richardson-style iteration (\cref{alg:richardson}) is to communicate, in each step,  only a partial number of bits of the residual vector. The idea here is that as the solution converges, the bits with high place values do not change much between consecutive iterations and therefore need not be sent every time. Using a similar idea, we show that the Richardson iteration is robust to a small amount of noise, which helps us avoid updating the lowest order bits.
 Overall, via a careful perturbation analysis, we show that communicating the updates on only the $O(L \log \kappa)$ middle bits of each entry suffices to guarantee the convergence of Richardson's iteration. 
 
\subsubsection{High-Accuracy Linear Programming}
Similar to \cref{sec:technicalOverviewHighAccRegression}, we ask the question of solving \emph{linear programs} to a high accuracy. These require different techniques than ones in fast \textit{first-order} algorithms for linear programs with runtimes depending polynomially on $1/\varepsilon$~\cite{applegate2021practical, applegate2023faster, xiong2023computational}. Specifically, interior-point methods and cutting-plane methods are the standard approaches in the high-accuracy regime. Recent advances in fast high-accuracy algorithms for linear programs~\cite{ls14, DBLP:conf/focs/LeeS15, van2020solving} were spurred by developments in the novel use of the Lewis weight barrier, techniques for efficient maintenance of the approximate inverse of a slowly-changing matrix, and efficient data structures for various linear algebraic primitives. Our approach for a communication-efficient high-accuracy linear program solver builds upon these developments, effectively adapting them into the coordinator model.  

\looseness=-1We first describe the standard framework of interior-point methods. In this paradigm, one reduces solving the problem of $\min_{\vu\in \mathcal{S}} {\vc}^\top{\vu}$ to that of solving a sequence of slowly-changing unconstrained problems $\min_{\vu} \Psi_{t}(\vu)\defeq \left\{t \cdot {\vc}^\top{\vu} + \psi_{\mathcal{S}}(\vu)\right\}$ parametrized by $t$, 
with  a \emph{self-concordant barrier} $\psi_{\mathcal{S}}$  that enforces feasibility by becoming unbounded as  $\vx\rightarrow\partial\mathcal{S}$.
The algorithm starts at $t = 0$, for which an approximate minimizer $\vx_0^\star$ of $\psi_{\mathcal{S}}$ is known, and it alternates between increasing $t$ and updating, via Newton's method, $\vx$ to an approximate minimizer $\vx_t^\star$ of the new $\Psi_{t}$. For a sufficiently large $t$, the minimizer $\vx_t^\star$ also approximately optimizes the original problem $\min_{\vu\in \mathcal{S}} {\vc}^\top{\vu}$ with sub-optimality gap $O(\nu/t)$, where $\nu$ is the self-concordance parameter of the barrier function used. This self-concordance parameter typically also appears in the iteration complexity. 

\looseness=-1 While this is the classical interior-point method as pioneered by~\cite{nesterov1994interior}, there has been a flurry of recent effort focusing on improving different components of this paradigm. The papers we use for our purposes are those by ~\cite{ls14, DBLP:conf/focs/LeeS15,  van2020solving}, which developed variants of the aforementioned central path method essentially by reducing the original LP to certain data structure problems such as inverse maintenance and heavy-hitters.  
Adapting these approaches to the coordinator model, we provide an algorithm for approximately solving LPs with $\Otil((sd^{1.5} (L+\log (\kappa R)) +d^2 L)\cdot \log(\eps^{-1}))$ bits of communication 
(\cref{thm:ipm}). Among the tools we employ for our analysis are those for matrix spectral approximation developed  in \cref{sec:linear_regression} and  our result for the communication complexity of leverage scores (\cref{lem:levScoreOverestimates}), which we use to bound the communication complexity of iteratively computing Lewis weights (\cref{lemma:lewis-weight-communication}) for computing an initial feasible solution. 

\subsubsection{Decomposable Nonsmooth Convex Optimization}\label{sec:decomposable_techniques}
\looseness=-1For decomposable nonsmooth convex optimization in the blackboard model, we improve an algorithm from the literature  and then adapt this improved algorithm to the distributed setting. Specifically, we study $\min_{\vx\in\R^d} \sum_{i = 1}^s f_i(\vx)$, where each $f_i$ is $\mu$-Lipschitz, convex, and dependent on some $d_i$ coordinates of $\vx$ --- note that the different $f_i$ could have overlapping supports --- and the $i^\mathrm{th}$ machine has subgradient-oracle access to the $i^\mathrm{th}$ function.

\looseness=-1Most prior works~\cite{roux2012stochastic, shalev2013stochastic, johnson2013accelerating} and their accelerated variants~\cite{lin2015universal, frostig2015regularizing, zhang2015stochastic, agarwal2015lower, allen2017katyusha} designed in the non-distributed variant of finite-sum minimization assumed $f_i$ to be smooth and strongly convex. Those designed for the distributed setting~\cite{BPCPE11, s14, zhang2015disco} also typically imposed this assumption (some exceptions include~\cite{d12}), but additionally also used as their performance metric only the \emph{number of rounds of communication}, as opposed to the \emph{total number of bits communicated}, which is what we focus on. Variants of gradient descent~\cite{nesterov1983method} that are typically applicable to this problem also require a bounded condition number. There has also been work on non-smooth empirical risk minimization, but usually it requires that the objective be a sum of a smooth loss and a non-smooth regularizer. The  formulation we study is a more general form of empirical risk minimization: In particular, our setting allows all $f_i$ to be non-smooth. 

\looseness=-1The work of \cite{dong2022decomposable} combines ideas from classical cutting-plane and interior-point methods to obtain a nearly-linear (in total effective dimension) number of subgradient oracle queries for solving the problem in the non-distributed setting. This is the algorithm we modify and adapt to the distributed setting; our modification also yields improvements in the non-distributed setting. 

\looseness=-1We first describe the result obtained by simply adapting the algorithm of \cite{dong2022decomposable} to the blackboard model. Following \cite{dong2022decomposable}, we first use a simple epigraph trick to  reduce this problem to one with a linear objective and constrained to be on an intersection of parametrized epigraphs of $f_i$: $\min_{\vx: \vxi\in \mathcal{K}_i \subseteq \R^{d_i + 1} \forall i\in [s], \ma\vx=\vb}  \vc^\top\vx$, where $\mathcal{K}_i$ are all convex sets.  All servers hold identical copies of the problem data at all times. However, each server $i$ has only separation-oracle access to the set $\ki$, which comes from the equivalence to the subgradient-oracle access to $f_i$ by the result of \cite{lee2018efficient}. 

\looseness=-1We maintain crude outer and inner set approximations, $\kini$ and $\kouti$, to each set $\ki$ (such that $\kini \subseteq \ki\subseteq\kouti$) and update $\vx$, our candidate minimizer of the (new) objective $\vc^\top \vx$ using an interior-point method. Ideally, for some choice of barrier function defined over $\kcal\cap \{\ma\vx=\vb\}$, we would update our candidate minimizer to move along the central path through this set. However, since we do not explicitly know $\kcal$, we instead use a barrier function defined over its proxy, $\kout\cap \{\ma\vx=\vb\}$. 
We improve our approximations of $\kin$ and $\kout$ using ideas inspired from classical cutting plane methods \cite{vaidya1989new}. Thus, our algorithm essentially alternates between performing a cutting-plane step (to improve our set approximation of $\kcal$) and performing an interior-point method step (to enable the candidate minimizer $\vx$ to make progress along the central path). 

Each server runs a copy of the  above algorithm. After updating the parameter $t$ and computing $\vxos$ --- the current target for the interior-point method step --- each server tests feasibility of $\vxos$. If there is a potential infeasibility of  the $i^\mathrm{th}$ block, $\vxosi$, then the server queries $\vxosi$ (the $i^\mathrm{th}$ block of the current target point) and sends to the blackboard a separating hyperplane to update $\kouti$ or a bit to indicate otherwise. The other servers then read this information and update either the set $\kouti$ or $\kini$ on their ends. 
It was shown in \cite{dong2022decomposable} that  this algorithm (without the distributed setting) has an  oracle query complexity of $\widetilde{O}(\sum_{i=1}^s d_i)$. In the distributed setting, this would translate to a communication complexity of $\widetilde{O}(\max_{j\in [s]} d_jL \cdot\sum_{i=1}^s d_i)$. 

Our main novelty is to modify the prior analysis (and slightly modify a specific parameter of the algorithm) so as to obtain the more fine-grained oracle cost of $\widetilde{O}(\sum_{i=1}^s w_i d_iL)$, for any arbitrary weight vector $\vw\geq 0$. 
In our distributed algorithm, we set $w_i = d_i$, since the only communication that happens in a  round is when a server sends hyperplane information to the blackboard. Thus, this translates to the communication cost of $\widetilde{O}(\sum_{i=1}^s d_i^2 L)$, an improvement over the bound of $\widetilde{O}(\max_{j\in [s]} d_jL \cdot\sum_{i=1}^s d_i)$ obtained from adapting \cite{dong2022decomposable} to the blackboard setting.  
 
\subsubsection{Lower Bounds}\label{sec:introduction_lower_bounds} 

We are interested in obtaining tight lower bounds for least squares regression and low-rank approximation that capture the dependence on the bit complexity $L$. When proving such lower bounds, it is common to reduce from  communication games such as multi-player set-disjointness
\cite{roughgarden2016communication}.
However, it is not at all clear how one could encode such a combinatorial problem into an instance of regression that would yield a good bit-complexity lower bound. Indeed, most natural reductions from the standard communication problems would result in a single bit entry of $\ma.$ This motivates us to introduce a new communication game (\cref{prob:s_player_inner_prod}) that forces the players to communicate a large number of bits of their inputs.

\begin{restatable}[]{problem}{probSPlayerInnerProd}\label[prob]{prob:s_player_inner_prod}
The coordinator holds an (infinite-precision) unit vector $\vv\in\R^d$ with $d\geq 3$, and the $s$ servers hold unit vectors $\vw_1, \ldots, \vw_s \in \R^d$ respectively.  The coordinator must decide between (a) ${\vv}^\top{\vw_k} = 0$ for all $k\in [s]$ and (b) For some $k$, $\abs{{\vv}^\top{\vw_k}}\geq \frac{\eps}{d}$ and ${\vv}^\top{\vw_i} = 0$ for all $i\neq k$.
\end{restatable}

The two player-version of \cref{prob:s_player_inner_prod} is reminiscent of the promise inner product problem ($\text{PromiseIP}_d$) over $\mathbf{F}_p$, where the goal is to distinguish between ${{\vv}}^\top{\vw} = 0$ and ${{\vv}}^\top{\vw} = 1$ for $\vv,\vw \in \mathbf{F}_p^d.$  This problem was introduced by \cite{sun2012relationship} who gave an $\Omega(d\log p)$ lower bound and further considered by \cite{li2014communication} who developed an $s$-player version.  We note that their $s$-player version is for the ``generalized inner product" and is therefore quite different from the game that we introduce.  Furthermore we are not aware of a version of $\text{PromiseIP}_d$ over $\R$ that is suitable for our purposes, even though real versions of the inner product problem have been studied \cite{andoni2022communication}.
 
\noindent We give the following lower bound for our problem:

\begin{restatable}{theorem}{thmSPlayerInnerProd}
\label{thm:s_player_inner_product_hardness}
A protocol that solves \cref{prob:s_player_inner_prod} with probability at least $0.9$ requires at least $\Omega(sd\log(\eps^{-1}))$ communication for protocols that exchange at most $c\log(\eps^{-1})/\log\log(\eps^{-1})$ rounds of messages with each server.
\end{restatable}

To prove this, we begin by considering $d=3$, and $s=1$ so that the game involves two players, say Alice and Bob, holding vectors $\vv$ and $\vw$ in $\R^3$. We borrow techniques from Fourier analysis on the sphere to prove an $\Omega(\eps^{-1})$ communication lower bound. Our techniques are reminiscent of those in \cite{regev2011quantum}, although we require somewhat less sophisticated machinery.  
One might wonder why we choose to start with $d=3$ rather than $d=2.$ It turns out that when $d=2$ the $\Omega(\eps^{-1})$ lower bound does not hold! Indeed Alice can form the vector $\mathcal{\vv^{\perp}}$ so that the problem reduces to checking if $\mathcal{\vv^{\perp}}$ and $\mathcal{\vw}$ are approximately equal up to sign.  This reduces to checking exact equality after truncating to approximately $\log(\eps^{-1})$ bits.  But this is easy to accomplish with $O(1)$ bits of communication by communicating an appropriate hash. It is not immediately clear whether a similar trick could apply in higher dimensions.  In particular any proof of the lower bound must explain the difference between $d=2$ and $d=3.$  The difference turns out to be that the spherical Radon transform is smoothing in dimensions $3$ and higher, but not in dimension $2.$

Given the $d=3$ case, we boost our result to higher dimensions by a viewing a $d$-dimensional vector as the concatenation of $d/3$ vectors each of $3$-dimensions and then applying the direct-sum technique of \cite{bar2004information}.  This requires us to first prove an information lower bound on a particular input distribution.  This turns out to be easier to accomplish for public-coin protocols, and we then upgrade to general (private-coin) protocols using a ``reverse-Newman" result of \cite{braverman2014public}.  This last step is where our bounded round assumption arises from.  We note that this is a purely technical artifact of our proof and can likely be avoided. Finally, we show how to extend our lower bound from two players to $s$ players. With this result, we are able to deduce new lower bounds for least-squares regression and testing feasibility of linear programs.

\paragraph{Least Squares Regression.}
\cite{vww20} studied the communication complexity of the least squares regression problem and showed a communication lower bound of $\widetilde{\Omega}(sd + d^2L).$  
We show that obtaining a constant factor approximation to a least-squares regression problem requires $\Omega(sdL)$ communication, at least for protocols that use at most roughly $L$ rounds of communication. This bounded round assumption is mild since our algorithms need only $\tildeO(1)$ rounds, which is desirable. 

\looseness=-1The reduction is from \cref{prob:s_player_inner_prod} above.  Our approach is to construct a matrix  from the inputs whose smallest singular value is roughly $2^{-L}$ in case (a) and roughly $2^{-L/2}$ in case (b). To create such a matrix $\ma$ we stack the vectors $\alpha \vv, \vw_1, \ldots, \vw_s$ and additionally append an orthonormal basis for $\vv^{\perp}.$  We choose $\alpha$ to be an extremely small constant so that in either case, $\vv$ is approximately the singular vector of $\ma$ corresponding to $\sigma_{\min}(\ma).$  In case (a) we will arrange for $\sigma_{\min}(\ma)$ to be roughly $2^{-L}$ whereas in case (b) $\vw_k$ will cause $\sigma_{\min}(\ma)$ to increase since $\vw_k$ has positive inner product with $\vv.$  While the additive change in $\sigma_{\min}(\ma)$ is small, the multiplicative effect will be large.  We then set up a regression problem involving $\ma$ so that an approximate least squares solution has norm roughly $1/\sigma_{\min}(\ma)$ in either case, allowing us to distinguish cases (a) and (b).

\paragraph{Linear programming.}  Given our new communication lower bound, our reduction to linear feasibility is simple.  We pick a collection of linear constraints that forces a feasible point $\vx$ to satisfy $\vx = \vv$ and ${\vx}^\top{\vw_k} = 0$. In fact, this is just a linear system so how can our lower bound apply to it, given the better upper bounds for linear systems in \cite{vww20}? The issue is that we need our linear constraints to have fixed bit precision whereas \cref{prob:s_player_inner_prod} involves vectors with infinite precision. So we create inequalities enforcing the machine precision instead of requiring inner products exactly zero. Our lower bound gives a new way to obtain lower bounds depending on the condition number in this context, which may be useful for other problems. 

\paragraph{High-Accuracy Regression.}
Finally, in the high-accuracy regime we show an $\widetilde{\Omega}(sd\log(\eps^{-1}))$ lower bound for solving least squares regression to $\eps$ additive error. This shows that the $sd\log(\eps^{-1})$ dependence in our high precision algorithm is unavoidable, and in fact shows that our upper bound is tight in the common setting where $L$ and $\log\kappa$ are $O(1).$

\begin{restatable}[]{theorem}{thmHighAccuracyLowerBound}\label{thm:high_precision_regression_lower_bound} 
Consider a distributed least squares regression problem with the rows $\ma$ and $\vb$ distributed across $s$ servers, and with $\norm{\vb}{}=1.$ Let $\vx_\star$ minimize $\norm{\ma\vx_\star - b}{2}.$
A protocol in the coordinator model that produces $\widehat{\vx}$ satisfying
\begin{center}
$\norm{\ma \widehat{\vx} - \vb}{2} \leq \eps + \norm{\ma \vx_\star 
- \vb}{2}$ 
\end{center}
with probability at least $0.98$ requires $\widetilde{\Omega}\left(sd\min(\log(\eps^{-1}),L)\right)$ communication. This lower bound holds even if $\ma$ is promised to have condition number $O(1).$ 
\end{restatable}

\subsection{Notation and Preliminaries}\label{sec:all_prelims}

\paragraph{Matrices.} We use boldface letters to represent matrices and vectors.  We use $\mai{i}\in\R^{n_i \times d}$ to denote the matrix stored in the $i^{\mathrm{th}}$ machine and $\mathbf{a}_i\in\R^d$ to denote the $i^{\mathrm{th}}$ row of $\ma$ --- note that $\mathbf{a}_i$ is a column vector. When referring to the $j^{\mathrm{th}}$ row of $\mai{i}$, we use the notation $\mathbf{a}_i^{(j)}$. We define $[\ma^{(i)}]:=  \begin{bmatrix} \mai{1} \\ \vdots \\ \mai{s}\end{bmatrix}$ to be the matrix obtained by stacking all $i\in [s]$ matrices and $[\vb^{(i)}]:=  \begin{bmatrix} \vbi{1} \\ \vdots \\ \vbi{s}\end{bmatrix}$ to be the vector obtained by stacking all $i\in [s]$ vectors $\vbi{i}$. Given matrices $\ma\in\R^{n_1 \times n_2}$ and $\mb\in\R^{m_1\times m_2}$, we define the Kronecker product $\ma\otimes\mb\in\R^{m_1n_1 \times m_2n_2}$ as $\begin{bmatrix}
    a_{11} \mb \hdots a_{1n_2}\mb \\ 
    \vdots \ddots \vdots\\ 
    a_{n_1 1}\mb \hdots a_{n_1 n_2}\mb 
\end{bmatrix}.$   Given matrices $\ma\in\R^{n\times d}$ and $\mb\in\R^{m\times d}$, we denote the matrix formed by vertically stacking them on top of each other as $[\ma; \mb]$. Frequently, for a vector $\vx\in \R^n$, we use $\mx\in \R^{n\times n}$ to denote the diagonal matrix such that $\mx_{ii} = x_i$. Given vectors $\vx$ and $\vy$ of the same length, we use the notation $\vx/\vy$ (or $\vx\odot \vy$) to mean the vector formed by element-wise division (resp. multiplication). Similarly, given diagonal matrices $\mx$ and $\ms$ of the same dimensions, we use the notation $\mx/\ms$ to denote the diagonal matrix formed by element-wise division. We use $\nullsp$ to denote the nullspace (kernel) associated with a linear transformation. We say $\vu \in \nullsp(\ma)$ to mean that the vector $\vu$ lies in the kernel of $\ma$; we say $\vu\in \nullsp(\ma)^\perp$ to mean that vector $\vu$ is orthogonal to the kernel of $\ma$. A matrix $\mathbf{M}$ is a projection if it satisfies $\mathbf{M}^2 = \mathbf{M}$. It is an \emph{orthogonal} projection if it additionally satisfies $\mathbf{M} = \mathbf{M}^\top$.

\begin{lemma}[Johnson-Lindenstrauss Random Projection \cite{johnson1984extensions, achlioptas2001database}]
\label{lemma:random-projection}
Let $\vx\in\R^d$. Assume the entries in $\mg\in \R^{r\times d}$ are sampled independently from $\{-1,+1\}$. 
Then,
\[
\textrm{Pr}\left((1-\epsilon)\norm{\vx}{2}^2 \leq \norm{\frac{1}{\sqrt{r}}\mg \vx}{2}^2 \leq (1+\epsilon) \norm{\vx}{2}^2\right)
\geq
1- 2e^{-\left(\epsilon^2-\epsilon^3\right) r /4}
\]
\end{lemma}

\paragraph{Matrix Operations.} For a matrix $\ma$, we denote by $\|\ma\|_2$ its operator norm (i.e., the largest singular value). For a positive definite matrix $\mathbf{M}\in\R^{n\times n}$, we refer to the $\mathbf{M}$-norm of a vector $\vx$ to mean the weighted Euclidean norm $\|\vx\|_{\mathbf{M}} = \sqrt{\langle \vx, \mathbf{M}\vx\rangle}$.  We also use the notation $\norm{\mm}{p,2}$ to denote the $\ell_{p,2}$ norm of a matrix, i.e. the $\ell_p$ norm of the $\ell_2$ norms of its row vectors: $\left(\sum_{i=1}^n \norm{\mathbf{m}_i}{2}^p\right)^{1/p}.$ We denote the Moore-Penrose inverse (also called the pseudo-inverse) of a matrix $\ma$ with $\ma^\dagger$. The condition number of the matrix is then defined as $\kappa(\ma):=\norm{\ma}{2} \cdot \norm{\ma^\dagger}{2}$. 

\paragraph{Matrix Identities.}

To reduce communication costs while maintaining correctness, we extensively use spectrally sampled matrices, for which we need the following notation. 
\begin{definition}[Spectral Approximation]\label{def:specApprox}
For $\lambda\geq 1$, a matrix $\widetilde{\ma}\in \R^{n^\prime \times d}$ is said to be a $\lambda$-spectral approximation of $\ma\in \R^{n \times d}$ if \[ \frac{1}{\lambda} \ma^\top \ma \preceq \widetilde{\ma}^\top \widetilde{\ma} \preceq \ma^\top \ma, \] where $\preceq$ is used to denote the Loewner ordering of matrices. 
\end{definition}

\paragraph{Row Sampling Techniques.} All our algorithms  extensively use $\ell_p$ Lewis weights, which were initially discovered in the functional analysis literature by \cite{lewis1978finite} where they were employed to derive optimal bounds on distances, in the Banach-Mazur sense, between subspaces of $\ell_2$ and $\ell_p$.  The utilization of Lewis weights as sampling probabilities for the approximation of $d$-dimensional subspaces of  $\ell_p$ was first introduced by \cite{schechtman1987more}. Subsequent refinements and extensions were made by \cite{bourgain1989approximation, talagrand1995embedding, ledoux1991probability, schechtman2001embedding}. This technique has then been popularized in the algorithms community by \cite{cohen2015lp}, and we provide their definition below.

\begin{definition}[$\ell_p$ Lewis Weights;\cite{lewis1978finite, cohen2015lp}]\label{def:lewisweights}
For a full-rank matrix $\ma\in \R^{n\times d}$ and a scalar $0<p<\infty$, the $\ell_p$ Lewis weights  are the coordinates of the unique vector $\vw\in \R_{\geq 0}^n$ that satisfies the equation $$\vw_i^{2/p} = \mathbf{a}_i^\top (\ma^\top {\mw}^{1-2/p} \ma)^{-1} \mathbf{a}_i \text{  for all $i\in [n]$},  \label{eq_def_lewis_weights}$$ where $\mathbf{a}_i$ is the $i$'th row of matrix $\ma$, and ${\mw}$ is the diagonal matrix with the vector $\vw$ on its diagonal. The matrix $\ma^\top {\mw}^{1-2/p} \ma\in\R^{d\times d}$ is known as the $\ell_p$ Lewis quadratic form of $\ma$.
\end{definition} 

While this definition is recursive since $\mathbf{w}$ appears on both sides of the equation, the existence and uniqueness of such weights
is nonetheless proven by \cite{lewis1978finite, schechtman2001embedding, cohen2015lp}. Furthermore, efficient algorithms for approximating these weights have been given in \cite{cohen2015lp} and \cite{fazel2022computing}. 

An important special instance of $\ell_p$ Lewis weights are $\ell_2$ Lewis weights, commonly called leverage scores. Leverage scores have an explicit closed-form expression, and a higher leverage score indicates a higher degree of importance of  the corresponding row in composing the rowspace of the matrix. 
\begin{definition}[Leverage Scores
\cite{lewis1978finite, cohen2015lp}]\label{def:levscores}
For a full-rank matrix $\ma\in \R^{n\times d}$, the $i^{\mathrm{th}}$ leverage score is defined by $$\tau_i(\ma) = \mathbf{a}_i^\top (\ma^\top \ma)^{-1} \mathbf{a}_i \text{  for all $i\in [n]$},  \label{eq_def_levscores}$$ where $\mathbf{a}_i$ is the $i$'th row of matrix $\ma$. The leverage scores satisfy $\tau_i(\ma)\in (0, 1]$, and $\sum_{i=1}^n \tau_i(\ma) \leq d$. The \emph{generalized} leverage scores of matrix $\ma$ with respect to a matrix $\mathbf{B}\in \R^{n^\prime \times d}$ are defined as
\[
\tau_i^{\mb}(\ma)=
\begin{cases}
\mathbf{a}_i^\top (\mb^\top \mb)^{\dagger} \mathbf{a}_i & \text{ if } \mathbf{a}_i \perp \mathcal{N}(\mb) \\
\infty & \text{ otherwise }
\end{cases}, \numberthis\label{eq:genlevscoresdef}
\] 
where $(\mb^\top \mb)^\dagger\in\R^{d\times d}$ is the Moore-Penrose pseudoinverse of $\mb^\top \mb$. 
\end{definition}
 Leverage scores constitute a fundamental tool used in obtaining a small-sized (in terms of the number of rows) spectral approximation of a given matrix. In particular, it is known~\cite{drineas2006sampling, spielman2008graph} that $O(d\log d)$ rows sampled with probability proportional to the corresponding leverage scores give a spectral approximation to the original matrix. Conversely, as demonstrated by the following lemma, we may use a matrix that spectrally approximates another to  approximate the true leverage scores by constructing generalized leverage scores. 
\begin{lemma}[Leverage Score Approximation via Spectral Approximation; ~\cite{li2013iterative}]\label{lem:LiMillerPeng}
If $\mb$ is a $\lambda$-spectral approximation of $\ma$ such that $\frac{1}{\lambda}\ma^\top\ma\preceq \mb^\top \mb \preceq \ma^\top \ma$, then $\tau_i(\ma)\leq \tau_i^{\mb}(\ma) \leq \lambda \cdot \tau_i(\ma)$. 
\end{lemma}

We also need the following definition.
\begin{definition}[Ridge Leverage Scores~\cite{alaoui2015fast}]\label{defn:ridgeLevScores} 
    Given a matrix $\ma\in \R^{n\times d}$ and scaling factor $\lambda>0$, we define the $\lambda$-ridge leverage scores of $\ma$ as the leverage scores of the rows of $\ma$ computed with respect to  the matrix $[\ma; \sqrt{\lambda}\mathbf{I}]$. We denote the $i^{\mathrm{th}}$ $\lambda$-ridge leverage score as $\tau_i^{\lambda}(\ma)$, and its closed-form expression is $\mathbf{a}_i^\top (\ma^\top \ma + \lambda \mathbf{I})^{-1} \mathbf{a}_i$. 
\end{definition}

Ridge leverage scores have been used in \cite{cohen2017input, mccurdy2018ridge} for input sparsity time low rank approximation. 

\paragraph{Bit Complexity.}
We say a number is represented with $L$ bits in fixed-point arithmetic if it has at most $L$ bits before the decimal point and at most $L$ bits after the decimal point. Therefore, such a number is in the set $\{0\} \cup [2^{-L}, 2^{L}-1] \cup [-2^{L}+1, -2^{-L}]$. Note that the condition number of a full column-rank $n\times d$ matrix $\ma$ with $L$ bits in fixed-point arithmetic is $e^{O(d(L+\log(dn))}$.\footnote{To see this, note that the top singular value of $\ma^{\top} \ma$ is bounded by $2^{L} d \sqrt{n}$ simply by the bound on the entries of $\ma$. On the other hand $\det(\ma^{\top} \ma)=\lambda_1\cdots \lambda_d$ is a nonzero positive integer so $\lambda_{\min}(\ma^{\top} \ma) \geq 1/\lambda_{\max}(\ma^{\top} \ma)^{d-1}.$}

\paragraph{Approximations.} Given scalars $x, y,$ and $\lambda \geq 1$, we use $x\approx_{\lambda} y$ to denote $y\cdot e^{-\lambda} \leq x \leq y \cdot e^{\lambda}$. In the case of matrices, we overload notation and denote $\ma\approx_{\lambda}\mb$ to mean that $\ma$ is a $\lambda$-spectral approximation of $\mb$, i.e., $\frac{1}{\lambda} \ma^\top \ma \preceq \mb^\top \mb \preceq \ma^\top \ma$.

\paragraph{Time and Probability.}  The notation $\Otil$ hides factors of $\poly\log(sndL)$ and $\poly\log\log(\kappa\epsilon^{-1} \delta^{-1})$. We also use the terminology ``with high probability (w.h.p.)'' to mean ``with probability at least $1-n^{-C}$  for some arbitrarily large constant $C$''. 

\begin{fact}\label{fact:OperatorNormBoundForBitBoundedMatrix}
    \looseness=-1If the bit complexity of each entry of $\ma\in \R^{n\times d}$ to be $L$, we have  $\|\ma^\top \ma\|_2\leq 2^{2L}\cdot nd$. 
\end{fact}
\begin{proof}
    Since we assume each entry of $\ma$ to have a bit complexity of at most $L$, it implies that each entry has a value of at most $2^L$. Since the row dimension of $\ma$ is $n$, this implies each entry of $\ma^\top\ma$, obtained by an inner product of two $n$-dimensional vectors, is at most $2^{2L} n$. To compute the operator norm of $\ma^\top \ma$, we want to bound  $\max_{\vx:\|\vx\|_2=1} \|\ma^\top \ma \vx\|_2$. Since $\|\vx\|_2 =1 $, every entry of the vector  $\ma^\top \ma\vx \in \R^d$ has a value of at most $2^{2L} n \sqrt{d}$. Therefore, the $\ell_2$ norm of the $d$-dimensional vector $\ma^\top \ma \vx$ can be bounded by $2^{2L} nd$, as claimed. 
\end{proof}

\section{Regression in the Coordinator Model}

The main export of this section is a protocol with the following guarantee.
\thmEllpRegressionLessThanTwo*

\looseness=-1 Towards proving this result, we provide two protocols: \cref{alg:relativeLevScoreSampling} in \cref{subsec:l2_lev_score_sampling} and one in \cref{sec:ell2_subspace_embedding_via_block_lev_scores} for $p=2$, based on sketching of block leverage scores. The latter approach is easier to implement with a better dependence in the $\log$ factors, and we therefore expect this approach to yield better practical performance for $p=2$.  
With either approach, our algorithm works by constructing an $\ell_p$ subspace embedding via Lewis weight or leverage score sampling, from which regression is an immediate corollary.  

\subsection{First Protocol: Non-Adaptive Adaptive Sampling}
\label{subsec:l2_lev_score_sampling}

\looseness=-1We first discuss the special case of our protocol when $p=2$.  In this setting, our first protocol, \cref{alg:relativeLevScoreSampling},  revisits the recursive sampling framework of \cite{cohen2015uniform}. The idea there is to sample a nested sequence $S_1 \supseteq S_2 \supseteq S_3 \cdots \supseteq S_{O(\log n)}$ of rows of $\ma$ uniformly, where $S_1$ consists of all the rows of $\ma$, and each subsequent set samples a random subset of about half the size of the current set. 
Thus, $|S_{O(\log n)}| = O(d)$. One recursively computes a subspace embedding of $\ma(S_i)$ --- the matrix $\ma$ restricted to rows in $S_i$ --- and uses this embedding to compute a subspace embedding of $\ma(S_{i-1})$. 

One could hope to use this algorithm in the coordinator model since communicating $\widetilde{O}(d)$ uniform rows takes only $\widetilde{O}(d^2)$ communication rather than $\widetilde{O}(sd^2)$ communication, so if there is a way to do the distributed sampling of the next $\widetilde{O}(d^2)$ rows from the previous $\widetilde{O}(d^2)$ rows that the coordinator learned, using only $\widetilde{O}(sd + d^2)$ additional communication, then overall this would give $\widetilde{O}(sd + d^2)$ total communication in each of $O(\log n)$ rounds (note that there is an $\Omega(sd)$ lower bound -- the reason for this term will become clear later).

\looseness=-1To try to implement this idea, for each $j \in S_i$, one needs to compute the {\it generalized leverage score} $\tau_j^{\ma(S_{i+1})}(\ma)$ 
of the row $\mathbf{a}_j$ 
with respect to $\ma(S_{i+1})$. The coordinator inductively maintains a subspace embedding $\mb_{i+1} \in \mathbb{R}^{O(d) \times d}$ of $\ma(S_{i+1})$. Ideally it could send $(\mb_{i+1}^\top \mb_{i+1})^{-1/2} \vg$ to each server, for a random Gaussian vector $\vg$. Then for a row $\mathbf{a}_k$ held by a server, the server can compute $\|\mathbf{a}_k^\top (\mb_{i+1}^\top \mb_{i+1})^{-1/2} \vg\|_2$, which using the Johnson-Lindenstrauss lemma, can be used to approximate $\tau_k^{\mb_{i+1}}(\ma)$
(and hence, approximately, $\tau_k^{\ma(S_{i+1})}(\ma)$). 

\paragraph{The Issue of Bit Complexity.} Unfortunately, while $(\mb_{i+1}^\top \mb_{i+1})^{-1/2} \vg$ is only a $d$-dimensional vector, a major issue is
that the {\it bit complexity} of describing this vector is potentially $\Omega(d^2)$ due to the poor conditioning of $\mb_{i+1}^\top \mb_{i+1}$ which can cause its inverse to have high bit complexity, and thus it can require $\Omega(sd^2)$ bits to be communicated from the coordinator to the $s$ servers. Indeed, even if the entries of $\ma$ were in $\{0,1\}$, the non-zero singular values of submatrices could be exponentially small in $d$, making the entries of $(\mb_{i+1}^\top \mb_{i+1})^{-1/2} \vg$  exponentially large, each requiring $\Omega(d)$ bits of precision, and it is not clear how to round them to preserve relative error. 

To circumvent this issue, for each $j\in [s]$, we can have the $j$-th server compute a sketch $\ms^{(j)} \ma^{(j)}$ of its matrix $\ma^{(j)}$ (which could be a JL sketch for example in the case of $\ell_2$ regression) and send this to the coordinator. This has low bit complexity since the input $\ma^{(j)}$ is assumed to have low bit complexity.
The coordinator can now locally post-multiply by $(\mb_{i+1}^\top \mb_{i+1})^{-1/2}$ and obtain 
$\ma^{(j)} (\mb_{i+1}^\top \mb_{i+1})^{-1/2}$. By choosing $\ms^{(j)}$ to be a sketch for approximately preserving the Frobenius norm, if 
$\| \ma^{(j)} (\mb_{i+1}^\top \mb_{i+1})^{-1/2}\|_F^2\leq 1$
, then sampling the rows of $\ma^{(j)}$ using their generalized leverage scores computed with respect to $\mb_{i+1}$ 
is equivalent to squared row norm sampling from the matrix
$\ma^{(j)} (\mb_{i+1}^\top \mb_{i+1})^{-1/2}$. This 
can be accomplished using a sketch $\mt^{(j)}$ for sampling rows according to their squared $2$-norm \cite{mahabadi2020non}, i.e., server $j$ can compute $\mt^{(j)} \ma^{(j)}$ and send it to the coordinator, who can then post-multiply to obtain $\mt^{(j)} \ma^{(j)} (\mb_{i+1}^\top \mb_{i+1})^{-1/2}$, from which a sample can be extracted. This idea of post-multiplying by a change of basis is referred to as non-adaptive adaptive sampling in the streaming literature \cite{mahabadi2020non}.

The case 
$\|\ma^{(j)} (\mb_{i+1}^\top \mb_{i+1})^{-1/2}\|_F^2\geq 1$
may occur either because we should take more than one sample from the $j$-th server or because for a row $\mathbf{a}_k$, we have $\|\mathbf{a}_k^\top (\mb_{i+1}^\top \mb_{i+1})^{-1/2}\|_2\geq1$, 
in which case we would like to treat its leverage score as $1$. From bounds in \cite{cohen2015uniform}, there can be at most $\widetilde{O}(d)$ rows that need their scores to be truncated, and these can all be found with $\widetilde{O}(d^2 + sd)$ communication by using non-adaptive adaptive sampling, reported back to the servers owning such rows, and removed from their local matrices. After doing so we reduce to the first case, in which case  
$\|\ma^{(j)} (\mb_{i+1}^\top \mb_{i+1})^{-1/2}\|_F^2$ 
is proportional to the number of samples the coordinator should obtain from the $j$-th server, and we can then adjust the size of the row sampling sketch for each server accordingly. We note that this Frobenius norm can be approximated by the coordinator, simply by having server $j$ sketch $\ma^{(j)}$ on the left by a JL embedding.

\looseness=-1We extend the above argument to $\ell_p$-regression, where we instead use a sketch for sampling rows proportional to their Euclidean norms  raised to their $p^{\text{th}}$ power, rather than two, and  follow the recursive Lewis weight sampling algorithm in \cite{cohen2015lp}.  While the details are slightly more complicated than the leverage score algorithm described above, the key ideas remain essentially the same.

\subsubsection{Details of Our Algorithm}
\looseness=-1We now give procedures for the coordinator to construct $\ell_p$ subspace embeddings (for $1\leq p \leq 2$) of a matrix $\ma = [\ma^{(1)}; \ldots; \ma^{(s)}]$ distributed among $s$ servers.
In particular, this allows the coordinator to solve $\ell_1$ and $\ell_2$ regression problems with $\eps$ error, i.e. the coordinator recovers an $\widehat{\vx}$ with $
\norm{\ma\widehat{\vx} - \vb}{p} \leq (1+\eps)\norm{\ma\vx^\star - \vb}{p},$
where $\vx^\star$ is the optimal solution to the regression problem and $p\in[1,2].$

\begin{figure}[!ht]
\begin{framed}

\textbf{Input.} Matrix $\mm\in\R^{d\times d}$ held by the coordinator, matrices $\mai{j}\in\R^{n_j \times d}$ on servers $j\in[s]$, parameters $p$, $\delta$, $r$ and $T$, where the first three parameters are as in \cref{prob:rel_lev_score_est}, and $T$ roughly quantifies the number of samples taken

    \vspace{.5em}

\textbf{Output.} $N$ (rescaled) samples $\mathbf{a}_{i_1}, \mathbf{a}_{i_2},\ldots, \mathbf{a}_{i_N}$ from the rows of $\mai{j}$ sent to the coordinator 
\vspace{.5em}

\begin{enumerate}[itemsep = .1em, leftmargin = 1.7em, topsep = .4em, label=\protect\circled{\arabic*}]
    \item \label{item:ComputeFj} Use \cref{lem:ell_p_two_norm_estimation} to obtain constant factor estimates $F_j$ for $\|{\ma^{(j)}\mm}\|_{p,2}^p$ for all servers $j\in [s]$

    \item Initialize $\mathcal{O} = \{\cdot{}\}$. 
    \item Repeat $O(T{\delta^{-1}}\log T)$ times
    \begin{enumerate}[itemsep = .1em]
        
        \item Sample a server $\ell$ from the distribution obtained by normalizing $(F_1, \ldots, F_s).$

        \item Call \textbf{SampleFromBlock}$(\ma^{(\ell)}, \mm, p, 1/2, 1/2)$
 in \cref{alg:auxillary_sampling_procedure} to sample a row $\mathbf{a}^{(\ell)}_i$ from $\ma^{(\ell)}$
                
        \item \label{item:update_Fj} Coordinator computes $\langle\mathbf{a}^{(\ell)}_i, \mm \mm^\top \mathbf{a}^{(\ell)}_i\rangle.$  If it is outlying (at least $1$), add $(\ell, i)$ to $\mathcal{O}$, remove row $i$ from $\mai{\ell}$ and recompute $F_{\ell}$ as  in \cref{item:ComputeFj} above
    \end{enumerate}
    
    \item For each server $j\in [s]$, let $n_j$ be the number of its rows in $\mathcal{O}$ 

    \item \label{step:do_the_sampling}Repeat $N$ times:
    \begin{enumerate}[itemsep=.1em]
    
        \item With probability $q := \sum_{j\in [s]} n_j/\sum_{j\in [s]} (n_j + F_j)$ sample a row from $\mathcal{O}$ uniformly, and rescale by $(Nq)^{-1/p}$
        
        \item Otherwise sample a server $\ell$ from the distribution obtained by normalizing $(F_1, \ldots, F_s)$. Then call \textbf{SampleFromBlock}$(\ma^{(\ell)}, \mm, p, 1/2, 1/2)$ 
to sample a row $\mathbf{a}^{(\ell)}_i$ and obtain a probability estimate $\widetilde{p}_i$. 
        Rescale $\mathbf{a}^{(\ell)}_i$ appropriately, by a factor of $\left((1-q)\cdot\frac{F_{\ell}}{\sum_{i\in [s]} F_i}\cdot\widetilde{p}_i\right)^{-1/p}$ 
    \end{enumerate}
\end{enumerate}
\end{framed}
\caption{Relative Lewis weight sampling}
\label[alg]{alg:relativeLevScoreSampling}
\end{figure}

Our algorithm is a communication-efficient implementation of the recursive samping procedure in \cite{cohen2015lp}. In the special case $p=2$, this procedure is essentially the leverage score computation algorithm of \cite{cohen2015uniform}, which iteratively computes improved spectral approximations to $\ma$. These improvements are realized by alternating between computing generalized leverage scores (with respect to the current best spectral approximation) and using the current generalized leverage scores to compute an improved spectral approximation. To implement this algorithm, we need a procedure (\cref{alg:relativeLevScoreSampling}) to carry out leverage score sampling with respect to an intermediate spectral approximation.  We will give a subroutine (\cref{alg:auxillary_sampling_procedure}) to solve the following slightly more general sampling problem that will be useful for $\ell_p$ regression.  In our application, $\mm$ will be the inverse of these spectral approximations to $\ma^{\top} \ma$ (or to the inverse Lewis quadratic form for $p\neq 2$; cf. \cref{def:lewisweights} for the definition of Lewis quadratic form.)

\begin{problem}
\label[prob]{prob:rel_lev_score_est}
Let $\mm\in\R^{d\times d}$ be a positive semidefinite matrix owned by the coordinator. We have matrices $\ma^{(1)}, \ma^{(2)}, \hdots, \ma^{(s)}$ on the coordinators, with $\ma^{(j)} \in\R^{n_j \times d}$ for all $j\in [s]$. Denote the $i^\mathrm{th}$ row of the $j^\mathrm{th}$ matrix by $\mathbf{a}_i^{(j)}$. Then, define the following quantities: 
\begin{align*}
\uij &:= \norm{\mm\mathbf{a}_i^{(j)}}{2}^2,\,\,\,\,
\vij := \min\left(\norm{\mm\mathbf{a}_i^{(j)}}{2}^p, 1\right).
\end{align*}
The problem asks to output $r$ i.i.d. rows of $\ma$ from the probability distribution obtained by normalizing the $\vij$'s.  Specifically, for each individual sample, the probability of choosing row $i_0$ from block $j_0$ should be
\[
(1\pm \eps)\frac{v_{i_0}^{j_0}}{\sum_{j\in [s],i\in [n_j]} \vij} \pm O(n^{-c}),\]
where $n$ is an upper bound on $n_j$, the number of rows in  $\ma^{(j)}$, and $C$ is an absolute constant.  For each sampled row, the problem also requires outputting a $(1\pm\eps)$-factor approximation to the sampling probability.
\end{problem}

We note that the $n^{-c}$ term in the problem statement above comes from the $\ell_p$-sampler that we borrow. This term can be taken to be $n^{-3}$ for example, with no extra cost.  It will therefore be irrelevant for us, as we will never apply the sampler more than $n$ times per block.

To solve \cref{prob:rel_lev_score_est},  we use the $\ell_p$ sampling sketch of \cite{jayaram2021perfect}, which we employ to sample a row of a matrix proportional to the $p^{\text{th}}$ power of its norm. Such sketches using polylogarithmic space exist only for $p\in[1,2]$.  This is  the reason for our restricted range of $p$.

For completeness, we give a statement of their result here.
\begin{theorem}
\label{thm:ell_p_sampling}
($\ell_p$-sampling.) There is a sketching matrix $\ms \in \R^{m\times n}$ such that given $\ms v \in \R^m$ one can output an index $i$ of $v$ such that the probability of outputting index $i$ is $(1\pm \nu)\frac{\abs{v_i}{}^p}{\norm{v}{p}^p} \pm \frac{1}{\poly(n)}.$
The sketching dimension $m$ can be taken to be $O(\log^2 n \log\frac{1}{\delta}(\log\log n)^2)$. Moreover the entries of $\ms$ can be taken to have $O(\poly\log n)$ bits of precision for $\eta \geq n^{-c}$.
\end{theorem}

\begin{lemma}
\label{lem:sample_from_block_guarantee}
($\ell_{p,2}$ sampling procedure) 
With an appropriate $\ell_p$-sampling sketch, the probability that \cref{alg:auxillary_sampling_procedure}, with input matrices $\mm$ and $\mx$, outputs row $\mathbf{x}_i$ 
is $(1\pm \eps)p_i$ where 
\[p_i := \frac{\norm{\mathbf{e}_i^\top \mx \mm}{2}^p}{\norm{\mx \mm}{p,2}^p}.\] Furthermore, the output probability estimate $\hat{p}_i$  of \cref{alg:auxillary_sampling_procedure}  satisfies $\hat{p}_i = (1\pm \eps) p_i.$
These guarantees fail with probability at most $\delta$, and the total communication used by the protocol is \[\widetildeO({d}{\eps^{-2}}\log^2 n \log^2({\delta^{-1}}) + dL).\] 
\end{lemma}

\begin{figure}[!ht]
\begin{framed}
    
\textbf{Input.} A matrix
$\mm \in \R^{d\times d}$ held by the coordinator, matrix $\mx\in\R^{n\times d}$ held by a server, accuracy parameter $\eps$, failure tolerance $\delta$, norm $p\in [1,2]$
\vspace{0.2cm}

\textbf{Output.}
On the coordinator side:
(1) A row $\vx_k$ of $\mx$ such that the probability of sampling row $i$ is $(1\pm \eps)p_i$ where $p_i := \frac{\norm{\mathbf{e}_i^\top \mx \mm}{2}^p}{\norm{\mx \mm}{p,2}^p}$ and $\norm{\cdot}{p,2}$ is the $\ell_p$ norm of the vector of row $\ell_2$ norms.
and (2) An estimate $\widehat{p}_k = (1\pm \eps) p_k.$
\\

\textbf{procedure SampleFromBlock}($\mx$, $\mm$, $p$, $\eps$, $\delta$):

    \begin{enumerate}[itemsep = .1em, leftmargin = 1.7em, topsep = .4em, label=\protect\circled{\arabic*}]

    \item Set $r = {\eps^{-2}}(\log(n) + \log(\delta^{-1}))$.  Coordinator forms $\mm' = \mm \otimes \mathbf{I}_r\in\R^{rd\times rd}$. 
    Server forms $\mx' = \mx\otimes \mathbf{I}_r\in\R^{nr \times dr}$ 

    \item\label{item:jwSketchApp} Server draws an $\ell_p$ sampling sketch, 
    computes $\ms\mx'$, and sends it to the coordinator

    \item Coordinator samples $\vg \sim c_p^{1/p}\mathcal{N}(\mathbf{0}, \mathbf{I}_{rn}) \in \R^{rd}$ and computes
    $\ms\mx'\mm'\vg$ where $c_p = 1/\E(\abs{\mathcal{N}(0,1)}^p)$ 

    \item \label{step:apply_sampling_sketch} Coordinator uses the $\ell_p$ sampling sketch $\ms\mx'$
    to sample 
    a row index $k$ of $\mx'\mm'\vg$

    \item \label{step:request_row_sample_from_block}Coordinator requests row $\mathbf{x}_k$ from the server who then sends that row

    \item
    \label{step:estimate_sampling_prob}  Coordinator produces probability estimate $\widehat{p}_k$:  Coordinator first uses $\mathbf{x}_k$ to compute $\norm{\ve_i^T \mX \mM}{2}^p$ and an uses the $\ell_{p,2}$ norm estimation protocol from \cref{lem:ell_p_two_norm_estimation} to approximate $\norm{\mX\mM}{p,2}^p$ to $(1\pm \eps)$ multiplicative error.
    \end{enumerate}
\end{framed}
\caption{$\ell_{p,2}$ sampling procedure used by \cref{alg:relativeLevScoreSampling}}
\label[alg]{alg:auxillary_sampling_procedure}
\end{figure}

\begin{proof} We first prove the communication cost, followed by the correctness of \cref{alg:auxillary_sampling_procedure}. 
\paragraph{Communication cost. } We use the $\ell_p$-sampling sketch $\ms$ of \cite{jayaram2021perfect}, as stated above in \cref{thm:ell_p_sampling}.
Recall that given the sketch  $\mathbf{v} \in \R^n$, this this result allows us to sample an index $i$ of $\mathbf{v}$ with probability proportional to $\abs{v_i}^p,$ up to $\poly(1/n)$ additive error on the sampling probabilities.
Note that Section 6 of \cite{jayaram2021perfect},  also gives a sketch to estimate the frequency moment of a sample, which we need here. To obtain a failure probability of $O(\delta)$ and an accuracy $O(\eps)$ as above, their result gives such a sketch $\ms$ with space $\widetildeO\left(\log^2(rn)\log(\delta^{-1}) + {\eps^{-2}}\log(rn)\log^2(\delta^{-1})\right)$ when applied to a single vector of length $rn.$ 
In \cref{item:jwSketchApp}, the server applies $\ms$ to the matrix $\mx\otimes \mathbf{I}_r\in \R^{nr \times dr}$ which 
has $dr$ columns before sending it to the coordinator, giving the communication bound stated above. This proves the  communication cost, since there is no communication cost in the other steps.  

\paragraph{Correctness.} Next, we observe that the $\ell_p^p$ mass corresponding to each row of $\mx\mm$ is preserved by the Gaussian sketch.
Note that the entries of $\mx'\mm'\vg\in\R^{rn}$ 
consist of $rn$ Gaussians, 
where for each $i\in[n]$, there are $r$ mutually independent entries distributed 
as $c_p^{1/p} \mathcal{N}(0,\norm{\mathbf{e}_i^\top\mx\mm}{2}^2)\sim c_p^{1/p}\norm{\mathbf{e}_i^\top\mx\mm}{2} \mathcal{N}(0,1).$  For $(i,j)\in [n]\times [r]$ let $z_{i,j}$ be independent and  standard normal.  Then the total $\ell_p^p$ mass corresponding to row $i$ of $\mx\mm$ is distributed as
\[
\frac{\sum_{j\in [r]} c_p\norm{\mathbf{e}_i^\top\mx\mm}{2}^p\abs{z_{i,j}}^p} {\sum_{(i,j)\in[n]\times[r]} c_p\norm{\mathbf{e}_i^\top\mx\mm}{2}^p\abs{z_{i,j}}^p }
=
\frac{\norm{\mathbf{e}_i^\top\mx\mm}{2}^p \sum_{j\in [r]}c_p\abs{z_{i,j}}^p} {\sum_{i\in[n]} \norm{\mathbf{e}_i^\top\mx\mm}{2}^p \sum_{j\in[r]} c_p \abs{z_{i,j}}^p }.
\]
For all $i,$ we have that $\sum_{j\in[r]} c_p \abs{z_{i,j}}^p$ is a sum of independent subexponential random variables with mean $1.$ So by Bernstein's inequality~\cite{vershynin2018high}, each fixed sum $Z_i := \sum_{j\in[r]} c_p \abs{z_{i,j}}^p$ satisfies
\[
\pr\left(\abs{Z_i - r} \geq \eps r\right) \leq 2\exp(-c \eps^2 r).
\]
So $Z_i = (1\pm \eps) r$ with probability  $1-2\exp(-c\eps^2 r).$ For large enough $r=O(\frac{1}{\eps^2}(\log n +\log\frac{1}{\delta})),$ this probability is at most $\delta/n$ 
, and therefore the bound on $Z_i$ holds for all $i\in[n]$ simultaneously with probability all but $\delta.$ Conditioned on these bounds, the $\ell_p^p$ mass on row $i$ of $\mx\mm$ is 
\[
\frac{\norm{\mathbf{e}_i^\top\mx\mm}{2}^p Z_i} {\sum_{i\in[n]} \norm{\mathbf{e}_i^\top\mx\mm}{2}^p Z_i }
= 
(1\pm O(\eps))\frac{\norm{\mathbf{e}_i^\top\mx\mm}{2}^p} {\sum_{i\in[n]} \norm{\mathbf{e}_i^\top\mx\mm}{2}^p }
=
(1\pm O(\eps))\frac{\norm{\mathbf{e}_i^\top\mx\mm}{2}^p}{\norm{(\mx\mm)}{p,2}^p}
\] The claim now follows from correctness of the $\ell_p$ sampler.
\end{proof}

We need a sketch to approximate the $\ell_p$ norm of a vector.  Such sketches are based on the $p$-stable random variables, which require infinite bits to represent \cite{indyk2006stable,kane2011fast}.  For completeness, we show that using appropriately rounded Cauchy random variables is sufficient.  One might be concerned that the rounding could cause problems as we will sketch vectors with potentially exponentially large entries (corresponding to poorly conditioned linear systems). However this is not actually an issue.
\begin{lemma}
\label{lem:ell_p_estimation_with_bit_precision}
Let $p\in [1,2].$ There is a sketching matrix $S \in \R^{m}{n}$ that sketches a vector $x \in \R^n$ and outputs a $1\pm \eps$ multiplicative approximation to $\norm{x}{p}$ with probability at least $1-\delta$.  This guarantee is achieved with a sketching dimension of $m = O(\frac{1}{\eps^2}\log\frac{1}{\delta})$ .  Moreover, each entry of $S$ is represented with $O(\log \frac{n}{\eps})$ bits of precision.
\end{lemma}
\begin{proof}
We first recall the structure of a typical $\ell_p$ norm sketch.  Each row is independent and consists of i.i.d. $p$-stable random variables $x_1, \ldots, x_n \sim p-\text{stable}.$  Now suppose that we round each $x_i$ and instead use $x_i' = x_i + f_i$ where $\mathbf{f}$ is the vector of rounding errors.
Then
\[
\inner{x'}{v} = \inner{x}{v} + \inner{f}{v} \sim \norm{v}{p}Y + \inner{f}{v},
\]
where $Y$ is a $p$-stable random variable with the same distribution as the $x_i$'s.  Since $p\leq 2$, we have
\[
\inner{f}{v} \leq \norm{f}{2}\norm{v}{2} \leq \norm{f}{2}\norm{v}{p}.
\]
Thus, each row of the sketch is distributed as $\norm{v}{p}Y$ but perturbed by at most $\norm{f}{2}\norm{v}{p}.$  The $\ell_p$ estimation sketch works by taking the median over $1/\eps^2$ rows, thereby approximating $\norm{v}{p}$ to within $\eps\norm{v}{p}$.   After rounding, the median is perturbed at most $\norm{f}{2}\norm{v}{p}$, so it suffices to arrange for $\norm{f}{2} \leq \eps.$  This can be achieved by rounding each coordinate by no more than $\eps/\sqrt{n},$ which means we can round to $O(\log\frac{n}{\eps})$ bits of precision.
\end{proof}

We also note that the same approach of embedding $\ell_{p,2}$ into $\ell_p$ when combined with an $\ell_p$ norm estimator, gives an $\ell_{p,2}$ norm estimation protocol in our setting.

\begin{lemma}
\label{lem:ell_p_two_norm_estimation}
Let $\mx\in\R^{n\times d}$ be held by a server and $\mm \in \R^{d\times d}$ a matrix held by the coordinator. For $1\leq p \leq 2$, there is a protocol using communication $\widetildeO(d\eps^{-2} L\log n\log(\delta^{-1}))$ that allows the coordinator to estimate $\norm{\mx\mm}{p,2}^p$ to $1\pm \eps$ multiplicative error.
\end{lemma}

\begin{proof}
Similar to the setup of  \cref{lem:sample_from_block_guarantee}, consider sampling $r = O({\eps^{-2}}(\log n + \log(\delta^{-1})))$ independent  Gaussians $\vg_1,\ldots \vg_r \sim c_p^{1/p}\mathcal{N}(\mathbf{0}, \mathbf{I}_d)$ in $\R^d,$ where $c_p = 1/\E(\abs{\mathcal{N}(0,1)}^p).$  The protocol that achieves the claimed approximation guarantee is for the server to send a $p$-moment estimation sketch
 $\ms \mx$ to the coordinator (gotten from \cref{lem:ell_p_estimation_with_bit_precision}), who then computes $\ms\mx\mm\vg_1, \ldots, \ms\mx\mm\vg_r.$  Note that
\[
\norm{\mx\mm\vg_1}{p}^p + \ldots + \norm{\mx\mm\vg_r}{p}^p
=
\sum_{j\in [r]} \sum_{i\in [n]} \abs{\mathbf{e}_i^\top{\mx\mm \vg_j}}^p.
\]
As in the proof above, each inner sum $\sum_{i\in [n]} \abs{\mathbf{e}_i^\top{\mx\mm \vg_j}}^p$ is $(1\pm \eps)\norm{\mathbf{e}_i^\top\mx\mm}{2}^p$ 
with failure probability at most $\delta.$  So $\norm{\mx\mm\vg_1}{p}^p + \ldots + \norm{\mx\mm\vg_r}{p}^p$ is a $(1\pm \eps)$-approximation to $\norm{\mx\mm}{p,2}^p$. 
We choose the sketch to have failure probability most $\delta/r$ so that each term above is well approximated. The sketch uses \[\widetildeO\left({L}{\eps^{-2}}\log n\log(\delta^{-1})\right) \text{ space per column},\] giving the claimed communication bound (since we have $d$ columns in $\mx$). 
\end{proof}

\begin{lemma}
    \label{lem:estimate_relative_leverage_scores}
    There is a protocol which solves \cref{prob:rel_lev_score_est} with failure probability at most $\delta$ and 
    $\tildeO\left((T\log T \log({\delta^{-1}}) + r)({d}{\eps^{-2}}\log^2n + dL) + s{\eps^{-2}}\log n \log({\delta^{-1}}) L\right)$
    communication, where the parameters here are as given in \cref{prob:rel_lev_score_est}.
\end{lemma}

\begin{remark}
Note that later, in the proof of \cref{thm:ell1_subspace_embedding}, we will apply this result with $T,r = O(d).$
\end{remark}
\begin{proof}
    The first step of \cref{alg:relativeLevScoreSampling} is to produce an estimate of
    \[
    B^{(j)} := \sum_i \vij
    \]
    for all $j$.  To do this, we first estimate 
    \[
    B^{(j)} := \sum_i (\uij)^{p/2} = \norm{\Aj \mm}{p,2}^p
    \]
    to within $(1\pm \eps)$ multiplicative error using \cref{lem:ell_p_two_norm_estimation} above.
    
    In order to handle the truncation to $1$, we would next like to find all $\uij\geq1.$ Call the corresponding rows ``outlying". To identify the outlying rows we use the $\ell_{p,2}$ sampling protocol given in \cref{alg:auxillary_sampling_procedure}.

    To implement $\ell_{p,2}$ sampling across the blocks, the coordinator first uses the values of $\widetilde{B}^{(j)}$ to choose a server $\ell$ from which to sample.  Then they run the protocol
    \[\textbf{SampleFromBlock}(\ma^{(\ell)}, \mm, p, 0.5, 0.5)\] 
    from \cref{alg:auxillary_sampling_procedure} to sample from the $p$th moment of the row norms, up to a $(1\pm 0.5)$ factor on the sampling probabilities.

    We would like to identify all outlying rows and (temporarily) remove them as we go.  We sample rows via the $\ell_p$ sampling method described above.  Each such row is sent to the coordinator which then checks by direct computation whether or not it is outlying.  If it is, then the server temporarily removes that row and sends a new sketch to the coordinator, so that the coordinator can update its norm estimate for that server (as in step \ref{item:update_Fj}). 
    
    Suppose that there are $k$ outlying rows remaining.  Then the total $\ell_p^p$ mass contributed by those rows is at least $k$, and the total $\ell_p^p$ mass of the non-outlying rows is at most $T$ (from the statement of \cref{prob:rel_lev_score_est}).  Hence the probability of sampling an outlying row is at least $\frac{ck}{T+k}$ for a constant $c$, which satisfies $\frac{ck}{T+k} \geq \frac{ck}{2T}$ since $k\leq T.$  The expected number of samples needed to encounter an outlying row is therefore at most $O(\frac{2T}{k}).$ There are at most $T$ outlying rows to start with, so up to a constant factor, the expected number of samples needed to find all outlying rows is at most 
    \[
    \frac{2T}{T} + \frac{2T}{T-1} + \ldots +\frac{2T}{1} = 2T\left(\frac{1}{1} + \frac{1}{2} + \ldots + \frac{1}{T}\right) \leq O(T\log T).
    \]
    This means that after $O(T\log T)$ rounds of sampling we identify all outlying rows with constant probability. To boost the success probability to $1-\delta$, we can apply the standard median trick and obtain a bound of $O(\log\frac{1}{\delta} T\log T)$ on the required number of rounds of sampling to identify all outlying rows with probability  $1-\delta$.

    Next, the coordinator counts the total number of outlying rows for each server, and adds this to a $(1\pm \eps)$-factor approximation of the $\ell_{p,2}$ norm of the remaining rows (computed using a sketch as above).  This gives the coordinator $(1\pm \eps)$-factor approximations $\widehat{B}^{(j)}$ to $B^{(j)}$.
    
    To sample from the $\vij$ distribution, the coordinator first decides whether to sample an outlying row, with probability proportional to the number of outlying rows, as described in \cref{step:do_the_sampling}.
    
    Otherwise, the coordinator picks a server $j$ with probability proportional to $\widehat{B}^{(j)}$, and requests an $\ell_{p,2}$ sampling procedure with failure probability $O(\delta/r)$ as given in \cref{alg:auxillary_sampling_procedure}, for all the non-outlying rows.  The coordinator uses this sketch to sample a row index and then requests the corresponding row from the server, as in \cref{step:request_row_sample_from_block}.
    The probability estimate given by \cref{alg:auxillary_sampling_procedure} suffices to produce the estimates required by \cref{prob:rel_lev_score_est}. This procedure is then repeated $r$ times.

    The communication cost comes from  the $\ell_{p,2}$ sampling procedure, which is run $O(\log({\delta^{-1}})T\log T + r)$ times. We also request an $\ell_{p,2}$ norm estimation sketch from each server, incurring an additional $\tildeO(s\log n \log(L\delta^{-1}))$ cost. By \cref{lem:sample_from_block_guarantee} this gives the stated communication cost. 
\end{proof}

\subsubsection{Proof of Main Results for $\ell_p$ Regression}

\begin{figure}[!ht]
\begin{framed}
    
\textbf{Input.} A matrix $\ma \in \R^{n\times d}$, $p\in [1,2]$

\vspace{.5em}
 \textbf{Output.} A (weighted) row-sampling matrix $\mathbf{S}$ with $O(d\log d)$ rows such that\\ $\norm{\mathbf{S}\ma \x}{p} = (1\pm \eps)\norm{\ma x}{p}$ 
 for all $\vx \in \R^d.$

\begin{enumerate}[itemsep = .1em, leftmargin = 1.7em, topsep = .4em, label=\protect\circled{\arabic*}]
    \item Compute $\mathbf{Q} = \textbf{ApproxLewisForm}(\ma,p)$ an $O(1)$ spectral approximation to the Lewis quadratic form of $\ma$
    \item Sample $O(\frac{1}{\eps^2}d\log d)$ indices of $\ma$ as in step \ref{step:get_A_prime} below, and return the corresponding sampling matrix $\ms.$
\end{enumerate}

\textbf{procedure ApproxLewisForm}($\ma$, $p$)
    \begin{enumerate}[itemsep = .1em, leftmargin = 1.7em, topsep = .4em, label=\protect\circled{\arabic*}]
        \item If $\ma$ has at at most $d$ rows, return $\ma.$ 
        \item Uniformly sample half the rows of $\ma$ to obtain $\widehat{\ma}$
        \item Compute $\widehat{\mathbf{Q}}$ by recursively calling \texttt{ApproxLewisForm}($\widehat{\ma}$, $p$)
        \item\label{step:ComputeUiLewisFormApprox} For all row indices $i$ of $\ma$, compute $u_i$ as constant-factor approximations to $\mathbf{a}_i^\top \widehat{\mathbf{Q}}^{\dagger} \mathbf{a}_i$
        \item For all row indices $i$ of $\ma$, let $q_i = \min(1, u_i^{p/2})$ 
        and normalize to obtain $p_i = q_i/\sum_j q_i$
        \item \label{step:get_A_prime}
        Obtain $\ma'$ by sampling $N = O(d\log d)$ i.i.d. rows of $\ma$ from the $p_i$ distribution (each rescaled by $(N p_i)^{-1/p}$, where the scaling is accurate up to a constant factor) 

        \item Compute the Lewis quadratic form for $\ma'$ (see \cref{def:lewisweights}) and return it.
    \end{enumerate}
\end{framed}
\caption{The recursive sampling algorithm from \cite{cohen2015lp}, specialized to the range $1\leq p\leq 2$.  When $p=2$ this algorithm is essentially the repeated halving algorithm of \cite{cohen2015uniform}.}
\label[alg]{alg:recursiveSamplingCohenPeng}
\end{figure}

\thmEllpRegressionLessThanTwo*
\begin{proof}
    We use the recursive Lewis weight sampling algorithm of \cite{cohen2015lp}, a version of which is reproduced in \cref{alg:recursiveSamplingCohenPeng}.  First we note a few minor modifications to their original algorithm.  We consider only $1\leq p \leq 2$, which as noted in \cite{cohen2015lp} avoids the extra $\poly(d)$ factor in the number of row samples required per recursive call. Further, constant-factor approximations to the sampling probabilities in \cref{step:get_A_prime} suffice, following the original analysis in \cite{cohen2015lp} who show that computing constant-factor approximations to the Lewis weights suffices in each recursive call; constant-factor errors on the row-norms of $\ma$ have precisely the same effect on the Lewis quadratic form. 

    To modify this algorithm for the distributed setting, we note that uniformly sampling rows is easy for the coordinator to simulate.  Given the number of rows on each server, the coordinator samples a random collection of row indices (but does not request the actual rows).

    To implement the sampling described in \cref{step:ComputeUiLewisFormApprox} through \cref{step:get_A_prime}, we note that this is precisely the setting of our \cref{prob:rel_lev_score_est} with $r = O(d\log d), T = O(d)$ and $\eps = O(1)$, where the latter statement is given in the proof of  \cite[Lemma 3.2]{cohen2015lp}.  Note that \cref{alg:relativeLevScoreSampling}
    is run once per recursive call, of which there are at most $O(\log(dn)).$ 

    This yields a constant factor approximation to the Lewis quadratic form on the coordinator side.  Then we apply  \cref{alg:relativeLevScoreSampling} to
    sample (and rescale) an additional $O({d}{\eps^{-2}}(\log n + \log{\delta^{-1}}))$ rows.  To rescale appropriately, we additionally estimate the sampling probabilities to within $1\pm \eps$ error. This lets us obtain an $\eps$ distortion $\ell_p$ embedding for the column space of $\ma$.  That the resulting subspace embedding is correct with probability $1-\delta$ for these parameters follows from \cite{woodruff2023online}.

\end{proof}

\subsubsection{Proof of Main Result for $\ell_2$-Regression with Known Conditioning}
Finally for $\ell_2$ regression, we observe that we can improve the bit complexity of our protocol provided that the system is known to be well conditioned.

To do this, we first note that the recursive halving algorithm of \cite{cohen2015uniform} (effectively \cref{alg:recursiveSamplingCohenPeng} for $p=2$) works equally well for obtaining a spectral approximation to $\ma^\top\ma + \lambda \mathbf{I}$.  One simply runs the recursive leverage sampling algorithm on the matrix $[\ma; \sqrt{\lambda}\mathbf{I}].$ If run directly the intermediate spectral approximations could be poorly conditioned.  To fix this, when (recursively) running the ``uniform sampling step" we simply insist on additionally sampling all rows in the $\sqrt{\lambda}\mathbf{I}$ block with probability $1.$  This gives a variant of the recursive leverage sampling procedure for which all intermediate spectral approximations have smallest eigenvalue lower bounded by $\lambda.$  We make use of this algorithm below.

\thmLowAccuracyLinearRegression*
\begin{proof}
We use a slight regularization to modify our protocol for computing a constant factor spectral approximation to $\ma\in\R^{n\times d}$. 
Let $\ma^\prime=[\ma; \sqrt{\lambda} \mathbf{I}_d]\in\R^{(n+d)\times d}$ be the matrix formed by vertically concatenating $\ma$ and $\sqrt{\lambda} \mathbf{I}_d$, where $\lambda>0$ is chosen with $\sqrt{\lambda} \leq \sigma_{\min}(\ma)/2.$  Then $\frac{1}{2}\ma^\top \ma \preceq {(\ma^\prime)}^\top (\ma^\prime) \preceq \frac{3}{2}\ma^\top \ma$. Hence, to spectrally approximate $\ma$, it suffices to find a constant-factor spectral approximation for $\ma^\prime.$

To do this we run our algorithm 
discussed in \cref{thm:ell2_subspace_embedding} above, but use the repeated halving algorithm above for obtaining a spectral approximation to $\ma^\prime.$ This simply requires a procedure for solving \cref{prob:rel_lev_score_est} under the assumption that $\mm\mm^\top$ has largest eigenvalue at most $1/\lambda.$  To replace the $sdL$ dependence with $sd$ in \cref{lem:estimate_relative_leverage_scores} we  need to communicate only the JL sketches $\ms^{(j)} \ma^{(j)}$ more efficiently.

Rather than communicating the sketches with bit-precision, suppose that the servers instead communicate $\ms^{(j)} \ma^{(j)} + \mathbf{E}^{(j)}$ where $\mathbf{E}^{(j)}$ represents some small error.  As in \cref{lem:estimate_relative_leverage_scores}, the coordinator will then compute
\[
\sqrt{\tildeBj} = \norm{(\ms^{(j)}\ma^{(j)} + \mathbf{E}^{(j)}) \mm}{F} = \norm{\ms^{(j)}\ma^{(j)}\mm}{F} + \norm{\mathbf{E}^{(j)}\mm}{F}.  
\]
We would like this quantity to approximate $\norm{\ms^{(j)}\ma^{(j)}\mm}{F}$ to within a constant factor, so it suffices to ensure that 
\[
\norm{\mathbf{E}^{(j)}\mm}{F} \leq \frac{1}{10}\norm{\ms^{(j)}\ma^{(j)}\mm}{F}.
\]
Note that $\norm{\mathbf{E}^{(j)}\mm}{F} \leq \sigma_{\max}(\mm)\norm{\mathbf{E}^{(j)}}{F}$ and that $\norm{\ms^{(j)}\ma^{(j)}\mm}{F} \geq \sigma_{\min}(\mm) \norm{\ms^{(j)}\ma^{(j)}}{F}.$  So we just need
\[
\norm{\mathbf{E}^{(j)}}{F} \leq \frac{\sigma_{\min}(\mm)}{\sigma_{\max}(\mm)} \norm{\ms^{(j)}\ma^{(j)}}{F} = \frac{1}{\kappa(\mm)} \norm{\ms^{(j)}\ma^{(j)}}{F}.
\]

To accomplish this bound on the term, it is sufficient to communicate each entry of $\ms^{(j)}\ma^{(j)}$ to within $1\pm \kappa(\mm)^{-1}$ multiplicative error.  Since $\mm^\top \mm$ is obtained by subsampling from $\ma^\prime,$ note that $\lambda_{\max}(\mm^\top\mm) \leq \lambda_{\max}(\ma^\top\ma) + \lambda.$  Therefore $\kappa(\mm) \leq 1 + \frac{\sigma_{\max}(\ma)}{\lambda}.$

To ensure $\lambda$ is chosen small enough, set $\lambda = \norm{\ma}{F}/(2 \kappa \sqrt{d}) \leq \sigma_{\min}(\ma)/2.$  Note also that $\sigma_{\max}(\ma)/\lambda = \frac{2\kappa\sqrt{d}\sigma_{\max}(\ma)}{\norm{\ma}{F}} \leq 2\kappa^2,$ which implies that $\kappa(\mm) \leq 1 + 2\kappa^2.$

We now round each entry of $\ms^{(j)}\ma^{(j)}$ to the nearest value in
\[\{\pm (1+\kappa(\mm)^{-1})^j : j\in \Z\}.\]  
To communicate one of these numbers, we send only the sign and exponent.  Moreover our $\ms^{(j)}$ sketch can be taken to be Rademacher, so the entries of $\ms^{(j)}\ma^{(j)}$ can be expressed with $L + \log d$ bits. Hence we need only $O(\log(\kappa(\mm)L) = O(\log(\kappa L))$ bits per entry to send our approximation of $\ms^{(j)}\ma^{(j)}$. The improvement of $1/\eps^2$ to $1/\eps$ follows from Sarlos's trick mentioned in \cref{thm:ell2_subspace_embedding}
\end{proof}

\subsection{Second Protocol: Block Leverage Score Sketching}
\label{sec:ell2_subspace_embedding_via_block_lev_scores}
\looseness=-1In this section we use block leverage scores to design an algorithm for constructing an $\ell_2$ subspace embedding. The coordinator  needs to send only $O(\log d)$ bits to the servers over the course of $d$ rounds.  We summarize the result here which will be proven in the following subsections.

\begin{theorem}
    \label{thm:nearly_one_way_protocol}
    There is a protocol that constructs an $(1\pm\eps)$ distortion $\ell_2$ subspace embedding for $\ma$ which runs in $O(\log d)$ rounds, and uses $\tildeO(sd + \eps^{-2}d^2)$ communication.  Moreover the servers collectively  receive  a total of  only $\tildeO(s)$ bits from the coordinator.
\end{theorem}

\begin{proof}
\cref{thm:block_lev_estimation_analysis} in the following section will give a protocol for estimating the block leverage scores of $\ma$ to within constant factors.  Given such estimates of the block leverage scores, \cref{thm:sketched_block_leverage_sampling} gives a sampling algorithm to compute a subspace embedding for $\ma.$  Combining these algorithms gives the guarantee stated here.
\end{proof}

As an corollary, we can obtain slightly improved $\eps$ dependence for $\ell_2$ regression.
\begin{corollary}
\label{thm:ell2_subspace_embedding}
There is a protocol using $\tildeO(sdL + d^2{\eps^{-2}}L)$ communication that allows the coordinator to produce an $\eps$-distortion $\ell_2$ subspace embedding for the column span of $\ma$, with a failure probability of at most $\delta.$  As a consequence, the linear regression problem can be solved with communication $\tildeO(sdL + d^2 \eps^{-1}L).$
\end{corollary}
\begin{proof}
    This  follows immediately  from the $\ell_p$ result in \cref{thm:ell1_subspace_embedding}.     
    However, the above procedure requires a $ d\eps^{-2}$ dependence to produce a subspace embedding for $\ma.$  For simply solving the linear regression problem to $1+\eps$ multiplicative accuracy, we note that a trick due to Sarlos \cite{sarlos2006improved} allows this to be improved to $\frac{d}{\eps}$.
    
    An argument for this is given in the proof of  \cite[Theorem 3.1]{w14} when the subspace embedding is a random sign matrix.  However, the same proof applies for leverage score sampling matrices, by using  \cite[Lemma 7.3]{clarkson2017low}, which shows that leverage score sampling matrices yield an approximate-matrix-multiplication guarantee (and noting that if $\ma = \mathbf{U} \mathbf{\Sigma} \mathbf{V^T}$ is the  singular value decomposition of $\ma$, then the leverage scores of $\mathbf{U}$ and $\ma$ coincide).

    In fact, it is true more generally (although not widely known) that an $O(\sqrt{\eps})$-distortion subspace embedding suffices to solve the $\ell_2$-regression problem to $1 + \eps$ accuracy \cite{bourgain2015toward}.
\end{proof}

\subsubsection{Block Leverage Scores: Definition and Basic Properties}
\looseness=-1Our second protocol is based on the notion of {\it block leverage scores}, which may be of independent interest.  The block leverage score is simply the sum of the leverage scores of the rows in a block (i.e., matrix stored on a server), where the leverage score of each row is with respect to the entire matrix. Intuitively, therefore, the block leverage score captures the importance of the block in the overall rowspace of the matrix $\ma$. It is known~\cite{cohen2015uniform} that the sum of the (block) leverage scores is at most the rank of the matrix (and hence by $d$ in our setting). 

\looseness=-1Our key technical result is that a simple sketch suffices to estimate the block leverage scores of $\ma = [\ma^{(1)};\ldots; \ma^{(s)}] $.  Our approach is to sketch each block down to roughly an $O(k) \times d$ matrix using a separate Rademacher (or other) sketch $\ms^{(i)}$ for each block. The block leverage scores are then estimated to be the those of $[\ms^{(1)}\ma^{(1)}; \ldots ; \ms^{(s)}\ma^{(s)}]$. Unfortunately, this does not necessarily yield good estimates for all blocks.  Indeed the block leverage scores of the sketched matrix are all bounded by $O(k)$, so we may underestimate the scores of outlying blocks.  However, by introducing a characterization of the block leverage scores as a \textit{block sensitivity}, we show that we obtain good (over-)estimates for all block leverage scores smaller than $Ck.$ Additionally, we can detect those blocks for which our estimates are not good; their leverage score estimates are guaranteed to be larger than $Ck.$ These observations yield a simple iterative procedure for computing block leverage score estimates. 

\paragraph{Estimating Block Leverage Scores.} 
\looseness=-1 In the first round, the coordinator requests a roughly $1\times d$ sketch from each block, yielding good estimates of the block leverage scores for all blocks with block leverage score at most $1.$  The coordinator now needs to focus only on those blocks with estimated block leverage score larger than $1$, of which there are at most $O(d)$, since the block leverage scores sum to at most $d$.  The coordinator then requests roughly a $2\times d$ sketch from each server whose block leverage score estimate is large (where the estimate came via the sketch) 
This yields good estimates for blocks with leverage score at most $2$ and underestimates for those with leverage scores larger than $2$, of which there are at most $d/2.$  This procedure is repeated, doubling the number of rows requested in each round while halving the number of servers from which these rows are requested. More generally, in round $r$ the server requests a sketch of size $2^r \times d$ from each of approximately $d/2^r$ servers. So the procedure requires $O(d^2)$ communication per round, with a cost of $O(sd)$ in the very first round. After $O(\log d)$ rounds, we find good estimates for all the blocks, and hence use $\tildeO(sd + d^2)$ communication.

Given estimates for the block leverage scores, a ``block'' version of the standard leverage score sampling algorithm suffices to construct a subspace embedding for the column span of $\ma$:  We choose blocks proportional to the estimated block leverage scores and then sample a $1 \times d$ sketch of the rows from that block.  By taking $d\eps^{-2}\log d$ such samples
we obtain our desired subspace embedding with distortion $\eps.$ The entire algorithm (estimating the block leverage scores and sampling) can be implemented simultaneously by having each server send twice as many rows during the leverage score estimation algorithm.  The extra rows can then be used later during the sampling phase.  Hence the algorithm is nearly one-way in the sense that coordinator needs to communicate only $O(s + d)$ bits in total, and only for the purposes of notifying the servers that are active in that round.

\begin{definition}
Let $\ma = [\ma^{(1)};\ldots; \ma^{(s)}].$  We define the block leverage score of block $\Ai$ to be 
\begin{equation}
    \blev_i(\ma) = \Tr\left(\Ai (\ma^\top \ma)^{-1} (\Ai)^\top\right).
\end{equation}
\end{definition}

For use throughout, we list a few basic properties of the block leverage scores.

\begin{proposition}
\label{prop:block_lev_properties}
Let $\ma=[\ma^{(1)};\ldots; \ma^{(s)}].$ The following properties hold:
\begin{enumerate}
\item If $\Ai \in \R^{k\times d}$, and the rows of $\Ai$ have leverage scores $\ell_{i1},\ldots, \ell_{ik},$ (computed with respect to $\ma$) then the $i^\mathrm{th}$ block leverage score $\blev_i(\ma) = \sum_{j=1}^k \ell_{ij}.$ 
\item The sum of all block leverage scores of $\ma$ satisfies $\sum_{i=1}^s \blev(\Ai) = \rank(\ma).$
\item Consider the matrix $\widetilde{\ma} = [\Aone; \ldots; \ma^{(s)}; \ma^{(s+1)}].$ For all $i\in [s]$, we have $\blev_i(\ma) \geq \blev_i(\widetilde{\ma}).$
\end{enumerate}
\end{proposition}

\begin{proof}
For property 1, we observe that 
\[
\blev_i(\ma) = \Tr(\Ai (\ma^\top \ma)^{-1} {\Ai}^\top) = \sum_{j=1}^k \mathbf{A}^{(i)}_{j:} (\ma^\top \ma)^{-1} {\mathbf{A}^{(i)}_{j:}}^\top = \sum_{j=1}^k\ell_{ij}.
\]
In light of property 1, properties 2 and 3 follow from the corresponding facts for classical leverage scores.
\end{proof} We also give a characterization of the block leverage score as a block sensitivity,
which will be the key to our analysis of the block leverage score sketch in the following section.
The \emph{sensitivity sampling} framework was introduced by \cite{langberg2010universal} and was shown recently in \cite{wy2023ICML} to be advantageous over Lewis weights sampling \cite{cohen2015lp} in numerous settings such as under small total sensitivity and structured data matrix (such as low-rank, sparsity, etc. as studied in, e.g., \cite{meyer2022fast}). Algorithms for approximating $\ell_p$ sensitivities have been improved by works such as \cite{padmanabhan2023computing}. 

\begin{proposition}
\label{prop:lev_score_as_sensitivity}
Given a full column rank matrix $\ma$ consisting of blocks $\ma^{(1)}, \ldots, \ma^{(s)}$, we have
    \[
     \blev_i(\ma) = \sup_{\mathbf{X}}   \frac{\|\ma^{(i)} \mathbf{X}\|_F^2}{\|\ma\mathbf{X}\|_2^2},
    \]
where the supremum is is over all matrices with dimensions compatible with $\ma$.
\end{proposition}

\begin{proof}
    Let $\mathbf{U} \mathbf{D} \mathbf{V}^\top$ be the singular value decomposition for $\ma$, and let $\mathbf{U}_j$ have a subset of the rows of $\mathbf{U}$ so that $\mathbf{U}_j \mathbf{D} \mathbf{V}^\top = \ma^{(j)}$. We are interested in maximizing $\|(\mathbf{U}_j \mathbf{D} \mathbf{V}^\top) \mathbf{X}\|_F^2 = ||\mathbf{U}_j (\mathbf{D} \mathbf{V}^\top \mathbf{X})\|_F^2$ subject to $\|\mathbf{U} \mathbf{D} \mathbf{V}^\top \mathbf{X}\|_2=1$.  Since $\mathbf{U}$ is orthonormal, the constraint becomes $\|\mathbf{D} \mathbf{V}^\top \mathbf{X}\|_2 = 1$. $\mathbf{D}\mathbf{V}^\top$ has full rank so the optimization problem is equivalent to maximizing $\|\mathbf{U}_j \mathbf{Y}\|_F^2$ s.t. $\|\mathbf{Y}\|_2 \leq 1$.  This is optimized for $\mathbf{Y}=\mathbf{I}$ and the objective is the sum of squares of row norms for $\mathbf{U}_j$ which is the sum of leverage scores of the rows of $\ma^{(j)}$.
\end{proof}

\subsubsection{Sketching Block Leverage Scores}\label{sec:sketchingBlockLevScores}

We use our sensitivity characterization of the block leverage scores to show that sketching a block does not cause its leverage score to drop too much.

\begin{lemma}\label{lem:sketchingABlockDoesNotDropLevScoreMuch}
Let $\mg^{(1)}$ be a sketching matrix which is an $O(1)$ distortion oblivious subspace embedding for a $k$-dimensional subspace with a probability of at least $1-\delta$. Then, with a probability of at least $1-\delta$, we have that 
\[
\blev_1\left([\mg^{(1)}\ma^{(1)}, \ma^{(2)}, \ldots, \ma^{(s)}]\right) \geq C \min(k, \blev_1(\ma)).
\]
\end{lemma}

\begin{proof}
  By \cref{prop:lev_score_as_sensitivity}, we can choose an $\mx$ with $\|\ma^{(1)} \mx\|_F^2 / \|\ma\mx\|_2^2 = \blev_1(\ma)$.  \cite[Theorem 1]{cohen2015uniform} implies that 
      \[\|\mg^{(1)} \ma^{(1)} \mx\|_2^2 \lesssim \|\ma^{(1)} \mx\|_2^2 + (1/k) \|\ma^{(1)} \mx\|_F^2,
      \] 
      and so
     \[
     \frac{\|\mg^{(1)} \ma^{(1)} \mx\|_2^2}{\|\ma\mx\|_2^2} \lesssim 1 + \frac1k \frac{\|\ma^{(1)} \mx\|_F^2}{ \|\ma\mx\|_2^2} = 1 + \frac{\blev_1(\ma)}{k}.
     \] Then 
\[
\|\mg^{(1)} \ma \mx\|_2 \lesssim ||\mg^{(1)} \ma^{(1)} \mx||_2 + ||\ma\mx||_2 \lesssim \left(1 + \frac1k \blev_1(\ma)\right)||\ma\mx||_2.
\]
Hence 
\[
     \frac{\|\mg^{(1)} \ma^{(1)} \mx\|_F^2}{ \|\mg^{(1)} \ma \mx\|_2^2} \gtrsim \frac{\|\ma^{(1)} X\|_F^2}{\|\mg^{(1)} \ma\mx\|_2^2} \gtrsim \min(k, \blev_1(\ma)),
\]
by Johnson-Lindenstrauss, and the previous bound.  Hence $\mx$ witnesses a sensitivity of at least $C \min(k, \blev_1(\ma))$ for the first block of $\mg^{(1)} \ma$ as desired.
\end{proof}

Next we analyze the situation where all but one block is sketched. To streamline the argument we first lead with a few elementary claims.

\begin{proposition}
\label{prop:simple_inv_inequality}
Let $\mx\in \R^{d\times d}$ be positive semidefinite, let $\mathbf{U}\in \R^{d\times m}$, and suppose that $\Tr(\mathbf{U}^\top \mx \mathbf{U}) \leq \Tr(\mathbf{U}^\top \mathbf{U}).$ Then $\Tr(\mathbf{U}^\top \mx^{-1} \mathbf{U}) \geq \Tr(\mathbf{U}^\top \mathbf{U}).$ 
\end{proposition}

\begin{proof}
Since $\mx$ and $\mx^{-1}$ are simultaneously diagonalizable, the L\"oewner order inequality $\mx + \mx^{-1} \geq 2\mathbf{I}$ follows from the scalar inequality $x + 1/x \geq 2$ for $x\geq 0.$ Thus $\mathbf{U}^\top (\mx+\mx^{-1}) \mathbf{U} \geq 2\mathbf{U}^\top \mathbf{U}$ and so $\Tr(\mathbf{U}^\top(\mx+\mx^{-1})\mathbf{U}) \geq 2\Tr(\mathbf{U}^\top \mathbf{U}).$  Therefore by the assumption $\Tr(\mathbf{U}^\top \mx \mathbf{U})\leq \Tr(\mathbf{U}^\top \mathbf{U})$, we have 
\[
\Tr(\mathbf{U}^\top \mx^{-1} \mathbf{U}) \geq 2\Tr(\mathbf{U}^\top\mathbf{U}) - \Tr(\mathbf{U}^\top \mx \mathbf{U})
\geq 2\Tr(\mathbf{U}^\top \mathbf{U}) - \Tr(\mathbf{U}^\top\mathbf{U})
= \Tr(\mathbf{U}^\top \mathbf{U}).
\qedhere
\]
\end{proof}

\begin{proposition}
\label{prop:prob_lowener_inv}
Let $\mx$ be a random $d\times d$ matrix which is a.s. PSD, and let $\ma$ be fixed matrix which is PSD and non-singular. Suppose that for every $\mathbf{U}$ in $\R^{d\times m}$ it holds that 
\[
\pr\left(\Tr(\mathbf{U}^\top \mx \mathbf{U}) \leq \Tr(\mathbf{U}^\top \ma \mathbf{U}) \right) \geq 1-\delta.
\]
Then for every $\mathbf{V}$ in $\R^{d\times m}$ it also holds that
\[
\pr\left(\Tr(\mathbf{V}^\top \mx^{-1} \mathbf{V}) \geq \Tr(\mathbf{V}^\top \ma^{-1} \mathbf{V}) \right) \geq 1-\delta.
\]
\end{proposition}

\begin{proof}
Plugging $\ma^{-1}\mathbf{V}$ into the hypothesis gives that for all $\mathbf{V}$,
\[
\pr(\Tr(\mathbf{V}^\top \ma^{-1} \mx \ma^{-1} \mathbf{V}) \leq \Tr(\mathbf{V}^\top \ma^{-1} \mathbf{V}) ) \geq 1-\delta,
\]
or equivalently
\[
\pr\left( \Tr((\ma^{-1/2}\mathbf{V})^\top \ma^{-1/2} \mx \ma^{-1/2}(\ma^{-1/2}\mathbf{V})) \leq \Tr((\ma^{-1/2}\mathbf{V})^\top(\ma^{-1/2} \mathbf{V}))  \right) \geq 1-\delta.
\]
By \cref{prop:simple_inv_inequality}, this gives that for all $V$,
\[
\pr\left( \Tr((\ma^{-1/2}\matv)^\top \ma^{1/2} \mx^{-1} \ma^{1/2}(\ma^{-1/2}\mv)) \geq \Tr((\ma^{-1/2}\matv)^\top(\ma^{-1/2} \matv)) \right) \geq 1-\delta,
\]
which simplifies to the desired conclusion.
\end{proof}

\begin{lemma}
\label{lem:sketch_preserves_row_score}
Let $\ma = [\ma^{(1)}; \ldots; \ma^{(s)}]$ be non-singular and let $\ms^{(1)},\ldots \ms^{(s)}$ be random sketching matrices of appropriate dimension so that the products $S^{(i)} A^{(i)}$ are defined. Assume that each $S^{(i)}$ satisfies the $(1, \delta,2)$-JL-moment property. Let $\matv\in\R^{d\times m}$. Then with probability at least $1-\delta$,
\[
\Tr\left(\matv^\top\left(\sum_{i=1}^s {\ma^{(i)}}^\top {\ms^{(i)}}^\top \ms^{(i)} \ma^{(i)}\right)^{-1} \matv \right)
\geq 
\frac12 \Tr\left(\matv^\top (\ma^\top \ma)^{-1}\matv\right).
\]

\end{lemma}

\begin{proof}
We apply \cref{prop:prob_lowener_inv}. Let $\mU\in \R^{d\times m}$ be an arbitrary fixed matrix.  Then we have
\[
\Tr\left(\mathbf{U}^\top \left(\sum_{i=1}^s {\ma^{(i)}}^\top {\ms^{(i)}}^\top \ms^{(i)} \ma^{(i)}\right) \mU\right) = \sum_{i=1}^s \norm{\ms^{(i)} \ma^{(i)} \mU}{F}^2.\]

The block matrix $\ms^{(1)} \oplus \cdots \oplus \ms^{(s)}$ also has the $(1,\delta,2)$-JL-moment property (see for example Lemma 13 of \cite{ahle2020oblivious}). So with probability at least $1-\delta$, \[\sum_{i=1}^s \norm{\ms^{(i)} \ma^{(i)} \mU}{F}^2 \leq 2 \norm{\ma\mU}{F}^2 = 2 \Tr(\mathbf{U}^\top \ma^\top \ma \mU).\]
The claim now follows by \cref{prop:prob_lowener_inv}.
\end{proof}

\begin{lemma}
\label{lem:all_but_one_sketched_block_lev_lower_bound}
Let $\ma = [\ma^{(1)}; \dots, \ma^{(s)}]$, let $\ms^{(1)},\ldots, \ms^{(s)}$ be sketching matrices satisfying the $(1,\delta,2)$-JL-moment property, and let $\widetilde{\ma} = [\ms^{(1)}\ma^{(1)};\ldots ;\ms^{(s)} \ma^{(s)}],$ where $\ms_k=\mathbf{I}$ for a fixed $k$. Then with probability at least $1-\delta$,
$\blev_k(\widetilde{\ma}) \geq \frac{1}{2}\blev_k(\ma).$
\end{lemma}

\begin{proof}
By \cref{lem:sketch_preserves_row_score} (and the hypothesis that $\ms^{(k)} = \mathbf{I}$), we have
\[
\blev_k(\widetilde{\ma}) 
= \Tr\left(\ma^{(k)} \left(\sum_{i=1}^s {\ma^{(i)}}^\top {\ms^{(i)}}^\top {\ms^{(i)}}^\top \ma^{(i)} \right)^{-1} \ma^{(k)} \right)
\geq \frac{1}{2} \Tr\left(\ma^{(k)} \left(\sum_{i=1}^s {\ma^{(i)}}^\top \ma^{(i)} \right)^{-1} \ma^{(k)} \right)
= \frac{1}{2}\blev_k(\ma).
\]
\end{proof}

Combining the two block sketching results \cref{lem:sketchingABlockDoesNotDropLevScoreMuch} and \cref{lem:all_but_one_sketched_block_lev_lower_bound} gives the following.

\begin{lemma}
\label{lem:sketched_block_leverage_lower_bound}
Let $\ms^{(1)},\ldots, \ms^{(s)}$, each with $d$ columns, all be (normalized) Rademacher with $O(k\log(s/\delta))$ rows. Then for each $i$, 
\[
\blev_i([\ms^{(1)}\ma^{(1)}; \ldots ; \ms^{(s)} \ma^{(s)}]) 
\geq C \min\left(k, \blev_i([\ma^{(1)}; \ldots ; \ma^{(s)}])\right),
\]
with probability at least $1-\delta$.
\end{lemma}

\begin{remark}
The sketches in the above result were taken to be Rademacher only for convenience.  The same argument applies to sparse sketches for example.
\end{remark}

\subsubsection{Estimating Block Leverage Scores}

The sketch from the previous section shows that we can accurately (over-)estimate a given block leverage score by sketching down to dimension roughly $k.$  Unfortunately the block leverage scores can be as large as $d$ and we are unable to take a sketch of $d$ rows from all servers.  Fortunately, not many servers can have large block leverage score, so by iteratively pruning off the ones that don't, we can focus on the servers with the most information. 

\begin{figure}[!htb]

\begin{framed}

\textbf{Input. } For $i\in [s]$, each server $i$ has the block $\ma^{(i)}$ of $\ma$

\vspace{.5em}

\textbf{Output. } List $L$ where $L[i]$ estimates of the block leverage score of $\ma^{(i)}$

\vspace{.5em}

\textbf{Initialize.} $L \gets [\perp, \ldots, \perp]$ of length $s$, $\mathcal{S}_0 \gets \{1,\ldots, s\}$

\vspace{.5em}

\textbf{For} iterations $r = 0, 1, 2, \dotsc, \lceil \log d\rceil$: 

\begin{enumerate}[itemsep = .1em, leftmargin = 1.7em, topsep = .4em, label=\protect\circled{\arabic*}]       
        \item  $k_r\gets 2^r$

        \item \textbf{For } $i \in \mathcal{S}_r$
        
        \begin{enumerate}[itemsep = .1em, leftmargin = 1.7em, topsep = .4em] 
        
        	\item  Each server $i$ draws $\ms_{r,i}$ a constant distortion $k_r$-dimensional oblivious subspace embedding, and sends $\ms_{r,i}\ma^{(i)}$ to the coordinator
        
        \end{enumerate} 
        
        \item   Coordinator forms block matrix $\ma_{\langle r \rangle}$ with blocks given by $\ms_{r,i}\ma^{(i)}$ for $i$ in $\mathcal{S}_r$, and where the blocks are indexed by $\mathcal{S}_r$ 
        
        \item For all $i\in \mathcal{S}_r$, coordinator computes $\widehat{\mathcal{L}}_{r,i} = \blev_i(\ma_{\langle r \rangle})$
        
        \item  $\mathcal{S}_{r+1} \gets \{i \in \mathcal{S}_r : \widehat{\mathcal{L}}_{r,i} \geq C k_r\}$
        
        \item \textbf{For } $i$ in $\mathcal{S}_r \setminus \mathcal{S}_{r+1}$ 
        
         \begin{enumerate}[itemsep = .1em, leftmargin = 1.7em, topsep = .4em] 
        
        	\item   $L[i] \gets \widehat{\mathcal{L}}_{r,i}$
        	
        \end{enumerate} 
        
\end{enumerate}

\textbf{For} all $i\in[s]$, if $L[i] = \perp$, then set $L[i]\gets d$ 

\textbf{Return:} $L$. 

\end{framed}
\caption{Block leverage score estimation.}
\label[alg]{alg:block_lev_est}
\end{figure}

\begin{theorem}
\label{thm:block_lev_estimation_analysis}
\cref{alg:block_lev_est} 
runs with $O(\log d)$ rounds of communication, and returns a list $L$ satisfying 
\begin{enumerate}[label=(\roman*)]
    \item $L[i] \geq C \blev_i([\ma^{(1)};\ldots; \ma^{(s)}])$ for all $i$
    \item $\sum_{i=1}^s L[i] \leq O(d\log d)$
\end{enumerate}
Moreover the servers collectively send at most $O(\sketchfactor s + \sketchfactor d\log d)$ vectors of length $d$ to the coordinator.
\end{theorem}

\begin{proof}

We start by bounding the number of servers which are active in a given round.  In round $0$, $|\mathcal{S}_0| = s.$  For $r\geq 1$, note that for every server $i$ in $\mathcal{S}_r$, $\widehat{\mathcal{L}}_{r-1,i} \geq C k_{r-1}.$  On the other hand there cannot be many such servers, since by \cref{prop:block_lev_properties},
\[
\sum_{i\in \mathcal{S}_{r}} \widehat{\mathcal{L}}_{r-1,i}
\leq d,
\]
which implies that $|\mathcal{S}_r| \leq \frac{d}{C k_{r-1}}.$ This immediately gives a bound on the communication cost. Summing the number of vectors transmitted in each round gives a total of
\[
\sum_{i=0}^{\lceil \log d \rceil} \sketchfactor k_r |\mathcal{S}_r|
\leq \sketchfactor \left(sk_0 + \sum_{i=1}^{\lceil \log d \rceil} k_r\frac{d}{C k_{r-1}} \right)
= c\left(s + \frac{2d}{C} \lceil \log d \rceil\right)
\]
vectors sent to the coordinator.  Next we show that (i) holds.  By the algorithm, note that either $L[i]=d$ or on round $r$ we set $L[i] = \widehat{\mathcal{L}}_{r,i}$.  In the first case (i) is trivial since all block leverage scores are at most $d.$ In the latter case, $\widehat{\mathcal{L}}_{r,i} \leq C k_r.$  But $\widehat{\mathcal{L}}_{r,i} \geq C \min\left(k_r, \blev_i(\ma_{\langle r \rangle})\right)$ by \cref{lem:sketched_block_leverage_lower_bound}, which is at least $\min(Ck_r, C \blev_i(\ma))$ by monotonicity.
So we have $k_r \geq \blev_i(\ma),$ which implies that \[
\widehat{\mathcal{L}}_{r,i} \geq C \min(k_r, \blev_i(A)) = C \blev_i(\ma).
\] Finally we show (ii). Since we have
\[
\sum_{i \in \mathcal{S}_r} \widehat{\mathcal{L}}_{r,i}
= \sum_{i \in \mathcal{S}_r} \mathcal{L}_i(\ma_{\langle r \rangle})
\leq d,
\]
it follows that the entries of $L$ which are set in round $r$ sum to at most $d$.  Hence the sum of the entries of $L$ set in the outer for-loop is at most $(\lceil \log d \rceil + 1)d.$  By the argument given above for the communication cost, there are at most $\frac{d}{C k_r}\leq \frac{1}{C}$ entries of $L$ which are not set after the loop.
These entries are set to $d$, which gives
$\sum_{i=1}^s L[i] \leq (\lceil \log d \rceil + 1)d + \frac{d}{C}
\leq O(d\log d).$ 
\end{proof}

\subsubsection{Sampling via Block Leverage Scores}
Given the overestimates computed for the block leverage scores in the previous section, a straightforward concentration bound allows to get a spectral approximation via block leverage sampling. By combining with the algorithm for estimating the block leverage scores, this immediately yields an algorithm with $\tilde{O}(sd + d^2\eps^{-2})$ communication in the coordinator model, for computing an $\eps$ distortion subspace embedding for the columns of $\ma.$

\begin{figure}[!htb]

\begin{framed}

\textbf{Input. } Sampling distribution $p$, number of samples $N$

\vspace{.5em}

\textbf{Output. } A spectral approximation $\widehat{\ma}$ of $\ma$

\vspace{.5em}

\textbf{Initialize.} Coordinator sets $\widehat{\ma}=\mathbf{0}$ $\in\R^{N\times d}$

\vspace{.5em}

\textbf{For} iterations $\ell = 1, 2, \dotsc, N$: 

\begin{enumerate}[itemsep = .1em, leftmargin = 1.7em, topsep = .4em, label=\protect\circled{\arabic*}]       
        \item\label{item:sampleServerFromDist} Sample server $j$ from the distribution $p$
        \item\label{item:serverJGeneratesRademacher}   Server $j$ generates a Rademacher random vector $\vg\in \R^{m_j}$ and sends $\vg^\top \ma^{(j)}$ to coordinator
        \item\label{item: coordinatorAppendsvgTopMa} Coordinator appends $\vg^\top \ma^{(j)}$ to $\widehat{\ma}$ 
        
\end{enumerate}
    
\textbf{Return:} $\widehat{\ma}$. 

\end{framed}
\caption{Block leverage score sampling.}
\label[alg]{alg:block_lev_sampling}
\end{figure}

As is standard for analyses of leverage score sampling, we rely on the Matrix Chernoff bound (see \cite{w14} for example).  We state a version here which follows from \cite{tropp2012user}.  The version we use is slightly less general, but more convenient for our purposes.

\begin{theorem}
\label{thm:matrix_chernoff}
Let $\mx_1,\ldots, \mx_d \in \R^{d\times d}$ be random matrices which are independent and symmetric PSD, with $\mu_{\min}I \leq \E \mx_i \leq \mu_{\max}I,$ and $\norm{\mx_i}{} \leq R$ a.s.  Let $\overline{\mx} = \frac{1}{N}\sum_{i=1}^N \mx_i.$ Then for all $\delta\in [0,1),$ \[
\pr\parens{\lambda_{\max}(\overline{\mx}) \geq (1+\delta)\mu_{\max}} \leq d \exp\left(-\delta^2 \frac{N\mu_{\max}}{3R}\right).
\]

\[
\pr\parens{\lambda_{\max}(\overline{\mx}) \leq (1-\delta)\mu_{\min}} \leq d \exp\left(-\delta^2 \frac{N\mu_{\min}}{2R}\right).
\]
\end{theorem}

The quantities that we apply the matrix Chernoff bound to will have operator norm given by a Hutchinson trace estimator \cite{hutchinson1989stochastic}. This standard application of Matrix Chernoff is the core argument.  The additional work simply fixes a technical issue.

Hutchinson's estimator may very occasionally be much larger than expected, which would require $R$ in \cref{thm:matrix_chernoff} to be undesirably large.  Fortunately Hutchinson's estimator has exponential tail decay, and so these potential large values may be safely ignored with high probability. To make this precise, we will use the following technical fact, which is effectively a restatement of results in \cite{dharangutte2021dynamic}.

\begin{proposition}
\label{prop:hutchinson_subexponential_tail}
Let $\ma\in \R^{d\times d}$ be symmetric positive semidefinite, and let $g\in\R^d$ be a Rademacher random vector.  Let $\mu = \E(\vg^\top \ma \vg) = \Tr(\ma).$ Then
\[
\pr(\vg^\top \ma \vg \geq t) \leq c_1 e^{-c_2 t/\mu},
\]
for all $t\geq c_3 \mu.$ The $c_i$'s are positive absolute constants.
\end{proposition}

\begin{proof}
We set $\ell=1$ in Claim A.3 of \cite{dharangutte2021dynamic}.  By bounding $\norm{\ma}{2}$ and $\norm{\ma}{F}$ each by $\Tr(\ma) = \mu$, we may take $\nu = c_4 \mu$ and $\beta = c_5 \mu$ in Claim A.3.  By properties of subexponential random variables given in \cite{wainwright2015basic}, it then follows that
$\pr(\vg^\top \ma \vg \geq \mu + t) \leq 2 e^{-c_6 t/\mu}$ for $t \geq \nu^2/\beta = c_7 \mu,$ which by adjusting constants rearranges to claim above.
\end{proof}

\begin{theorem}
\label{thm:sketched_block_leverage_sampling}
Suppose that the input to \cref{alg:block_lev_sampling} satisfies $p_i \geq \beta \frac{\blev_i(\ma)}{d}$ for some $\beta\in (0,1]$, and with $N\geq  \Omega\left(\frac{d}{\beta \eps^2}\log\parens{\frac{d}{\beta\eps}}\log d\right),$ where $\eps < 1$. Then the output $\widehat{\ma}$ of \cref{alg:block_lev_sampling} satisfies
\[
(1-\eps)\ma^\top\ma \leq \widehat{\ma}^\top \widehat{\ma} \leq (1+\eps)\ma^\top\ma.
\]
\end{theorem}

\begin{proof}

Let $X_k$ be distributed as $\frac{1}{p_j}\ma_j^\top \vg \vg^\top \ma_j$ where the index $j$ is drawn from $p$, and $\vg$ is independently drawn as a Rademacher random vector.  Then $\widehat{\ma}^\top\widehat{\ma}$ is distributed as $\frac{1}{N}\sum_{k=1}^N \mx_k$, so we show concentration for this average. As is standard in such arguments, we show that the following equivalent statement holds with the desired probability:
\[
(1-\eps)\mathbf{I} \leq (\ma^\top \ma)^{-1/2} \mx (\ma^\top \ma)^{-1/2} \leq (1+\eps)\mathbf{I}.
\] First note that 
\[\E(\mx_k) = \sum_{i=1}^s p_i \E_{\vg}(\frac{1}{p_i}\Ai^\top \vg^\top \vg \Ai) = \sum_{i=1}^s  \Ai^\top \E_{\vg}(\vg^\top \vg) \Ai = \sum_{i=1}^s \Ai^\top \Ai = \ma^\top \ma,\] 
since $\E(\vg^\top\vg) = I.$ Let $\mathbf{Y}_k = (\ma^\top \ma)^{-1/2} \mx_k (\ma^\top \ma)^{-1/2}$, and note that by the above, $\E(\mathbf{Y}_k) = \mathbf{I}.$ Next we have
\[
\norm{\mathbf{Y}_k}{} = \norm{(\ma^\top \ma)^{-1/2}\parens{\frac{1}{p_i} \Ai^\top \vg \vg^\top \Ai}(\ma^\top \ma)^{-1/2} }{} = \vg^\top \parens{\frac{1}{p_i} \Ai (\ma^\top \ma)^{-1}\Ai^\top}\vg.
\] For fixed $i$, this latter expression is the classic Hutchinson's trace estimator for  $\frac{1}{p_i} \Ai (\ma^\top \ma)^{-1}\Ai^\top$, which has mean
\[
\Tr\left(\frac{1}{p_i} \Ai (\ma^\top \ma)^{-1}\Ai^\top\right) = \frac{1}{p_i}\blev_i(\ma) \leq \frac{d}{\beta}.
\] So by \cref{prop:hutchinson_subexponential_tail},
\begin{equation}
    \label{eq:subexp_tail}
    \pr(\norm{\mathbf{Y}_k}{} \geq t) \leq c_1 e^{-c_2 t/\mu},
\end{equation}
for $t\geq c_3\mu$, where we set $\mu = d/\beta.$ At this point we would like to apply Matrix Chernoff to the $\mathbf{Y}_k$'s.  Unfortunately we cannot since the $\mathbf{Y}_k$'s are not bounded a.s.  Therefore we let $\widetilde{\mathbf{Y}}_k$ be the random variable obtained by conditioning $\mathbf{Y}_k$ on the event that $\norm{\mathbf{Y}_k}{} \leq M.$  We will show below that taking $M = c_5 \mu \log\frac{\mu}{\eps}$ gives $\norm{\E \mathbf{Y}_k  - \E \widetilde{\mathbf{Y}}_k}{} \leq \eps.$  As a consequence this gives
\[
(1-\eps)\mathbf{I} \leq \E(\widetilde{\mathbf{Y}}_k) \leq (1+\eps)\mathbf{I}.
\]
Given this choice of $M$, we apply the Matrix Chernoff bound to the $\widetilde{\mathbf{Y}}_k$'s, which now satisfy $\norm{\widetilde{\mathbf{Y}}_k}{} \leq M$. Setting $\overline{\widetilde{Y}} = \frac{1}{N}\sum_{i=1}^N \widetilde{Y}_k$, and plugging into \cref{thm:matrix_chernoff} gives
\[
\pr(\norm{\overline{\widetilde{\mathbf{Y}}} - \mathbf{I}}{} \geq 3\eps) \leq d \exp\parens{-c \eps^2 \frac{N}{M}},
\]
which is bounded by $0.05$ for $N\geq \Omega\parens{\frac{d}{\beta \eps^2}\log\parens{\frac{d}{\beta\eps}}\log d}.$ (We replace $\eps$ with $\eps/3$ to recover the statement in the theorem.) Possibly by adjusting the constant in the definition of $M$, we can arrange so that with probability at least $0.95$, all of the $N$ samples $\mathbf{Y}_k$ are such that $\norm{\mathbf{Y}_k}{} \leq M$ (this follows from the exponential tail bound on the $\mathbf{Y}_k$'s), and hence indistinguishable from the $\widetilde{\mathbf{Y}}_k$'s.  So with probability at least $0.9$ the $\mathbf{Y}_k$'s enjoy the same concentration bound as the $\widetilde{\mathbf{Y}}_k$'s above, which then implies the conclusion of the theorem.

Finally we conclude the argument by showing that $\E\widetilde{\mathbf{Y}_k}$ is approximately $\E \mathbf{Y}_k.$  To simplify notation, let $\mathbf{Y}$ and $\widetilde{\mathbf{Y}}$ be distributed as $\mathbf{Y}_k$ and $\widetilde{\mathbf{Y}}_k$ respectively.  We write
\[\E Y = \E(\widetilde{\mathbf{Y}}) \pr(\norm{\mathbf{Y}}{} \leq M) + \E(\mathbf{Y} | \norm{\mathbf{Y}}{}\geq M)\pr(\norm{\mathbf{Y}}{}\geq M).\]
 
Thus we have
\begin{align*}
    \norm{\E \mathbf{Y} - \E \widetilde{\mathbf{Y}}}{} 
    &\leq (1-\pr(\norm{\mathbf{Y}}{}\leq M ))\norm{\E \widetilde{\mathbf{Y}}}{} + \pr(\norm{\mathbf{Y}}{} > M) \norm{\E(\mathbf{Y} | \norm{\mathbf{Y}}{}\geq M)}{}\\
    &= \pr(\norm{\mathbf{Y}}{}> M )\norm{\E \widetilde{\mathbf{Y}}}{} + \pr(\norm{\mathbf{Y}}{} > M) \norm{\E(\mathbf{Y} | \norm{\mathbf{Y}}{}\geq M)}{}\\
    &\leq \pr(\norm{\mathbf{Y}}{}> M )\E \norm{\widetilde{\mathbf{Y}}}{} + \pr(\norm{\mathbf{Y}}{} > M) \E(\norm{\mathbf{Y}}{} | \norm{\mathbf{Y}}{}\geq M)\\
    &\leq \pr(\norm{\mathbf{Y}}{}> M )\E \norm{\mathbf{Y}}{} + \pr(\norm{\mathbf{Y}}{} > M) \E(\norm{\mathbf{Y}}{} | \norm{\mathbf{Y}}{}\geq M),
\end{align*}
where in the last step we observed that $\E \norm{\widetilde{\mathbf{Y}}}{} \leq \E \norm{\mathbf{Y}}{} $.  We bound each of the relevant terms.

We will take $M\geq c_3 \mu.$ As shown above, $\E \norm{\mathbf{Y}}{}\leq \mu.$ By \cref{eq:subexp_tail}, $\pr(\norm{\mathbf{Y}}{}> M ) \leq c_1 e^{-c_2 M/\mu}$.  To handle the last term, 
\[
    \E\left(\norm{\mathbf{Y}}{} \bigg| \norm{\mathbf{Y}}{}\geq M\right)
    = \int_0^{\infty} \pr\left(\norm{\mathbf{Y}}{} \geq t | \norm{\mathbf{Y}}{} \geq M \right)\,dt
    = M + \int_M^{\infty}\pr\left(\norm{\mathbf{Y}}{} \geq t | \norm{\mathbf{Y}}{} \geq M \right)\,dt.
\]
So 
\begin{align*}
    \pr(\norm{\mathbf{Y}}{} > M) \E(\norm{\mathbf{Y}}{} | \norm{\mathbf{Y}}{}\geq M)
    &= \pr(\norm{\mathbf{Y}}{} > M) M + \int_M^{\infty}\pr\left(\norm{\mathbf{Y}}{} \geq t \right)\,dt\\
    &\leq c_1 e^{-c_2 M/\mu} M + \int_M^{\infty}c_1e^{-c_2t/\mu}\,dt\\
    &= c_1 e^{-c_2 M/\mu} M + c_4 \mu e^{-c_2 M/\mu}. 
\end{align*}
Putting the pieces together gives
\[
\norm{\E \mathbf{Y} - \E \widetilde{\mathbf{Y}}}{} \leq  c_1 \mu e^{-c_2 M/\mu} + c_1 
M e^{-c_2 M/\mu} + c_4 \mu e^{-c_2 M/\mu},
\]
which is bounded by $\eps$ for $M\geq c_5 \mu\log\left(\frac{\mu}{\eps}\right).$
\end{proof}

\subsection{Low-Rank Matrix Approximation}

\looseness=-1In the preceding section, we showed how the coordinator can learn a subspace embedding matrix $\ms$ for $\ma$.  We now show how to utilize this embedding to learn with efficient communication a projection giving a good rank-$k$ approximation to $\ma$.

\thmLowRankMatrixApproximation*

\begin{proof}

    We first reduce the number of columns
    of $\ma$ by right-multiplying by a Rademacher random matrix $\mathbf{R}$ with $O(k/\eps)$ columns. Note that this computation can be carried out locally on each server; each server simply computes $\ma^{(i)} \mathbf{R}$, where $\mathbf{R}$ is known using shared randomness.  Since $\mathbf{R}$ is Rademacher, the bit complexity of each $\ma^{(i)}\mathbf{R}$ is at most $L + O(\log d).$
    By \cite[Theorem 4.2]{clarkson2009numerical}, we have that
    \[
    \min_{\textrm{rank}(\mx) = k} \norm{\ma\mathbf{R}\mx - \ma}{F} \leq (1 + \eps/3) \norm{\ma_k - \ma}{F},
    \] where $\ma_k$ is defined in \cref{ineq:desired-frob-norm-ineq-low-rank-approx} as the best rank-$k$ approximation of $\ma.$  We would like to find an approximate minimizer $\mx$, since then $\ma\mathbf{R}\mx$ is a rank-$k$ approximation of $\ma.$  To do this efficiently, we use our protocol from \cref{subsec:l2_lev_score_sampling}\footnote{One could also apply our protocol from \cref{sec:ell2_subspace_embedding_via_block_lev_scores} (with minor modifications when we apply the black-boxed Lemma 32 below as the resulting matrix is not a true leverage score sampling matrix, since we sample sketches of the rows from each block).}  to construct a row-sampling matrix $\ms$ from the leverage score distribution of $\ma\mathbf{R}$ (up to constant factors on the probabilities), with $\tildeO(k/\eps)$ rows.  Running our protocol requires $\tildeO\left((\frac{sk}{\eps} + \frac{k^2}{\eps^2})\cdot L\right)$ communication (note that the rows of $\ma\mathbf{R}$ have dimension $O(k/\eps)$).
    Then Lemma 32 and Theorem 36 of \cite{clarkson2017low} together imply that
    \[
    \min_{\rank(\mx) = k} \norm{\ms\ma\mathbf{R}\mx - \ms\ma}{F} = \min_{\rank(\mx) = k} \norm{\ms(\ma\mathbf{R}\mx - \ma)}{F} \leq (1+\eps/3) \min_{\rank(\mx) = k} \norm{\ma\mathbf{R}\mx - \ma}{F}.
    \]  The coordinator can learn $\ms\ma\mathbf{R}$ and $\ms\ma$ using total communication $\tildeO(kdL/\eps^2).$  This allows the coordinator to compute $\widehat{\mx}$ with 
    \[
    \norm{\ma\mathbf{R}\widehat{\mx} - \ma}{F} \leq (1+\eps/3)\min_{\rank(\mx) = k} \norm{\ma\mathbf{R}\mx - \ma}{F}
    \leq (1+\eps)\norm{\ma_k - \ma}{F}.
    \]  
    Let $\Pi$ be the orthogonal projection onto the row space of $\mathbf{R}\widehat{\mx}$ which has dimension at most $k$ by construction. For a given row $\mathbf{a}_i$ of $\ma$, $\mathbf{a}_i^\top \mathbf{R} \widehat{\mx}$ is in the row space of $\mathbf{R} \widehat{\mx}$.  So $\norm{\mathbf{a}_i^\top \Pi - \mathbf{a}_i^\top}{}^2 \leq \norm{\mathbf{a}_i^\top \mathbf{R}\widehat{\mx} - \mathbf{a}_i^\top}{}^2$, and hence 
 $\norm{\ma \Pi - \ma}{F} 
    \leq \norm{\ma \mathbf{R} \widehat{\mx} - \ma}{F} 
    \leq (1+\eps)\norm{\ma_k - \ma}{F}.$
    
\end{proof}

\section{High-Accuracy Linear Regression in the Coordinator Model}\label{sec:linear_regression}

\looseness=-1In this section, we present our result for the communication complexity of computing high-accuracy solutions to linear regression problems as in the setup of \cref{def:lin-reg-setting}. The primary result of this section is the following.

\thmHighAccuracyLinearRegression* 

\looseness=-1Our framework for achieving the results in \cref{thm:main_lin_reg} builds upon Richardson's iteration with preconditioning. This circumvents the need to send to the coordinator the matrix $\ma^\top \ma$ --- the approach in the previous best result~\cite{vww20} for this problem, which incurs a communication cost $\Omega(sd^2 L)$\footnote{To see this, we observe that each of $s$ servers computes and sends to the coordinator $\langle\ma^{(i)}, \ma^{(i)}\rangle\in \R^{d\times d}$, each of which takes $O(d^2)$ bits.}, thus exceeding our targeted budget. To ensure convergence with this Richardson-style iteration, the key conceptual idea is to use a matrix $\mm$ that spectrally approximates $\ma^\top\ma$ as a preconditioner. We construct $\mm$ via sampling with respect to overestimates of the leverage scores of $\ma$, which are computed using an iterative process we build upon the refinement sampling framework of \cite{cohen2015uniform}. The main novelty of our algorithm and analysis are in careful roundings and bit complexity analysis to guarantee convergence while ensuring a small number of bits are communicated. In \cref{sec:lin-reg-alg-overview}, we discuss the main components of our algorithm and analysis.

\subsection{An Overview of Our Algorithm and Analysis}
\label{sec:lin-reg-alg-overview}
Our main algorithm for solving the linear regression problem to high-accuracy is \cref{alg:lin-reg-coordinator-poly-cond}. It has three main components that we explain next.
\begin{figure}[!htb]
\begin{framed}
\textbf{Input.} A matrix $\ma:= [\mai{i}]\in  \R^{n \times d}$ and vector $\vb:= [\vbi{i}]\in \R^{n}$, where the $i^{\mathrm{th}}$ machine stores matrix $\mai{i}\in \R^{n_i\times d}$ and vector $\vbi{i}\in\R^{n_i}$; accuracy parameter $0<\epsilon<1$; probability parameter $c$. 

\vspace{.5em}

\textbf{Output.} Vector $\vxhat\in\R^{d}$ such that 
    \[
    \norm{\ma \vxhat - \vb}{2} \leq \epsilon \cdot \norm{\ma (\ma^\top \ma)^{-1} \ma^\top \vb}{2} + \min_{\vx \in \R^d} \norm{\ma \vx - \vb}{2}
    \]

\begin{enumerate}[itemsep = .1em, leftmargin = 1.7em, topsep = .4em, label=\protect\circled{\arabic*}]
    \item\label{item:firstForLoopPolyLinRegCoordinator} Run  \cref{alg:levscoresRefinementSampling} with inputs $\ma$ and $c$. This returns $\levover\in\R^n_{\geq0}$, which satisfies $\levover \geq \lev(\ma)$ and $\norm{\levover}{1} \leq 9d$. We store the coordinates of $\levover$ on the corresponding machines.
    \item\label{item:secondStepInMainAlgLinReg} Using $\levover$ from \cref{item:firstForLoopPolyLinRegCoordinator}, each machine $i\in[s]$ forms a diagonal sampling matrix $\R^{n \times n} \ni \ms^{(i)} = \texttt{Sample}(\levover, 100, c)$ (cf. \cref{def:sampleFnCohen}) and sends  to the coordinator the following objects: the nonzero entries of $\ms^{(i)}$ and the rows of $\mai{i}$ corresponding to the nonzero entries of $\ms^{(i)}$.
    \item\label{item:thirdStepInMainAlgLinReg} The coordinator forms the matrix 
    $\matil := \frac{1}{\sqrt{1.1}}[\mstil^{(i)} \ma^{(i)}]$, 
    where $\mstil^{(i)}$ is the matrix left after removing the zero rows of $\ms^{(i)}$. Set $\mm = \matil^\top \matil$.
    \item\label{item:finalStepInMainAlgLinReg} Implement the protocol in \cref{alg:richardson} with $\ma$, $\vb$, and $\mm$ as input.
\end{enumerate}
\end{framed}
\captionsetup{belowskip=-10pt}
\caption{Protocol for linear regression in the coordinator setting.}
\label[alg]{alg:lin-reg-coordinator-poly-cond}
\end{figure}

\paragraph{Leverage score computation.} \looseness=-1The first step of \cref{alg:lin-reg-coordinator-poly-cond}, as shown in \cref{item:firstForLoopPolyLinRegCoordinator}, is to compute $\widehat{\tau}$, a vector of sufficiently accurate overestimates of $\tau(\ma)$, the true leverage scores of $\ma$. We do this using \cref{alg:levscoresRefinementSampling}, which iteratively refines our initial crude overestimates. In each iteration of \cref{alg:levscoresRefinementSampling}, we sample a matrix $\widetilde{\ma}$, a spectral approximation of $\ma$, and use $(\widetilde{\ma}^\top \widetilde{\ma})^\dagger$ to  compute a new $\widehat{\tau}$ with improved accuracy. This reduces $\|\widehat{\tau}\|_1$ by a constant factor in each iteration, and therefore \cref{alg:levscoresRefinementSampling} runs for only $O(\log n)$ iterations before  $\|\widehat{\tau}\|_1 \leq O(d)$, the required accuracy of $\widehat{\tau}$. A detailed discussion is presented in \cref{sec:lev-score}. 

\looseness=-1Throughout \cref{alg:levscoresRefinementSampling}, the machines communicate to the coordinator only the $\widetilde{O}(d)$ rows they sample locally according to their leverage score overestimates. The coordinator forms $\widetilde{\ma}^\top \widetilde{\ma}$, a constant-factor spectral approximation to $\ma^\top \ma$ and sends to the machines its sketch formed by a combination of rounding and a Johnson-Lindenstrauss random projection. Then the machines use this sketch to update their overestimates to $\widehat{\tau}^{\widetilde{\ma}}$, an approximation of $\tau^{\widetilde{\ma}}(\ma)$, the generalized leverage scores of $\ma$ computed with respect  to $\matil$.
While this approach is similar to that of \cite{cohen2015uniform}, our target communication complexity necessitates additional approximations, e.g. via roundings, of the objects we communicate between the coordinator and the machines (cf. \cref{item:5h} and \cref{item:5j} of \cref{alg:levscoresRefinementSampling}), whereas \cite{cohen2015uniform} is analyzed in exact arithmetic. We elaborate this in \cref{lemma:approx-dot-product} and the overall communication complexity bound for computing the leverage score overestimates in \cref{lem:levScoreOverestimates}. 

\paragraph{Sampling a spectral approximation.} \looseness=-1After computing a sufficiently accurate $\widehat{\tau}$, in the next step (\cref{item:secondStepInMainAlgLinReg} of \cref{alg:lin-reg-coordinator-poly-cond}), each machine locally samples a set of rows  according to its leverage score overestimates and  communicates these to the coordinator. The coordinator uses these rows to form $\mm$, the final spectral approximation of $\ma^\top \ma$,  as stated in \cref{item:thirdStepInMainAlgLinReg} of \cref{alg:lin-reg-coordinator-poly-cond}.

\paragraph{Richardson-type iteration.} \looseness=-1Finally, in \cref{item:finalStepInMainAlgLinReg} of \cref{alg:lin-reg-coordinator-poly-cond}, $\mm^{-1}$ is used as a preconditioner in a Richardson-type algorithm, displayed in \cref{alg:richardson}. In each iteration of \cref{alg:richardson}, a residual vector is computed on each machine and  communicated to the coordinator. In  \cref{item:richardson-round-rule} and \cref{item:richardson-coord-to-machine} of \cref{alg:richardson}, the coordinator linearly combines these residual vectors from all the machines and communicates back to the machines a vector carefully designed to optimize communication complexity. The output of \cref{alg:richardson} is the solution to the given regression problem at the specified accuracy, with a communication cost of $\Otil(sd(L + \log \kappa)\log(\epsilon^{-1}))$ (cf.~\cref{lemma:richardsonExpanded}). A detailed discussion is presented in \cref{sec:richardson}

\begin{remark}
\looseness=-1The primary theme of all the components of \cref{alg:lin-reg-coordinator-poly-cond} is the simple idea that we communicate only the bits that are \textit{necessary} for the convergence of our algorithms. This idea is perhaps most apparent in the last component (\cref{alg:richardson}). To this end, we carefully round intermediate vectors and matrices (i.e., discard bits with low place values) and reuse some bits (e.g.,  by  not communicating bits with high place values). The latter is possible essentially because in an iterative algorithm, when the solution is converging, the bits with high place values are the same from one iteration to the next. 
\end{remark}

\paragraph{Discussion of our analysis.}

\looseness=-1We conclude this overview by reiterating that although our approach is conceptually fairly simple, the overall algorithm and analysis are involved due to careful bit modifications in intermediate steps such as the JL projection, matrix inversions, and Richardson's iteration to reduce the communication complexity while ensuring convergence. The careful bit modification and analysis in turn plays a crucial role in yielding improvements compared to \cite{vww20}, particularly for matrices with a condition number of $e^{o(dL)}$, as summarized in \cref{tab:coordinator_model_regression}. In subsequent sections, we explain each of the components of our approach. We first discuss the computation of leverage score overestimates in \cref{sec:lev-score} and our iterative preconditioning approach for solving the linear regression problem in \cref{sec:richardson}. We finally prove the main result of this section in \cref{subsec:lin-reg-main-proof}.

\renewcommand{\arraystretch}{1.4}

\subsection{Leverage Score Overestimates}\label{sec:lev-score}
\looseness=-1We now focus on the point-to-point communication complexity of computing a vector of overestimates for leverage scores of the matrix. We emphasize that although \cite{vww20} considers computing the leverage scores in the blackboard model, this is significantly more complicated in the coordinator (point-to-point) setting due to bit complexity issues involving inverse of matrices. The main result of this section is \cref{lem:levScoreOverestimates}.

\begin{restatable}{lemma}{lemlevScoreOverestimates}\label{lem:levScoreOverestimates}\looseness=-1Given the linear regression setting of \cref{def:lin-reg-setting} with matrix $\ma=[\ma^{(i)}]\in \R^{n\times d}$, and $n\geq 5$, there is a randomized algorithm that, with \[\Otil(d^2 L+sd(L+\log(\kappa))) \text{ bits of communication}\] and, with high probability, computes a vector $\levover\in\R^n$ such that $\norm{\levover}{1} \leq 9d$ and $\levover_i\geq\lev_i(\ma)$, for all $i\in[n]$, where $\lev_i(\ma)$ is the $i^\mathrm{th}$ leverage score of the matrix $\ma$. Each entry of $\levover$ is stored on the machine that contains the corresponding row.
\end{restatable}

\looseness=-1We achieve the results in \cref{lem:levScoreOverestimates} via \cref{alg:levscoresRefinementSampling}. As described in \cref{sec:lin-reg-alg-overview}, this algorithm is then used as a black-box in our main algorithm (\cref{alg:lin-reg-coordinator-poly-cond}) as a pre-cursor to construct a spectral approximation of $\ma$. 
 \cref{alg:levscoresRefinementSampling} is based on the \texttt{Refinement Sampling} algorithm of \cite{cohen2015uniform} to approximately compute leverage scores of a matrix,  with appropriate modifications for the coordinator setting. We provide an overview of the distributed version of \texttt{Refinement Sampling}  in \cref{subsec:refinement-sampling} and prove \cref{lem:levScoreOverestimates} in \cref{sec:proofOfCCCLevScoreOverests}. While the results of \cite{cohen2015uniform} are provided under the exact arithmetic  model, whereas we need to employ careful bit complexity analysis under fixed-point arithmetic. 

\subsubsection{An Overview of Refinement Sampling}
\label{subsec:refinement-sampling}
\looseness=-1The \texttt{Refinement Sampling} algorithm returns a small-sized spectral approximation to an input matrix $\ma$ by iteratively refining (hence the name) $\widehat{\tau}$, the leverage score overestimates of $\ma$; when $\|\widehat{\tau}\|_1 \leq O(d)$, it is accurate enough for use (with appropriate scaling) in sampling a matrix with the desired approximation guarantee, leading to termination of the algorithm. 

\looseness=-1The algorithm starts by setting $\widehat{\tau} = \mathbf{1}$, the vector of all ones. In each iteration, the algorithm uses $\levover$ to construct a matrix $\matil\in \R^{\widetilde{n}\times d}$ composed of a subset of $\widetilde{n}$ (rescaled) rows of $\ma$, where each row $i$ is sampled independently with probability $\vp_i\propto \alpha \levover_i$ for some sampling rate $\alpha>0$ (cf. \cref{def:sampleFnCohen} for the precise row sampling). The vector $\levover$ is then updated to $\tau^{\matil}(\ma)$, the generalized leverage scores of $\ma$ with respect to $\matil$ (cf. \cref{eq:genlevscoresdef}), and this process continues iteratively.

\looseness=-1We pick $\alpha$ such that the number of rows of $\matil$, which is proportional to $\alpha \cdot\|\levover\|_1$, is $\Otil(d)$ with high probability. This ensures (see \cref{item:5e}) that the number of rows communicated between the coordinator and machines is only $\widetilde{O}(d)$, incurring a bit complexity of $\widetilde{O}(d^2L)$.
Therefore, iteratively reducing $\|\levover\|_1$ by a constant factor enables a corresponding increase in $\alpha$, which eventually yields the desired spectral approximation (\cref{lem:specApproxViaLevScoreSampling} and \cref{thm:LevScoreApproxViaUndersampling}, first stated in \cite{cohen2015uniform}), all the while maintaining a row size of $\Otil(d)$ for the matrix used to compute generalized leverage scores. We state these formal guarantees next since they are used in our proof of correctness of \cref{alg:levscoresRefinementSampling}. 

\begin{definition}[Sampling Function~\cite{cohen2015uniform}]\label{def:sampleFnCohen}
Given a vector $\vu\in \R^n_{\geq 0}$, a parameter $\alpha> 0$, and a positive constant $c$, we define vector $\vp\in \R^n_{\geq 0}$ as $\vp_i = \min(1, \alpha c \log d \cdot \vu_i)$. We  define the function $\texttt{Sample}(\vu, \alpha, c)$ to be one which returns
a random diagonal $n\times n$ matrix $\mathbf{S}$ with independently chosen entries: \[\mathbf{S}_{ii} = \twopartdef{\frac{1}{\sqrt{\vp_i}}}{\textrm{with probability } \vp_i}{0}{\textrm{otherwise}}.\]  
\end{definition}

\noindent
The spectral approximation guarantees of this sampling approach are formalized in the following lemma.

\begin{lemma}[Spectral Approximation via Leverage Score Sampling; Lemma 4 of~\cite{cohen2015uniform}]\label{lem:specApproxViaLevScoreSampling} Given a matrix $\ma\in \R^{n\times d}$, a sampling rate $\alpha>1$, and a fixed constant $c>0$. Let  $\vu\in \R^n_{\geq 0}$ be a vector of leverage score overestimates, that is, \[\vu\geq \tau(\ma),\,\, \text{ and } \mathbf{S} := \texttt{Sample}(\vu, \alpha, c)\] as in \cref{def:sampleFnCohen}. Then, with probability at least $1-d^{-c/3} - (3/4)^d$, the following results hold: \[\textbf{nnz}({\mathbf{S}}) = 2c \alpha \|\vu\|_1 \log d\text{ and } \frac{1}{\sqrt{1+\alpha^{-1/2}}} \mathbf{SA}\approx_{\left(\frac{1+\alpha^{-1/2}}{1-\alpha^{-1/2}}\right)}\ma.\]  
\end{lemma}

The following lemma formalizes how $\norm{\vu}{1}$ shrinks when we use a spectral approximation $\matil$ obtained by the sampling approach of \cref{def:sampleFnCohen} (cf. \cref{lem:specApproxViaLevScoreSampling}) to update $\widehat{\tau}$ (the leverage score estimates) to $\tau^{\widetilde{\ma}}(\ma)$ (the generalized leverage scores of $\ma$ with respect to $\matil$).

\begin{lemma}[Leverage Score Estimate Update via Undersampling; Theorem 3 of~\cite{cohen2015uniform}]\label{thm:LevScoreApproxViaUndersampling}
Given a matrix $\ma\in \R^{n\times d}$ and an undersampling parameter $\alpha\in (0, 1]$, let $\vu\in \R^n_{\geq 0}$ and $\vu^{(\textrm{new})}\in \R^n_{\geq 0}$ be vectors such that:  \[\vu \geq \tau(\ma),\,\, \mathbf{S} := \sqrt{\frac{3\alpha}{4}} \texttt{Sample}(\vu, 9\alpha, c), \text{ and } \vu_i^{(\textrm{new})} := \min\{\tau_i^{\mathbf{S} \ma}(\ma), \vu_i\}\text{ for all } i\in [n].\] 
Then, with probability at least $1-d^{-c/3} - (3/4)^d$, $\vu_i^{(\textrm{new})}$ is a leverage score overestimate, i.e., $\vu_i^{(\textrm{new})}\geq\tau_i(\ma)$. Furthermore, \[\| \vu^{(\textrm{new})}\|_1 \leq 3d/\alpha \text{ and } \textbf{nnz}(\mathbf{S})=O(\alpha \cdot \|\vu\|_1 \cdot \log d).\] 
\end{lemma}

\subsubsection{Invariance of Leverage Score Overestimates: $\widehat{\tau} \geq \tau(\ma)$}\label{sec:levScoresRefinementSamplingGuarantees}
\looseness=-1The bulk of communication in \cref{alg:levscoresRefinementSampling} happens in
\cref{item:5e} and \cref{item:5h}. In \cref{item:5e}, we send
only the rows of the original matrix and the associated vector of probabilities. We construct the probabilities to be powers of two so that they can be communicated with a small number of bits. However,  in \cref{item:5h}, we need to send the product of the inverse of a matrix with other matrices. To do this, we need to round the product and communicate the rounded version. The following two technical lemmas (\cref{lemma:approx-dot-product} and \cref{lemma:ridge-and-lev-connection}) help us bound the error in computing the generalized leverage scores arising from this rounding process.

\begin{figure}[!ht]
\begin{framed}
\textbf{Input.} A matrix $\ma := [\ma^{(i)}]\in \R^{n\times d}$, where $n = \sum_{i=1}^s n_i$ and the $i^{\mathrm{th}}$ machine stores matrix $\mai{i}\in \R^{n_i\times d}$; probability parameter $c$.
\vspace{.5em}

\textbf{Output.} Vector $\widehat{\tau}$ that satisfies  $\levover\geq \lev(\ma)$ and $\|\levover\|_1 \leq 9d$. The $i^{\mathrm{th}}$ machine stores $\levover^{(i)}$, the set of coordinates of $\levover$ corresponding to the rows of $\ma$ stored on that machine.

\vspace{.5em}

\textbf{Initialize.} Set the total number of iterations $T = \ceil{\log_2(n/d)}$ and $r=10^6\cdot \log n$.  Set the leverage scores estimate vector $\levover = \mathbf{1}$. Set $\lambda = \frac{1}{100 \kappa^2}$, where $\kappa$ is the condition number of $\ma$. 

\vspace{.5em}

\textbf{For} iterations $\ell = 1, 2, \dots, T$:

\begin{enumerate}[itemsep = .1em, leftmargin = 2.2em, topsep = .4em, label=\protect\circled{\arabic*}]
        \item \label{item:5a} Each machine $i$ computes $t_i = \|\levover^{(i)}\|_1$, the sum of its leverage score overestimates, and sends it to the coordinator.
        
        \item \label{item:5c} The coordinator computes $t \defeq \sum_{i=1}^s t_i$ and sends it to all the machines.

        \item\label{item:5e}
        All machines set $\widehat{\alpha}$ to the smallest power of half
        that is at least $\frac{25d \log d}{t}$. Set $\alpha=\min\{1, \frac{\widehat{\alpha}}{\log d}\}$.
        Each machine $i$ forms an $n_i$-by-$n_i$ random diagonal matrix $\ms^{(i)} = \texttt{Sample}(1.01 \levover^{(i)}, 9\alpha, c)$ as per \cref{def:sampleFnCohen} and  sends to the coordinator those rows of $\ma^{(i)}$ and coordinates of  $\levover^{(i)}$  that correspond to the nonzero entries of $\ms^{(i)}$.

        \item\label{item:5g} Using the received entries of  $\levover^{(i)}$ for all $i\in[s]$, the coordinator computes $\mstil^{(i)}$, the non-zero rows of $\ms^{(i)}$ 
        as per \cref{def:sampleFnCohen}. It then constructs the matrices $\matil^{(i)} = \sqrt{\frac{3\alpha}{4}} \mstil^{(i)} \ma^{(i)}$ and sets $\matil=[\matil^{(i)}]\in \R^{\widetilde{n} \times d}$ 
        and $\mb = [\matil; \sqrt{\lambda} \cdot \mi]\in \R^{(\widetilde{n}+d)\times d}$.
        
        \item\label{item:5h} The coordinator samples a JL sketching matrix $\mg\in\{-1,+1\}^{r \times (\widetilde{n}+d)}$ and uses $\mb$ from the previous step to compute
        $\mghat:=\frac{\sqrt{1.01}}{ 0.99\sqrt{r}}\mg \mb (\mb^\top \mb)^{-1}\in \R^{r\times d}$. It then generates $\mj\in \R^{r\times d}$ by rounding the entries of $\mghat$ so that for all $i_1,i_2$, $\abs{\widehat{g}_{i_1 i_2}-j_{i_1 i_2}} < \frac{1}{10^4 \cdot  n^2 d\sqrt{r} \cdot 2^L}$ and sends $\mj$ to all the machines. 

        \item\label{item:5j} The coordinator computes an integer basis for the kernel of $[\mstil^{(i)} \ma^{(i)}]$ (e.g., by Gaussian elimination). It computes a linear combination $\vecv$ of the kernel basis by picking independent and uniformly random coefficients for in $[-2^{dL},2^{dL}]$. It picks $z=\ceil{100 \log(n)}$ random prime numbers $y_1,\ldots,y_{z}$ less than $(dL)^2$. For each prime number $y_k$, it sets $\vecv^{(k)}$ to be $\vecv$ modulo $y_k$. It then sends all $\vecv^{(k)}$'s and $y_k$'s to all the machines.
        
        \item\label{item:6} 
        If for all $k\in[z]$, $\mathbf{a}_j^\top\vecv^{(k)} = 0$ modulo $y_k$, then we set $\widehat{\tau}^{\matil}_j=\infty$ (on the machine holding row $\mathbf{a}_j$). Otherwise,
        we set $\widehat{\tau}^{\matil}_j = \norm{\mj \mathbf{a}_j}{2}^2$. Then we update $\levover_j^{\textrm{new}}$ to the smallest power of two that is at least $\max\left\{\min\left\{\levover_{j}, \widehat{\tau}^{\matil}_j\right\}, \frac{1}{2n^2}\right\}$.

\item \label{item:lev-final-step}
Each machine updates its overestimates with $\levover_j = 1.01\cdot \levover_j^{\textrm{new}}$.
\end{enumerate}

\end{framed}
\captionsetup{belowskip=-10pt}
\caption{Protocol for computing leverage score overestimates in the coordinator setting}
\label[alg]{alg:levscoresRefinementSampling}
\end{figure}

\begin{restatable}{proposition}{LEMapproxDotProd}\label{lemma:approx-dot-product}
Let $0<\veps<1$ be an accuracy parameter, and let $\vu, \vutil \in \R^d$ be vectors satisfying, for all $i\in[d]$, that $\abs{u_i-\widetilde{u}_i}\leq \veps$. Then for any $\vv\in\R^d$, we have $|{\vutil^\top \vv - \vu^\top \vv}| \leq \sqrt{d} \cdot \veps \norm{\vv}{2}$. Moreover if $\mb,\mbtil\in\R^{m\times d}$ such that for all $i\in[m],j\in[d]$, $|b_{ij}-\widetilde{b}_{ij}| \leq \veps$, then 
\[|\norm{\mb \vv}{2} - \norm{\mbtil \vv}{2}|\leq \sqrt{md} \cdot \veps \norm{\vv}{2}.\]
\end{restatable}

\begin{restatable}[Approximating Ridge Leverage Scores with Leverage Scores]{proposition}{LEMridgeAndLecConenction}
\label{lemma:ridge-and-lev-connection}
Let
$\lev(\ma)$ and $\lev^\lambda(\ma)$ be the vector of leverage scores and $\lambda$-ridge leverage scores (see \cref{defn:ridgeLevScores}), respectively, of $\ma \in \R^{n\times d}$, $n\geq d$. Let $\lambda \geq 0$ and $\sigma_{\min}$ be the smallest \emph{nonzero} singular value of $\ma$. Then 
\[
\frac{\sigma_{\min}^2+\lambda}{\sigma_{\min}^2} \cdot\lev^{\lambda}_j(\ma) \geq \lev_j(\ma) \geq \lev^{\lambda}_j(\ma)\text{ for all $j\in [n]$. }
\]
\end{restatable} \looseness=-1Equipped with these technical results, we prove the different components of \cref{lem:levScoreOverestimates}, starting with the invariant of \cref{alg:levscoresRefinementSampling} that the leverage score overestimates vector  $\levover$ indeed always remains larger than $\lev^{\lambda}(\ma)$, the $\lambda$-ridge leverage scores vector of $\ma$. As we see in \cref{lem:levScoreOverestimates}, this implies $\levover\geq \lev(\ma)$. 

\begin{lemma}[Invariance of Leverage Score Overestimates]\label{lem:levScoreOverestimatesInvariance}
    Assume that at the start of each iteration $\ell\in[T]$ of \cref{alg:levscoresRefinementSampling},  we have $\levover\geq \lev^{\lambda}(\ma)$. Then, at the end of this iteration, we have $\widehat{\lev}^{\textrm{(new)}}\geq \lev^{\lambda}(\ma)$.
\end{lemma}
\begin{proof}[Proof of \cref{lem:levScoreOverestimatesInvariance}]
     The random vector $\vecv$ computed in \cref{item:5j} of \cref{alg:levscoresRefinementSampling} satisfies $\vecv\in \nullsp(\matil)$  (we remove the $\sqrt{\frac{3\alpha}{4}}$ coefficient only to be sure that the matrix is an integer matrix and we can find an integer basis for the kernel). Therefore if a vector $\vw\in \nullsp(\matil)^{\perp}$, then  $\vw^\top \vecv = 0$ as well as  $\vw^\top\vecv^{(k)}=0$.  Otherwise, if $\vw\notin \nullsp(\matil)^{\perp}$, then with high probability, $\vw^\top\vecv\neq 0$. Since both $\vecv$ and $\mathbf{a}_j$ are integer vectors, we can look at $\mathbf{a}_j^\top\vecv$ modulo $y_k$. Since we are using Gaussian elimination to compute the integer basis for the kernel and since the $\ma$ is an integer matrix with bit complexity $L$, $|\mathbf{a}_j^\top \vecv| \leq \poly(n) \cdot 2^{dL}$. Therefore the number of prime factors of $|\mathbf{a}_j^\top \vecv|$ is $O(dL \log n)$. Therefore if we select a random prime number less than $(dL)^2$, then with a large probability, we pick a prime $y_k$ that is not a factor of $|\mathbf{a}_j^\top \vecv| \leq \poly(n) \cdot 2^{dL}$. Therefore if $|\mathbf{a}_j^\top \vecv|\neq 0$, then  with high probability, $|\mathbf{a}_j^\top \vecv|$ modulo $y_k$ is also not zero. In \cref{item:5j} of \cref{alg:levscoresRefinementSampling}, we select multiple random primes independently to boost this probability even further. Then by taking union bound over all the rows of $\ma$, the algorithm, with high probability, can detect which ones are orthogonal to the kernel of $\matil$. Thus if $\mathbf{a}_j\notin \mathcal{N}(\matil)^\perp$, 
then with high probability, we set $\levover^{\matil}_j=\infty$, and therefore (from \cref{item:6}), we have the desired inequality \[\levover_j^{(\textrm{new})} \geq \max\left\{\min\left\{\levover_{j}, \widehat{\tau}^{\matil}_j\right\}, \frac{1}{2n^2}\right\} \geq \levover_j \geq \lev^{\lambda}_j(\ma), \numberthis\label[ineq]{eq:invariant_levover_case1}\] where the final inequality follows from the lemma's assumption.

Next, consider the case in which $\mathbf{a}_j\in \nullsp(\matil)^{\perp}$. In this case, \cref{item:6} of \cref{alg:levscoresRefinementSampling} sets $\widehat{\tau}^{\matil}_j = \norm{\mj \mathbf{a}_j}{2}^2$ for $\mj$ as defined in \cref{item:5h} of \cref{alg:levscoresRefinementSampling}. We now proceed to show $\norm{\mj \mathbf{a}_j}{2}^2 \geq \lev^{\lambda}_j(\ma)$; since the lemma assumes $\levover_j\geq \lev^{\lambda}_j(\ma)$, we then have that \[ \|\mj \mathbf{a}_j\|_2^2 \geq \tau^{\lambda}_j(\ma) \implies  \widehat{\tau}_j^{\textrm{new}} \geq  \max\left\{\min\left\{\levover_{j}, \widehat{\tau}^{\matil}_j\right\}, \frac{1}{2n^2}\right\} \geq  \lev^{\lambda}_j(\ma),\numberthis\label[ineq]{eq:invariant_levover_case2}\] which is the desired inequality. 
In the rest of this proof, we show $\norm{\mj \mathbf{a}_j}{2}^2 \geq \lev^{\lambda}_j(\ma)$.

First note that $\sqrt{\frac{3\alpha}{4}} \ms^{(i)}:=\sqrt{\frac{3\alpha}{4}} \texttt{Sample}(1.01\cdot\widehat{\tau}^{(i)}, 9\alpha, c)$ (as set in \cref{item:5e} of \cref{alg:levscoresRefinementSampling}) has a distribution equivalent to
$\mshat^{(i)}:=\sqrt{\frac{3}{4}} \texttt{Sample}(1.01\cdot \levover^{(i)}, 9, c)$ with some nonzero entries set to zero. This is because in \cref{item:5e} of \cref{alg:levscoresRefinementSampling}, we set $\alpha\leq 1$, which implies that  $\texttt{Sample}(1.01\cdot \levover^{(i)}, 9\alpha, c)$ has a smaller sampling rate compared to $\texttt{Sample}(1.01\cdot \levover^{(i)}, 9, c)$ for all entries, but the value of any entry selected in both matrices  (i.e., nonzero entry) is the same in $\sqrt{\frac{3\alpha}{4}} \ms^{(i)}$ and $\mshat^{(i)}$. Therefore by \cref{lem:specApproxViaLevScoreSampling} and since $1.01\cdot \levover$ is a vector of overestimates for leverage scores of $\ma$, with high probability,
\[
\matil^\top \matil \preceq \ma^\top \mshat^\top \mshat \ma \preceq \ma^\top \ma,
\]
where $\mshat \in\R^{n\times n}$ is the diagonal matrix obtained by putting $\mshat^{(i)}$ together (as block-diagonals of $\ms$). Therefore, for $\mb = [\ma; \sqrt{\lambda}\mathbf{I}]$ as defined in \cref{item:5g}, we have 
\[
\mb^\top \mb = \matil^\top \matil + \lambda \cdot \mi \preceq \ma^\top \ma + \lambda \cdot \mi.
\]
Thus $(\mb^\top \mb)^{-1} \succeq (\ma^\top \ma + \lambda \cdot \mi)^{-1}$, and 
$\lev^{\mb}_j(\ma) \geq \lev^{\lambda}_j(\ma)$, for all $j\in[n]$. Moreover,
\[
\lev^{\mb}_j(\ma) = \mathbf{a}_j^\top (\mb^\top \mb)^\dagger \mathbf{a}_j = \mathbf{a}_j^\top (\mb^\top \mb)^\dagger \mb^\top \mb (\mb^\top \mb)^\dagger \mathbf{a}_j = \norm{\mb (\mb^\top \mb)^\dagger \mathbf{a}_j}{2}^2.
\]
Consequently, by \cref{lemma:random-projection}, for $\mghat$, as defined in \cref{item:5h} of \cref{alg:levscoresRefinementSampling}, and any $j\in[n]$, since $\frac{(0.01^2 - 0.01^3)\cdot r}{4} > 20 \log n$,
\[
\frac{101}{99} \lev^{\lambda}_j(\ma)\leq \frac{101}{99} \lev^{\mb}_j(\ma) \leq \norm{\mghat \mathbf{a}_j}{2}^2 \leq \frac{101^2}{99^2} \cdot \lev^{\mb}_j(\ma), \numberthis\label[ineq]{ineq:jl-sketch-lev-ineq}
\]
with probability at least $1-2n^{-20}$. Taking the union bound, with probability of at least $1-2n^{-19}$, \cref{ineq:jl-sketch-lev-ineq} holds for all $j\in[n]$, where $\mathbf{a}_j \in \mathcal{N}(\matil)^\perp$. Now
by \cref{lemma:approx-dot-product}, we have
\[
\abs{\norm{\mghat \mathbf{a}_j}{2} - \norm{\mj \mathbf{a}_j}{2}} \leq \sqrt{r d} \cdot \epsilon \norm{ \mathbf{a}_j }{2} \leq d \sqrt{r} \cdot \frac{1}{10^4 \cdot  n^2 d\sqrt{r} 2^L} \cdot  2^L = \frac{1}{10^4 \cdot n^2}.
\numberthis\label[ineq]{ineq:lev-score-rounding-error}
\]
We now consider two cases. 
\begin{caseof} 
\case{$\frac{1}{2n^2}\geq \tau_j^{\lambda}(\ma)$.}{Then by the construction in \cref{item:6} of \cref{alg:levscoresRefinementSampling} we have \[\levover^{(\textrm{new})}_j \geq \frac{1}{2n^2}\geq \lev^{\lambda}_j(\ma),\numberthis\label[ineq]{eq:invariant_levover_case4}\] which is the claim of the lemma.} 
\case{$\lev^{\lambda}_j(\ma) > \frac{1}{2n^2}$.}{
Recall that our goal is to show $\norm{\mj \mathbf{a}_j}{2}^2 \geq \lev^{\lambda}_j(\ma).$ For the sake of contradiction, suppose instead that $\norm{\mj \mathbf{a}_j}{2}^2 < \lev^{\lambda}_j(\ma)$; then we have
\[
\left( \norm{\mj \mathbf{a}_j}{2} + \frac{1}{10^4 \cdot n^2} \right)^2 < \left( \lev^{\lambda}_j(\ma) + \frac{1}{10^4 \cdot n^2} \right)^2 \leq \lev^{\lambda}_j(\ma) + \frac{1}{5000 n^2} + \frac{1}{10^8 n^4} \leq \lev^{\lambda}_j(\ma) + \frac{1}{2500 n^2},\numberthis\label[ineq]{ineq:asdfasdfasdf}
\]
where the penultimate inequality follows from the fact that $\tau^{\lambda}_j(\ma) \leq 1$. Further, observe that we may combine \cref{ineq:jl-sketch-lev-ineq} and \cref{ineq:lev-score-rounding-error} to conclude
\[
\frac{101}{99} \lev^{\lambda}_j(\ma) \leq \norm{\mghat \mathbf{a}_j}{2}^2 \leq \left( \norm{\mj \mathbf{a}_j}{2} + \frac{1}{10^4 \cdot n^2} \right)^2.
\numberthis\label[ineq]{ineq:lev-J-up-bound-error}
\]
Combining \cref{ineq:asdfasdfasdf} with \cref{ineq:lev-J-up-bound-error},
we have
$\frac{2}{99} \lev^{\lambda}_j(\ma) < \frac{1}{2500 n^2},$
which is a contradiction to our assumption that $\lev^{\lambda}_j(\ma) > \frac{1}{2n^2}$. This shows $\norm{\mj \mathbf{a}_j}{2}^2 \geq \lev^{\lambda}_j(\ma)$, which then finishes the proof because of \cref{eq:invariant_levover_case2}. }
\end{caseof}
\end{proof}

\subsubsection{$\ell_1$ Bound on Leverage Score Overestimates: $\|\widehat{\tau}\|_1 \leq O(d)$}
\begin{lemma}\label{lem:LevScoreOverestimatesTotalSumAtMostD}
 The vector $\widehat{\tau}$ returned by  \cref{alg:levscoresRefinementSampling} satisfies $\|\widehat{\tau}\|_1\leq O(d)$. 
\end{lemma}
\begin{proof}  
Define the vector \[\vu^{(\textrm{new})}_j := \min\{\tau_j^{\matil}, \levover_j\}.\] Note that the vector $\vu^{\textrm{(new)}}_j$ is not the same as $\min\{\widehat{\tau}^{\widetilde{\ma}}_j, \widehat{\tau}_j\}$ as seen in \cref{item:6} of \cref{alg:levscoresRefinementSampling}. In \cref{item:6} of \cref{alg:levscoresRefinementSampling}, if $\mathbf{a}_j \notin \mathcal{N}(\matil)^\perp$, then we set $\widehat{\lev}^{\matil}_j = \lev^{\matil}_j=\infty$, with high probability. Otherwise, we have $\mathbf{a}_j \in \mathcal{N}(\matil)^\perp$, and  we set $\widehat{\tau}_j = \|\mj \mathbf{a}_j\|_2^2$. Then, by \cref{ineq:jl-sketch-lev-ineq} and the  fact that $\matil^\top \matil \preceq \matil^\top \matil + \lambda \cdot \mi := \mb^\top \mb$ and $\mathbf{a}_{j}\in \mathcal{N}(\matil)^\perp$, we have 
\[
\norm{\mghat \mathbf{a}_j}{2}^2 \leq \frac{101^2}{99^2} \cdot \tau^{\mb}_j \leq \frac{101^2}{99^2} \cdot \tau^{\matil}_j.
\numberthis\label[ineq]{ineq:jl-sketch-lev-ineq-2}
\]
\begin{caseof}
\case{$\norm{\mj \mathbf{a}_j}{2} < \norm{\mghat \mathbf{a}_j}{2}$.}{By combining the assumed inequality with  \cref{ineq:jl-sketch-lev-ineq-2}, we get
\[
\norm{\mj \mathbf{a}_j}{2}^2 \leq  \|\mghat \mathbf{a}_j\|_2^2 \leq \frac{101^2}{99^2} \cdot \tau^{\matil}_j. \numberthis\label[ineq]{ineq:ineq1}
\]
Therefore, we may combine \cref{item:6} of \cref{alg:levscoresRefinementSampling} with the value of $\widehat{\tau}^{\matil}_j$, definition of $\vu_j^{\textrm{new}}$ and \cref{ineq:ineq1} to get
\[
\levover^{(\textrm{new})}_j \leq 2\cdot\min\{\widehat{\tau}_j, \widehat{\tau}^{\widetilde{\ma}}_j\} + 2 \cdot\frac{1}{2n^2} = 2 \cdot
\min\{ \widehat{\tau}_j, \|\mj \mathbf{a}_j\|_2^2  \} + \frac{1}{n^2} \leq 2 \cdot \frac{101^2}{99^2} \vu^{(\textrm{new})}_j + \frac{1}{n^2} \leq 2.1 \cdot \vu^{(\textrm{new})}_j + \frac{1}{n^2}.
\numberthis\label[ineq]{ineq:lev-over-update-error-1}
\]}
\case{$\norm{\mj \mathbf{a}_j }{2} \geq \norm{\mghat \mathbf{a}_j}{2}$.}{In this case, we may combine \cref{ineq:lev-score-rounding-error} and \cref{ineq:jl-sketch-lev-ineq-2} to obtain 
\[
\norm{\mj \mathbf{a}_j}{2}^2 \leq \left(\norm{\mghat \mathbf{a}_j}{2} + \frac{1}{10^4 \cdot n^2} \right)^2 \leq \left(\frac{101}{99} \cdot \sqrt{\lev^{\matil}_j} + \frac{1}{10^4 \cdot n^2} \right)^2 \leq 1.05 \cdot \lev^{\matil}_j + \frac{1}{2500 n^2}, \numberthis\label[ineq]{ineq:ineq2}
\]
where the last inequality follows from $\sqrt{x} \leq \max(x, 1)$. 
Therefore by the construction of \cref{item:6} of \cref{alg:levscoresRefinementSampling},  the value of $\widehat{\tau}^{\matil}_j$, definition of $\vu_j^{\textrm{new}}$, and \cref{ineq:ineq2}, we get
\[
\levover^{(\textrm{new})}_j \leq 2\cdot\min\{\widehat{\tau}_j, \widehat{\tau}^{\widetilde{\ma}}_j\} + 2 \cdot\frac{1}{2n^2} = 2 \cdot
\min\{ \widehat{\tau}_j, \|\mj \mathbf{a}_j\|_2^2  \} + \frac{1}{n^2} \leq 2 \cdot \frac{101^2}{99^2} \vu^{(\textrm{new})}_j + \frac{1}{n^2} \leq 2.1 \cdot \vu^{(\textrm{new})}_j + \frac{2502}{2500}\cdot \frac{1}{n^2}.
\numberthis\label[ineq]{ineq:lev-over-update-error-2}
\]}
\end{caseof}

Combining \cref{ineq:lev-over-update-error-1} and \cref{ineq:lev-over-update-error-2} from the two cases and using \cref{thm:LevScoreApproxViaUndersampling}, we have 
\[
\norm{\levover^{(\textrm{new})}}{1} \leq 2.1 \cdot \norm{\vu^{(\textrm{new})}} + \frac{2502}{2500} \cdot \frac{1}{n} \leq 2.1 \cdot \frac{3d}{\alpha} + \frac{2502}{2500} \cdot \frac{1}{n}.
\]
If $\alpha=1$, since $n,d\geq 1$,
\[
\norm{\levover^{(\textrm{new})}}{1} \leq 8d.
\]
Otherwise, $\alpha \geq \frac{25 d}{\norm{\levover}{1}}$. Therefore since $\norm{\levover}{1}\geq 1$ and $n\geq 5$,
\[
\norm{\levover^{(\textrm{new})}}{1} \leq \frac{6.3}{25} \cdot \norm{\levover}{1} + \frac{2502}{2500} \cdot \frac{1}{n} \leq \frac{1}{2} \cdot \norm{\levover}{1}.
\]
Therefore after $\ceil{\log_2(n/d)}$ iterations, either $\alpha$ becomes one, which means $\norm{\levover^{(\textrm{new})}}{1} \leq 8d$, or $\norm{\levover}{1}$ is cut by half in each iteration which means since at the beginning $\norm{\levover}{1}=n$, at the end $\norm{\levover}{1} \leq d$. Finally in \cref{item:lev-final-step}, we multiply the vector $\levover$ by $1.01$ which by \cref{lemma:ridge-and-lev-connection} and choice of $\lambda$ is guaranteed to be a vector of leverage score overestimates. Moreover by the above discussion for the final $\levover$, we have $\norm{\levover}{1} \leq 1.01 \cdot 8 d \leq 9 d$.

\end{proof}

\subsubsection{Correctness and Communication Complexity of Leverage Score Overestimates}\label{sec:proofOfCCCLevScoreOverests}
We are now ready to prove the result stated at the start of \cref{sec:lev-score}, which we first restate below. 
\lemlevScoreOverestimates*
\begin{proof}[Proof of \cref{lem:levScoreOverestimates}] 
We first prove the correctness of \cref{alg:levscoresRefinementSampling}, and then we bound its communication complexity. 

\paragraph{Correctness. } The algorithm initializes $\levover = \mathbf{1}$, which makes it a vector of leverage score overestimates. In \cref{lem:levScoreOverestimatesInvariance}, we show that throughout the algorithm, we maintain the invariant  $\widehat{\tau} \geq \tau^{\lambda}$. This implies for each $j\in [n]$:
\[
1.01 \cdot \levover_j
\geq 
\frac{\sigma_{\min}^2 + \lambda}{\sigma_{\min}^2} \cdot  \levover_j \geq \frac{\sigma_{\min}^2 + \lambda}{\sigma_{\min}^2} \cdot \lev^{\lambda}_j \geq \lev_j,
\numberthis\label[ineq]{ineq:ridge-gives-lev}
\] 
where the first inequality follows from $\sigma_{\min} \geq \frac{1}{\kappa}$ and $\lambda \leq \frac{1}{100 \kappa^2}$, the second inequality
is by the invariant $\widehat{\tau}_j\geq \tau^{\lambda}_j$, and the third inequality is 
by \cref{lemma:ridge-and-lev-connection}. This implies
that $1.01 \cdot \levover$, which is used in \cref{item:5e} of \cref{alg:levscoresRefinementSampling} for sampling $\matil$, satisfies $1.01 \cdot\levover \geq \tau$. Finally, \cref{lem:LevScoreOverestimatesTotalSumAtMostD} implies the claimed bound $\|\widehat{\tau}\|_1 \leq 9d$. 

\paragraph{Communication complexity. }\looseness=-1 Having proved the correctness of the returned output, we next bound the communication complexity of \cref{alg:levscoresRefinementSampling}.
First, note that by construction, the entries of $\levover$ are always powers of two and in the interval $[\frac{1}{2n^2},1]$. Therefore communicating their summations in \cref{item:5a} and \cref{item:5c} of \cref{alg:levscoresRefinementSampling} can be done with $O(s\log n)$ bits in each iteration. Therefore the cost of these steps over the course of the algorithm is $O(s\log^2 n)$. The part of \cref{item:5e} that forms $\widehat{\alpha}$ and $\alpha$ does not pose any communication cost since the machines can compute them given $t$.

In \cref{item:5e}, by choice of $\alpha$, with high probability, we only need to communicate $O(d\log d)$ rows of $\ma$ to the coordinator --- note that the extra factor of $1.01$ for the sampling only increases the number of sampled rows by a constant factor. Since each row is a vector of size $d$ with bit complexity $L$, the cost of sending the rows to the coordinator is $O(d^2 L\log d)$ in each iteration. Moreover, as we mentioned, since the entries of $\levover$ are always powers of two and in the interval $[\frac{1}{2n^2},1]$, we can send the entries of leverage score overestimates to the coordinator with $O(d\log(d) \log(n))$ bits in each iteration. Therefore the total cost of this step over the course of the algorithm is $O(d^2 L \log(d) \log^2(n))$. \cref{item:5g} does not incur any communication cost.

Note that $\mb^\top \mb \succeq \lambda \cdot \mi$. Therefore $(\mb^\top \mb)^{-1} \preceq \frac{1}{\lambda} \cdot \mi$. Moreover since $\levover \geq \frac{1}{2n^2}$ and by choice of $\lambda$,
we have
\[
\norm{\mghat}{F} \leq \poly(n) \cdot 2^L \cdot \frac{1}{\lambda} = \poly(n) \cdot \kappa^2 2^L.
\]
Therefore due to error tolerance for $\mj$ described in \cref{item:5h}, the bit complexity of $\mj$ is $O((L + \log \kappa ) \cdot  \log n)$. Therefore since the number of rows and columns of $\mj$ is $O(\log n)$ and $d$, respectively, it can be sent to all the machines with $O(sd \cdot (L + \log \kappa ) \cdot  \log n)$ bits of communication, which concludes  the cost of \cref{item:5h} in each iteration.

In \cref{item:5j}, since the primes are less than $(dL)^2$, the primes and the entries of the vectors $\vecv^{(k)}$ only need $O(\log(dL))$ bits of communication. Therefore the total cost of \cref{item:5j} in each iteration is $O(d\log(dL) \log n)$.

Finally, note that \cref{item:6} and \cref{item:lev-final-step} do not pose any communication cost. The total communication complexity of \cref{alg:levscoresRefinementSampling} is then bounded by the above discussion.

\end{proof}

\subsection{Richardson-Type Iteration with Preconditioning}
\label{sec:richardson}

The main export of this section is \cref{lemma:richardsonExpanded}, wherein we use an approximate preconditioner (obtained by sampling according to leverage score overestimate computed using \cref{alg:levscoresRefinementSampling} and \cref{lem:levScoreOverestimates}) within a Richardson-type iteration to solve our linear regression problem iteratively.  As alluded to earlier, in \cref{alg:richardson}, in addition to appropriate roundings, we communicate only the functions of difference of consecutive vectors in the computation, i.e., $\vx^{(k+1)} - \vx^{(k)}$. This prevents redundant communication and permits the claimed communication bounds in \cref{lemma:richardsonExpanded}.

\looseness=-1\cref{alg:richardson} differs from the classic Richardson's iteration in a few key aspects, which allow for smaller communication complexity. Richardson's method requires $\log(\epsilon^{-1})$ iterations, in each of which, we multiply the current vector by the Gram matrix of the original matrix and inverse of the preconditioner. Such a multiplication increases the bit complexity of the vector by $\log(\kappa)+L$ even if we use a rounded version of the inverse of the preconditioner. Therefore at iteration $k$, the bit complexity of the vector is $k (L + \log \kappa)$. Communicating such vectors between the coordinator and the machines would result in a communication complexity of
\[
\Omega(sd \cdot \sum_{k=1}^{\ceil{\log \tfrac{1}{\epsilon}}} k (L + \log \kappa)) = \Omega (sd \cdot (L + \log \kappa) \log^2 (\epsilon^{-1})).
\]
In this section, we show that one of the $\log(\epsilon^{-1})$ factors in the above expression can be avoided. \cref{item:richardson-classic} of \cref{alg:richardson} is a classic Richardson's step which gives $\vxtil^{(k+1)}_i$ --- see the example in \cref{eq:PlaceValueRuleForBits}, which we display here again. After this computation, we modify the bits of the $i^{\mathrm{th}}$ coordinate $\vxtil^{(k+1)}_i$ to obtain $\vx^{(k+1)}_i$ by setting the bits with small place values to their values in $\vx^{(k+1)}_i$ (cf. \cref{item:richardson-round-rule} of \cref{alg:richardson}). The underlying principle is that since the place value of these bits is small, this modification does not affect the convergence of Richardson. Moreover, since our algorithm is converging (proven in \cref{lemma:richardsonExpanded}), the bits with high place values do not change from one iteration to the next. These bits are represented by blue in the following equation. Thus, since the machines store $\vx^{(k)}_i$, upon receiving the bits of $\vx^{(k+1)}_i$ represented by red, they can construct $\vx^{(k+1)}_i$. Our proof essentially boils down to showing the number of bits in this red middle part of $\vx^{(k+1)}_i$ is $O(L+\log \kappa)$ and that the modification of bits with small place value does not affect the convergence. This then implies both the communication complexity bounds and correctness as stated in \cref{lemma:richardsonExpanded}. We note that \eqref{eq:PlaceValueRuleForBits} does not completely reflect all cases. For example, it is possible that $\vx^{(k)}_i=1.0000$ and $\vxtil^{(k+1)}_i=0.1111$. In this case, it might appear that the bits with high place values are not equal. However, in the difference $\vxtil^{(k+1)}_i - \vx^{(k)}_i = 0.0001$, the bits with high place value are equal to zero. This is carefully written and analyzed in our algorithms and proofs.

\begin{equation}
\label{eq:PlaceValueRuleForBits}
\begin{split}
\vxtil^{(k+1)}_i & = 
\underbrace{\textcolor{blue}{101.0110101001}}_\text{is same as $\vx^{(k)}_i$}\underbrace{\textcolor{red}{011011010010101}}_\text{is different than $\vx^{(k)}_i$}\underbrace{\textcolor{red}{0010000010101\cdots}}_\text{is different than $\vx^{(k)}_i$} \\
\vx^{(k+1)}_i & = 
\underbrace{\textcolor{blue}{101.0110101001}}_\text{is same as $\vx^{(k)}_i$}\underbrace{\textcolor{red}{011011010010101}}_\text{is same as $\vxtil^{(k+1)}_i$}\underbrace{\textcolor{blue}{1010100001111\cdots}}_\text{is same as $\vx^{(k)}_i$}
\end{split}  
\end{equation}

\begin{figure}[t]
\begin{framed}
\textbf{Input.} A matrix $\ma := [\ma^{(i)}]\in \R^{n\times d}$, where $n = \sum_{i=1}^s n_i$, and vector $\vb := [\vb^{(i)}]\in \R^n$, where the $i^{\mathrm{th}}$ machine stores matrix $\mai{i}\in \R^{n_i\times d}$ and vector $\vbi{i}\in\R^{n_i}$. The matrix $\mm$ with $\ma^\top \ma \preceq \mm \preceq \lambda \ma^\top \ma$, with $\lambda\geq 1$, stored on the coordinator machine; accuracy parameter $0<\epsilon<1$. 
\vspace{.5em}

\textbf{Output.} A vector $\vxhat\in\R^d$ on all machines such that $\norm{\ma \vxhat - \vb}{2} \leq \epsilon \cdot \norm{\ma (\ma^\top \ma)^{-1} \ma \vb}{2} + \min_{\vx} \norm{\ma \vx - \vb}{2}$.

\vspace{.5em}

\textbf{Initialize.} Set $\vx^{(0)} = \vec{0}$, $\veps=\frac{1}{2\lambda}$, and $T \in\N$ be the smallest integer such that $(1-\frac{1}{2\lambda})^{T} \leq \epsilon$. Moreover each machine $i\in[s]$, computes $(\ma^{(i)})^\top \vb^{(i)}$ and sends it to the coordinator.

\vspace{.5em}

\textbf{For} iterations $k = 0, 1, \dots, T-1$:

\begin{enumerate}[itemsep = .1em, leftmargin = 1.7em, topsep = .4em, label=\protect\circled{\arabic*}]
        \item\label{item:richardson-classic} The coordinator computes $\vxtil^{(k+1)} = \vx^{(k)} - \mm^{-1} (\ma^\top \ma \vx^{(k)} - \ma^\top \vb)$.\label{item:firstStepRichardsonIter}
        
        \item\label{item:richardson-round-rule} The coordinator sets each coordinate $j$ of $\vx^{(k+1)}$ equal to the corresponding coordinate $j$ of  $\vxtil^{(k+1)}$ at all bits except those with a place value less than $\veps \cdot \frac{\norm{\vxtil^{(k+1)} - \vx^{(k)}}{\mm}}{\lambda \cdot nd^2 \cdot 2^{2L+1}}$, at which bits the coordinator sets the bit to the corresponding bit of coordinate $j$ of $\vx^{(k)}$.\label{item:secondStepRichardsonIter}
        
        \item\label{item:richardson-coord-to-machine} The coordinator sends the vector $\vx^{(k+1)} - \vx^{(k)}\in \R^d$ to the machines.
        
        \item\label{item:richardson-machine-to-coord} Each machine $i\in [s]$ computes the vector $(\ma^{(i)})^\top \ma^{(i)} (\vx^{(k+1)} - \vx^{(k)})\in \R^d$ and sends it to the coordinator.
\end{enumerate}

\end{framed}
\captionsetup{belowskip=-10pt}
\caption{Protocol for computing the solution of a linear regression problem given a preconditioner.}
\label[alg]{alg:richardson}
\end{figure}

In the next result, we assume the matrix is full-rank. However, this is not a limitation of our approach since we can concatenate the matrix with a small-scaled identity matrix --- see Appendix A in \cite{ghadiri2023bit}.

\begin{lemma}[Richardson-Type Iteration]
\label{lemma:richardsonExpanded}
Given $\epsilon>0$ and the linear regression setting of \cref{def:lin-reg-setting} with input matrix $\ma=[\ma_i]\in \R^{n\times d}$ and vector $\vb=[\vb_i]\in\R^n$, let  $\mm \in \R^{d\times d}$ be a matrix stored in the coordinator such that $\ma^\top \ma \preceq \mm \preceq \lambda \cdot \ma^\top \ma$, for constant $\lambda \geq 1$. Then there is an algorithm that outputs a vector $\vxhat$ such that
\[
\norm{\ma \vxhat - \vb}{2}^2 \leq \epsilon \cdot \norm{\ma (\ma^\top \ma)^{-1} \ma^\top \vb}{2}^2 + \min_{\vx} \norm{\ma \vx - \vb}{2}^2
\] 
using  $\Otil(sd (L+\log\kappa)  \log(\epsilon^{-1}))$ bits of communication. 
Moreover, the vector $\vxhat$ is available on all the machines at the end of the algorithm.
\end{lemma}

\begin{proof}
We show that \cref{alg:richardson} returns the specified solution with the specified communication complexity bounds.
As a first step, we argue that the vectors that are required for computation at the coordinator can indeed be computed from the information that is communicated to the coordinator. 

First, note that $\ma^\top \vb = \sum_{i=1}^s (\ma^{(i)})^\top \vb^{(i)}.$  The quantity $\ma^\top \vb$ needs to be communicated to the coordinator only once (since it does not change through the algorithm). Since the bit complexities of $\ma$ and $\vb$ are $L$ and each machine sends to the coordinator $(\ma^{(i)})^\top \vb^{(i)}$, it  takes a total of $O(sdL \log n)$ to communicate $\ma^\top \vb$. 

We now show by induction that the coordinator can compute $\ma^\top \ma \vx^{(k)}$. Initially $\vx^{(0)}=0$. Therefore $\ma^\top \ma \vx^{(0)}$, and the coordinator has this information. Then, assume as the base case of the induction that the coordinator has the value of $\ma^\top \ma \vx^{(k)}$ and the machines are sending the vectors $(\ma^{(i)})^\top \ma^{(i)} (\vx^{(k+1)} 
- \vx^{(k)})$ to the coordinator (see \cref{item:richardson-machine-to-coord} of \cref{alg:richardson}). We observe that
\begin{align*}
\ma^\top \ma \vx^{(k+1)} 
& = \ma^\top \ma \vx^{(k)} + \ma^\top \ma (\vx^{(k+1)} - \vx^{(k)})= \ma^\top \ma \vx^{(k)} + \sum_{i=1}^s (\ma^{(i)})^\top \ma^{(i)} (\vx^{(k+1)} - \vx^{(k)}).
\end{align*}
\looseness=-1Therefore, if the coordinator has $\ma^\top \ma \vx^{(k)}$ then \cref{item:richardson-machine-to-coord} enables it to compute $\ma^\top \ma \vx^{(k+1)}$. Moreover, since $\mm$ and $\vx^{(k)}$ are also stored on the coordinator, it can perform the computation in \cref{item:richardson-classic}.

We now prove that \cref{alg:richardson} converges. The optimal vector $\vxstar = \arg\min_{\vx} \|\ma\vx-\vb\|_2$ satisfies $\vxstar = (\ma^\top \ma)^{-1} \ma^\top \vb$. Subtracting it from the iterate $\vxtil^{(k+1)}$ and applying the update from \cref{item:firstStepRichardsonIter} of \cref{alg:richardson} in each iteration, we have 
\begin{align*}
\vxtil^{(k+1)} - \vxstar
 &= 
\vx^{(k)} - \mm^{-1} (\ma^\top \ma \vx^{(k)}- \ma^\top \vb) - \vxstar = (\mi - \mm^{-1} \ma^\top \ma) (\vx^{(k)} - \vxstar). \numberthis\label{eq:richardson-update}
\end{align*}
To obtain a bound on the rate of shrinkage of the distance (in $\mm$-norm) of the algorithm's iterate from the true optimizer, we study the squared $\mm$-norm of the right-hand side in \cref{eq:richardson-update}. 
\begin{align*}
\norm{(\mi - \mm^{-1} \ma^\top \ma) (\vx^{(k)} - \vxstar)}{\mm}^2
& =
(\vx^{(k)} - \vxstar)^\top (\mi -  \ma^\top \ma \mm^{-1}) \cdot \mm \cdot (\mi - \mm^{-1} \ma^\top \ma) (\vx^{(k)} - \vxstar). \numberthis\label{eq:app4}
\end{align*}
To simplify notion, we define $\mh \defeq \mm^{-1/2} \cdot \ma^\top \ma \cdot \mm^{-1/2}$. Therefore, we can check the computation 
\[
(\mi -  \ma^\top \ma \mm^{-1})\cdot \mm\cdot (\mi - \mm^{-1} \ma^\top \ma) = \mm^{1/2} (\mi - \mh)^2 \mm^{1/2}.
\]
Moreover, by the assumption $0 \prec \ma^\top \ma \preceq \mm\preceq \lambda \cdot\ma^\top \ma$, the definition of $\mh$, and the properties stated in \cref{fact:psdOrderingFacts}, we have 
\[
\frac{1}{\lambda} \mi = \frac{1}{\lambda} \mm^{-1/2} \mm \mm^{-1/2} \preceq \mh \preceq \mm^{-1/2} \mm \mm^{-1/2} = \mi.
\]
Therefore $0\preceq \mi - \mh \preceq (1-\tfrac{1}{\lambda}) \mi$; applying to this \cref{fact:psdOrderingFacts}, we have
\[
0 \preceq (\mi - \ma^\top \ma \mm^{-1}) \cdot\mm\cdot (\mi - \mm^{-1} \ma^\top \ma) \preceq (1-\lambda^{-1})^2\cdot \mm.\numberthis\label[ineq]{app5}
\]
Plugging \cref{app5} into \cref{eq:richardson-update} and \cref{eq:app4} gives 
\[
\norm{\vxtil^{(k+1)} - \vxstar}{\mm}^2 =
\norm{(\mi - \mm^{-1} \ma^\top \ma) (\vx^{(k)} - \vxstar)}{\mm}^2 \leq  (1-\lambda^{-1})^2 \norm{\vx^{(k)} - \vxstar}{\mm}^2. \numberthis\label[ineq]{app2}
\]
By the update rule in \cref{item:firstStepRichardsonIter} of \cref{alg:richardson} and the expression for $\vxstar$, we have
\begin{align*}
\vxtil^{(k+1)} - \vx^{(k)} = - \mm^{-1} (\ma^\top \ma \vx^{(k)}- \ma^\top \vb) =
- \mm^{-1} (\ma^\top \ma \vx^{(k)} - \ma^\top \ma \vxstar) =
- \mm^{-1} \ma^\top \ma (\vx^{(k)} - \vxstar). \numberthis\label{eq:xkPtilToxstarMnorm}
\end{align*}
By applying the assumption $\ma^\top \ma \preceq \mm$ to \cref{eq:xkPtilToxstarMnorm} and using \cref{fact:psdOrderingFacts}, we get  
\begin{align*}
\norm{\vxtil^{(k+1)} - \vx^{(k)}}{\mm}^2 
 & = 
(\vx^{(k)} - \vxstar)^\top \ma^\top \ma \mm^{-1} \ma^\top \ma (\vx^{(k)} - \vxstar) \\
&\leq
(\vx^{(k)} - \vxstar)^\top \ma^\top \ma (\ma^\top \ma)^{-1} \ma^\top \ma (\vx^{(k)} - \vxstar) \\
& = 
(\vx^{(k)} - \vxstar)^\top \ma^\top \ma (\vx^{(k)} - \vxstar)\\
&\leq
\norm{\vx^{(k)} - \vxstar}{\mm}^2. \numberthis\label[ineq]{eq:xKpTilToxKMnormBoundInt}
\end{align*}
We then have 
\begin{align*}
\norm{\vx^{(k+1)} - \vxtil^{(k+1)}}{\mm} \leq \lambda \cdot nd \cdot 2^{2L} \norm{\vx^{(k+1)} - \vxtil^{(k+1)}}{2} \leq \veps \norm{\vxtil^{(k+1)} - \vx^{(k)}}{\mm} \leq \veps \norm{\vx^{(k)} - \vxstar}{\mm},\numberthis\label[ineq]{eq:xMinusxTilBound}
\end{align*} where the first step uses the assumptions that $\ma$ has a bit complexity of $L$ and associated \cref{fact:OperatorNormBoundForBitBoundedMatrix} and  also that $\mm\preceq \lambda\ma^\top \ma$, the second step uses \cref{item:secondStepRichardsonIter} from \cref{alg:richardson} (specifically, the example visually depicted in \cref{eq:PlaceValueRuleForBits}), and the third step uses \cref{eq:xKpTilToxKMnormBoundInt}. 
Combining \cref{app2} and \cref{eq:xMinusxTilBound} along with the triangle inequality, we have
\begin{align*}
\norm{\vx^{(k+1)} - \vxstar}{\mm} \leq (1-\lambda^{-1} + \veps) \norm{\vx^{(k)} - \vxstar}{\mm} = (1-\tfrac{1}{2\lambda}) \norm{\vx^{(k)} - \vxstar}{\mm}.\numberthis\label[ineq]{ineq:m-norm-guarantee-richardson}
\end{align*}
Since $\lambda$ is a constant, this implies that after $O(\log(\frac{1}{\epsilon}))$ iterations, we achieve the required accuracy. Starting from the zero vector and since $\ma^\top \ma \preceq \mm \preceq \lambda \ma^\top \ma$, we have
\[
\norm{\vx^{(T)} - \vxstar}{\ma^\top \ma} \leq \epsilon \norm{\vxstar}{\ma^\top \ma}.\numberthis\label[ineq]{ineq:DistXtToOptInTermsOfOpt}
\]
We then have
\begin{align*}
\norm{\ma\vx^{(T)} - \vb}{2}^2 
& = 
\norm{\ma\vx^{(T)} - \ma (\ma^\top \ma)^{-1} \ma^\top \vb - (\mi - \ma (\ma^\top \ma)^{-1} \ma^\top )\vb}{2}^2
\\ & =
\norm{\ma\vx^{(T)} - \ma (\ma^\top \ma)^{-1} \ma^\top \vb}{2}^2 + \norm{(\mi - \ma (\ma^\top \ma)^{-1} \ma^\top )\vb}{2}^2
\\ & =
\norm{\ma(\vx^{(T)} - \vxstar)}{2}^2 + \min_{\vx} \norm{\ma \vx - \vb}{2}^2.
\end{align*}
Then applying \cref{ineq:DistXtToOptInTermsOfOpt} to the above inequality, we have the claimed convergence bound: 
\begin{align*}
\norm{\ma\vx^{(T)} - \vb}{2}^2 \leq \epsilon^2 \cdot \norm{\ma (\ma^\top \ma)^{-1} \ma^\top \vb}{2}^2 + \min_{\vx} \norm{\ma \vx - \vb}{2}^2.
\end{align*}
We now bound the bit complexity of each iteration of  \cref{alg:richardson}. Denoting $\sigma_{\textrm{min}}$ to be the minimum singular value of $\mm$, we have
\begin{align*}
    \norm{\vx^{(k+1)} - \vx^{(k)}}{2} 
    & \leq 
    \frac{1}{\sigma_{\textrm{min}}} \norm{\vx^{(k+1)} - \vx^{(k)}}{\mm} \\ & \leq 
    \frac{1}{\sigma_{\textrm{min}}} \left( \norm{\vx^{(k+1)} - \vxtil^{(k+1)}}{\mm} + \norm{\vxtil^{(k+1)} - \vx^{(k)}}{\mm} \right) 
    \\ & \leq 
    \frac{1+\veps}{\sigma_{\textrm{min}}} \cdot  \norm{\vxtil^{(k+1)} - \vx^{(k)}}{\mm},
\end{align*} where the first step uses $\|\mathbf{u}\|_{\mm} \geq \sigma_{\textrm{min}} \|\mathbf{u}\|_2$ for any vector $\mathbf{u}$, the second step is by triangle inequality, and the third step uses the rounding we performed in \cref{item:richardson-round-rule} of \cref{alg:richardson}. 
Moreover note that if $\norm{\vx^{(k+1)} - \vx^{(k)}}{2} \neq 0$, by construction, it is more than $\veps\cdot\frac{\norm{\vxtil^{(k+1)} - \vx^{(k)}}{\mm}}{\lambda \cdot n d^2 2^{2L+2}}$. Let $p$ be an integer such that $2^p \leq \veps \cdot \frac{\norm{\vxtil^{(k+1)} - \vx^{(k)}}{\mm}}{\lambda \cdot n d^2 2^{2L+2}} < 2^{p+1}$. Then we have
for any $j$ corresponding to a nonzero entry of $\vx^{(k+1)} - \vx^{(k)}$,
\begin{align*}
    1 \leq 2^{-p}\abs{(\vx^{(k+1)} - \vx^{(k)})_j} \leq 2^{-p} \cdot \frac{1+\veps}{\sigma_{\min}} \cdot  \norm{\vxtil^{(k+1)} - \vx^{(k)}}{\mm} < \frac{(1+\veps)}{\veps \cdot\sigma_{\min}} \cdot \lambda \cdot n d^2 2^{2L+3}.
\end{align*}
Therefore we only need to communicate $\log(p) + \log(\frac{(1+\veps)}{\veps \cdot\sigma_{\min}} \cdot \lambda \cdot n d^2 2^{2L+3})$ bits to each machine in \cref{item:richardson-machine-to-coord}. Finally, in \cref{item:richardson-machine-to-coord},  when the machines multiply $\vx^{(k+1)} - \vx^{(k)}$ by $\ma^\top \ma$ and send it back to the coordinator this just adds
$O(L\log n)$ bits.
\end{proof}

\subsection{Proof of Main Result on High-Accuracy Linear Regression}
\label{subsec:lin-reg-main-proof}

In this section, we prove our main result for the communication complexity of linear regression in the point-to-point model of communication.
\begin{proof}[Proof of \cref{thm:main_lin_reg}]
We show that \cref{alg:lin-reg-coordinator-poly-cond} returns the correct output and satisfies the communication complexity bounds. We start with the proof of correctness.

\paragraph{Correctness.} By \cref{lem:levScoreOverestimates}, \cref{item:firstForLoopPolyLinRegCoordinator} of \cref{alg:lin-reg-coordinator-poly-cond} returns a vector of overestimates $\levover \geq \lev$ with $\norm{\levover}{1} \leq 9d$. Then by \cref{lem:specApproxViaLevScoreSampling}, sampling according to these leverage score overestimates with $\alpha=100$ produces the matrix $\matil$ in \cref{item:thirdStepInMainAlgLinReg} which is a $\frac{1.1}{0.9}$-spectral approximation of the matrix $\ma$. Finally, in \cref{item:finalStepInMainAlgLinReg}, we use this spectral approximation to return a high-accuracy solution as proved in \cref{lemma:richardsonExpanded}.

\paragraph{Communication complexity.} First, observe that \cref{item:thirdStepInMainAlgLinReg} incurs no communication. Next, the communication complexities of \cref{item:firstForLoopPolyLinRegCoordinator} and \cref{item:finalStepInMainAlgLinReg} directly follow from \cref{lem:levScoreOverestimates} and \cref{lemma:richardsonExpanded} and are 
$\Otil(d^2 L + sd(L+\log \kappa))$ and $\Otil(sd (L + \log \kappa) \log(\epsilon^{-1}))$, respectively. Finally, \cref{item:secondStepInMainAlgLinReg} requires the communication of the selected rows and their probabilities, which we now compute. Since \cref{item:firstForLoopPolyLinRegCoordinator} provides vector $\levover$ of leverage score overestimates that satisfies $\|\levover\|_1\leq 9d$, we are guaranteed by \cref{lem:specApproxViaLevScoreSampling} that, with high probability, there are only $O(d\log d)$ non-zero rows in \cref{item:secondStepInMainAlgLinReg}. These can be communicated with $\Otil(d^2 L)$ bits since the bit complexity of the matrix $\ma$ is $L$ and each vector has $d$ coordinates. As per \cref{def:sampleFnCohen}, the sampling probabilities can be computed at the coordinator by knowing the leverage score overestimates for the selected rows at the coordinator. Therefore since before the last step of \cref{alg:levscoresRefinementSampling}, the leverage score overestimates are powers of two, we can communicate these powers of two (that are in the interval $[\frac{1}{2n^2},1]$) with $O(d\log(d)\cdot \log(n) )$, bits and then the coordinator can obtain the leverage score overestimates by multiplying these by $1.01$ and computing the probabilities accordingly.
\end{proof}

\section{Linear Programming in the Coordinator Model}\label{sec:linearprogramming}
In this section, we bound the communication complexity of solving the following   linear program in the setup of \cref{def:lin-reg-setting}:
\begin{align}
\label{eq:LP-formulation}
 \min_{\vx \in \R^{n}_{\geq 0}: \ma^\top \vx = \vb} \vc^\top \vx, 
\end{align} 
with $\ma=[\ma^{(i)}]_{i\in [s]}$, $\vc=[\vc^{(i)}]_{i\in [s]}$,  the $i^\mathrm{th}$ machine holding $\ma^{(i)}\in\R^{n_i \times d}$ and $\vc^{(i)} \in \R^{n_i}$. We assume that the bit complexities of $\ma,\vb,\vc$ are bounded by $L$, and the vector $\vb$ is available to all machines. 
The main result of this section is the following.

\ipmthm*

Our framework for achieving the results in \cref{thm:ipm} is an adaptation of the algorithm of \cite{van2020solving} to the coordinator setting, 
which in turn builds upon the techniques of \cite{ls14}. While we describe our procedure in more detail in \cref{sec:LPoverview}, 
essentially each iteration of our interior-point method involves solving a linear system of the form $\ma^\top \mw \ma \vecv = \ma^\top \mw^{1/2} \vg$, where $\vecv$ is the variable vector and $\mw$ is a nonnegative diagonal matrix --- see \cref{eq:kkt-for-updates}. This is equivalent to solving a linear regression of the form $\min_{\vecv} \norm{\mw^{1/2} \ma \vecv - \vg}{2}^2$. To this end, we use a randomized approach based on leverage score sampling. 
Since the randomness of previous iterations determines the solutions of next iterations, this might cause adaptive adversary issues --- see \cite{DBLP:conf/focs/LeeS15}.
To prevent such an issue, we solve the linear regression instance in each iteration to high-accuracy by \cref{thm:main_lin_reg}.

\subsection{An Overview of Our Algorithm}\label{sec:LPoverview}
\looseness=-1Set in the framework of primal-dual path-following interior-point method~\cite{renegar1988polynomial, renegar2001mathematical}, our algorithm maintains a primal feasible point $\vx\in\R^n_{\geq0}$ satisfying (the primal feasibility condition) $\ma^\top \vx = \vb$ and a dual feasible point $\vy\in\R^d_{\geq 0}$ satisfying (the dual feasibility condition) $\ma\vy - \vc\geq 0$ (note that all vector inequalities in this section are coordinate-wise.) The goal is to steadily decrease the primal-dual gap \[ \vc^\top \vx - \vb^\top \vy = (\ma\vy + \vs)^\top \vx - \vb^\top \vy = \vs^\top \vx, \] where $\vs:= \vc-\ma\vy$ is the \emph{slack} variable, and the other terms cancel out due to the aforementioned feasibility conditions. At each iteration, a primal-dual path-following algorithm trades off decreasing the gap $\vx^\top\vs$ against maintaining feasibility. Collecting all these requirements yields the following optimality conditions we want the algorithm to satisfy in each iteration: 
 \[ 
\vx \odot \vs = \mu \cdot \tau (\vx, \vs), \,\, \ma^\top \vx = \vb, \,\,  \ma \vy + \vs  = \vc,\,\,  
\vx, \vs  \geq 0, \numberthis\label{eq:kkt-lp}
\]
where we use $\odot$ to denote coordinate-wise product of two vectors, $\tau(\vx,\vs)$ is  a carefully designed function that tracks the primal-dual gap,  and $\mu$ is a parameter the algorithm gradually decreases. The (unique) set $(\vx,\vy,\vs)$ satisfying this system of equations is said to follow the weighted central path. Decreasing $\mu$ puts more emphasis on decreasing the primal-dual gap (i.e., making progress on the objective value), and the algorithm alternates between decreasing $\mu$ and taking a Newton-like step from the current approximate solution of \cref{eq:kkt-lp} towards that with the updated $\mu$. 

\looseness=-1The choice of weight function $\tau$ plays an immensely critical role in the convergence rate of the overall algorithm. For example, setting $\tau$ to be the all-ones vector yields the standard log-barrier-based central path~\cite{renegar1988polynomial, mehrotra1992implementation, gonzaga1992path, nesterov1994interior, ye1994nl} which takes $\widetilde{O}(\sqrt{n})$ iterations for convergence \cite{karmarkar1984new,vaidya1989speeding}. Drawing on geometric connections between Lewis weights and ellipsoidal approximations of polytopes, the breakthrough work of \cite{ls14} developed an algorithm with $\widetilde{O}(\sqrt{\mathrm{rank}(\ma)})$ iterations by choosing $\tau$ to be $\ell_p$ Lewis weights of a certain matrix for $p = \Omega(\log n)$. The work of \cite{van2020solving} simplified this to the leverage scores of a certain matrix, with an added regularizer for simplicity of analysis, and this is the  weight function we also use: \[\tau(\vx,\vs) := \sigma(\ms^{-1/2-\alpha} \mx^{1/2-\alpha} \ma) + \frac{d}{n} \cdot \boldsymbol{1}, \numberthis\label{eq:ourCentrality} \] where $\sigma(\mb)$ denotes the vector of leverage scores of matrix $\mb$, the uppercase $\mx$ and $\ms$ denote, respectively, the diagonal matrices formed using the vectors $\vx$ and $\vs$, and the parameter $\alpha=O(\tfrac{1}{\log(n/d)})$.

\looseness=-1In order to improve the proximity of the next point to the true central path described by \cref{eq:kkt-lp}, the algorithm attempts to ensure, for the updated $\mu$, that $\vx\odot \vs \approx \mu \tau (\vx, \vs)$ defined in \cref{eq:ourCentrality}. This proximity is called \emph{centrality}~\cite{ls14}, and to measure and track it, prior works \cite{cohen2021solving, lee2019solving, van2020deterministic, van2020solving} have successfully used  the following ``soft-max''-like potential  \[  \Phi(\vecv) = \sum_{i=1}^n \exp(\lambda(\vvi{i}-1))+\exp(-\lambda(\vvi{i}-1)),\numberthis\label{eq:BrandPotential} \] with $v=\tfrac{\mu \tau(\vx, \vs)}{\vx\odot \vs}$ and $\lambda = O(\epsilon^{-1})$; this is what we also therefore use. Making fast progress along the central path therefore entails updating $\vx$ and $\vs$ so that \cref{eq:BrandPotential} decreases sufficiently fast. 

\looseness=-1To achieve this goal, one observes that the updated points $\vx+\delta_{\vx}$ and $\vs + \delta_{\vs}$ must also satisfy \cref{eq:kkt-lp}; considering the fact that $\vx$ and $\vs$ also satisfy this equation, one infers that the updates $\delta_{\vx}$ and $\delta_{\vs}$ must satisfy 
\[ \mx \delta_{\vs} + \ms \delta_{\vs} = \delta_{\widetilde{\mu}}, \,\, \ma^\top \delta_{\vx} = 0, \,\, \ma\delta_{\vy} + \delta_{\vs} = 0,\numberthis\label{eq:kkt-for-updates}\] 
where we have introduced the notation $\mutil := \mu\cdot\tau(\vx,\vs) = \mu \cdot (\sigma(\ms^{-1/2-1/p} \mx^{1/2-1/p} \ma) + \frac{d}{n} 1)$, and  $\dmutil \approx \nabla \Phi(\tfrac{\mutil}{\vx\odot\vs})$ is the update step we choose to decrease the potential, thereby improving the centrality of the next point. Plugging this value of $\delta_{\mutil}$ back into \cref{eq:kkt-for-updates} yields the solution 
\begin{align}
\label{eq:ipm_iter_sol}
\dx = \sqrt{\tfrac{\mx}{\ms}} (\mi-\mproj)  \tfrac{\mi}{\sqrt{\mx \ms}}\cdot\dmutil  \text{  and    } \ds = \sqrt{\tfrac{\ms}{\mx}} \mproj  \tfrac{\mi}{\sqrt{\mx \ms}}\cdot\dmutil,
\end{align}
where $\mproj$ is the orthogonal projection matrix given by the following closed-form expression 
\begin{align}
\label{eq:proj_matrix}
\mproj &= \sqrt{\tfrac{\mx}{\ms}} \ma \left(\ma^\top \tfrac{\mx}{\ms} \ma\right)^{-1} \ma^\top \sqrt{\tfrac{\mx}{\ms}}. 
\end{align}

\looseness=-1Our high-level strategy then is to first compute an initial feasible solution (\cref{sec:init-feasible-sol}) and then iteratively compute the updates described in \cref{eq:ipm_iter_sol}. To perform these updates efficiently, we use inverse maintenance, whose iteration count and communication complexity we discuss in \cref{sec:inv-maintain}. Finally in \cref{sec:ipm-finish}, we put these together to prove \cref{thm:ipm}.

\subsection{Finding an Initial Feasible Point}
\label{sec:init-feasible-sol}
We  construct our initial iterates following the approach of \cite{van2020solving}. We first obtain an initial set of primal and dual  points $\vx\approx\mathbf{1}$ and $\vs\approx\mathbf{1}$, which are feasible for the modified linear program in \cref{def:modified-tall-LP}, and as stated in \cref{lemma:init-feasible-tall} (proved in \cite{van2020solving}), these vectors may be easily transformed to be a set of feasible points for the original linear program.   
\begin{definition}[Theorem 12 of \cite{van2020solving}]
\label{def:modified-tall-LP}
For a linear program $\min_{\ma^\top \vx = \vb, \vx\geq 0} \vc^\top \vx$ with outer radius $R$, and any $\epsilon \in (0,1]$, we define the modified linear program to
$
\min_{\mabar^\top \vxbar=\vbbar, \vxbar\geq 0} \vcbar^\top \vxbar,
$
where
\[
\mabar = \begin{bmatrix}
\ma & \norm{\ma}{\fro} \cdot \boldsymbol{1}_n \\
0 & \norm{\ma}{\fro} \\
\frac{1}{R} \vb^\top - \boldsymbol{1}^\top \ma & 0
\end{bmatrix} \in \R^{(n+2)\times (d+1)}, \vbbar = \begin{bmatrix}
    \frac{1}{R}\vb \\ (n+1) \norm{\ma}{\fro}
\end{bmatrix}\in\R^{d+1}, \vcbar=\begin{bmatrix}
    \frac{\epsilon}{\norm{\vc}{2}} \cdot \vc \\ 0 \\ 1
\end{bmatrix}\in\R^{n+2}.
\]
\end{definition}

\begin{lemma}[Theorem 12 of \cite{van2020solving}]
\label{lemma:init-feasible-tall}
The following are feasible primal and dual vectors for the modified linear program stated in \cref{def:modified-tall-LP}:
\[
\vxbar=\boldsymbol{1} \in \R^{n+2}, \,\, \vybar = \begin{bmatrix}
    \boldsymbol{0}_d \\ -1
\end{bmatrix} \in \R^{d+1}, \,\, \vsbar = \begin{bmatrix}
    \boldsymbol{1}_n + \frac{\epsilon}{\norm{c}{2}} \cdot \vc \\ 1 \\ 1
\end{bmatrix}\in\R^{n+2}. 
\]
Let $(\vxhat,\vyhat,\vshat)$ be an arbitrary set of primal dual vectors of the modified linear program, and let  $\rho := \frac{1}{\mu} \cdot \|\mabar^\top \vxhat - \vbbar\|^2_{(\mabar^\top \mxhat \mshat^{-1} \mabar)^{-1}} = O(1)$. Then, 
\[
\norm{\vxhat}{\infty} = O(n) \cdot (1+ \rho/\mu).
\]
Moreover for $\mu < \frac{\epsilon^2}{8d}$ and $\vxhat,\vshat$ such that $\vxhat \odot \vshat \approx_{0.5} \mu \cdot \tau(\vxhat, \vshat)$ (where $\tau$ is defined in \cref{eq:ourCentrality}), if we set $\vx = R \cdot  \vxhat_{1:n}$ (where $\vxhat_{1:n}$ is the vector of the first $n$ coordinates of $\vxhat$), then $\vx\geq 0$ is an approximate solution to the original linear program in the following sense: 
\[ \vc^\top \vx  \leq \min_{\ma^\top \vx = \vb,\vx\geq 0} \vc^\top \vx + O(n R \cdot \norm{\vc}{2}) \cdot (\sqrt{\rho} + \epsilon), \,\, \norm{\ma^\top \vx - \vb}{2}  \leq O(n^2) \cdot (\norm{\ma}{\fro} \cdot R + \norm{\vb}{2}) \cdot (\sqrt{\rho} + \epsilon).   \]
\end{lemma}

\looseness=-1It is critical to note that although $(\vxbar,\vybar,\vsbar)$ in \cref{lemma:init-feasible-tall} is feasible, it is not necessarily on the central path (or even near it) since $\vxbar \odot \vsbar$ is not necessarily near $\mu \cdot \tau(\vxbar,\vsbar)$. To enforce this centrality with $\vx=\mathbf{1}_{n+2}$ and $\mu = 1$, we must have $\vs \approx \tau(\mathbf{1}, \vs),$ which corresponds roughly to the definition of regularized $\ell_p$ Lewis weights of $\ma$ for $p=\tfrac{1}{1+\alpha}.$  
The fact that $\alpha >0$ corresponds to $p<1$, which is a range of $p$ for which $\ell_p$ Lewis weights may be efficiently computed via a simple fixed-point iteration~\cite{cohen2015lp}. We note that  for the interval complementary to that in \cite{cohen2015lp}, \cite{fazel2022computing} provides an efficient (but different) algorithm for computing $\ell_p$ Lewis weights. Since in our setting, $p < 1$, we use the algorithm by \cite{cohen2015lp}, whose main guarantee we restate below.  

\begin{lemma}[Lemma 3.2 of \cite{cohen2015lp}]
\label{lemma:lewis-closer}
Given $\ma\in\R^{n\times d}$, $0<p<4$, and vectors $\vv,\vw\in\R^{n}$ such that $\vv\approx_{\zeta} \vw$, define  
$
\vvhat_i = (\va_i^\top (\ma^\top \mv^{1-2/p} \ma)^{-1} \va_i)^{p/2}$ and $\vwhat_i = (\va_i^\top (\ma^\top \mw^{1-2/p} \ma)^{-1} \va_i)^{p/2}.  
$
Then $\vvhat$ and $\vwhat$ satisfy the approximation \[\vvhat \approx_{\zeta\cdot{\abs{p/2-1}}} \vwhat,\] where, for some given $\alpha>0$, we use the notation  $x\approx_{\alpha}y$  to denote $y \cdot e^{-\alpha}\leq x\leq y\cdot e^{\alpha}$. 
\end{lemma}

The proof of \cite{cohen2015lp} for convergence to the Lewis weights, for $p \in (0, 4)$, through the iterative process described above is based on the fact that $\abs{p/2-1}<1$ and hence results in a contraction. We also use this observation to  bound the communication complexity of computing the (approximate) regularized Lewis weights in \cref{lemma:lewis-weight-communication}, adapted from Corollary 3.4 of \cite{cohen2015lp} and Theorem 13 of \cite{van2020solving}. Before proving this lemma, we first establish some notation for easier readability, followed by two technical results, which we invoke in order to prove \cref{lemma:lewis-weight-communication}. 

\begin{definition}\label{def:map_T_map_Ttilde}
    Given a vector $\vx$ and a diagonal matrix $\mx$ formed using $\vx$, we define the functions  \[\lewisWtApx(\vx):= \mathbf{a}_i^\top (\ma^\top \mx^{1-2/p} \ma)^{-1} \mathbf{a}_i\text{ and } \mathcal{T}_i(\vx) := (\lewisWtApx(\vx) + \eta \vx_i^{2/p-1})^{p/2}, \numberthis\label{eq:map_T}\] and the following approximation to $\mathcal{T}$, with both additive and multiplicative error, \[ \widetilde{\mathcal{T}}_i(\vx):= (\widetilde{\mathbf{q}}_i(\vx)+ \eta \vx_i^{2/p-1})^{p/2}, \text{ where } \lewisWtApx(\vx) + e^{-\varepsilon_1} \cdot\frac{\eta^{2/p}}{n^2} \leq \widetilde{\mathbf{q}}_i(\vx) \leq e^{\varepsilon_1}\cdot
\lewisWtApx(\vx) + \frac{\eta^{2/p}}{n^2}.\numberthis\label{eq:map_T_tilde} \]We remark that our $\mathcal{T}$ is identical to $T$ in Theorem $13$ of \cite{van2020solving}.
\end{definition}

\begin{claim}[\cite{van2020solving}]\label{claim:x_approx_y_implies_Tx_approx_Ty}
    For some $\alpha>0$, let $x$ and $y$ be positive numbers satisfying $x\approx_{\alpha} y$. Then, for $\lambda, u, k>0$, we have $\lambda x\approx_{\alpha}\lambda y $, $(x+u) \approx_{\alpha} (y+u)$, and $x^k \approx_{\alpha k} y^k$. 
\end{claim}

Next, we extend \cref{claim:x_approx_y_implies_Tx_approx_Ty} to the following lemma, which essentially says that if two vectors $\vx$ and $\vy$ are $\alpha$-approximations of each other (in the sense stated in \cref{lemma:lewis-closer}), then $\widetilde{\mathcal{T}}(\vx)$ and $\widetilde{\mathcal{T}}(\vy)$ are also close. We use this lemma to inductively prove \cref{lemma:lewis-weight-communication}.  
\begin{lemma}\label{lem:x_y_close_Ttildex_Ttildey_close}
    Consider the notation from \cref{def:map_T_map_Ttilde}. Suppose for some $\alpha >0$, the vectors $\vx$ and $\vy$ satisfy $\vx\approx_{\alpha}\vy$. Then we have $\widetilde{\mathcal{T}}_i(\vx)\approx_{p/2(\epsilon_1 + \alpha|1-2/p|)}\widetilde{\mathcal{T}}_i(\vy)$, where $\epsilon_1 = \begin{cases} 
     p \cdot \epsilon/4 & \text{ if } p\leq 2 \\ 
     \frac{\epsilon \cdot (4-p)}{4}  & \text{ if } p>2
    \end{cases}$. 
\end{lemma}
\begin{proof}
Since $\vx\approx_{\alpha}\vy$, by \cref{claim:x_approx_y_implies_Tx_approx_Ty}, we have 
\[
\mathbf{q}_{i}(\vy)\approx_{\alpha|1-2/p|}\mathbf{q}_{i}(\vx)\text{ and }\vy_{i}^{2/p-1}\approx_{\alpha|1-2/p|}\vx_{i}^{2/p-1}.
\] 
Let $\gamma=\epsilon_{1}+\alpha|1-2/p|$ for some $\epsilon_{1}>0.$
Then the above approximations may equivalently be expressed as
\[
\mathbf{q}_{i}(\vx)\approx_{\gamma-\epsilon_{1}}\mathbf{q}_{i}(\vy)\text{ and }\vx_{i}^{2/p-1}\approx_{\gamma-\epsilon_{1}}\vy_{i}^{2/p-1}.
\] This implies, for the above choice of $\gamma,$ the following inequalities
hold: 
\[
e^{\epsilon_{1}}\mathbf{q}_{i}(\vy)+\frac{\eta^{2/p}}{n^{2}}+\eta\vy_{i}^{2/p-1}\leq e^{\gamma}\mathbf{q}_{i}(\vx)+e^{\gamma-\epsilon_{1}}\frac{\eta^{2/p}}{n^{2}}+\eta e^{\gamma}\vx_{i}^{2/p-1},
\]
\[
e^{\epsilon_{1}}\mathbf{q}_{i}(\vx)+\frac{\eta^{2/p}}{n^{2}}+\eta\vx_{i}^{2/p-1}\leq e^{\gamma}\mathbf{q}_{i}(\vy)+e^{\gamma-\epsilon_{1}}\cdot\frac{\eta^{2/p}}{n^{2}}+\eta e^{\gamma}\vy_{i}^{2/p-1}.
\]
Chaining these inequalities with the definition of $\widetilde{\mathcal{T}}$ in \cref{def:map_T_map_Ttilde}
finishes the proof.
\end{proof}

\begin{lemma}
\label{lemma:lewis-weight-communication}
Given $1>\eta>0$, $0<\epsilon<0.5$, $p\in(0,4)$ and the setting of \cref{def:lin-reg-setting} with input matrix $\ma=[\ma^{(i)}]\in \R^{n\times d}$, there is a randomized algorithm that, with high probability, outputs a vector $\vwhat$ such that
\[
\vwhat \approx_{\epsilon} \sigma(\mwhat^{1/2 - 1/p} \ma) + \eta \cdot \boldsymbol{1}.
\] 
The number of bits of communication this algorithm uses is \[\Otil\left((d^2 L+ sd (L+\log\kappa + p^{-1}\log(\eta^{-1})))\cdot  \frac{\log(\eta^{-1}\epsilon^{-1} p^{-1})}{1-\abs{p/2-1}}\right).\]   
Each machine has access to $\vwhat_i\in\R^{n_i}$, which is the part of $\vwhat$ corresponding to the rows of $\ma^{(i)}$.
\end{lemma}
\begin{proof} For  conciseness inside this proof, we use the  notation from \cref{def:map_T_map_Ttilde}. The algorithm that achieves this lemma's stated guarantee constructs  iterates  $\vwhat^{(t)}$ as follows: \[ 
\vwhat^{(0)}=\eta \cdot \boldsymbol{1}\,\, \text{ and }\,\,  \vwhat^{(t)} = \widetilde{\mathcal{T}}_i(\vwhat^{(t-1)}).\numberthis\label{eq:wHatEvolution}
\] 
Similarly, the algorithm constructs iterates $\mathbf{w}^{(t)}$ as follows: 
\[
\vw^{(0)}=\eta \cdot \boldsymbol{1} \,\, \text{ and }\,\,
\vw^{(t)}_i= \mathcal{T}_i(\vwhat^{(t-1)}) ,\numberthis\label{eq:wEvolution}
\]
noting that the difference from \cref{eq:wHatEvolution} is in the use of $\mathcal{T}$ instead of $\widetilde{\mathcal{T}}$. 
Starting with the definition of $\vwhat^{(0)}$ and update rule in \cref{eq:wHatEvolution}, we obtain the following lower bound on $\vwhat_i^{(1)}$: 
\begin{align}
\label[ineq]{eq:lewis-bound-1}
\vwhat_i^{(1)} & = 
(\widetilde{\mathbf{q}}_i(\widehat{\vw}_i^{(0)}) + \eta^{2/p})^{p/2}
 \geq
(\eta^{2/p-1}\cdot \va_i^\top (\ma^\top \ma)^{-1} \va_i + \eta^{2/p})^{p/2} \geq \eta = \vwhat_i^{(0)}.
\end{align}
To obtain an upper bound on $\vwhat_i^{(1)}$, we start with the definition of $\vwhat^{(1)}$, followed by the upper bound from \cref{eq:map_T_tilde} with  $\epsilon_1 < 0.04$: 
\begin{align*}
\vwhat_i^{(1)} & =(\widetilde{\mathbf{q}}_i(\widehat{\vw}_i^{(0)}) + \eta^{2/p})^{p/2}
 \leq (e^{\epsilon_1} \cdot \eta^{2/p-1} \va_i^\top (\ma^\top \ma)^{-1} \va_i + \frac{\eta^{2/p}}{n^2} + \eta^{2/p})^{p/2}
 \\ & \leq \nonumber 
 \eta \cdot (1 + 1.1 \cdot \frac{1}{\eta} + \frac{1}{n^2})^{p/2}
 \leq
 \eta \cdot \exp((p/2)\cdot(1.1 \cdot \frac{1}{\eta} +\frac{1}{n^2}))
 \\ & = 
 \vwhat^{(0)}_i \cdot \exp((p/2)\cdot(1.1 \cdot \frac{1}{\eta} +\frac{1}{n^2})), \numberthis\label[ineq]{eq:lewis-bound-2}
\end{align*}
where the second step uses  \cref{eq:map_T_tilde}, the third step uses the upper bound $\mathbf{a}_i (\mathbf{A}^\top\mathbf{A})^{-1} \mathbf{a}_i\leq 1$, and the  fourth step uses the inequality $1+x\leq e^x$.
Therefore, from \cref{eq:lewis-bound-1} and \cref{eq:lewis-bound-2}, 
we may conclude that $\vwhat_i^{(1)}\approx_{\beta} \vwhat_i^{(0)}$ for $\beta = (p/2)\cdot(1.1 \cdot \frac{1}{\eta} +\frac{1}{n^2})$. 
We may now invoke \cref{lem:x_y_close_Ttildex_Ttildey_close} on $\vwhat^{(1)}$ and $\vwhat^{(0)}$ inductively to claim that  
\[
\vwhat^{(t+1)} \approx_{\beta\cdot \abs{p/2-1}^t + \theta_t} \vwhat^{(t)},
\]
where $\theta_t = \epsilon_1 \cdot (p/2) \cdot \sum_{j=1}^t |p/2-1|^{j-1}$. If $p\leq 2$, then this formula implies that $\theta_t \leq \epsilon_1$;  if $p>2$, then we have $\theta_t \leq \epsilon_1\cdot \frac{p}{4-p}$.
For $p\leq 2$, we pick $\epsilon_1 = p \cdot \epsilon/4$, and if $p>2$, we pick $\epsilon_1 =  \frac{\epsilon \cdot (4-p)}{4}$.
Since $0<p<4$, $0<\eta<1$, and $\abs{\frac{p}{2}-1}\leq \exp(\abs{\frac{p}{2}-1} -1)$,
after $O(\frac{\log(\eta^{-1}\epsilon^{-1} p^{-1})}{1-\abs{p/2-1}})$ iterations, we have $k$ such that
\[
\vwhat^{(k)} \approx_{p\cdot \epsilon/2} \vwhat^{(k-1)}.
\]
Next, based on the update rule for $\vw_i^{(t)}$ in \cref{eq:wEvolution}, we infer  
$
\vw^{(k+1)}_i 
=
(\vwhat^{(k)}_i)^{1-p/2} \cdot (\sigma_i((\mwhat^{(k)})^{1/2-1/p} \ma) + \eta)^{p/2}$. 
By \cref{lem:x_y_close_Ttildex_Ttildey_close} and arguments similar to above,
\[
\vw^{(k+1)} \approx_{p\cdot \epsilon/2} \vwhat^{(k)}.
\]
Therefore
\[
(\vwhat^{(k)}_i)^{1-p/2} \cdot (\sigma_i((\mwhat^{(k)})^{1/2-1/p} \ma) + \eta)^{p/2} \approx_{p\cdot \epsilon/2} \vwhat^{(k)}_i.
\]
Thus
\[
(\sigma_i((\mwhat^{(k)})^{1/2-1/p} \ma) + \eta)^{p/2} \approx_{p\cdot \epsilon/2} (\vwhat^{(k)}_i)^{p/2}.
\]
Raising above to the power of $2/p$, we have
\[
\sigma_i((\mwhat^{(k)})^{1/2-1/p} \ma) + \eta \approx_{\epsilon} \vwhat^{(k)}_i.
\]
The communication complexity follows from the number of iterations and \cref{lem:levScoreOverestimates}. Note that the guarantee of \cref{eq:map_T_tilde} follows from the proof of \cref{lem:levScoreOverestimates} and the term $sd \cdot\frac{\log(\eta^{-1})}{p})$ appears because we need a $\frac{\eta^{2/p}}{n^2}$ additive error for the leverage score computation.
\end{proof}

\begin{figure}
\begin{framed}

\textbf{Input. } A matrix $\ma \in \R^{n\times d}$ and vectors $\vb\in\R^{d},\vc\in\R^n$ with parameters in \cref{thm:ipm}; Error parameters $0 < \epsilon < 1$.

\vspace{.5em}

\textbf{Output. } A vector $\vxhat\in\R^{n}_{\geq 0}$ satisfying \cref{eq:lpSolutionConditions}. 

\vspace{.5em}
    
\begin{enumerate}[itemsep = .1em, leftmargin = 1.7em, topsep = .4em, label=\protect\circled{\arabic*}]

\item Let $\alpha := \tfrac{1}{4\log(4n/d)}$, $\lambda := \tfrac{\alpha}{32000} \cdot \log (2^{16} n \tfrac{\sqrt{d}}{\alpha^2})$, and $\gamma := \min\{\tfrac{\alpha}{64000}, \tfrac{\alpha}{50 \lambda}\}$, $\veps=0.1$

\item Let $\mabar,\vbbar,\vcbar,\vxbar,\vybar,\vsbar$ be as defined in \cref{def:modified-tall-LP} for the modified linear program.

\item Set $\vshat\in\R^{n+2}$ to a vector with $\vshat \approx_{\veps} \sigma(\mshat^{1/2 - \alpha} \ma) + \frac{d}{n} \cdot \boldsymbol{1}$ (see \cref{lemma:lewis-weight-communication}).

\item Set $\vchat = \vshat$, $\tauhatv=\vshat$, $\vyhat = \boldsymbol{0}$, $\widehat{\mu}=1$ \label{step:lewis-weight-initial}

\item Let $(\vxhat^{(\final)},\vyhat^{(\final)},\vshat^{(\final)}, \tauhatv^{(\final)}, \widehat{\mu}^{(\final)}) = \ipm(\mabar,\vbbar,\vchat,\vxbar,\vyhat,\vshat, \tauhatv^{(\final)}, \widehat{\mu}, \theta(n^2 \sqrt{d}/(\gamma \alpha^2)))$ \label{alg-step:ipm-init-modified-initial} 

\item Set $\vx^{(0)}=\vxbar^{(\final)}$, $\vy^{(0)}=\vybar^{(\final)}$ and $\vs^{(0)}=\vsbar^{(\final)} + \vcbar - \vchat$, $\boldsymbol{\tau}^{(0)}=\tauhatv^{(\final)}$, $\mu^{(0)} = \widehat{\mu}^{(\final)}$

\item Let $(\vx^{(\final)},\vy^{(\final)},\vs^{(\final)},\boldsymbol{\tau}^{(\final)},\mu^{(\final)}) = \ipm(\mabar,\vbbar,\vcbar,\vx^{(0)},\vy^{(0)},\vs^{(0)},\boldsymbol{\tau}^{(0)}, \mu^{(0)},\epsilon^2 /(512 \cdot n^4 d))$ \label{alg-step:ipm-init-original-initial}

\item \Return $R \cdot \vx^{(\final)}_{1:n}$
\end{enumerate}
\end{framed}
\captionsetup{belowskip=-10pt}
\caption{LP Solver}
\label[alg]{alg:ipm-init}
\end{figure}

\begin{figure}[!htb]
\begin{framed}
\textbf{Input.} A matrix $\ma\in  \R^{n \times d}$ and vectors $\vd^{(0)},\vd^{(1)},\ldots,\vd^{(r)}\in\R^{n}$ and their corresponding diagonal matrices $\md^{(0)},\md^{(1)},\ldots,\md^{(r)}$. For $i\geq 1$, each $\vd^{(i)}$ is received after we returned the output for $\vd^{(i-1)}$.  

\vspace{.5em}

\textbf{Output.} For $i\geq 0$, the output is a spectral approximation to $(\ma^\top \md^{(i)}\ma)^{-1}$.

\vspace{.5em}

\begin{enumerate}[itemsep = .1em, leftmargin = 1.7em, topsep = .4em, label=\protect\circled{\arabic*}]
    \item Let $\gamma = 1000 C \cdot \log d$ (where $C$ is some absolute constant)
    
    \item Compute (with high probability using \cref{alg:levscoresRefinementSampling}) $\lev^{(\appr)}\in\R^{n}$ such that $0.99\cdot\lev^{(\appr)}_i \leq \lev_{\ma}(\vd^{(0)})_i \leq 1.01\cdot\lev^{(\appr)}_i$
    
    \item Set $\vd^{(\old)} = \vd^{(0)}$ and $\lev^{(\old)} = \lev^{(\appr)}$
    
    \item For each $i\in[n]$, let 
    \[
    \vh_{i}^{(0)} = \begin{cases}
    \vd_i^{(0)} / \min\{1, \gamma \cdot \lev_{i}^{\appr}\} & , \text{ with probability } \min \{1, \gamma \cdot \lev_{i}^{\appr}\} 
    \\
    0 & , \text{ otherwise} 
    \end{cases}
    \]
    
    \item Set $\mk^{(0)} = (\ma^\top \mh^{(0)} \ma)^{-1}$
    
    \item For $j = 1,\ldots, r$
    \begin{enumerate}
        \item Use Richardson's iteration and $\mk^{(j-1)}$ to compute JL sketchings of the form $(\ma^\top \md^{(j)} \ma)^{-1} \ma^\top (\md^{(j)})^{1/2} \mg$ (Each machine computes its own part of $\ma^\top (\md^{(j)})^{1/2} \mg$ and sends it to the coordinator and then there is a back and forth to solve the linear system using Richardson). Then compute $(\md^{(j)})^{1/2} \ma (\ma^\top \md^{(j)} \ma)^{-1} \ma^\top (\md^{(j)})^{1/2} \mg$ and use it to compute $\lev^{(\appr)}\in\R^{n}$ such that $0.99\cdot\lev^{(\appr)}_i \leq \lev_{\ma}(\vd^{(j)})_i \leq 1.01\cdot\lev^{(\appr)}_i$
        
        \item For each $i\in[n]$, if $\abs{\lev^{(\appr)}_i - \lev^{(\old)}_i}/\lev^{(\old)}_i > 0.1$ or $\abs{\vd^{(j)}_i - \vd^{(\old)}_i}/\vd^{(\old)}_i > 0.1$, then
        \begin{enumerate}
            \item Set $\vd^{(\old)}_i = \vd^{(j)}_i$, $\lev^{(\old)}_i = \lev^{(\appr)}_i$, and 
            \[
            \vh_{i}^{(j)} = \begin{cases}
            \vd_i^{(j)} / \min\{1, \gamma \cdot \lev_{i}^{\appr}\} & , \text{ with probability } \min \{1, \gamma \cdot \lev_{i}^{\appr}\} 
            \\
            0 & , \text{ otherwise} 
            \end{cases}
            \]
        \end{enumerate}
        \item else
        \begin{enumerate}
            \item $\vh_{i}^{(j)} = \vh_{i}^{(j-1)}$
        \end{enumerate}
        \item Set $\mk^{(j)} = (\ma^\top \mh^{(j)} \ma)^{-1}$
    \end{enumerate}
\end{enumerate}
\end{framed}
\captionsetup{belowskip=-10pt}
\caption{Inverse Maintenance for LP}
\label[alg]{alg:inv-maintenance}
\end{figure}

As we discussed, after computing $\vshat$ such that $\vshat \approx_{\veps} \sigma(\mshat^{-1/2-\alpha} \mabar)+\frac{d}{n} \boldsymbol{1}$, we have a point near the central path for the modified linear program for the modified objective vector $\vchat=\vshat$. Therefore, we can run the IPM so that $\mu$ is small enough. This part is illustrated as Step 7 of \cref{alg:ipm-init}. By notation of \cref{alg:ipm-init}, and the guarantees of our IPM, we have
\[
\vxhat^{(\final)} \odot \vshat{(\final)} = \widehat{\mu}^{(\final)} \cdot \tauhatv^{(\final)}.
\]
Since we have $\ma \vyhat^{(\final)} + \vshat^{(\final)} = \vchat$ (this equality is exact since we do not explicitly update $\vyhat$) and $\vchat=\vshat$, $\ma \vyhat^{(\final)} + \vs^{(0)} - \vcbar + \vchat = \vchat$. Thus $\ma \vyhat^{(\final)} + \vs^{(0)} = \vcbar$. Moreover, by construction $\vs^{(0)}=\vsbar^{(\final)} + \vcbar - \vchat$. Therefore $\frac{\vs^{(0)} - \vs^{(\final)}}{\vs^{(0)}} = \frac{\vcbar - \vchat}{\vs^{(0)}}$. Since by the guarantees of the IPM $\vs^{(0)} \odot \vx^{(0)} \approx_{\veps} \mu^{(0)} \cdot \boldsymbol{\tau}^{(0)}$, and $\veps \leq 0.5$, for each $i\in [n]$, we have
\[
\vs^{(0)}_i \geq \frac{\mu^{(0)}}{2\vx^{(0)}_i} \cdot \frac{d}{n} \geq \frac{\mu}{\Omega(n^2)}, 
\]
where the second inequality follows from \cref{lemma:init-feasible-tall} by $\norm{\vx}{\infty}=O(n)$. Moreover, note that $\norm{\vcbar}{\infty} \leq 1$ and $\norm{\vchat}{\infty} = \norm{\vshat}{\infty} \leq 3$. Therefore by picking the appropriate constant in Step 7 of \cref{alg:ipm-init}, for any constant $\beta$, we have
\[
\norm{\frac{\vs^{(0)} - \vs^{(\final)}}{\vs^{(0)}}}{\infty} = \norm{\frac{\vcbar - \vchat}{\vs^{(0)}}}{\infty} \leq \frac{16 n^2}{\mu^{(0)}} \leq \frac{\gamma \cdot \alpha^2}{\beta \cdot \sqrt{d}}.
\]
Therefore, we can pick $\beta$ small  enough so that $\vx^{(0)} \odot \vs^{(0)} \approx_{2\veps} \mu^{(0)} \boldsymbol{\tau}^{(0)}$, where \[\boldsymbol{\tau}^{(0)} \approx_{\gamma/4} \sigma((\ms^{(0)})^{-1/2-\alpha} (\mx^{(0)})^{1/2-\alpha}).\]
Thus after Step 9 of \cref{alg:ipm-init}, we have $\vx^{(\final)} \odot \vs^{(\final)} \approx_{\veps} \mu^{(\final)} \cdot \boldsymbol{\tau}^{(\final)}$.

We finally need to bound the condition number of the matrix $\mabar$. To do this, we use the following result.

\begin{lemma}[Lemma 5.15 of \cite{ghadiri2023bit}]
\label{lemma:cond-num-column-add}
Let $\ma \in \R^{n\times d}$, $n > d$, be a matrix with full column rank. Moreover let $\vg \in \R^n$.
Suppose $\kappa>1$, and
\[\norm{\ma^\top \ma}{\fro},\norm{(\ma^{\top} \ma)^{-1}}{\fro},\norm{\vg}{2}, 1/\norm{(\mi-\ma(\ma^{\top} \ma)^{-1} \ma^\top)\vg}{2} \leq \kappa.\] 
Then
$\norm{\mabar^\top \mabar}{\fro}, \norm{(\mabar^\top \mabar)^{-1}}{\fro}\leq 8\kappa^7$,
where
$
\mabar = \begin{bmatrix}
\ma | \vg
\end{bmatrix}.
$
\end{lemma}

Note that setting
\[
\matil = \begin{bmatrix}
\ma  \\
0  \\
\frac{1}{R} \vb^\top - \boldsymbol{1}^\top \ma
\end{bmatrix},
\]
since $\matil^{\top} \matil = \ma^\top \ma + (\frac{1}{R} \vb^\top - \boldsymbol{1}^{\top} \ma)^\top (\frac{1}{R} \vb^\top - \boldsymbol{1}^{\top} \ma)$, the condition number of $\matil$ is smaller than the condition number of $\ma$. Now by \cref{lemma:cond-num-column-add}, we have \[\norm{\mabar^\top \mabar}{\fro}, \norm{(\mabar^\top \mabar)^{-1}}{\fro}\leq 8(\kappa(\ma)+2\norm{\ma}{\fro})^7=O((\kappa(\ma)+2^L \sqrt{nd})^7).\]

\subsection{Inverse Maintenance for IPM}
\label{sec:inv-maintain}

In this section, we present the subprocedure $\ipm$ of \cref{alg:ipm-init}. This is presented in \cref{alg:lin-prog-coordinator-poly-cond}. Essentially, in each iteration, in \cref{item:ipm_compute_delta_mu}, each machine computes its own part of the gradient of the potential at a specific point and sends it to the coordinator. The coordinator then in \cref{item:ipm_compute_spec_apx} sums the vectors obtained from the machines and pre-multiplies this sum by a matrix $\mh^{-1}$ that spectrally approximates $(\ma^\top \mx \ms^{-1} \ma)^{-1}$ and sends the result to all the machines. In \cref{item:ipm_updateX_updateS_machine}, each machine uses this vector to update its own part of the primal and slack vectors. As discussed in \cref{lemma:init-feasible-tall}, we do not need exact feasibility, and therefore, the vectors sent to the coordinator or to the machine can be rounded down. Essentially, for each entry, we only need to send $L+\log(\kappa R/r)$ of its bits to guarantee convergence (see  \cite[Theorem 32]{van2020solving}).

The only remaining part for bounding the communication complexity of the algorithm is to show that \cref{step:send-resample} of \cref{alg:lin-reg-coordinator-poly-cond} does not resample too many rows over the course of the algorithm.

To improve the running time of the IPMs based on the Lee-Sidford barrier, \cite{DBLP:conf/focs/LeeS15} introduced the following $\sigma$-stability property. As discussed in Lemma 21 of \cite{van2020solving}, the primal and slack vectors in \cref{alg:lin-reg-coordinator-poly-cond} satisfy this property. This then allows us to bound the number of resampled rows in  \cref{step:send-resample} of \cref{alg:lin-reg-coordinator-poly-cond}.

\begin{definition}[$\sigma$-stability assumption]
We say that the inverse maintenance problem satisfies the $\sigma$-stability assumption if for each $k\in[r]$ (where $r$ is the number of rounds/iterations of the algorithm), we have 
\[
\norm{\log(\vd^{(k)})-\log(\vd^{(k-1)})}{\sigma(\md^{(k)} \ma)} \leq 0.1,
\]
\[
\norm{\log(\vd^{(k)})-\log(\vd^{(k-1)})}{\infty} \leq 0.1,
\]
and
\[
\beta^{-1} \ma^\top \md^{(0)} \ma \preceq \ma^\top \md^{(k)} \ma \preceq \beta \ma^\top \md^{(0)} \ma,
\]
for $\beta=\poly(n)$.
\end{definition}

\begin{figure}[!htb]
\begin{framed}

 \textbf{Input.} A matrix $\ma:= [\ma^{(i)}]\in  \R^{n \times d}$, vector $\vb\in \R^{d}$, and vector $\vc:= [\vc^{(i)}]\in \R^{n}$, where the $i^{\mathrm{th}}$ machine stores matrix $\ma^{(i)}\in \R^{n_i\times d}$; initial feasible primal, dual, slack and weight vectors $\vx^{(0)}$, $\vy^{(0)}$, $\vs^{(0)}$, $\boldsymbol{\tau}^{(0)}$, respectively. Initial and final centrality parameter $\mu^{(0)}$ and $\mu^{(\final)}$. 

  \vspace{.5em}

 \textbf{Output.} Vector $\vxhat\in\R^{d}$.

 \vspace{.5em}

\textbf{procedure }$\ipm(\ma,\vb,\vc,\vx^{(0)},\vy^{(0)},\vs^{(0)}, \boldsymbol{\tau}^{(0)}, \mu^{(0)}, \mu^{(\final)})$:

\begin{enumerate}[itemsep = .1em, leftmargin = 1.7em, topsep = .4em, label=\protect\circled{\arabic*}]
    
    \item Coordinator sets $\alpha = \tfrac{1}{4\log(4n/d)}$, $\lambda=\tfrac{2}{\epsilon}\log(\frac{2^{16}n\sqrt{d}}{\alpha^2})$, $\gamma=\min(\frac{\epsilon}{4},\frac{\alpha}{50\lambda})$, $\mu = \mu^{(0)}$.

    \item Each machine sets its components of the vectors as follows: $\vxbar = \vx^{(0)}$, $\vsbar = \vs^{(0)}$, $\overline{\boldsymbol{\tau}} = \boldsymbol{\tau}^{(0)}$, $\vx^{(\text{tmp})} = \vx^{(0)}$, $\vs^{(\text{tmp})} = \vs^{(0)}$, $\boldsymbol{\tau}^{(\text{tmp})} = \boldsymbol{\tau}^{(0)}$, $\vwbar = \mxbar \vsbar$.

    \item The coordinator computes  $\mh\in \R^{d\times d}$ with $\mh \approx_{\tfrac{\beta\veps}{d^{1/4} \log^3 n}}\ma^\top \tfrac{\mxbar}{\msbar} \ma$  and  $\mh^{-1}$. \label{step:first-spectral-approx}
    
    \item \textbf{while} true \textbf{do}
    \begin{enumerate}

        \item Let $\Phi(\vecv):=\exp(\lambda(\vecv-1))+\exp(-\lambda(\vecv-1))$ for $\vecv\in\R^n$.
    
        \item\label{item:ipm_compute_delta_mu} Each machine $i$ computes  ${(\ms^{(i)})}^{-1}\dmutil^{(i)} \in \R^{n_i}$ and ${(\mai{i})}^\top {(\ms^{(i)})}^{-1}\dmutil^{(i)} \in \R^{d}$, where $\dmutil \approx \nabla \Phi(\mwbar^{-1}\mutil)$, and sends it to the coordinator. \label{step:send-gradient}
        
        \item\label{item:ipm_compute_spec_apx} The coordinator computes $\mh^{-1} \sum_{i=1}^s {\mai{i}}^\top {(\ms^{(i)})}^{-1}\dmutil^{(i)}\in \R^{d}$ and sends it to all of the machines. \label{step:send-hessian-gradient}
        
        \item\label{item:ipm_updateX_updateS_machine}  Each machine $i$ computes  $\frac{\mxbar^{(i)}}{\msbar^{(i)}} \mai{i}\mh^{-1} {\ma}^\top {\msbar}^{-1}\dmutil\in \R^{n_i}$ and uses it to compute $\dx^{(i)}$ and $\ds^{(i)}$ to update, respectively, $\vxbar^{(i)}$ and $\vsbar^{(i)}$.

        \item Each machine sets its portion of the vectors as the following: $\vx^{(\text{tmp})}_i = \vxbar_i$ if $\vxbar_i \approx_{\gamma/8} \vx^{(\text{tmp})}_i$; $\vs^{(\text{tmp})}_i = \vsbar_i $ if $\vsbar_i \approx_{\gamma/8} \vs^{(\text{tmp})}_i$; $ \boldsymbol{\tau}^{(\text{tmp})}_i = \overline{\boldsymbol{\tau}}_i$ if $\overline{\boldsymbol{\tau}}_i \approx_{\gamma/8} \boldsymbol{\tau}^{(\text{tmp})}_i$. Set $\vwbar = \mxbar \vsbar$ and $\vvbar = \mu \mwbar^{-1} \overline{\boldsymbol{\tau}}$.

        \item For any $i$ where either of $\vx^{(\text{tmp})}_i,\vs^{(\text{tmp})}_i$ or $\boldsymbol{\tau}^{(\text{tmp})}_i$ has changed, we resample the $i$'th row according to its leverage scores and send it to the coordinator with the corresponding probability. The coordinator updates $\mh^{-1}$ accordingly and sends a sketch of it to the machines. \label{step:send-resample}
        
        \item The machines use the sketch to update their leverage scores.

        \item If $\mu > \mu^{(\final)}$, set $\mu = \max\{\mu^{(\final)}, (1-\frac{\gamma \alpha}{2^{15} \sqrt{d}})\mu\}$. Otherwise, set $\mu = \max\{\mu^{(\final)}, (1+\frac{\gamma \alpha}{2^{15} \sqrt{d}})\mu\}$.

        \item If $\mu = \mu^{(\final)}$ and $\Phi(\vvbar)\leq \frac{2^{16} n \sqrt{d}}{\alpha^2}$, \textbf{break}. \label{step:compute-potential}
        
    \end{enumerate}
    \item Return $(\vxbar,\vybar,\vsbar,\overline{\boldsymbol{\tau}},\mu)$
\end{enumerate}
\end{framed}
\captionsetup{belowskip=-10pt}
\caption{Protocol for linear programming in the coordinator setting}
\label[alg]{alg:lin-prog-coordinator-poly-cond}
\end{figure}

We use \cref{alg:inv-maintenance} for our inverse maintenance (i.e., to maintain a spectral approximation of $(\ma^\top \mx \ms^{-1} \ma)^{-1}$ in the coordinator). This algorithm is inspired by Algorithm 3 of \cite{DBLP:conf/focs/LeeS15} and based on the following which is  \cite[Lemma 15]{DBLP:conf/focs/LeeS15}, the number of changes in \cref{alg:inv-maintenance} is bounded by $\Otil(d \log(\epsilon^{-1}))$ when used with \cref{alg:lin-reg-coordinator-poly-cond} since the number of iterations of \cref{alg:lin-reg-coordinator-poly-cond} is $\Otil(\sqrt{d} \cdot \log^2(\epsilon^{-1}))$.

\begin{lemma}[\cite{DBLP:conf/focs/LeeS15}]
\label{lemma:num-changes}
Suppose changes of $\vd$ and the error occurred in computing leverage scores is independent of the sampled matrix. Under the $\sigma$ stability guarantee, during the first $r$ iterations of  \cref{alg:inv-maintenance}, the expected number of coordinate changes in $\mh^{(k)}$ over all iterations $k\in[r]$ is $O(r^2 \log(d))$.
\end{lemma}

\subsection{Proof of Main Result on High-Accuracy Linear Programming}
\label{sec:ipm-finish}

We are now equipped to prove the main theorem for the communication complexity of linear programming in the point-to-point model of communication. The correctness and number of iterations of the algorithm follow from \cite{van2020solving}. Therefore, we focus on discussing only the communication complexity bounds.

\begin{proof}[Proof of \cref{thm:ipm}]
First note that, the only parts of \cref{alg:ipm-init} that requires communication are \cref{step:lewis-weight-initial}, \cref{alg-step:ipm-init-modified-initial}, \cref{alg-step:ipm-init-original-initial}. By \cref{lemma:lewis-weight-communication}, \cref{step:lewis-weight-initial} only requires $\Otil(d^2 L+ sd (L+\log\kappa))$ bits of communication.

\cref{alg-step:ipm-init-modified-initial} and \cref{alg-step:ipm-init-original-initial} of \cref{alg:ipm-init} both call the $\ipm$ procedure of \cref{alg:lin-prog-coordinator-poly-cond}. The only parts of this algorithm with communication are \cref{step:first-spectral-approx}, 
\cref{step:send-gradient}, 
\cref{step:send-hessian-gradient}, 
\cref{step:send-resample}, 
\cref{step:compute-potential}. By \cref{lem:levScoreOverestimates}, the communication cost of \cref{step:first-spectral-approx} is $\Otil(d^2 L + sd(L+\log (R\kappa/r)))$. For each iteration, the communication cost of \cref{step:send-gradient} and 
\cref{step:send-hessian-gradient} is $\Otil(sd(L+\log (R\kappa/(r\epsilon))))$. Therefore since there are $\Otil(\sqrt{d} \log(\epsilon^{-1}))$ iterations, the total cost of these steps is $\Otil(sd^{1.5}(L+\log (R\kappa/(r\epsilon))) \log(\epsilon^{-1}))$. By \cref{lemma:num-changes} and \cite[Lemma 21]{van2020solving}, the total communication cost of \cref{step:send-resample} is $\Otil(d^2 L \log^2(\epsilon^{-1}))$ 
because there are at most a total of $\Otil(d\log^2(\epsilon^{-1}))$ rows that need to be sent to the coordinator over the course of the algorithm. For \cref{step:compute-potential}, note that each machine needs to compute a number and send it to the coordinator, so the coordinator be able to compute $\Phi(\vvbar)$. Since there are $\Otil(\sqrt{d}\log(\epsilon^{-1}))$ iterations, the total cost of this over the whole course of the algorithm is $\Otil(s \sqrt{d} (L+\log(\kappa R r^{-1}\epsilon^{-1})) \cdot \log(\epsilon^{-1}))$.
\end{proof}

\section{Finite-Sum Minimization in the Blackboard Model}\label{sec:finiteSumBlackboard}
In this section, we study finite-sum minimization (in the distributed setting), i.e., \[\text{minimize}_{\vx\in\R^d} \sum_{i=1}^s f_i(\vx),\numberthis\label[prob]{prob:finSumMinGen}\] where each $f_i$ is convex, Lipschitz, and supported on only $d_i$ coordinates of $\vx$. While finite-sum minimization itself is a general problem class encompassing, for example, empirical risk minimization, with each $f_i$ measuring the loss incurred by the $i^\mathrm{th}$ data point from the training set, the additional structural assumption (of dependence on $d_i$ coordinates) is also seen in prominent problem classes like decomposable submodular function minimization~\cite{axiotis2021decomposable}. There exist numerous fast variants of stochastic gradient descent \cite{robbins1951stochastic,bottou2003large,zhang2004solving, bottou2012stochastic} for \cref{prob:finSumMinGen} such as \cite{roux2012stochastic,shalev2013stochastic,johnson2013accelerating,mahdavi2013mixed, defazio2014saga, mairal2015incremental, allen2016improved, hazan2016variance,schmidt2017minimizing} but most of these algorithms depend on the problem's condition number, which could be quite large (and hence undesirable) for non-smooth  $f_i$. Conversely, both the cutting-plane methods~\cite{lee2015faster} and robust interior-point methods~\cite{lee2021tutorial} exchange their reliance on condition number for worse dependencies on the problem dimension.

In this section, our goal is to solve this problem with efficient communication complexity. We now formally state our problem setup and all the main results of this section. 

\thmmainFinSumMain* 

\noindent In order to obtain our communication bound, we first derive  the following fine-grained cost (in terms of certain weights), which also yields improved rates for submodular function minimization. 
 
\corFinSumMinDist* 

\corSFMusingFinSumMinImproved* 

\subsection{An Overview of Our Algorithm} \label[sec]{sec:AlgOverview}

\looseness=-1The goal of this section is to prove  \cref{thm:mainFinSumMin}. We obtain this result via \cref{alg:min-sum-convex-blackboard} 
 obtained by  generalizing a technique introduced in \cite{dong2022decomposable} and then setting it in the distributed framework. 

\looseness=-1Every machine holds a copy of all the data (i.e., the current variable $\vx$ 
, and each machine $i$ holds the corresponding function's subgradient oracle $\oi$. 
Per the technique of \cite{dong2022decomposable}, we first use the standard epigraph trick to reduce \cref{prob:finSumMinGen} to the following formulation,
\[
\begin{array}{ll}
    \mbox{minimize} & \vc^\top\vx,  \\
     \mbox{subject to} & \vxi\in\ki\subseteq \R^{d_i + 1} \;\forall i\in[s]\\
     &  \ma\vx=\vb.
\end{array}\numberthis \label[prob]{eq:1main}
\]
where $\vx = [\vx_i]$ concatenates the $s$ vectors $\vx_i\in \R^{d_i}$, and all the sets $\ki$ are disjoint and convex.
\cref{eq:1main} transfers  the overlap in support between the original $f_i$'s into  $\ma\vx=\vb$. This reduction  requires only the knowledge of support of each $f_i$, so this reduction can be done using $O(\sum_{i=1}^s d_i L)$  bits of communications. After the reduction, all the machines hold all the data --- vectors $\vc, \vx, \vb$, matrix $\ma$ --- and the $i^\mathrm{th}$ machine holds the separation oracle $\oi$ for the $i^\mathrm{th}$ set $\ki$. This oracle is the only means to access the sets $\ki$ and is obtained via a reduction from the corresponding subgradient oracle for $f_i$, see \cite{lee2018efficient}. 
Specifically, at any queried point $\vzi$, the oracle either asserts $\vzi\in\ki$, or returns a separating
hyperplane that separates $\vzi$ from $\ki$.  Formally, we prove the theorem below.

\begin{theorem}[Main theorem for \cref{eq:1main}]\label{thm:MainThmOfKiProblem}
Consider the convex program 
described in \cref{eq:1main},
with every machine holding all the data and the $i^\mathrm{th}$ machine having a separation oracle access to $\ki$.
Denote $\kcal = \kcal_1 \times \kcal_2 \times \dotsc \times \kcal_s$. Assuming we have
\begin{itemizec}
    \item outer radius $R$: For any $\vx_i\in \ki$, we have $\|\vx_i\|_2 \leq R$, and
    \item inner radius $r$: There exists a $\vz \in \R^d$ such that $\ma\vz=\vb$ and $\ball(\vz,r)\subset \kcal$,
\end{itemizec} then, for any $0<\epsilon<\frac12$, we can find a point $\vx \in \kcal$ satisfying $\ma \vx = \vb$ and
\[
\vc^\top\vx \leq \min_{\substack{\vxi\in\ki\subseteq \R^{d_i + 1}\forall i\in[s],\\ \ma\vx=\vb}}\vc^\top\vx + \epsilon \cdot \|\vc\|_2 \cdot R,
\] 
in $\poly(sd\log(R/(r\epsilon))$ time in $O\left(\sum_{i=1}^s d_i^2 \log(dR/r\epsilon) \cdot L\right)$ bits of communication where $L=O(\log(dR/r)$. 
\end{theorem}

We explain our main algorithm (\cref{alg:min-sum-convex-blackboard}) to obtain our guarantee of \cref{thm:MainThmOfKiProblem} for solving \cref{eq:1main}.   
Our algorithm's inputs are the functions $f_i$ (and their corresponding first-order oracles). 
All parameters of this algorithm are set in the proof of \cref{thm:MainThmOfKiProblem}. 
\looseness=-1Before explaining \cref{alg:min-sum-convex-blackboard}, we briefly review the algorithm of \cite{dong2022decomposable}  and then describe the algorithm obtained by directly adopting this in the blackboard model. 

\paragraph{Overview of \cite{dong2022decomposable}.} This algorithm updates the variable $\vx$ via updates inspired by interior-point methods, following a central path parametrized by a parameter $t$, with a barrier function over the set $\kcal$. 
Therefore, the ideal iterates would follow the points \[\arg\min_{\vx\in\R^d:\ma\vx=\vb} \left\{t \vc^\top \vx + \sum_{i=1}^s \phi_i (\vxi)\right\},\] 
with $t$ being updated, per the classical theory of interior point methods, based on the complexity parameter of the self-concordance barrier $\sum_{i=1}^s \phi_i$ defined over $\kcal$. However,  the algorithm does \emph{not} have an explicit closed-form expression for $\kcal$. Therefore, as a proxy to $\kcal$, it maintains inner and outer set approximations $\kin$ and $\kout$ respectively, satisfying  $\kini\subseteq \ki\subseteq\kouti$ for each $i\in [s]$, and performs IPM-style updates with a barrier on $\kout$. Thus, this algorithm essentially alternates between making progress on $\vx$ or $t$ (per the IPM framework) and updating the set approximation for $\kcal$ (using ideas from classical cutting-plane methods). 

\paragraph{Blackboard adaptation of \cite{dong2022decomposable}.}\looseness=-1In the blackboard model adaptation of the above algorithm, each server would run a copy of the above algorithm. The only step where communication happens is in \cref{item:CommunicationHappens}: Before updating $\vx$ to the currently set target point $\vxos$, each server checks for feasibility of $\vxos$; if there is potential infeasibility of the $i^\mathrm{th}$ block $\vxosi$ (\cref{item:TestFeasibilityOfVxos}), then server holding the oracle $\oi$ sends the output of querying $\oi$ on $\vxosi$ to the blackboard, while all other servers read this (for free, as allowed by the model).

\looseness=-1 The above framework, developed in \cite{dong2022decomposable}, would yield a communication complexity of $\widetilde{O}(d_{\max}  \cdot\sum_{i=1}^s d_i L)$ where $d_{\max}\defeq \max_{i=1}^s d_i$, which comes from scaling the oracle query complexity of \cite{dong2022decomposable} by the worst-case cost of communication (i.e., sending the hyperplane with the maximum $d_i$). This $d_{\max}$ factor arises under the assumption that each subgradient oracle has the same cost. However, this assumption
does not align with our communication complexity setting, where we charge for each subgradient oracle call by the length of the vector it outputs. 

\paragraph{Our modification.} To capture the above nuance, we consider a \textit{weighted} version of oracle complexity, which we denote by \emph{oracle cost}. Given some arbitrary but fixed weight vector $\vw\in \R^s_{\geq 1}$, the oracle \emph{cost} is $\sum_{i=1}^s w_i\cdot n_i$ where $n_i$ is the number of times the $i$-th subgradient oracle queried.

Thus, one of our technical contributions  is  a more fine-grained analysis of this technique by using a \emph{weighted} potential. (In terms of the algorithm, the only change that happens is the rate of update to $t$, which we do at $O\left(\tfrac{\eta}{\sum_{i=1}^s w_i d_i}\right)$ (as opposed to the previous rate of $O\left(\tfrac{\eta}{\sum_{i=1}^s d_i}\right)$).) As a result of this change in potential function, we can now conclude that the total \emph{cost} of communication is $\sum_i w_i d_i L$, where we can choose what $w_i$ are. This is in contrast with the previous result, which was about the total \emph{oracle} complexity. As a result of this change, we now choose $w_i = d_i $, which gives us the improved communication cost of $\sum_{i=1}^s d_i^2 L$. This is an improvement over the previous rate when the largest $d_i$ is much larger than the rest. In the following subsection, we go into more detail into the individual steps of \cref{alg:min-sum-convex-blackboard}, which would help in understanding the analysis. 

\subsubsection{Details of Our Algorithm}
Given the current outer approximating set $\kout$,  the current central path parameter $t$, and some self-concordant barrier $\varphiouti$  defined on each set $\kouti$, we define the total barrier and the analytic center of $\kout\cap \{\ma\vx=\vb\}$ with respect to this barrier as
\begin{equation}
    \varphiout(\vx) \defeq \sum_{i=1}^s w_i \varphiouti(\vx_i), 
\text{ and } \vxos\defeq\arg\min_{\ma \vx=\vb}\left\{ t\cdot \vc^\top \vx+\sum_{i=1}^s w_i \varphiouti(\vx_i)\right\}. \label{eq:xoutstar-min}
\end{equation} 
As mentioned earlier, not knowing $\kcal$ explicitly forces us to choose, as the constraint set, between its proxies $\kin$ and $\kout$; we choose $\kout$ to ensure that we do not miss a potential solution. 

\looseness=-1Having computed the current target $\vxos$, we move the current candidate $\vx$ towards it via a Newton step, provided certain conditions of feasibility  and minimum progress are satisfied. 
If either condition is violated, we first update either $\kin$, $\kout$, or the parameter $t$, then recompute $\vxos$ and repeat the checks until they are satisfied. In our Newton step update (\cref{line:MoveX}), we normalize by   by the radius of the appropriate Dikin ellipsoid,
which ensures the feasibility of the updated $\vx$ since, by self-concordance, the unit radius Dikin ball lies
inside the domain of the self-concordance barrier. 

The conditions we check before moving our candidate $\vx$ towards $\vxos$ are that $\vxos\in\kin$ (\cref{item:TestFeasibilityOfVxos}) and the central  path parameter $t$ is large enough (\cref{item:CheckSuboptimality}). Recall that $t$ determines the suboptimality gap at the current candidate $\vx$: so, if  $\vc^\top (\vx - \vxos)\leq O(\frac{1}{t}{\sum_{i=1}^s w_i \nu_i})$, then we need to first update our next goal along the central path by updating $t$. 
If we have already reached an approximate optimum, which we verify by checking whether $t\geq O(1/\epsilon)$ in \cref{line:EndAlg}, then the algorithm  terminates by returning
\[
\vx^{\mathrm{ret}} = \arg\min_{\vx: \ma\vx=\vb} \left\{t \cdot \vc^\top{\vx} + \sum_{i = 1}^s \barrini(\vxi)\right\}.
\]
The point $\vx^{\mathrm{ret}}$ is feasible because it is in $\kin$ by definition, and  $O(1/\tend) = O(\epsilon)$ ensures that it is an approximate optimum for the original problem. 
Otherwise, following classical interior-point method theory, we increase $t$ by a scaling factor of $1 + O(1/\sum_{i=1}^s w_i d_i)$ in \cref{line:updateT-Rule}  to set the next ``target suboptimality''. 
We then recompute $\vxos$ by \cref{eq:xoutstar-min}.
Since $\vc^\top\vx > {\vc}^\top{\vxos} + O(1/t)$ is not guaranteed with the new $t$ and $\vxos$, the algorithm jumps back to the start of the main loop.

\looseness=-1To check if $\vxos\in\kin$, we check if $\langle \nabla \barrini(\vxi), \vxosi-\vxi\rangle + \eta\cdot \|\vxosi-\vxi \|_{\vxi} \leq O(\nu_i)$ for all $i\in[s]$ and for some constant $\eta>0$. Any point in the domain of a self-concordant barrier satisfies the inequalities in \cref{thm:sc1} and \cref{thm:sc2}, 
hence violating this condition
implies that $\vxosi$ is far from $\kini$, and as a result, $\vxos$ is potentially not (yet) a good candidate to move $\vx$ towards.

Therefore, to rectify the situation of $\vxos\notin\kin$, we must update either $\kini$ or $\kouti$ and compute a new $\vxos$. 
To decide which option to take, 
we query $\oi$ at $\vxosi$: 
if the oracle indicates that $\vxosi\in\ki$, 
then we incorporate $\vxosi$ into $\kini$ by redefining $\kini = \textrm{conv}\left\{\kini, \vxosi\right\}$ to be the convex hull of the current $\kini$ and $\vxosi$ (\cref{line:KinUpdated}). 
If, on the other hand, $\vxosi\notin \ki$, we choose to update the outer set $\kouti$.
Then we redefine $\kouti = \kouti \cap \mathcal{H}_i$ (\cref{line:KoutUpdated}).
In either case, the only communication that takes place is when the server that queries the oracle sends the output of the oracle to the blackboard for all other servers to read (and update their data). 
After processing this update of the sets, the algorithm recomputes $\vxos$ and returns to the main loop 
since updating the sets does not necessarily imply that the new $\vxos$ satisfies $\vxos\in \kin$.
Updating a set only after checking the very specific condition $\vxosi\notin \kini$ dramatically reduces the number of calls to the separation oracle (since this is the only part of the algorithm it is invoked) as compared to arbitrarily checking all sets. 
Over the course of the algorithm, we gradually expand $\kin$ and shrink $\kout$, until they well approximate $\kcal$, and the algorithm's final output is approximately optimal.  

\begin{figure}[!htb]

\begin{framed}

\textbf{Input. }  A total of $s$ servers and a blackboard, with the $i^{\mathrm{th}}$ server storing indices $D_i\subseteq [d]$ where $d_i = |D_i|$ and a convex function $f_i:\R^{d_i} \to \R$. Each server $i$ has access to the subgradient oracle of $f_i$. An initial vector $\theta^{(0)}\in \R^d$  and $R\in \R$  such that $\|\theta^{(0)}-\theta^\star\|_2\leq R$.

\vspace{.5em}

\textbf{Output. } A vector $\theta\in \R^d$ on all machines such that $\sum_{i=1}^s f_i(\vtheta) \leq \sum_{i=1}^s f_i(\theta^\star)  + \epsilon \cdot \mu R.$

\vspace{.5em}

\textbf{Initialize.} Set initial parameters
$m=\sum_{i=1}^s w_i d_i$, $t =  \frac{m\log m}{\sqrt{n} \|\vc\|_2 R}$, $\tend=\frac{8m}{\epsilon\|\vc\|_{2}R}$, and $\eta =\frac{1}{100}$. Reformulate the problem into \cref{eq:1main} using \cref{lem:problem-reduction}. Find the initial $\kin$ using \cref{thm:initOne}.
Modify the program according to \cref{lem:initial-point} and obtain initial feasible $\vx$ with $\alpha=2^{16}\frac{m^{2.5} R}{r\epsilon}$. Compute $\vxos$ (via \cref{eq:xoutstar-min}) 

\vspace{.5em}

\textbf{while} true \textbf{do} 

\begin{enumerate}[itemsep = .1em, leftmargin = 1.7em, topsep = .4em, label=\protect\circled{\arabic*}]        \item\label{item:CheckSuboptimality} All servers check if $\vc^\top \vx\leq \vc^\top \vxos + \frac{4m}{t}$, and if true  \label[line]{line:updateT}

        \begin{enumerate}
            \item All servers check $t \geq \tend$; if true, they compute and return $\arg\min_{\vx:\vx\in \kin, \ma\vx=\vb} \left\{ t\cdot \vc^\top\vx + \sum_{i=1}^s w_i \barrini(\vxi)\right\}$ \label[line]{line:EndAlg}

            \item All servers update $t$ to $t\cdot\left(1+\frac{\eta}{4m}\right)$ \label[line]{line:updateT-Rule}

            \item All servers update $\vxos$ (\cref{eq:xoutstar-min}) and jumps to the start of the \textbf{while} loop
        \end{enumerate}
        \item\label{item:TestFeasibilityOfVxos} All servers find the smallest index $\mathsf{idx}$ such that $\langle \nabla\barrini(\vxi), \vxosi - \vxi\rangle + \eta \|\vxosi - \vxi\|_{\vxi}\geq 4 \nu_i$
        \item \textbf{For} all $i\in [s]$ \textbf{do} 
        \begin{enumerate}
            \item\label{item:CommunicationHappens}  Each server $i$ checks if $i= \mathsf{idx}$; if so then query $\oi$ at $\vxosi$ and send $\oi(\vxosi)$ to the blackboard, otherwise read the result of $\oi(\vxosi)$ from blackboard.
            \begin{enumerate}
                \item If $\vxosi \in \ki$ then 
                 set $\kini = \textrm{conv}\left\{\kini, \vxosi\right\}$ \label[line]{line:KinUpdated}
                \item Otherwise, 
                      set $\kouti = \kouti\cap \hi$  \label{line:KoutUpdated}
                \item Update $\vxos$ and jump to   the start of the \textbf{while} loop
            \end{enumerate}
        \end{enumerate}

        \item 	 Set $\delta_{\vx}\defeq\frac{\eta}{2}\cdot\frac{\vxos-\vx}{\|\vxos-\vx\|_{\vx,1}},$where
$\|\mathbf{u}\|_{\vx,1}\defeq\sum_{i=1}^{s}\|\mathbf{u}\|_{\vxi}$\label[line]{line:IPM-step}

     \item  Set $\vx\leftarrow\vx+\delta_{\vx}$ \label[line]{line:MoveX}  
    \end{enumerate}
    \textbf{Return:} Recover $\theta$ from $\vx$ according to \cref{lem:problem-reduction,lem:initial-point}. 
\end{framed}
\caption{Minimizing sum of convex functions in the blackboard model.}
\label[alg]{alg:min-sum-convex-blackboard}
\end{figure}

\subsection{An Overview of Our Analysis \label{subsec:PotentialFns}} 
\looseness=-1We note that in \cref{alg:min-sum-convex-blackboard}, after initialization, machines send a message to the blackboard only if the separation oracle is queried; each such message encodes a halfspace in $\R^{d_i}$, which can be encoded using $O(d_i \cdot L)$  bits. Therefore, the communication complexity of \cref{alg:min-sum-convex-blackboard}
 can be bounded as $\sum_{i=1}^s n_i d_i L$, where $n_i$ is the number of separation oracle calls on $\ki$.

\looseness=-1In this section, we show that for given any weight vector $\vw\in \R_{\geq 1}^s$, the total cost  of  all separation oracle queries is at most  
\[
    \sum_{i=1}^s w_i n_i \leq \sum_{i=1}^s w_i d_i \log(sdR/(\epsilon r)).
\]
Although we will simply use $w_i=d_i$ in this paper, we believe that the use of other weights could be useful in other applications.

\looseness=-1To analyze the oracle cost of  \cref{alg:min-sum-convex-blackboard}, 
we define a potential function
that captures the changes in $\kini$, $\kouti$, $t$, and $\vx$
in each iteration. We define $\varphiout(\vx) = \sum_{i=1}^s w_i \varphiouti(\vx_i)$ and use 
$\varphiouthat$ to denote $\varphiout$ restricted to the set $\{\vu:\ma\vu=\vb\}$. We further use $f^*$ to denote the Fenchel conjugate of the function $f$. Then we define our potential as 
\begin{equation}
\pot\defeq \underbrace{t\cdot \vc^\top{\vx} + \varphiouthat^*(-t\vc)}_{\text{outer potential terms}}+\underbrace{\sum_{i\in[s]}w_i\barrini(\vxi)}_{\text{inner potential terms}},\label{eq:TotalPotential}
\end{equation} where $\varphiouthat^*(-t\vc)=\max_{\ma\vu = \vb} - t\vc^\top \vu - \varphiout(\vu) = -\left(\min_{\ma\vu = \vb}  t\vc^\top \vu + \varphiout(\vu)\right).
$
Following the choice of the barrier functions in \cite{dong2022decomposable}, we use the universal barrier \cite{lee2021universal} for $\barrini$ and the entropic barrier\cite{bubeck2015entropic, chewi2021entropic}\footnote{ While \cite{dong2022decomposable} uses a simple entropic barrier for the entire $\kout$, we are using a weighted sum of entropic barriers on all $\kouti$'s.} for $\barrouti$. 
In the subsequent
sections, we study the changes in each of these potential functions along with obtaining bounds on the initial and final potentials 
and combine them to bound the algorithm's separation oracle complexity. 

\subsubsection{Potential Change Upon Shrinking an Outer Set \label{sec:OuterPotChange}}
Ideally, we want to show that the potential decreases 
uniformly after each separation oracle query to $\ki$.
Formally, we want to use a self-concordant barrier function satisfying the condition below, which we conjecture holds for all self-concordant functions. 
In this paper, we prove that the entropic barrier satisfies it (and use this fact in our analysis). 
\begin{conjecture}\label[conj]{assup:potn-func} 
    Let $\phi$ be a self-concordant barrier. Denote $\phi_{\mathcal{S}}$ to be $\phi$ restricted to the set $\mathcal{S}$. Given a bounded convex body $\kcal$ and cost vector $\vc$, define $\xstar = \arg\min \vc^\top \vx+ \phi_\kcal(\vx)$, and let hyperplane $\mathcal{H}$ contain a point $\vz$ such that $\|\xstar - \vz\|_{(\nabla^2 \phi_\kcal(\xstar))^{-1}}<0.01$.
 Then
    \[
    \min_{\vx} \left\{\vc^\top \vx + \phi_{\kcal \cap \mathcal{H}}(\vx)\right\} \geq \min_{\vx} \left\{\vc^\top \vx + \phi_{\kcal}(\vx)\right\} + 0.1.
    \]
\end{conjecture}

The analysis in \cite{vaidya1989new} suggests that the volumetric barrier might satisfy the condition above. In this paper, we instead use  a weighted version of the entropic barrier  on $\kout$, which is easier to analyze and implement.
Before we prove that the entropic barrier satisfies the conditions above, we first show how  \cref{assup:potn-func} implies the desired potential decrease.
\begin{lemma}\label{lem:outer-progress}
Consider the Cartesian product $\kcal = \kcal_1\times \ldots \times \kcal_s$, where each $\kcal_i$ is a bounded convex set associated with a self-concordant barrier function $\varphi_i$ satisfying the inequality in \cref{assup:potn-func}. Given a weight vector $\vw_{\geq 1} \in \R^n$, and the cost vector $\vc$, we define the weighted barrier  $\varphi(\vx) \defeq \sum_{i=1}^s w_i \varphi_i(\vx_i)$ and the corresponding analytic center $\mu \defeq  \arg \min_{\vx} \vc^\top \vx + \varphi(\vx)$. Fix $i$, let $\kcal_i^{\text{new}} \defeq \kcal_i \cap \mathcal{H}$, where $\mathcal H$ contains a point $\vz$ such that $\|\vz-\mu_i\|_{(\nabla^2 w_i\varphi_i(\mu_i))^{-1}} \leq 0.01$. We define $\kcal^{\text{new}}$ and $\varphi^{\text{new}}$ correspondingly.
Then, for any $\ma$ and $\vb$, we have
\[
\min_{\ma\vx=\vb} \vc^\top \vx + \varphi^{\text{new}}(\vx) \geq \min_{\ma\vx=\vb} \vc^\top \vx +  \varphi(\vx) + 0.1 w_i 
\]
\end{lemma}

\begin{proof}
    Let $\xstar \defeq \arg\min_{\ma\vx=\vb} \vc^\top \vx + \varphi(\vx)$. Using the appropiate Lagrange multiplier $\vz$, we can find $\ctil \defeq \vc+\ma^\top \vz$ such that 
    \[
        \xstar = \arg\min_{\vx} \ctil^\top \vx + \varphi(\vx).
    \]
    These definitions of $\xstar$ and $\ctil$ imply the following connection:
    \begin{equation}\label{eq:defXstarOriginal}
         \min_{\vx} \ctil^\top \vx + \varphi(\vx) = \ctil^\top \xstar + \varphi(\xstar)  = \min_{\ma\vx=\vb} \vc^\top \vx +\varphi(\vx). 
    \end{equation}
       Then, by  the definition of $\varphi^{\text{new}}$, we have:
    \begin{align}
        \min_{\vx} \ctil^\top \vx +  \varphi^{\text{new}}(\vx) &= \min_{\vx} \left\{ \ctil^\top \vx + \sum_{j\neq i} w_j \varphi_j(\vx_j) + w_i \varphi_i^{\text{new}} (\vx_i) \right\}  \nonumber \\ 
        &\geq  \min_{\vx} \left\{ \ctil^\top \vx + \sum_{j\neq i} w_j \varphi_j(\vx_j) + w_i (\varphi_i(\vx_i) + 0.1) \right\} \nonumber \\ 
        &=0.1w_i + \min_{\vx}\ctil^\top  \vx + \varphi(\vx) \nonumber \\ 
        &= 0.1 w + \min_{\ma\vx=\vb} \vc^\top \vx +\varphi(\vx)+ \vz^\top \ma\xstar, \label[ineq]{eq:ineq-minx-lowerbound}
    \end{align}
    where  the second is by \cref{assup:potn-func} applied to $\varphi_i$, the third step follows by definition of $\varphi$, and the last step follows by \cref{eq:defXstarOriginal}. Since our desired bound is on $\min_{\ma\vx=\vb} \vc^\top \vx +  \varphi^{\text{new}}(\vx)$, we now define the minimizer
    \begin{align}
    \xstarnew\defeq \arg\min_{\ma\vx=\vb}  \vc^\top \vx +  \varphi^{\text{new}}(\vx)\numberthis\label{eq:defXstarNew}.
    \end{align} We may now observe the following  upper bound on
$\min_{\vx} \ctil^\top \vx +  \varphi^{\text{new}}(\vx) $:
\begin{align*}
\min_{\vx} \ctil^\top \vx +  \varphi^{\text{new}}(\vx) 
&\leq \min_{\ma\vx=\vb} \ctil^\top \vx +  \varphi^{\text{new}}(\vx) \\ 
&\leq \ctil^\top \xstarnew + \varphi^{\text{new}}(\xstarnew)\\
&= \vc^\top \xstarnew+ \varphi^{\text{new}}(\xstarnew) +  \vz^{\top}\ma\xstarnew \\
&=  \min_{\ma\vx=\vb}  \vc^\top \vx +  \varphi^{\text{new}}(\vx) + \vz^{\top}\ma\xstarnew, \numberthis\label[ineq]{eq:ineq-minx-upperbound}
\end{align*}
where the first step is  by restricting the set of minimization, the second step uses the fact that $\xstarnew\in \{\vx:\ma\vx=\vb\}$, the third step uses the definition of $\ctil$, and the final step is by plugging in the definition of $\xstarnew$ from \cref{eq:defXstarNew}. Chaining  \cref{eq:ineq-minx-lowerbound} and \cref{eq:ineq-minx-upperbound} gives
\begin{align*}
\min_{\ma\vx=\vb} \vc^\top \vx + \varphi^{\text{new}}(\vx) - \min_{\ma\vx=\vb} \vc^\top \vx & \geq 0.1 w_{i} + \vz^{\top}\ma(\xstar-\xstarnew) =0.1 w_{i}, \numberthis\label[ineq]{eq:finalIneqOuterPotProof}
\end{align*}
where the second step is because  the definitions of  $\xstarnew$ and $\xstar$ imply that both
satisfy $\ma\vx=\vb.$ 
\end{proof}

\noindent Now, we show that the entropic barrier satisfies the condition in \cref{assup:potn-func}.

\begin{lemma}\label{lem:pot-change-approx}
    Given a bounded convex body $\kcal$ and a vector $\vc$, let $\phi$ be the entropic barrier and $\xstar \defeq \arg\min \vc^\top \vx + \phi_\kcal(\vx)$. Let the halfspace $\mathcal{H}$ contain a point $\vz$ such that $\|\vz - \vx^\star\|_{(\nabla^2 \phi(\vx^\star)^{-1}} \leq t$, then 
    \[
    \min_{\vx} \left\{\vc^\top \vx + \phi_{\kcal \cap \mathcal{H}}(\vx)\right\} \geq \min_{\vx} \left\{\vc^\top \vx + \phi_{\kcal}(\vx)\right\}  -\log(1-1/e+t).
    \]
\end{lemma}
\begin{proof}
    Let $\phi_\kcal$ be the entropic barrier, we note that 
    the  Fenchel conjugate of $\phi_\kcal$, $\phi_\kcal^*(-\vc)$ is 
    \[
       \phi_\kcal^*(-\vc) =  -1\cdot \min_{\vx} ({\vc^\top \vx + \phi_\kcal(\vx)}).
    \]
    Recall that $\phi^*(\cdot )$ is the logarithmic Laplace transform of the uniform measure on $\kcal$~\cite{bubeck2015entropic}, we have
    \[
        \phi_\kcal^*(-\vc) =\log \int_{\kcal} \exp(-\langle \vc,\vu \rangle ) d\vu.
    \]
    Consider the distribution over $\R^n$, with following density measure:
    \[
    p(\vx)= \exp(-\langle\vc,\vx\rangle - \phi_\kcal^*(-\vc))\cdot\mathbb{I}\{\vx\in \kcal\}.
    \]
    Using \cref{prop:EquivalenceOfCentroidAndMinimizer}, we have
    $\xstar = \E_{\vx\sim p}[\vx]$, is the centroid of $p$. It is an established fact that (see e.g.\cite[Lemma 1]{bubeck2015entropic})
    \[
        \E_{x\sim p}[(\vx-\xstar) (\vx-\xstar)^\top ] = \nabla^2 \phi^*_\kcal(-\vc)  = (\nabla^2 \phi_\kcal(\xstar))^{-1}.
    \]
     Then, by Gr\"unbaum's Theorem (\cref{thm:GrunbaumGeneral}), we have we have 
    \[
        \int_{\kcal \cap \mathcal{H}} p(\vx) d\vx \geq 1/e - t,
    \]
 which implies that  
    \[
\int_{\kcal \cap \mathcal{H}} p(\vx) d\vx = \frac{\int_{\kcal \cap \mathcal{H}} \exp(-\langle \vc, \vx )\rangle d\vx}{\phi_\kcal^*(-\vc)} \leq 1-1/e + t.
    \]
    Taking the logarithm of both sides, we get
    \[
    \phi^*_{\kcal \cap \mathcal{H}}(-\vc)  - \phi^*_{\kcal}(-\vc)  \leq\log(1-1/e + t).
    \]
    This finishes the claim:  
    \[
\min_{\vx} \left\{\vc^\top \vx + \phi_{\kcal \cap \mathcal{H}}(\vx)\right\} -  \min_{\vx} \left\{\vc^\top \vx + \phi_{\kcal}(\vx)\right\}  =  -\phi^*_{\kcal \cap \mathcal{H}}(-\vc)+\phi^*_{\kcal}(-\vc) \geq - \log(1-1/e + t) .
    \]
\end{proof}

\subsubsection{Potential Change Upon Increasing $t$}\label{sec:PotentialChangeTChange}
To capture the change in potential due to the update in $t$, we require a technical result derived from properties of conjugates  of self-concordant barriers. To obtain this result, we use a helper result from \cite{dong2022decomposable}, and based on this lemma, we prove a more general one in \cref{lem:EtChangeConstrained}.

\begin{lemma}[Lemma $4.2$ of \cite{dong2022decomposable}]
\label[lem]{lem:EtChange}
Consider a $\nu$-self-concordant barrier $\barr:\textrm{int}(\kcal)\rightarrow\R$
over the interior of a convex set $\kcal\subseteq\R^d$. Define 
\begin{equation}
\regBarr{\barr}t{\vc} \defeq\min_{\vx}\left[t\cdot \inprod{\vc}{\vx}+\barr(\vx)\right] \text{ and }\vxt\defeq\arg\min_{\vx}\regBarr{\barr}{t}{\vc}(\vx).\label{eq:defEt}
\end{equation}
Then for $0\leq h\leq\frac{1}{3\sqrt{\nu}}$, we have 
\[
\min_{\vx}\regBarr{\barr}{t}{\vc}(\vx)+th \cdot\vc^\top \vxt\geq \min_{\vx}\regBarr{\barr}{t(1+h)}{\vc}(\vx)\geq\min_{\vx}\regBarr{\barr}{t}{\vc}(\vx) + ht\cdot\vc^\top{\vxt}-h^{2}\nu.
\]
\end{lemma} 

\begin{lemma}
\label[lem]{lem:EtChangeConstrained}
Consider a $\nu$-self-concordant barrier $\barr:\textrm{int}(\kcal)\rightarrow\R$
over the interior of a convex set $\kcal\subseteq\R^d$. Define 
\begin{equation}
\vxhatt\defeq\arg\min_{\ma\vx=\vb} \regBarr{\barr}{t}{\vc}(\vx).\label{eq:defEtGeneral}
\end{equation}
Then for $0\leq h\leq\frac{1}{3\sqrt{\nu}}$, we have 
\[
\min_{\ma\vx=\vb}\regBarr{\barr}{t}{\vc}(\vx) + th \cdot{\vc}^\top \vxhatt\geq\min_{\ma\vx=\vb}\regBarr{\barr}{t(1+h)}{\vc}(\vx)\geq\min_{\ma\vx=\vb}\regBarr{\barr}{t}{\vc}(\vx)+ht\cdot{\vc}^\top{\vxhatt}-h^{2}\nu.
\]
\end{lemma} 
\begin{proof}
The first inequality holds for
any function $\barr$ with the specified definition of $\vxhatt$ and $\regBarr{\barr}{t(1+h)}{\vc}$:  
\begin{align}
\min_{\ma\vx=\vb} \regBarr{\barr}{t(1+h)}{\vc}(\vx)
 & \leq t(1+h)\cdot \vc^\top \myhighlighter{$\vxhatt$}+\barr(\myhighlighter{$\vxhatt$}) =\min_{\ma\vx=\vb} \regBarr{\barr}{t}{\vc}(\vx)+th \cdot\vc^\top \vxhatt,\nonumber 
\end{align} where the first inequality is by plugging in $\vxhatt$ into $\regBarr{\barr}{t(1+h)}{\vc}(\vx)$ and the second step by using the definition of $\vxhatt$ from \cref{eq:defEtGeneral}. To prove the second inequality, we use the self-concordance of $\barr$. First, using the appropriate Lagrange multiplier $\vz$, one can define $\ctil:= \vc + \ma^\top \vz$ to express $\vxhatt$ as the minimizer of an unconstrained problem as follows: 
\begin{equation}
\vxhatt=\arg\min_{\ma\vx=\vb} \regBarr{\barr}{t}{\vc}(\vx) = \arg\min_{\vx} \regBarr{\barr}{t}{\myhighlighter{$\ctil$}}(\vx), \text{ for } \ctil:= \vc + \ma^\top \vz. \label{eq:defXstarConstrained}
\end{equation} As a result, we may now apply \cref{lem:EtChange}, which gives
\begin{align*}
    \min_{\vx} \regBarr{\barr}{t(1+h)}{\ctil}(\vx) &\geq \min_{\vx} \regBarr{\barr}{t}{\ctil}(\vx) + ht\cdot \ctil^\top \vxhatt - h^2\nu \\
    &= \regBarr{\barr}{t}{\ctil}(\myhighlighter{$\vxhatt$})  + ht\cdot \ctil^\top \vxhatt - h^2\nu \\
    &= \regBarr{\barr}{t}{\myhighlighter{$\vc$}}(\vxhatt) + ht\cdot \myhighlighter{$\vc$}^\top \vxhatt - h^2\nu + \myhighlighter{$(1+h)t \cdot \vz^\top \ma\vxhatt$}, \\
    &= \min_{\ma\vx=\vb} \regBarr{\barr}{t}{\vc}(\vx) + ht\cdot \vc^\top \vxhatt - h^2\nu + (1+h)t \cdot \vz^\top \ma\vxhatt,
    \numberthis\label[ineq]{eq:tChangeEtConstrained}
\end{align*} where the first and second steps are by \cref{lem:EtChange} applied to $\regBarr{\barr}{t(1+h)}{\ctil}$ and the definition of $\vxhatt$ from \cref{eq:defXstarConstrained}, the third step is by using the definition of $\ctil:=\vc+\ma^\top\vz$, and the final step is by applying the definition of $\vxhatt$ from \cref{eq:defEtGeneral}. We now define \[ \vxhatnewt = \arg\min_{\ma\vx=\vb} \regBarr{\barr}{t(1+h)}{\vc}(\vx).\numberthis\label{eq:vxhatnewtDefn}\] In the other direction, we have 
\begin{align*}
    \min_{\vx} \regBarr{\barr}{t(1+h)}{\ctil}(\vx) 
    &\leq 
    \min_{\myhighlighter{$\ma\vx=\vb$}} \regBarr{\barr}{t(1+h)}{\ctil}(\vx)\\
    &\leq 
    \regBarr{\barr}{t(1+h)}{\ctil}(\myhighlighter{$\vxhatnewt$}) \\
    &= \regBarr{\barr}{t(1+h)}{\myhighlighter{$\vc$}}(\vxhatnewt)  + (1+h)t \cdot \vz^\top \ma \vxhatnewt, \\
    &= 
    \min_{\ma\vx=\vb}\regBarr{\barr}{t(1+h)}{\vc}(\vx)  + (1+h)t \cdot \vz^\top \ma \vxhatnewt,
    \numberthis\label[ineq]{eq:EtConstrainedChangeEnv3}
\end{align*} where the first step is by constraining the minimization set, the second step is because, by definition of $\vxhatnewt$ from \cref{eq:vxhatnewtDefn}, it satisfies $\ma\vx=\vb$,  the third step is by replacing $\ctil:=\vc + \ma^\top \vz$, and the fourth step is by definition of $\vxhatnewt$ in \cref{eq:vxhatnewtDefn}. Therefore, we have 
\begin{align*}
    \min_{\ma\vx=\vb} \regBarr{\barr}{t(1+h)}{\vc}(\vx) 
    &\geq 
    \min_{\vx} \regBarr{\barr}{t(1+h)}{\ctil}(\vx) - (1+h)t \cdot \vz^\top \ma \vxhatnewt\\
    &\geq 
    \min_{\ma\vx=\vb} \regBarr{\barr}{t}{\vc}(\vx) + ht\cdot \vc^\top \vxhatt - h^2\nu + (1+h)t \cdot \vz^\top \ma(\vxhatt - \vxhatnewt),\\
    &= 
     \min_{\ma\vx=\vb} \regBarr{\barr}{t}{\vc}(\vx) + ht\cdot \vc^\top \vxhatt - h^2\nu ,
\end{align*} where the first step is by 
 rearranging \cref{eq:EtConstrainedChangeEnv3}, the second step is by \cref{eq:tChangeEtConstrained}, and the final step is by the fact that both $\vxhatt$ and $\vxhatnewt$ satisfy $\ma\vx=\vb$. This finishes the proof. 
\end{proof}
To finally compute the potential change due to $t,$ we combine the result from \cref{lem:EtChangeConstrained} with the bound guaranteed by \cref{line:updateT} of \cref{alg:min-sum-convex-blackboard} along with the  self-concordance parameter of the volumetric barrier. 
We may now compute the potential change due to change in $t$ in \cref{line:updateT-Rule}. 
\begin{lemma}\label[lem]{lem:PotChangeT}
When $t$ is updated to $t\cdot\left[1+\frac{\eta}{4 \sum_{i=1}^s w_i \nu_i}\right]$
in \cref{line:updateT-Rule} of \cref{alg:min-sum-convex-blackboard}, the potential $\pot$ \cref{eq:TotalPotential} increases to $\potnew$ as follows: $$\potnew \leq \pot + \eta+\eta^{2}.$$ 
\end{lemma}
\begin{proof}
From \cref{eq:TotalPotential},
the change in potential by changing $t$ to $t\cdot(1+h)$ for some $h>0$ may be expressed as 
\[
\potnew-\pot= -\underbrace{\left(\min_{\ma\vy=\vb} \regBarr{\varphiout}{t(1+h)}{\vc}(\vy) - \min_{\ma\vy=\vb} \regBarr{\varphiout}{t}{\vc}(\vy)\right)}_{\text{Bounded via \cref{lem:EtChangeConstrained}}}+th\cdot\vc^\top{\vx}.
\]
We may now apply \cref{lem:EtChangeConstrained}
in the preceding equation to obtain the following bound. 
\[
\potnew-\pot\leq th\cdot\vc^\top{\vx}-th\cdot\vc^\top{\vxhatt}+h^{2}\nu.
\]
We see that $\vxhatt$ as defined in \cref{eq:defEtGeneral} for $\barr= \sum_{i=1}^s w_i \barrouti$ from \cref{{eq:xoutstar-min}} is exactly identical to $\vxos$ from   \cref{eq:xoutstar-min}.
We can therefore apply the guarantee
$\vc^\top\vx\leq\vc^\top{\vxos}+\frac{4 \sum_{i=1}^s w_i \nu_i}{t}$ (from \cref{line:updateT} of \cref{alg:min-sum-convex-blackboard}) and 
$h=\frac{\eta}{4 \sum_{i=1}^s w_i \nu_i}$ and $\nu=\sum_{i=1}^s w_i \nu_i$ to obtain
\begin{align*}
\potnew-\pot & \leq th\cdot\frac{4 \sum_{i=1}^s w_i \nu_i}{t}+h^{2}\nu=\eta+\left(\frac{\eta}{4\sum_{i=1}^s w_i \nu_i}\right)^{2}\nu\leq\eta+\eta^{2}.
\end{align*}
\qedhere 
\end{proof}

\subsubsection{Potential Change Upon Growing an Inner Set \label{sec:InnerPotChange}}

Here, we state the technical lemma that describes the change in the universal barrier potential when we add a point to the convex set.
\begin{lemma}[{\cite[Lemma 4.6]{dong2022decomposable}}]
\label[lem]{lem:UniversalBarrierPotentialChange} Given a convex set $\kcal\subseteq\R^d$
and a point $\vx\in\kcal$, let $\barrUniK\defeq\log\vol(\kcal-\vx)^{\circ}$
be the universal barrier defined on $\kcal$ with respect to $\vx.$
Let  $\vy\notin\kcal$ be a point satisfying the following condition for some scalar $\eta \leq 1/4$
\begin{equation}
\inprod{\nabla\barrUniK(\vx)}{\vy-\vx}+\eta \|\vy-\vx\|_{\vx}\geq 4d.\label{eq:violatedDistCond}
\end{equation}
Then, the
universal barrier (cf. \Cref{def:UniversalBarr}) defined on the set  $\conv\left\{ \kcal,\vy\right\}$  with respect
to $\vx$ satisfies the following inequality: 
\[
\barrUniKNew(\vx)\defeq\psi_{\conv\left\{ \kcal,\vy\right\} }(\vx)=\log\vol(\conv\left\{\kcal,\vy\right\} -\vx)^{\circ}\leq\barrUniK(\vx)+\log(1-1/e+\eta).
\]
\end{lemma}

\subsubsection{Potential Change For the Update of \texorpdfstring{$\vx$}{x} \label{sec:IPM-analysis} }

In this section, we quantify the amount of progress made in \cref{line:IPM-step}
of \cref{alg:min-sum-convex-blackboard} by computing the change in the potential
$\pot$ as defined in \cref{eq:TotalPotential}.
\begin{lemma}\label[lem]{lem:potChangeX}
Consider the potential $\pot$ \cref{eq:TotalPotential}. Denote by $\potnew$ the value of this potential after $\vx$ takes the update step $\delta_{\vx}=\frac{\eta}{2}\cdot\frac{\vxos-\vx}{\|\vxos-\vx\|_{\vx,1}}$
as in \cref{line:IPM-step}.
Assume the following guarantees \begin{itemize}
    \item $\vc^{\top}\vxos+\frac{4 \sum_{i=1}^s w_i \nu_i}{t}\leq\vc^{\top}\vx$.
    \item $\inprod{\nabla\barrini(\vxi)}{\vxosi-\vxi}+\eta \cdot  \|\vxosi-\vxi\|_{\vxi}\leq 4 \nu_{i}$ for all $i\in [s]$.
\end{itemize} Then the potential $\pot$ incurs the following
minimum decrease.
\[
\potnew\leq\pot-\frac{\eta^{2}}{4}\sum_{i=1}^s w_i \frac{\|\vxosi - \vxi\|_{\vxi}}{\|\vxos-\vx\|_{\vx, 1}} \leq \Phi - \frac{\eta^2}{4}.
\]
\end{lemma}

\begin{proof}
The proof is similar to that of Lemma 4.7 in \cite{dong2022decomposable}, we include it here for completeness.
Taking the gradient of $\pot$ with respect to $\vx$ and rearranging
the terms gives
\begin{equation}
t\vc=\nabla_{\vx}\pot-\sum_{i=1}^{s}w_i \nabla\barrini(\vxi),\label{eq:potential-ipm-3}
\end{equation} where we are overloading notation in $\nabla \barrini(\vxi)$ to mean the $d$-dimensional vector equalling the appropriate entries at the $d_i$ coordinates corresponding to $\vxi$ and zero elsewhere. 
By replacing $t\vc$ with the expression on the right-hand side  of the preceding equation,
we get 
\begin{align}
\potnew-\pot & =t\inprod{\vc}{\vx+\delta_{\vx}}+\sum_{i=1}^{s}w_i\barrini(\vxi+\delta_{\vx,i})-t\inprod{\vc}{\vx}-\sum_{i=1}^{s}w_i\barrini(\vxi)\nonumber \\
 & =\langle\nabla_{\vx}\pot,\delta_{\vx}\rangle+\sum_{i=1}^{s} w_i\underbrace{\left[\barrini(\vxi+\deltaXi)-\barrini(\vxi)-\inprod{\nabla\barrini(\vxi)}{\deltaXi}\right]}_{q_{\barrini}(\vxi)}.\label{eq:potential-ipm-2}
\end{align} Note that in substituting \cref{eq:potential-ipm-3} above, we crucially use that $\vxi$ are all disjoint vectors whose coordinates completely cover those of $\vx$. 
The term $q_{\barrini}(\vxi)$ measures the error due to first-order
approximation of $\barrini$ around $\vxi$. Since each $\barrini(\vxi)$ is a self-concordant function and $\|\delta_{\vx,i}\|_{\vx_{i}}\leq\|\delta_{\vx}\|_{\vx,1}\leq\eta\leq1/4$, this error is known to be small; more precisely,
\cref{thm:QuadApproxErr} applies and gives 
\begin{equation}
\barrini(\vxi+\deltaXi)-\barrini(\vxi)-\inprod{\nabla\barrini(\vxi)}{\deltaXi}\leq\|\delta_{\vx,i}\|_{\vxi}^{2}.\label[ineq]{eq:potential-ipm-1}
\end{equation}
Plugging in \cref{eq:potential-ipm-1}
into \cref{eq:potential-ipm-2}, we get 
\begin{equation}
\potnew-\pot \leq\langle\nabla_{\vx}\pot,\delta_{\vx}\rangle+\sum_{i=1}^s w_i \|\delta_{\vx,i}\|_{\vxi}^{2}.\label[ineq]{eq:potential-ipm-4}
\end{equation}
We now bound the two terms on the right hand side one at a time. Using
the definition of $\deltaX$ (as given in the statement of the lemma)
and of $\nabla_{\vx}\pot$ from \cref{eq:potential-ipm-3} gives
\begin{align*}
\langle\nabla_{\vx}\pot,\delta_{\vx}\rangle 
 & =\frac{\eta}{2}\frac{1}{\|\vxos-\vx\|_{\vx,1}}\langle\nabla_{\vx}\pot,\vxos-\vx\rangle\nonumber \\
 & =\frac{\eta}{2}\frac{1}{\|\vxos-\vx\|_{\vx,1}}\left[\inprod{t\vc}{\vxos-\vx}+\sum_{i=1}^{s} w_i\inprod{\nabla\barrini(\vxi)}{\vx_{\textrm{out},i}^{\star}-\vxi}\right]\nonumber \\
 & \leq\frac{\eta}{2}\frac{1}{\|\vxos-\vx\|_{\vx,1}}\left[\inprod{t\vc}{\vxos-\vx}+\sum_{i=1}^{s} w_i \left(4 \nu_{i}-\eta \|\vxosi-\vxi\|_{\vx_{i}}\right)\right]\nonumber \\
 & =\frac{\eta}{2}\frac{1}{\|\vxos-\vx\|_{\vx,1}}\left[\inprod{t\vc}{\vxos-\vx}+4\sum_{i=1}^s w_i \nu_i-\eta\sum_{i=1}^s  w_i \|\vxosi-\vxi\|_{\vx_{i}}\right]\nonumber \\
 & \leq\frac{\eta}{2}\frac{1}{\|\vxos-\vx\|_{\vx,1}}\cdot\left(-\eta\sum_{i=1}^s  w_i \|\vxosi-\vxi\|_{\vx_{i}}\right)\nonumber  =-\frac{\eta^2}{2}\sum_{i=1}^s w_i \frac{\|\vxosi - \vxi\|_{\vxi}}{\|\vxos-\vx\|_{\vx, 1}}.\numberthis\label[ineq]{eq:pot-ipm-5}
\end{align*}
where the third step follows from the second assumption, and the fifth step follows
from the first assumption. To bound the second term
\cref{line:IPM-step} that 
\begin{equation}
\sum_{i=1}^s w_i \|\delta_{\vx,i}\|_{\vxi}^{2} =\frac{\eta^2}{4}\cdot\sum_{i=1}^s w_i \frac{\|\vxosi-\vxi\|_{\vxi}^2}{\|\vxos-\vx\|_{\vx,1}^2}.\label{eq:pot-ipm-6}
\end{equation}
Hence, we may plug in \cref{eq:pot-ipm-5} and \cref{eq:pot-ipm-6}
into \cref{eq:potential-ipm-4} to get the desired result. 
\end{proof}

\subsubsection{Total Oracle Cost}\label{sec:total-complexity}
Before we bound the communication complexity of the algorithm, we first bound the total potential change throughout the algorithm.

\begin{lemma}\label{lem:InitMinusFinalpotential-change}
Consider the potential function
\[
\Phi(t,\vx,\mathcal{K}_{\textrm{out}},\mathcal{K}_{\textrm{in}})\defeq t\cdot\vc^{\top}\vx-\left(\min_{\ma\vy=\vb}t\cdot\vc^{\top}\vy+\sum_{i=1}^s w_{i}\barrouti(\vy_{i})\right)+\sum_{i=1}^s w_{i}\barrini(\vx_{i})
\]
as defined in \cref{eq:TotalPotential} associated with \cref{alg:min-sum-convex-blackboard}. Let $\potinit$ be the potential at $t=\tinit$ of this algorithm, and let $\potend$ be the potential at $t = \tend$. Suppose at $t=\tinit$ in \cref{alg:min-sum-convex-blackboard}, we have, for some $\vz\in \kin$, that $\ball_m(\vz,\bar r)\subseteq \kin$ with $\bar r = r / \operatorname{poly}(m)$ and  $\kout\subseteq\ball_m(0,\bar R)$ for $\bar R = O(\sqrt nR)$. Then we have, under the assumptions of \cref{thm:MainThmOfKiProblem}, that 
\[
\potinit - \potend \leq O\left(\sum_{i=1}^s w_i d_i\log\left(\frac{m R}{\epsilon r}\right)\right).
\]
\end{lemma} 
\begin{proof}
    We bound the change in the potential term by term, starting with the following terms depending on the current iterate $\vx$, the current time step $t$, and the current outer set $\kout$:
    \begin{equation}
        \Psi(t,\vx,\mathcal{K}_{\textrm{out}})\defeq t\cdot\vc^{\top}\vx-\left(\min_{\ma\vy=\vb}t\cdot\vc^{\top}\vy+\sum_{i=1}^s w_{i}\barrouti(\vy_{i})\right).
    \end{equation} We introduce the notation $\Psi_{\textrm{init}}:= \Psi(\tinit, \vx(\tinit), \kout(\tinit))$ and simplify it as follows.
    \begin{align*}
            \Psi_{\textrm{init}} &=  \tinit\cdot\vc^{\top}\vx-\left(\min_{\ma\vy=\vb}\tinit\cdot\vc^{\top}\vy+\sum_{i=1}^s w_{i}\barrouti(\vy_{i})\right)\\
            &\leq \tinit\cdot(\vc^{\top}\vx-\min_{\ma\vy=\vb} \vc^{\top}\vy)-\min_{\ma\vy=\vb}\sum_{i=1}^s w_{i}\barrouti(\vy_{i}) \\ 
            &\leq o(m) -\min_{\vy}\sum_{i=1}^s w_{i}\barrouti(\vy_{i}),\numberthis\label[ineq]{ineq:OuterInitFirst}
    \end{align*}
    where the second inequality follows from our choice of  $\tinit$ implying $t\cdot \vc^\top \vx = o(m)$ for any $\vx \in \kout$ and also because expanding the set of minimization only decreases the minimum value; note that because of $\barrouti$, the variable $\vy$ is implicitly already restricted to $\kout$. We emphasize that $\barrouti$ here is the barrier function on the $\kout$ at $t=\tinit$. 
     Let $\vy^\star_i$ be the analytic center of $\barrouti$:\[\vy^\star_i =\arg\min_{\vy} \sum_{i=1}^s w_i \barrouti(\vy_i). \numberthis\label{eq:defYStarTheAnalyticCenter}\] Hence, we may rewrite $\Psi_{\textrm{init}}$ using \cref{eq:defYStarTheAnalyticCenter} to obtain 
    \[
    \Psi_{\textrm{init}} \leq o(m) - \sum_{i=1}^s w_i\barrout(\vy_i^\star).\numberthis\label[ineq]{ineq:OuterInitFinal}
    \] 
    To bound $\Psi_\textrm{end}$, we define:  
    \[
    \vy_{t} \defeq \arg \min_{\ma\vy=\vb}(t\cdot\vc^{\top}\vy+\sum_{i=1}^s w_{i}\barrouti(\vy_{i})) \numberthis\label{eq:defYtTheAnalyticCenter}
    \] Similar to $\Psi_{\textrm{init}}$, we define $\Psi_{\textrm{end}}:= \Psi(\tend, \vx(\tend), \kout(\tend))$. Then, with this notation, we may use $\vy_{\tend}$ to state the following lower bound
    \begin{align*}
        \Psi_\textrm{end} &=  \tend\cdot\vc^{\top}\vx-\min_{\ma\vy=\vb}(\tend\cdot\vc^{\top}\vy+\sum_{i=1}^s w_{i}\barrouti(\vy_{i}))\\
        &\geq \tend \cdot ( \vc^\top \vy_\infty -  \vc^\top \vy_{\tend}) - \sum_{i=1}^s w_{i}\barrouti(\vy_{\tend,i})). \numberthis\label[ineq]{PhiEndBoundAsSumOfTwoTerms}
    \end{align*}
    By \cref{{lem:two-sided-ineq}}, we may bound the difference $\vc^\top \vy_{\infty} - \vc^\top \vy_{\tend}$ as follows:
    \[
         \tend \cdot (\vc^\top \vy_\infty -  \cdot \vc^\top \vy_{\tend}) \geq -\frac{\sum_{i=1}^s w_i \nu_i}{\tend} \cdot \tend = -\sum_{i=1}^s w_i d_i.\numberthis\label[ineq]{ineq:ObjDiffBoundWiDi}
    \]
    where we used \cref{fact:SCofBarrierRestrictedToLinSubspace} to deduce that the self-concordance parameter of $\sum_{i=1}^s w_i \barrouti$ restricted to $\ma\vy=\vb$ is $\sum_{i=1}^s w_i \nu_i$ and each $\barrouti$ is $O(d_i)$ self-concordant.    
    We now claim that \[\barrouti(\vy_{\tend, i}) \leq \barrouti(\vy^\star_i) + O(d_i\log(dR/r)). \numberthis\label[ineq]{ineq:BarroutiYTendInTermsOfYStar}\] Before proving this claim, we see that by combining \cref{ineq:OuterInitFinal,PhiEndBoundAsSumOfTwoTerms,ineq:ObjDiffBoundWiDi,ineq:BarroutiYTendInTermsOfYStar},  that 
    \[
    \Psi_\textrm{init} - \Psi_\textrm{end} \leq \sum_{i=1}^s w_id_i\log(5d R/r) + o(m). \numberthis \label{ineq:PsiInitMinusPsiEnd}
    \]
  We now show the claim in \cref{ineq:BarroutiYTendInTermsOfYStar}.  In order to apply \cref{fact:FirstOrderApproxOfScb}, we consider the ray starting from by $\vy^\star_i$ that passes through $\vy_{\tend,i}$. Let $\vy_{\textrm{bdry},i}$ to be the point where the ray intersects with the set $\kout$. Note that there is a $s\in (0,1)$, such that 
  \[
  \vy_{\tend,i}=\vy^{\star}_i+s(\vy_{\textrm{brdy},i}-\vy^{\star}_i).
  \]
  We note that 
  \[
    s = \frac{\|\vy^\star_i-\vy_{\tend,i}\|}{\|\vy_{\textrm{brdy},i}-\vy^\star_i\|} \geq \frac{\|\vy^\star_i-\vy_{\tend,i}\|}{2R}.
  \]
    By \cref{lem:two-sided-ineq}, we have $\|\vy^\star_i-\vy_{\tend,i}\|\geq \frac{\bar r}{5d_i}$. We finish the proof of \cref{ineq:BarroutiYTendInTermsOfYStar} using \cref{fact:FirstOrderApproxOfScb}.

   Since the potential is a sum of universal barrier and entropic barrier terms, we now need to bound the change in the universal barrier. The proof of this change is identical to the corresponding proof in Lemma 4.8 of \cite{dong2022decomposable}, but we include it next for completeness. 
    Recall the definition of the universal barrier terms $\sum_{i \in [n]} w_i \barrini(\vxi)$, where \[\barrini(\vxi) = \log \vol(\kini^\circ(\vxi)).\] Our computation follows a purely volume-based argument based on our assumptions about the dimensions of the balls contained in and containing the sets $\kin$ and $\kout$ and repeated application of \cref{fact:polarReversal}, as we now elaborate. Define $\ball_{d}(0, r)$ to be the $d$-dimensional Euclidean ball centered at the origin and with radius $r$. We note by the radius assumption of \cref{thm:MainThmOfKiProblem} that  $\kini\subseteq\ki \subseteq \ball_{d_i}(0, \bar R)$ throughout the algorithm. By the assumption made in this lemma, we have at the start of \cref{alg:min-sum-convex-blackboard} the inclusion $\ball_{d_i}(\vz_i, \bar r) \subseteq \kini$. These two inclusion assumptions and \cref{fact:polarReversal} lead to the following bounds for any $\vx_i$.   

\[
\barrini^{\text{end}}(\vxi)\geq \log(\vol(\ball^\circ_{d_i}(0,\bar R)) \text{ and }  \barrini^{\text{init}}(\vxi) \leq  \log(\vol(\ball_{d_i}^{\circ}(\vz_i, \bar r))). \numberthis\label[ineq]{eq:InitEndPotBoundsVolBar}
\]

Combining \cref{eq:InitEndPotBoundsVolBar}, the fact that $\vol(\ball_d(0, r)) \propto r^d$, \cref{fact:polarReversal}, and summing over all $i \in [s]$ gives 
\begin{align*}
    \sum_{i \in [s]} w_i\left[\barrini^{\text{init}}(\vxi) -\barrini^{\text{end}}(\vxi)\right] 
    &\leq \sum_{i \in [s]} w_i \log\left(\frac{\vol(\ball_{d_i}(\vz_i, 1/\bar r))}{\vol(\ball_{d_i}(0, 1/\bar R))}\right) =\sum_{i \in [s]} w_i d_i\log(\bar R/\bar r)  \numberthis\label[ineq]{eq:PsiChangeInitEnd}
\end{align*}

We finishes the proof by combining the inequality above and \cref{ineq:PsiInitMinusPsiEnd}.
\end{proof}

\begin{lemma}[Total oracle cost]
    \label[lem]{lem:totalOracleComplexity} Suppose the inputs $\kin$ and $\kout$ to \cref{alg:min-sum-convex-blackboard} satisfy $\kout\subseteq\ball_m(0,\bar{R})$ with $\bar R = O(\sqrt nR)$ and $\kin \supseteq \ball(\vz, \bar r)$ with $\bar r = r / \operatorname{poly}(m)$. Then, when \cref{alg:min-sum-convex-blackboard} terminates at $t \geq \tend$, it outputs a solution $\vx$ that satisfies 
\[
\vc^\top \vx \leq \min_{\vx\in\kcal,A\vx=\vb} \vc^\top\vx+  \epsilon\cdot \|\vc\|_2R.
\]
Moreover, if the cost of the separation oracle on $\kcal_i$ is $w_i$ and $n_i$ is the number of times $\oi$ is queried for all $i\in[s]$, then the total cost of the separation oracle is at most $O\left( \sum_i w_i d_i \log\left( \frac{mR}{\epsilon r}\right)\right)$. Namely, 
\[
    \sum_i^s w_i n_i \leq O\left( \sum_i w_i d_i \log\left( \frac{mR}{\epsilon r}\right)\right).    
\]
\end{lemma} 
\begin{proof}
 Let $\nt$ be the number of times $t$ is updated; $\nini$ the number of times $\kini$ is updated; $\nouti$ the number of times $\kouti$ is updated; $\nx$ the number of times $\vx$ is updated, and $\nn$ the total number of iterations of the \texttt{while} loop before termination of \cref{alg:min-sum-convex-blackboard}. 
 Then, combining \cref{lem:outer-progress,lem:pot-change-approx,lem:PotChangeT,lem:UniversalBarrierPotentialChange,lem:potChangeX} gives
\[ 
    \potend\leq \potinit - \sum_{i=1}^s 0.1w_i\cdot\nouti  + \nt\cdot (\eta + \eta^2) +\sum_{i=1}^s \nini \cdot  w_i \log(1-1/e + \eta) + \nx\cdot \left(-\frac{\eta^2}{4}\right).\numberthis\label[ineq]{eq:totalPotChange}
\] 
The initialization step of \cref{alg:min-sum-convex-blackboard} chooses $\eta = 1/100$, $\tend = \frac{8 m }{\epsilon \|\vc\|_2 R}$, and $\tinit = \frac{m\log(m)}{\sqrt{n} \|\vc\|_2 R}$, and we always update $t$ by a multiplicative factor of $1+\frac{\eta}{4 \sum_i w_i \nu_i}$ (see \cref{line:updateT-Rule}); therefore, we have 
\[
\nt = O(\sum_i w_i \nu_i\log(mR/(\epsilon r)).\numberthis\label{eq:BoundOnNt}
\] 
From \cref{alg:min-sum-convex-blackboard}, the only times the separation oracle is invoked is when updating $\kin$ or $\kout$ in \cref{line:KinUpdated} and \cref{line:KoutUpdated}, respectively.  Therefore, the cost of separation oracle on $\ki$ is $w_i n_i =w_i (\nini + \nouti)$. Therefore, we have by applying the bound on $\potinit-\potend$ from \cref{lem:InitMinusFinalpotential-change} and the bound on $\nt$ from \cref{eq:BoundOnNt} 
\[
 \sum_{i=1}^s w_i n_i = \sum_{i=1}^s w_i (\nin + \nout) \leq O(1)\cdot \left[\potinit - \potend +\nt \right] =  O(\sum_{i=1}^s w_i d_i\log(mR/(\epsilon r)), 
 \]
  which is the claimed separation oracle complexity.  We now prove the guarantee on approximation. Let $\vx_{\text{output}}$ be the output of \cref{alg:min-sum-convex-blackboard} and $\vx$ be the point which entered \cref{line:updateT} right before termination. Note that the termination of \cref{alg:min-sum-convex-blackboard} implies, by  \cref{line:updateT}, that \[
 \vc^\top \vx_{\text{output}} \leq \vc^\top \vx + \frac{\nu}{\tend}\leq \vc^\top \vxos + \frac{4(n+m)}{\tend}  \leq \min_{\vx\in\kcal,A\vx=\vb} \vc^\top\vx+  \epsilon\cdot \|\vc\|_2 \cdot R
 \]
 where the first step is by the second inequality in \cref{lem:two-sided-ineq} (using the universal barrier) and the last step follows by our choice of $\tend$ and the definition of $\vxos$ and $\kout\supseteq \kcal$.
\end{proof}

\subsubsection{Proof of \cref{thm:MainThmOfKiProblem}}
We now use the results from the prior sections to complete our proof of \cref{thm:MainThmOfKiProblem}. 
\begin{proof}[Proof of \cref{thm:MainThmOfKiProblem}]
We apply \cref{thm:initOne} for each $\ki$ separately to find a solution $\vz_i$. Then $\vz =(\vz_1, \dots, \vz_n) \in \R^{m+n}$ satisfies $\ball_{m+n}(\vz,\bar r)\subset \kcal$ with $\bar r=\frac{r}{6d^{3.5}}$. 
Then, we modify the convex problem as in \cref{lem:initial-point} with $\alpha=2^{16}\frac{m^{2.5} R}{r\epsilon}$ and obtain the following:
\[
\begin{array}{ll}
    \mbox{minimize} & \bar{\vc}^\top{\bar \vx} \\ 
    \mbox{subject to} & \bar\ma \bar \vx = \bar \vb,\\
    &\bar \vx \in \bar \kcal\defeq\kcal \times \R_{\geq 0}^{m+n}\times\R_{\geq 0}^{m+n}\numberthis \label[prob]{eq:modifiedConvexProgram}
\end{array}
\]
with 
\[
\bar \ma = [\ma \mid \ma \mid -\ma], \bar \vb = \vb, \bar \vc = (\vc, \frac{\|\vc\|_2 s}{\sqrt{m+n}}\cdot \mathbf{1},\frac{\|\vc\|_2 s}{\sqrt{m+n}}\cdot \mathbf{1})^\top
\]
We solve the linear system $\ma \vy=\vb-\ma\vz$ for $\vy$. Then, we construct the initial $\ovx$ by setting $\ovx^{(1)}=\vz$, 
\[
\ovx^{(2)}_i = \begin{cases}
\vy_i &\text{if } \vy_i\geq 0,\\ 
0 &\text{otherwise.}
\end{cases}
\quad \text{ and } \quad 
\ovx^{(3)}_i = \begin{cases}
-\vy_i &\text{if } \vy_i< 0,\\ 
0 &\text{otherwise}.
\end{cases}
\]Then, we run \cref{alg:min-sum-convex-blackboard} on the \cref{eq:modifiedConvexProgram}, with 
initial $\ovx$ set above,
$\bar m = 3(m+n), 
\bar n = n +2, 
\bar \epsilon = \frac{\epsilon }{6\sqrt{n}s}, \overline{\mathcal{K}}_{\textrm{in}} =\{\vx^{(1)}\in B(\vz,\bar{r}),(\vx^{(2)},\vx^{(3)})\in\R_{\geq0}^{2n}\}$ and $ \kouthat = \ball_{\bar m}(\mathbf{0}, \sqrt{n}R)$. 
By our choice of $\tend$, we have \[
\bar{t}_{\textrm{end}} 
= \frac{8 \bar m}{\bar \epsilon \|\bar \vc\|_2\bar R } 
\leq\frac{48m}{\epsilon\|\vc\|_2 R}.
\] First, we check the condition that $\alpha\geq 48\bar \nu \bar{t}_{\textrm{end}} \sqrt{m+n} \frac{R^2}{r}\|\vc\|_2$, we note that
\[
48\bar \nu \bar{t}_{\textrm{end}} \sqrt{m+n} \frac{R^2}{r}\|\vc\|_2 
\leq 27648\frac{m^{2.5} R}{\epsilon r} \leq  2^{16}\frac{m^{2.5} R}{r\epsilon} = \alpha.
\] Let $\bar{\vx}_{output}=(\vx_{output}^{(1)},\vx_{output}^{(2)},\vx_{output}^{(3)})$ be the output of \cref{alg:min-sum-convex-blackboard}. 
Then, let $\vx_{output}=\vx_{output}^{(1)}+\vx_{output}^{(2)}-\vx_{output}^{(3)}$ as defined in  \cref{lem:initial-point}.
By \cref{lem:totalOracleComplexity}, we have 
\[
\min_{\vx\in \pin} \bar\vc^\top\ovx \leq \min_{\vx\in \mathcal{P}} \vc^\top\ovx +\gamma
\]
where $\gamma = \bar \epsilon \cdot \|\bar \vc\|_2 \cdot \bar R$. 
Applying (3) of \cref{lem:initial-point}, we have 
\[
\vc^\top \vx_{output} \leq \frac{\bar \nu+1}{\bar{t}_{\textrm{end}}}+\gamma+\min_{x\in \kcal,A\vx=\vb} \vc^\top \vx \leq  \min_{x\in \kcal,A\vx=\vb} \vc^\top \vx + \epsilon\cdot \|\vc\|_2 \cdot R.
\]
The last inequality follows by our choice of $\bar \epsilon$ and $\bar{t}_{\textrm{end}}$, we have 
$\gamma \leq \frac{\epsilon}{2}\|\vc\|_2R$ and $\frac{\bar \nu+1}{\bar{t}_{\textrm{end}}}\leq \frac{\epsilon}{2}\|\vc\|_2R$.
Plug this $\bar\epsilon$ in \cref{lem:totalOracleComplexity}, it gives the claimed oracle cost.

\qedhere
\end{proof}

\subsubsection{Proof of Main Result of Finite-Sum Minimization (\cref{thm:mainFinSumMin})}

\thmmainFinSumMain* 

\begin{proof}
    First, we reformulate the problem into \cref{eq:1main} using \cref{lem:problem-reduction}.
    Then, we apply \cref{thm:MainThmOfKiProblem} to the reduced problem to get the error guarantee. For the communication complexity, we note that during  initialization, each machine sends the initial $\kini$ it found, which takes $O(\sum_{i=1}^s d_i L)$ bits of communication. Then, in the main loop, each machine sends the output of the separation oracle, which takes $O(\sum_{i=1}^s d_i^2 \log(sd/\epsilon)\cdot  L)$ bits of communication by setting $w_i = d_i$. 

    Finally, we show that it suffices to take the word length $L=O(\log(dR/r))$, which is $O(\log (d))$ by our choice of $R$ and $r$.  Recall that in the algorithm, each message encodes the description of a halfspace $\mathcal H$, which is described using two vectors $\vu,\vv$ by $\mathcal{H}=\{\vx:\vu^\top(\vx-\vv)\geq 0\}$. We describe these vectors in the relative scale of $r$. \cref{lem:outer-progress} shows that if suffices to send $\vv$ such that $\|\vv_i-\mu_i\|_{(\nabla^2 w_i \varphi_i(\mu_i))^{-1}}\leq 0.01$. Using \cref{cor:hessian-lower-bound}, we know $(\nabla^2 w_i \varphi_i(\mu_i))^{-1} \succeq 4w_i R$, hence it suffices to set word length of $\vv$ to be $O(\log(d R/r))$. For $\vu$, we note that even if $\mathcal{H}$ cuts through $\kcal$, as long as the radius of $\kcal$ is decreasing at rate of $1/d^C$ for some large constant $C$, the algorithm still works. Therefore, it suffices to choose the word length of $\vu$ to be $O(\log(dR/r))$.
    \end{proof}

\subsection{Reductions and Initializations}
\begin{lemma}\label{lem:problem-reduction}
    Given the same setup of \cref{thm:mainFinSumMin}, there is an algorithm using $O(sd \log(d))$ bits of communication, which reduces the original problem to the following formulation,
    \[
    \begin{array}{ll}
        \mbox{minimize} & \vc^\top\vx,  \\
         \mbox{subject to} & \vxi\in\ki\subseteq \R^{d_i + 1} \;\forall i\in[s]\\
         &  \ma\vx=\vb.
    \end{array} 
    \]
    where $\vx = [\vx_i]$ concatenates the $s$ vectors $\vx_i\in \R^{d_i}$. Denote $\kcal = \kcal_1 \times \kcal_2 \times \dotsc \times \kcal_s$, $\kcal$ satisfies the following properties:
    \begin{itemizec}
        \item convexity: Each $\kcal_i$ is convex and disjoint with each other.
        \item outer radius $R$: For any $\vx\in \ki$, we have $\|\vx_i\|_2 \leq R$.
        \item inner radius $r$: There exists a $\vz \in \R^d$ such that $\ma\vz=\vb$ and $\ball(\vz,r)\subset \kcal$.
        \item radius ratio: $R/r = O(1)$.
    \end{itemizec}
    After the reduction, all the machines hold all the data --- radii $R$ and $r$, vectors $\vc, \vb,\vx$, matrix $\ma$ --- and the $i^\mathrm{th}$ machine holds the separation oracle $\oi$ for the $i^\mathrm{th}$ set $\ki$. 
\end{lemma}
\begin{proof}
    The reduction is standard by using using a change of variables and the epigraph trick, but we include it here for completeness.
    Suppose each $f_i$ depends on $d_i$ coordinates of $\vtheta$ given by $\{i_1, \dots, i_{d_i}\} \subseteq [d]$. Then,
    symbolically define $\vx_i = [x^{(i)}_{i_1}; x^{(i)}_{i_2}; \dots; x^{(i)}_{i_{d_i}}] \in \R^{d_i}$ for each $i \in [n]$.
    Since each $f_i$ is convex and supported on $d_i$ variables, its epigraph is convex and $d_i+1$ dimensional. So we may define the convex set
    \[
    \ki^{\textrm{unbounded}} = \left\{(\vx_i, z_i) \in \R^{d_i + 1}: f_i (\vx_i) \leq L z_i \right\}.
    \]
    Finally, we add linear constraints of the form $x^{(i)}_{k} = x^{(j)}_{k}$ for all $i,j,k$ where $f_i$ and $f_j$ both depend on $\vtheta_k$. We denote these by the matrix constraint $\ma \vx = \vb$.
    Then, the problem is equivalent to
    \begin{equation} \label[none]{eq:linear-formulation}
    \begin{array}{ll}
        \mbox{minimize}  & \sum_{i=1}^s L z_i \\
        \mbox{subject to} &  \ma \vx =\vb \\
          & (\vx_i, z_i) \in \ki^{\textrm{unbounded}} \text{ for each $i \in [n]$}.
    \end{array}
    \end{equation}
    Since we are given $\vtheta^{(0)}$ satisfying $\|{\vtheta^{(0)} - \vtheta^*}\|_2 \leq D$, 
    we define $\vx_i^{(0)} = [ \vtheta^{(0)}_{i_1}; \dots, \vtheta^{(0)}_{i_{d_i}}]$ and 
    $z_i^{(0)} = f_i(\vtheta^{(0)})/L$.
    Then, we can restrict the search space $\ki^{\textrm{unbounded}}$ to
    \begin{align*}
    \ki &= \ki^{\textrm{unbounded}}  \cap \{ (\vx_i, z_i) \in \R^{d_i+1}: \|\vx_i-\vx_i^{(0)}\|_2 \leq D \text{ and }
        z_i^{(0)} -2D \leq z_i \leq z_i^{(0)} + 2D \}.
    \end{align*}
    
    One can then check that $\kcal_i$ is contained in a ball of radius $5D$ centered at $(\vx^{(0)}_i, z_i^{(0)})$  
    and contains a ball of radius $D$ centered at $(\vx^{(0)}_i, z_i^{(0)})$. 
    The subgradient oracle for $f_i$ translates to a separation oracle for $\ki$.
    We note that this reduction only requires the knowledge of $L$, $R$,  $\theta^{(0)}$, and the support $D_i$ for each $f_i$. By sending these information to blackboard, each machine can apply this reduction on their own. This takes $O\left(\sum_{i=1}^s d_i \log(d)\right)$ many bits of communication.
\end{proof}

Now, we show how to construct an initial set $\kin$ and find a good initial point for \cref{alg:min-sum-convex-blackboard} by slightly modifying the convex program above. These results first appeared in \cite{dong2022decomposable}, which we slightly modify to suit our purpose. 

\begin{theorem}[{\cite[Lemma 5.1]{dong2022decomposable}}]\label{thm:initOne}
    Suppose we have separation oracle access to a convex set $\kcal$ 
    satisfying $\ball(\vz,r)\subseteq\kcal\subseteq \ball(\mathbf{0},R)$ for some $\vz\in\R^{d}$.
    Then, there is a randomized algorithm, which in $O(d\log(R/r))$ separation oracle calls to $\kcal$, outputs a point
    $\vx$ such that $\ball\left(\vx,\frac{r}{6 d^{3.5}}\right)\subseteq \kcal$. 
\end{theorem}

\noindent To find a good initialization for \cref{alg:min-sum-convex-blackboard}, we need to slightly modify the convex program, for which we simply invoke the following result from \cite{dong2022decomposable}. 
\begin{lemma}[{\cite[Lemma 5.6]{dong2022decomposable}}]\label{lem:initial-point}
    Given a convex program $\min_{\ma\vx=\vb,\vx\in\mathcal{K}\subseteq \R^d}\vc^{\top}\vx$
    with outer radius $R$ and some $\alpha > 0$,
    we define  $\vc_{1}=\vc,\vc_{2}=\vc_{3}=\frac{\alpha \|\vc\|_2}{\sqrt{d}} \cdot \mathbf{1}$  and $\mathcal{P}=\{\vx^{(1)}\in\mathcal{K},(\vx^{(2)},\vx^{(3)})\in\R_{\geq0}^{2d}:\ma(\vx^{(1)}+\vx^{(2)}-\vx^{(3)})=\vb\}$. We then define the {\em modified convex program} by
    \[
\min_{(\vx^{(1)},\vx^{(2)},\vx^{(3)})\in\mathcal{P}}\vc_{1}^{\top}\vx^{(1)}+\vc_{2}^{\top}\vx^{(2)}+\vc_{3}^{\top}\vx^{(3)}.
    \]
    Given some $\kin\in \kcal$ where inner radius $r$, and  an arbitrary $t\in \R_{\geq 0}$, we further define the function \[f_{t}(\vx^{(1)},\vx^{(2)},\vx^{(3)})=t(\vc_{1}^{\top}\vx^{(1)}+\vc_{2}^{\top}\vx^{(2)}+\vc_{3}^{\top}\vx^{(3)})+\barr_{\mathcal{P}_{\text{in}}}(\vx^{(1)},\vx^{(2)},\vx^{(3)})\] 
where $\barr_{\mathcal{P}_{\text{in}}}$ is some $\nu$ self-concordant barrier for the set
\[\mathcal{P}_{\text{in}}=\{\vx^{(1)}\in\kin,(\vx^{(2)},\vx^{(3)})\in\R_{\geq0}^{2d}:\ma(\vx^{(1)}+\vx^{(2)}-\vx^{(3)})=\vb\}.\]
Given $\overline{\vx}_{t}\defeq(\vx_{t}^{(1)},\vx_{t}^{(2)},\vx_{t}^{(3)})=\arg\min_{(\vx^{(1)},\vx^{(2)},\vx^{(3)})\in\pin}f_{t}(\vx^{(1)},\vx^{(2)},\vx^{(3)})$, we denote $\vx_{\text{in}}=\vx_{t}^{(1)}+\vx_{t}^{(2)}-\vx_{t}^{(3)}$.
Suppose $\min_{\ovx \in\pin} \bar\vc^\top \ovx  \leq \min_{\ovx\in\mathcal{P}	}\bar\vc^\top \ovx + \gamma$ and $\alpha \geq 48\nu t \sqrt{d}\cdot\frac{R}{r}\cdot\|c\|_{2}R$, then we have that $\ma\vx_{\text{in}}=\vb$,  $\vx_{\text{in}}\in\kin$,
and  $\vc^{\top}\vx_{\text{in}}\leq \min_{\vx\in\kcal,\ma\vx=\vb}\vc^{\top}\vx +\frac{\nu+1}{t}+\gamma$.
     We denote $(\vc_{1},\vc_{2},\vc_{3})$ 
    by $\overline{\vc}$.
\end{lemma}

\section{Lower Bounds}\label{sec:lower_bounds_full_section}
\subsection{Lower Bound Primitives}

We introduce two fundamental communication problems, the latter of which is an $s$-player version of the first.

\begin{problem}
\label{prob:two_player_inner_prod}
Alice an Bob hold unit vectors $\vv$ and $\vw$ respectively in $\R^d$.  They would like to decide between (a) $\inner{\vv}{\vw}=0$, and (b) $\abs{\inner{\vv}{\vw}}\geq \frac{\eps}{d}$ under the promise that one of these conditions holds.
\end{problem}

\probSPlayerInnerProd*

For each of these problems we prove a corresponding hardness hardness result. For the two player version of the game we have the following communication lower bound.

\begin{lemma}
\label{lem:two_player_inner_product_hardness}
A protocol solving \cref{prob:two_player_inner_prod} with probability at least $0.9$ requires at least $\Omega(d\log\frac{1}{\eps})$ communication for $r$-round protocols when $r\leq c\log\frac{1}{\eps}/\log\log\frac{1}{\eps}$ for an absolute constant $c$.
\end{lemma}

While our argument requires a technical assumption on the number of rounds, this can almost certainly be removed via a more careful analysis.  Moreover, in constant dimension our argument directly implies the bound above, with no requirement on the number of rounds.

Given the lemma above we will show how to boost it to an $s$-player lower bound.

\thmSPlayerInnerProd*

We will prove these results below.  Before presenting the proofs we give the reductions to linear regression and linear programming.

\subsection{Reduction to Linear Regression}
We now present our lower bound for \cref{prob:communication_of_linear_regression} and will then return to analyzing the communication complexity of the two problems above.

Our main interest is the following communication problem.
\begin{problem}
\label{prob:communication_of_linear_regression}
Each of $s$ servers holds a matrix $\ma^{(i)}$ and vector $\vb^{(i)}$ for $i=1,\ldots, s.$ Let $\ma$ and $\b$ the vertical stack of the the $\ma^{(i)}$'s and $\vb^{(i)}$'s respectively. All entries are held to $L$ bits of precision. The coordinator must produce a vector $\hat{\vx}$ with \[\norm{\ma \hat{\vx} - \vb}{} \leq 2 \norm{\ma \vx_* - \vb}{},\]
where $\vx_* = \argmin_x \norm{\ma \vx - \vb}{}.$
\end{problem}

\begin{theorem}
\label{thm:lower_bound_low_accuracy_regression}
A protocol that solves \cref{prob:communication_of_linear_regression} with at least $0.99$ probability requires $\tilde{\Omega}(sd\log L)$
bits of communication provided that there are at most $C \log L / \log\log L$ rounds of communication between the coordinator and each server for an absolute constant $C$.
Additionally, if $\ma$ is promised to have condition number at most $\kappa$ then assuming at most $C\log \kappa / \log \log \kappa$ rounds per server, any protocol requires at least $\tilde{\Omega}(sd\min(L,\log\kappa)$ communication.
\end{theorem}

\begin{proof}
We reduce from \cref{prob:s_player_inner_prod} above. To construct a linear system, the coordinator first computes an orthonormal basis $\vu_1,\ldots \vu_{d-1}$ for $\vv^{\perp}$ and then rounds each vector to $L$ bits of precision to obtain $\vu_1', \ldots, \vu_{d-1}'$.  The coordinator then forms a matrix with rows $\alpha \vv', \vu_1', \ldots, \vu_{d-1}',$ where $\alpha<1$ is a small parameter that will be chosen later and where $\vv'$ is $\vv$ rounded to $L$ bits of precision.  Each server $i$ simply holds the vector $\vw_i'$ which is $\vw_i$ rounded to $L$ bits of precision.  Let $\ma$ denote this matrix which is distributed across the servers.  Also set $\vb$ to have all entries equal to $0$, except with a $1$ in the entry corresponding to the row $\alpha \vv'$.  We will show that the norm of an approximate regression solution $\hat{\vx}$ allows us to distinguish between the two possibilities in \cref{prob:s_player_inner_prod}.

Suppose that for some $i$, $\abs{\inner{\vw_i}{\vv}} \geq \eps' := \eps/d.$
We lower bound the smallest singular value of $\ma.$
Let $\mm$ be the matrix with rows $\vu_1,\ldots, \vu_{d-1}, \vw_i.$  To lower bound the smallest singular value of $\mm$, let $\vx$ be an arbitrary unit vector and write $\vx = \vx_1 + \vx_2$ where $\vx_1$ is the projection of $\vx$ onto $\vv.$ Then $\norm{\mm \vx}{} \geq \norm{\vx_2}{}$ and also 
\begin{align*}
\norm{\mm \vx}{} \geq \abs{\inner{\vx}{\vw_i}}
&= \abs{\inner{\vx_1}{\vw_i} + \inner{\vx_2}{\vw_i}}\\
&\geq \abs{\inner{\vx_1}{\vw_i}} - \norm{\vx_2}{}\\
&= \abs{\inner{\norm{\vx_1}{}\vv}{\vw_i}} - \norm{\vx_2}{}\\
&\geq \eps' \norm{\vx_1}{} - \norm{\vx_2}{}.
\end{align*}

Thus 
\[\norm{\mm \vx}{} \geq \max(\norm{\vx_2}, \eps'\norm{\vx_1}{} - \norm{\vx_2}{}) 
\geq \max(\norm{\vx_2}, (\eps'/2)\norm{\vx_1}{}),
\] since the maximum of two numbers is at least their average. Since $\norm{\vx_1}{}^2 + \norm{\vx_2}{}^2 = 1,$ this latter quantity is at least $\eps'/3.$ This shows that $\sigma_{\min}(\mm) \geq \eps'/3.$

Let $\mm'$ be matrix $\mm$ with the rows rounded as above.  Note that all entries of $\mm-\mm'$ are bounded in absolute value by $2^{-L}$ so $\norm{\mm-\mm'}{} \leq 2^{-L} d.$  It follows that 
\[
\sigma_{\min}(\ma) \geq \sigma_{\min}(\mm') \geq \eps'/3 - 2^{-L}d \geq \eps'/4,
\]
provided that we later choose $\log(1/\eps)\leq L - \log(12d).$

Now let $\hat{\vx}$ satisfy $\norm{\ma\hat{\vx} - \vb}{}^2 \leq (1+\beta)\norm{\mm \vx_* - \vb}{}^2$ (we will later set $\beta=1$). Then we have
\[
(1+\beta)\norm{\ma \vx_* - \vb}{}^2 \geq \norm{\ma \hat{\vx} - \vb}{}^2 = \norm{\ma (\hat{\vx} - \vx_*)}{}^2 + \norm{\ma x_* - \vb}{}^2,
\]
and so $\norm{\ma (\hat{\vx}-\vx_*)}{} \leq \sqrt{\beta}$. So we get
\[
\sigma_{\min}(\ma)\norm{\hat{\vx}}{} \leq \norm{\ma\hat{\vx}}{} \leq \norm{\ma \vx_*}{} + \sqrt{\beta} \leq 1 + \sqrt{\beta}.
\]
Hence $\norm{\hat{\vx}}{} \leq \frac{4}{\eps'}(1 + \sqrt{\beta}).$

Now suppose that $\abs{\inner{\vw_i}{\vv}} = 0$ for all $i.$ Set $\eta = 2^{-L}d$. In this case, we have
\begin{align*}
\norm{\ma (\alpha^{-1} \vv) - \vb}{}^2 
&= (\inner{\vv}{\vv'} - 1)^2 + \sum_{i=1}^{d-1} \inner{\alpha^{-1} \vv}{\vu_i'}^2  + \sum_{i=1}^s \inner{\alpha^{-1} \vv}{\vw_i'}^2\\
&\leq \eta^2 + \alpha^{-2}d\eta^2 + \alpha^{-2} s \eta^2\\
&= (\alpha^{-2}s + \alpha^{-2}d + 1)\eta^2.
\end{align*}

Suppose that $\vx_*$ is the least squares solution and that $\norm{\ma \hat{x} - b}{}^2 \leq (1+\beta)\norm{\ma \vx_* - \vb}{}^2.$ Then
\begin{align*}
(1+\beta) (\alpha^{-2}s + \alpha^{-2}d + 1)\eta^2
&\geq (1+\beta)\norm{\ma (\alpha^{-1}\vv) - \vb}{}^2\\
&\geq (1+\beta)\norm{\ma \vx^* - \vb}{}^2\\
&\geq \norm{\ma \hat{\vx} - \vb}{}^2\\
&\geq (\inner{\hat{\vx}}{\alpha \vv} - 1)^2.
\end{align*}

This implies that
\[
\norm{\hat{\vx}}{}\geq \abs{\inner{\hat{\vx}}{\vv}} \geq \alpha^{-1} \left(1 - \eta \sqrt{(1+\beta) (\alpha^{-2}s + \alpha^{-2}d + 1)}\right).
\]

We can distinguish the two inputs using $\norm{\hat{x}}{}$ as long as
\[
\frac{4}{\eps'} = \frac{4d}{\eps} < \alpha^{-1} \left(1 - \eta \sqrt{(1+\beta) (\alpha^{-2}s + \alpha^{-2}d + 1)}\right).
\]
To make this happen we choose our parameters.  Set $\beta=1$.  The quantity in parentheses is at least $1/2$ as long as $\eta \alpha^{-1}\sqrt{s+d} \leq 1/4.$  So it suffices for this to hold, along with $\alpha^{-1} > \frac{8d}{\eps}.$ 

There is an $\alpha$ that satisfies these bounds, as long as $\frac{8d}{\eps} < \eta^{-1}(s+d)^{-1/2} = \frac{2^L}{d}(s+d)^{-1/2} $ which holds when $\log\frac{1}{\eps} \leq L - 4\log(s+d).$  Given this bound, we can then take $\alpha = O(\eps/d)$ (say rounded to an appropriate power of $2$ so that bit precision is unaffected), and take $\eps$ such that $\log\frac{1}{\eps} \leq L - c\log(s+d)$ for an abolute constant $c.$  Then \cref{prob:s_player_inner_prod} that we reduced from, requires $\Omega(sdL)$ communication by \cref{thm:s_player_inner_product_hardness}, provided that $L\geq 8(s+d)$, and therefore so does \cref{prob:communication_of_linear_regression}.

Finally we check the condition number of our hard instance. Note that we always have $\kappa(\ma) = \frac{\sigma_{\max}(\ma)}{\sigma_{\min}(\ma)} \leq \frac{\sqrt{2d+s}}{\alpha} = O(\frac{1}{\eps} \poly(s+d)).$  Choosing 
\[
\eps \approx \min(\poly(s+d)/\kappa, L - c\log(s+d))
\]
gives the second statement of the theorem.
\end{proof}

\subsection{Reduction to Linear Feasibility}

\begin{problem}
    \label[prob]{prob:linear_feasibility}
    (Linear Feasibility) Each of $s$ servers holds a matrix $\ma^{(i)}$  and vector $\vb^{(i)}.$ All entries are held to $L$ bits of precision.  They would like to decide whether there is an $\vx$ satisfying $\ma \vx \leq \b$.
\end{problem}

\thmLinFeasLowerBound*

\begin{proof}
We give a simple reduction from \cref{prob:s_player_inner_prod}.  Recall that in this problem, the coordinator holds a vector $\vv$ and the $s$ servers holds vectors $\vw_1,\ldots, \vw_s.$

Set $\eta = 2^{-L}\lceil\sqrt{d}\rceil$. Let $v'$ denote $v$ with each entry rounded towards $0$ to $L$ bits of precision, and similarly for $\vw_1',\ldots, \vw_s'.$  For our reduction, each server $i$ sets $\ma^{(i)}_1 = \vw_i'$, $\ma^{(i)}_2 = -\vw_i'$ and sets $\vb^{(i)}_1 = \vb^{(i)}_2 = 2\eta$, thereby creating the constraint $\abs{\inner{\vw_i'}{\vx}} \leq 2\eta.$

The coordinator similarly adds rows corresponding to the constraints $x_j \leq v'_j$ and $-x_j\leq v'_j$ for $j=1,\ldots,d$, which simply amounts to the constraint $\vx = \vv'.$

Suppose that $\inner{\vw_i}{\vv} = 0$ for all $i.$  Then 
\[\abs{\inner{\vv'}{\vw_i'}} 
= \abs{\inner{\vv'}{\vw_i'} - \inner{\vv}{\vw_i}}
= \abs{\inner{\vv - \vv'}{\vw'} + \inner{\vv'}{\vw_i - \vw_i'}}
\leq \norm{\vv-\vv'}{} + \norm{\vw_i - \vw_i'}{}
\leq 2\eta
\]
for all $i$ since $\norm{\vv' - \vv}{}\leq \eta$ and $\norm{\vw_i' - \vw_i}{}\leq \eta$.  This means that the constraints are satisfied by taking $\vx = \vv'.$

On the other hand, suppose that $\abs{\inner{\vw_i}{\vv}}\geq \eps/d$ for some $i.$ Then we must have $\vx = \vv'$, but then 
\[
\abs{\inner{\vx}{\vw_i'} - \inner{\vv'}{\vw_i'}} \leq 2\eta
\]
by the same calculation as above, which means that $\abs{\inner{\vx}{\vw_i'}}\geq \eps/d - 2\eta.$ So the constraints are not satisfiable as long as $\eps$ is chosen so that $\eps > 4\eta d.$  Therefore a protocol that solves the linear feasibility problem above can solve \cref{prob:s_player_inner_prod} for $\eps = O(\eta d)$, which means it requires at least 
$\Omega(sdL)$ communication by \cref{thm:s_player_inner_product_hardness}. 
\end{proof}

\subsection{Proof of \cref{thm:s_player_inner_product_hardness}}
\subsubsection{Harmonic Analysis Setup}

We will use a similar set of tools to \cite{regev2011quantum}.  To streamline the analysis, we very briefly recall some facts about Fourier analysis on the sphere.  We will use the notation $L^2(\sphere{d-1})$ to indicate the space of real-valued square-integrable functions on $\sphere{d-1}.$  This is a Hilbert space with inner product given by 
\[
\inner{f}{g}_{L^2} = \int_{\sphere{d-1}} f(\vx) g(\vx) d\sigma_{d-1}(\vx),
\] 
where $\sigma_{d-1}$ is the rotationally invariant probability measure on $\sphere{d-1}.$

The fundamental fact from Fourier analysis that we use is that any $L_2$ function $f$ on the $\sphere{d-1}$ can be decomposed into a sum of spherical harmonics as
\[
f = \sum_k \Pi_{\mathcal{H}_k} f
\]
where $\mathcal{H}_k$ is the subspace consisting of spherical harmonics of degree $k$ and $\Pi_{\mathcal{H}_k}$ is the orthogonal projection onto that subspace. Any operator on $L^2(\sphere{d-1})$ that commutes with rotations has an eigen-decomposition with eigenspaces given by the $\mathcal{H}_k$'s. 

The space $\mathcal{H}_k$ contains a special class of axially-symmetric functions known as the \textit{zonal spherical harmonics}. The zonal spherical harmonic $z_{k,\vv}$ of degree $k$ and axis of symmetry $\vv \in \sphere{d-1}$ is given by
\[z_{k,\vv}(\vx) = G_k^{(d/2 - 1)}(\inner{\vx}{\vv}),
\]
where $G_k^{(\alpha)}:[-1,1]\rightarrow \R$ is the so-called \textit{Gegenbauer polynomial} of degree $k$ with parameter $\alpha$.  There are many ways to define the Gegenbauer polynomials, for example by their generating function \cite{szego1962orthogonal}:
\begin{equation}
\label{eq:gegenbauer_generating_func}
\frac{1}{(1 - 2xt + t^2)^{\alpha}} = \sum_{k=0}^{\infty} G_k^{(\alpha)}(x)t^k.
\end{equation}
The Gegenbauer polynomials satisfy many interesting identities.  We will use the following identity for the derivative of $G_k^{(\alpha)}$ which follows from differentiating \cref{eq:gegenbauer_generating_func} \cite{szego1962orthogonal}:
\begin{equation}
\label{eq:gegenbauer_derivative}
\frac{d}{dx} G_k^{(\alpha)}(x) = 2\alpha G_{k-1}^{(\alpha+1)}(x).
\end{equation}
It is also simple to compute $G_k^{(\alpha)}(1)$ from \cref{eq:gegenbauer_generating_func}. Plugging in $x=1$ gives
\[
(1-t)^{-2\alpha} = \sum_{k=0}^{\infty} G_k^{(\alpha)}(1)t^k,
\]
so $G_k^{(\alpha)}(1) = (-1)^k\binom{-2\alpha}{k}$ by the generalized binomial theorem.  In particular we will use the following values below:
\begin{align}
\label{eq:specific_gegenbauer_values_at_1}
G_k^{(1/2)}(1) &= 1\\
G_k^{(3/2)}(1) &= \frac{(k+2)(k+1)}{2}.
\end{align}

The zonal spherical harmonics $z_{k,\vv}$ (in any dimension) are also well-known to have the maximum sup-norm among the degree $k$ spherical harmonics of fixed $L^2$ norm, with this maximum value achieved at $\vv.$  One way to see this is to recall that for an appropriate normalizing constant $c_{k,d}$, the zonal spherical harmonic $c_{k,d} z_{k,\vv}$ satisfies the reproducing property $\inner{c_{k,d} z_{k,\vv}}{f}_{L^2} = f(\vv)$ for all $f\in\mathcal{H}_k$ (see \cite{dai2013approximation} for example).  Therefore for $f\in \mathcal{H}(k)$ with $\norm{f}{L^2}=1$,
\[
\abs{f(\vv)} = \abs{\inner{c_{k,d} z_{k,\vv}}{f}_{L^2}}
\leq \inner{c_{k,d} z_{k,\vv}}{\frac{1}{\norm{z_{k,\vv}}{L^2}}z_{k,\vv}}_{L^2} 
\leq \frac{1}{\norm{z_{k,\vv}}{L^2}} z_{k,\vv}(\vv).
\]
As a particular consequence of this fact,
\begin{equation}
\label{eq:max_of_gegenbauers}
\sup_{x \in [-1,1]} \abs{G_k^{(n/2)}(x)} = G_k^{(n/2)}(1)\\
\end{equation}
for all natural numbers $n.$

We also recall the spherical Radon transform, also known as the Minkowski-Funk transform $R:L^2(\sphere{d-1})\rightarrow L^2(\sphere{d-1})$, which for a function $f$ on $\sphere{d-1}$ is defined by 
\[
Rf(x) = \int_{x^{\perp}} f(y) d\sigma_{x^{\perp}}(y),
\]
where $\sigma_{x^{\perp}}$ is the natural probability measure over $x^{\perp}\cap \sphere{d-1}.$  In other words, $Rf(x)$ is the average of $f$ over the spherical equator perpendicular to $x.$  The following computation of eigenvalues is classical, and due to Funk \cite{funk1911flachen}.  A similar formula holds in all dimensions, but we will only need the result for $\sphere{2}.$
\begin{proposition}
    \label{prop:eigenvalues_of_radon_transform}
    The eigenfunctions of the spherical Radon transform $R:L^2(\sphere{2})\rightarrow L^2(\sphere{2})$ are precisely the spherical harmonics of degree $k$, and the associated eigenvalues are
    \[
    \mu_k = (-1)^{k/2} \frac{(k-1)!!}{k!!}
    \]
    for $k$ even, and $0$ for $k$ odd.
\end{proposition}

To simplify our computations later, we give a simple upper bound on the eigenvalues of $R$ that follows immediately from the formula above.
\begin{proposition}
\label{prop:bound_on_mu_k}
Let $f_k$ be a unit spherical harmonic on $\sphere{2}$ of degree $k$.  Then 
\[
\mu_k^2 \leq \frac{1}{k}.
\]
\end{proposition}

\begin{proof}
By the formula above in \cref{prop:eigenvalues_of_radon_transform}, when $k$ is even,
\[
\abs{\mu_k} = \frac{1}{2}\cdot \frac{3}{4}\cdot \frac{5}{6}\cdot \ldots \cdot \frac{k-1}{k},
\]
so 
\[
\mu_k^2 \leq \frac{1}{2}\cdot \frac{2}{3}\cdot \frac{3}{4}\cdot \ldots \cdot \frac{k-1}{k} = \frac{1}{k},
\]
as desired.
\end{proof}

We will also make use of a local averaging operator $T_{\eps}$. Similar to the Radon transform we define $T_{\eps}f(\vx)$ to be the average of $f$ over the set $\{\vy\in \sphere{d-1}:\inner{\vx}{\vy} = 1 - \eps\}$ equipped with the probability measure that is invariant under rotations fixing $\vx.$

\begin{proposition}
\label{prop:eigenvals_of_T_eps}
The eigenfunctions of $T_{\eps}$ are the degree $k$ spherical harmonics, and when $d=3$, the associated eigenvalues $\xi_k$ satisfy the bound 
\[
\abs{1 - \xi_k} \leq \min(\eps k^{2},2).
\]
\end{proposition}

\begin{proof}
It is clear that $T_{\eps}$ commutes with rotations, so the $\mathcal{H}_k$'s are the eigenspaces for $T_{\eps}.$  We analyze the eigenvalues of $T_{\eps}$ by considering its action on the degree $k$ zonal spherical harmonic $z_{k,\vv} = G_k^{(1/2)}(\inner{x}{v})$. 

Since $z_{k,\vv}$ is known to be an eigenfunction of $T_{\eps}$, we must have
\[
\xi_k = \frac{T_{\eps} z_{k,\vv}(\vv)}{z_{k,\vv}(\vv)}
=\frac{G_k^{(1/2)}(1-\eps)}{G_k^{(1/2)}(1)}
= G_k^{(1/2)}(1-\eps).
\]
To bound this, recall the identity \cref{eq:gegenbauer_derivative} which gives
\[
\frac{d}{dx} G_k^{(1/2)}(x) = G_{k-1}^{(3/2)}(x).
\]
By the Mean Value Theorem, along with \cref{eq:max_of_gegenbauers} this gives
\[
\abs{G_k^{(1/2)}(1) - G_k^{(1/2)}(1-\eps)} 
\leq \eps\cdot\sup_{x} \abs{\frac{d}{dx}G_k^{(1/2)}(x)} 
= \eps\sup_x \abs{G_{k-1}^{(3/2)}(x)} 
\leq \eps G_{k-1}^{(3/2)}(1)
= \eps \cdot \frac{k(k+1)}{2}.
\]
Therefore,
\[
\abs{1 - \xi_k} 
= \abs{G_k^{(1/2)}(1) - G_k^{(1/2)}(1-\eps)}
\leq \eps k(k+1)
= \frac{\eps k(k+1)}{2}
\leq \eps k^2.
\]

Finally, recall that $|G_k^{(1/2)}(x)|$ attains it maximum on $[-1,1]$ at $x=1.$  So,
\[
\abs{1 - \xi_k} 
= \abs{G_k^{(1/2)}(1) - G_k^{(1/2)}(1-\eps)}
\leq \abs{G_k^{(1/2)}(1)} + \abs{G_k^{(1/2)}(1-\eps)}
\leq 2\abs{G_k^{(1/2)}(1)}
= 2.
\]

\end{proof}

This allows us to bound the operator norm of $R - T_{\eps} R$ which we will use in the next section.
\begin{proposition}
\label{prop:radon_operator_bound}
The operator norm of $R - T_{\eps}R$ satisfies the bound
\[
\norm{R-T_{\eps} R}{L^2 \rightarrow L^2} \leq 2\eps^{1/4}.
\]
\end{proposition}
\begin{proof}
Both $R$ and $T_{\eps}$ commute with rotations, and hence are diagonalized by the spherical harmonics.  The operator norm of $R-T_{\eps} R$ is bounded by its largest magnitude eigenvalue.  The eigenvalue of $R-T_{\eps} R$ on the degree $k$ spherical harmonics is $\mu_k (1 - \xi_k)$, which by \cref{prop:bound_on_mu_k} and \cref{prop:eigenvals_of_T_eps} is bounded in magnitude by
\[
\min(\eps k^{2}, 2) k^{-1/2} = \min(\eps k^{3/2}, 2k^{-1/2}) 
\leq (\eps k^{3/2})^{1/4}(2k^{-1/2})^{3/4}
\leq 2\eps^{1/4}.
\]
 \end{proof}

\subsubsection{A Discrepancy Bound for 3D Inner Product}

We define two distributions $\mathcal{D}_0$ and $\mathcal{D}_{\eps}$ on $\sphere{2}\times \sphere{2}$ that we we will show are hard to distinguish with low communication.  We let $\mathcal{D}_0$ be the rotationally invariant distribution over pairs of orthogonal vectors.  To describe $\mathcal{D}_{\eps}$ we give the procedure for taking a sample.  First sample $(v,w_0)$ from $\mathcal{D}_0$.  Then sample $w$ uniformly from $\{x:\inner{x}{w_0} = 1-\eps\}.$  The pair $(v,w)$ is a sample from $\mathcal{D}_{\eps}.$

To continue proving our lower bound, we recall the connection between communication protocols and rectangles.  If the universe for Alice's and Bob's inputs is $\mathcal{X}$ and $\mathcal{Y}$ then a communication protocol using $b$ bits of communication partitions $\mathcal{X}\times \mathcal{Y}$ into at most $2^b$ \textit{combinatorial rectangles}, i.e. sets of the form $A\times B$ where $A\subset \mathcal{X}$ and $B\subseteq \mathcal{Y}.$  The information that Alice and Bob learn from running the protocol is precisely the identity of the rectangle that their (joint) input lies in. Randomized protocols with public randomness can simply be thought of as samples from the space of deterministic protocols --- Alice and Bob still learn that their input lies in some rectangle.  For more details on the basic setup for communication lower bounds, we refer the reader to \cite{roughgarden2016communication}.

One specific approach for showing a communication lower is a so-called \textit{discrepancy bound} over rectangles.  To apply the technique, one chooses two distributions $\mathcal{D}_0$ and $\mathcal{D}_1$ over inputs and then shows for all $A\subset \mathcal{X}$ and $B\subseteq \mathcal{Y}$ the probability of an input landing in $A\times B$ is roughly the same under both distributions:
\[
\abs{\Pr_{(\vx,\vy)\sim \mathcal{D}_0}((x,y)\in A\times B) - \Pr_{(\vx,\vy)\sim \mathcal{D}_1}((\vx,\vy)\in A\times B))}
\leq \alpha.
\]
It is a standard fact that a discrepancy bound of the form above, implies that distinguishing $\mathcal{D}_0$ and $\mathcal{D}_1$ with constant advantage requires $\Omega(\log(1/\alpha))$ communication. The proof is simple:  For a fixed deterministic protocol using at most $N$ bits of communication, let $\tilde{\mathcal{D}}_i$ be the distribution over rectangles corresponding to samples from $\tilde{\mathcal{D}}_i$.  Each rectangle contributes at most $\alpha$ to the total variation distance between $\tilde{\mathcal{D}}_0$ and $\tilde{\mathcal{D}}_1$. There are at most $2^N$ rectangles, and so the total variation distance between the distributions is at most $2^N\cdot\alpha.$  The extension to randomized protocols is via Yao's minimax principle.

This discrepancy lower bound technique will give our lower bound for constant $d$ after the following proposition.  For large $d$, we will need to bootstrap our rectangle lower bound to an information lower bound.  This is reminiscent of the main result in \cite{braverman2016discrepancy}, however we will need a slightly stronger version of this lower bound, specialized to our setting, which doesn't seem to follow from their result.  We will therefore ultimately use the discrepancy lower bound in a somewhat different way from what is typical.

\begin{proposition}
    \label{prop:discrepancy_bound}
    Let $S,T\subset \sphere{2}$ be measurable sets.  The following discrepancy bound holds:
    \[
    \abs{\pr_{(\vx,\vy)\sim \mathcal{D}_0}\left((\vx,\vy)\in S\times T\right) - \pr_{(\vx,\vy)\sim \mathcal{D}_{\eps}}\left((\vx,\vy)\in S\times T\right)} \leq 2\eps^{1/4}.
    \]
\end{proposition}

\begin{proof}

    We first rewrite these probabilities in terms of $R$ and $T_{\eps}.$ Let $\chi_S$ and $\chi_T$ be the characteristic functions of $S$ and $T$ respectively. We have 
    \begin{align*}
    \pr_{(\vx,\vy)\sim \mathcal{D}_0}((\vx,\vy)\in S\times T) 
    &= \int_{\vx\in\sphere{3}} \int_{\vy\in x^{\perp}} \chi_S(\vx)\chi_T(\vy) d\sigma_{\vx^{\perp}}(\vy) d\sigma(\vx)\\
    &= \int_{\vx\in\sphere{3}} \chi_S(\vx) \int_{\vy\in \vx^{\perp}} \chi_T(\vy) d\sigma_{\vx^{\perp}}(\vy) d\sigma(\vx) \\
    &= \int_{\vx\in\sphere{3}} \chi_S(\vx) (R\chi_T)(\vx)\\
    &= \inner{R\chi_T}{\chi_S}_{L^2}.
    \end{align*}
    
    Similarly 
    \[
    \pr_{(\vx,\vy)\sim \mathcal{D}_{\eps}}\left((\vx,\vy)\in S\times T\right) = \inner{T_\eps R \chi_T}{\chi_S}_{L^2}.
    \]
    By Cauchy-Schwarz along with \cref{prop:radon_operator_bound} above,
    \[
    \abs{\inner{(R-T_{\eps} R)\chi_S}{\chi_T}_{L^2}} \leq 2\eps^{1/4}\norm{\chi_S}{L_2}\norm{\chi_T}{L_2}
    \leq 2\eps^{1/4},
    \]
    since $\norm{\chi_S}{L_2}$ and $\norm{\chi_T}{L_2}$ are both at most $1.$
\end{proof}

\begin{remark}
The above argument only applies in dimension $3$ and higher.  When $d=2$, $\mu_k$ does not decay with $k$ and so we can not bound the $\min$ above uniformly in $k$. This corresponds to the fact that the Radon transform is not smoothing on $\sphere{1}.$  Indeed for a function on $\sphere{1}$ which is symmetric about the origin, the Radon transform just performs a $90$ degree rotation.
\end{remark}

The following Proposition simply checks that samples $(\vv,\vw)\sim\mathcal{D}_{\eps}$ are likely to have inner product $\Omega(\eps).$

\begin{proposition}
\label{prop:D_eps_samples_have_large_inner_prod}
Let $(\vv,\vw)\sim \mathcal{D}_{\eps}.$ With probability at least $0.95$, $\abs{\inner{\vv}{\vw}} \geq \eps/20$.
\end{proposition}
\begin{proof}
By symmetry, we may assume that $\vv=[1,0,0]$ and $\vw_0 = [0,0,1]$ where $\vw_0$ is as in the definition of $\mathcal{D}_{\eps}$ above.  Let $\Pi$ be the projection onto the $xy$-plane.  Then $\inner{\vw}{\vv} = \inner{\Pi \vw}{\vv}$.  Note that \[
\norm{\Pi \vw}{} = \sqrt{1 - \inner{\vw}{\vw_0}^2} = \sqrt{2\eps - \eps^2}.
\] Letting $\theta$ be the angle that $\Pi w$ makes $\vv$, we have
\[
\abs{\inner{\vv}{\vw}} = \sqrt{2\eps - \eps^2}\abs{\cos\theta} \geq \sqrt{\eps}\abs{\cos\theta} \geq \eps \abs{\cos\theta}. 
\]
Note that $\theta$ is distributed uniformly over the unit circle, so with probability $0.95$, $\abs{\cos\theta} \geq \frac{1}{20}.$
\end{proof}

\subsubsection{Upgrading to Higher Dimensions}
\looseness=-1As discussed above, since we proved a discrepancy bound over rectangles, previous work \cite{braverman2016discrepancy} immediately implies an information complexity lower bound for the $3$-dimensional version of \cref{prob:two_player_inner_prod}.  We will boost this result to $d$ dimensions using a direct-sum type result from \cite{bar2004information} by viewing a $d$ dimensional vector as a concatenation of $d/3$ $3$-dimension vectors.  To obtain a direct sum result in our setting, we need the information lower bound on the input distribution $\mathcal{D}_0$. Our next goal is to obtain the necessary information lower bound for public coin protocols. Then we will borrow a ``reverse-Newman" result to upgrade to an information lower bound for private-coin protocols.

\looseness=-1We recall some basic definitions.  Given a protocol $\Pi$ depending on the inputs and public randomness $R$, and an input distribution $\mathcal{D}$, the (internal) information cost of the protocol is
\[
\IC_{\mathcal{D}}(\Pi) = I(\Pi; \vx|\vy,R) + I(\Pi; \vy|\vx,R)
\]
where $(\vx,\vy)\sim \mathcal{D}.$  In other words this is the amount that the players learn about each others' inputs.  

Below we will show that $\textit{any}$ protocol run on our input distribution $\mathcal{D}_0$ must either consist mostly of large rectangles, or have high internal information cost.  This will be in tension with our discrepancy bound above, which shows that having too many large rectangles is unhelpful.

In the proof of the next proposition we will use the following technical fact.

\begin{claim} 
\label{claim:great_circles}
Assign the natural rotationally invariant probability measure $\mu$ to sets of equators through the north pole. Let $T$ be a measurable set of equators, and let S be a measurable subset of $\sphere{2}$.  Suppose that each element of $T$ intersects $S$ in a set of measure at least $\alpha$, with respect to the rotationally invariant $1$-dimensional probability measure $\sigma_1$ along equators.  Let $\sigma_2$ be the spherical probability measure. Then $\sigma_2(S) \geq 4 \mu(T) \alpha^2$.
\end{claim}

\begin{proof}
Recall the well-known fact that the map $F$ given in cylindrical coordinates by 
\[
(r,\theta, z)\mapsto \left(\theta, z \right)
\]
from $\sphere{2}$ to $[0,2\pi]\times [0,1]$ (equipped with Lebesgue measure) is measure-preserving.  Geometrically, $F$ projects the sphere outwardly onto a circumscribing cylinder, and then unfolds this cylinder into a rectangle.  Note that $F$ maps great semi-circles through the north pole and south pole onto vertical lines of the form $\{\theta\} \times [0,1].$ We make the following easy subclaim.

\textbf{Subclaim.}  Let $f:[-\pi/2, \pi/2] \rightarrow [-1,1] $ be the projection $\theta\mapsto \sin\theta$  from the half-circle onto the $y$-axis.  Let $Y$ be a measurable subset of $[-\pi/2,\pi/2]$.  Then $\lambda_1(f(Y)) \geq 8\sigma_1(Y)^2,$ where $\lambda_1$ is the Lebesgue measure, and $\sigma_1$ is the probability measure on $\sphere{1}.$

Intuitively, this says that to minimize the measure of the projection we should push the mass to the top and bottom of the semicircle.  To see this more rigorously, first suppose that $Y\subseteq [0,\pi/2].$ Then we have 
\[\lambda_1(f(Y)) = 2\pi \int \abs{\cos\theta} \chi_Y(\theta) d\sigma_1(\theta)
\geq 2\pi \int_{\pi/2 - 2\pi\sigma_1(Y)}^{\pi/2} \cos\theta d\simga_1(\theta) 
= \int_{\pi/2 - 2\pi\sigma_1(Y)}^{\pi/2} \cos\theta d\lambda_1(\theta)
\]
since cosine is decreasing on $[0,\pi/2].$  The latter integral evaluates to 
\[
(1 - \cos(2\pi\sigma_1(Y)))
\geq \frac{4}{\pi^2}(2\pi \sigma_1(Y))^2
= 16 \sigma_1(Y)^2
\]
To finish off the subclaim, for $Y\subseteq [-\pi/2, \pi/2],$ partition $Y$ into $Y_1$ and $Y_2$ where $Y_1 \subseteq [0,\pi/2]$ and $Y_2 \subseteq [-\pi/2, 0).$ Then
\[
\lambda_1(f(Y)) = \lambda_1(f(Y_1)) + \lambda_1(f(Y_2)) 
\geq 16(\sigma_1(Y_1)^2 + \sigma_1(Y_2)^2)
\geq 8 (\sigma_1(Y_1) + \sigma_1(Y_2))^2 = 8\sigma_1(Y)^2.
\]

Given the subclaim and the stated conditions, we see that $F(S)$ intersects at least a $\frac{1}{2}\mu(T)\cdot 2\pi$ measure of vertical lines each in a set of measure at least $\frac{1}{2}\cdot 8\alpha^2.$  The claim follows from Fubini's Theorem.
\end{proof}

\begin{proposition} 
\label[proposition]{prop:low_info_implies_large_rects}
Let $\mathcal{D}_0$ be the rotationally invariant distribution over pairs of orthogonal vectors on $\sphere{2}$, and consider a deterministic protocol run on pairs $(\vv, \vw)\sim \mathcal{D}_0$.  Let $I$ be the internal information cost of the protocol on this distribution.  With probability at least $0.4$, $(\vv, \vw)$ is in a rectangle of measure at least $2^{-60 I}$.
\end{proposition}

\begin{proof}
Fix a deterministic protocol, and let $R = S_A \times S_B$ be the combinatorial rectangle that $(\vv, \vw)\sim \mathcal{D}$ lies in.

Let $E_1$ be the event that $\sigma_1(\vv^{\perp} \cap S_B)$ is at least $2^{-10I}$ where $\sigma_1$ is the natural probability measure over $\vv^{\perp}\cap \sphere{2}$. (In other words the event that the equator orthogonal to Alice’s vector has large intersection with $S_B$.) Note that $E_1$ occurs with probability at least $0.9$ (otherwise the information cost of the protocol would be larger than $I$).

Let $E_2$ be the event
\[
\sigma_1\left(\{\vu \in \vv^{\perp} \cap S_B: \sigma_1(\vu^{\perp} \cap S_A) \geq 2^{-10 I} \}\right) \geq \frac{1}{2} \sigma_1(\vv^{\perp} \cap S_B).
\]
In other words $E_2$ is the event that for a good fraction of $\vu \in (\vv^{\perp} \cap S_B)$, the orthogonal equator to $\vu$ has large intersection with $S_A$.

By the same reasoning as above, \[\pr(\sigma_1(\vw^{\perp} \cap S_A) \geq 2^{-10I}) \geq 0.9.\] For fixed $\vv'$ and $R' = S_A'\times S_B'$, let
\[
P_{\vv',R'} = P_{\vw}\left(\sigma_1(\vw^{\perp} \cap S_A) \geq 2^{-10I} \big|\, \vv=\vv'\,\,\text{and}\,\, (\vv', \vw) \in R'\right).
\] 
 
By the statement two lines above, with probability at least $0.8$ over $\vv$ and $R$, $P_{\vv,R} \geq 1/2$. Also, conditioned on $\{\vv=\vv', \text{and}\, (\vv', \vw) \in R'\}$, $\vw$ is distributed uniformly over $(\vv')^{\perp} \cap S_B'$.  Conditioned on $P_{\vv,R}\geq 1/2$ we have
\begin{align*}
&\Pr_{\vw\sim \text{Unif}(R\cap \vv^{\perp})}\left(\sigma_1(\vw^{\perp}\cap S_A) \geq 2^{-10 I} \right) \geq 1/2,
\end{align*}
which is equivalent to $E_2.$

Thus with probability at least $0.7$, events $E_1$ and $E_2$ hold simultaneously.  When this happens, $\vv^{\perp} \cap S_B$ is large, and at least half the points in that set have orthogonal equators that have large overlap with $S_A$.  So conditioned on $E_1$ and $E_2$, there is a subset $U$ of $v^{\perp} \cap S_B$, such that $m(U) \geq 0.5\cdot 2^{-10 I}$ and $\sigma_1(\vu^{\perp} \cap S_A) \geq 2^{-10 I}$ for all $\vu$ in $U$.

We have found that most of the time $S_A$ has large overlap with a large measure of the great circles passing through a fixed point.  This means that $S_A$ is typically large, as formalized by \cref{claim:great_circles}. Applying the claim to our situation gives 
\[
m_2(S_A) \geq \frac12 \cdot 4(0.5 \cdot 2^{-10 I}) (2^{-10 I})^2 \geq 2^{-30 I},
\]
with probability at least $0.7.$
(The extra factor of $1/2$ is because antipodal points correspond to the same equator.)  Symmetrically, the same bound applies to $S_B$ and a union bound finishes the argument.
\end{proof}

The above argument is sufficient to get a lower bound for public coin protocols, however we will need the analogous fact for private coin protocols.  To do this we use a ``reverse-Newman" type result from \cite{braverman2014public} which gives a slight improvement over an earlier result from \cite{brody2016towards}. This allows us to replace our public-coin lower bound with a private-coin lower bound, albeit at the cost of restricting the number of rounds.

\begin{lemma}
\label{lem:public_coin_lower_bound}
Let $\mathcal{D}_0$ and $\mathcal{D}_{\eps}$ be as above.  Let $\Pi = \Pi(x,y,R)$ be the transcript of a public coin protocol with coins $R$, that distinguishes $\mathcal{D}_0$ from $\mathcal{D}_{\eps}$ with probability at least $0.9.$  Then when $(\vx,\vy)\sim D_0$ we have $I(\Pi; \vy |\vx,R) + I(\Pi; \vx | \vy,R) \geq \frac{1}{60}\log\frac{1}{15\eps^{1/4}} .$
\end{lemma}

\begin{proof}
We are considering public-coin protocols, so by Yao's principle it suffices to consider deterministic protocols. We have an $2\eps^{1/4}$ discrepancy bound on rectangles for $\mathcal{D}_0$ and $\mathcal{D}_{\eps}$ by \cref{prop:discrepancy_bound}. This implies that a protocol that succeed with probability $0.9$ must have have probability at least $0.75$ of $(\vv,\vw)$ lying in rectangle of measure at most $15\eps^{1/4}$, when $(\vv,\vw)\sim\mathcal{D}_0$.

To see this, say that a rectangle is small if it has measure at most $15\eps^{1/4}$ and large otherwise.  Suppose that a sample $(\vv,\vw)\sim \mathcal{D}_0$ has probability at least $0.25$ of lying in a large rectangle.  There are at most $1/(15\eps^{1/4})$ large rectangles, so by the discrepancy bound, the probability that $(\vv,\vw)\sim \mathcal{D}_{\eps}$ lies in a large rectangle is at least $0.25 - 1/(15\eps^{1/4}) \cdot (2\eps^{1/4}) = 7/60.$  Now consider the distributions $\tilde{\mathcal{D}}_0$ and $\tilde{\mathcal{D}}_{\eps}$ over rectangles induced by $\mathcal{D}_0$ and $\mathcal{D}_{\eps}.$  We bound their total variation distance.  Summing the absolute differences in probabilities for $\tilde{\mathcal{D}}_0$ and $\tilde{\mathcal{D}}_{\eps}$ over large rectangles the gives at most $1/(15\eps^{1/4}) \cdot (2\eps^{1/4}) = 2/15$  The corresponding sum over small rectangles is at most $(1 - 0.25) + (1 - 7/60) \leq 1.64.$ So the total variation distance between $\tilde{\mathcal{D}}_0$ and $\tilde{\mathcal{D}}_{\eps}$ is at most $\frac{1}{2} (2/15 + 1.64) < 0.9,$ which contradicts the protocol succeeding with $0.9$ probability.

Combining with \cref{prop:low_info_implies_large_rects} (and noting that $0.4 + 0.75 > 1$), we see that a correct protocol with information cost $I$ on $\mathcal{D}$ must have $2^{-60I} \leq 15\eps^{1/4}$, from which the claim follows.
\end{proof}

As an easy consequence, we can construct a distribution $\mathcal{D}_{\eps}'$ for which $\mathcal{D}_0$ and $\mathcal{D}_{0}$ are hard to distinguish but where $\abs{\inner{\vx}{\vy}}\geq \eps$ a.s. when $(\vx,\vy)\sim\mathcal{D}_{\eps}.$  This will be more convenient below.

\begin{proposition}
\label{prop:distinguishing_D0_and_D_eps_prime}
Define $\mathcal{D}_{\eps}'$ as the conditional distribution $\mathcal{D}_{20\eps}\bigg|\{\abs{\inner{\vx}{\vy}}\geq \eps\}.$  Let $\Pi = \Pi(\vx,\vy,R)$ be the transcript of a public coin protocol with coins $R$, that correctly identifies $\mathcal{D}_0$ or $\mathcal{D}_{\eps}'$ with probability at least $0.95.$  Then when $(\vx,\vy)\sim D_0$ we have $I(\Pi; \vy |\vx,R) + I(\Pi; \vx | \vy,R) \geq \frac{1}{60}\log\frac{1}{50\eps^{1/4}} .$
\end{proposition}

\begin{proof}
Suppose that we have a protocol that decides between $\mathcal{D}_0$ and $\mathcal{D}_{\eps}'$ with probability $0.95.$  By \cref{prop:D_eps_samples_have_large_inner_prod} this protocol gives an algorithm to distinguish between $\mathcal{D}_0$ and $\mathcal{D}_{20\eps}$ with probability at least $0.9$, since with probability at least $0.95$ a sample $(\vx,\vy)$ from $\mathcal{D}_{20\eps}$ satisfies the condition $\abs{\inner{\vx}{\vy}}\geq \eps.$  The claim now follows from \cref{lem:public_coin_lower_bound}.
\end{proof}

We next state the version of the Reverse-Newman theorem that we use.  This is essentially Theorem 1.1 of \cite{braverman2014public} who states the result for single-round protocols.  \cite{brody2016towards} shows a similar (slightly weaker) result for one-round protocols and then inductively generalizes to $r$-round protocols.  The same induction applies to the one-round protocol of \cite{braverman2014public}.

\begin{proposition}
\label{prop:reverse_newman}
An $r$-round private coin protocol with internal information cost $I$ on an input distribution $\mu$ can be simulated by a public coin protocol with information cost $I + O(r \log I).$
\end{proposition}

\begin{proof}
Combine Theorem 1.1 of \cite{braverman2014public} with the inductive argument given in \cite[Section 3.2]{brody2016towards}, replacing $c\log(2n\ell)$ with $\log I$. Note that by \cite{braverman2014public}, $\log I$ is an upper bound on the information revealed in each round $j$. This is simply because the information revealed in round $j$ is bounded by the total information revealed by the protocol, which we assume is $I.$
\end{proof}

We also mention here a direct-sum technique due to \cite{bar2004information} that we apply below.  The approach is to consider an ``OR" of $d/3$ independent instances of the three-dimensional inner product problem, and to show that a correct protocol has high information cost on the  input distribution $\mathcal{D})^{d/3}$.  The idea is that, given an instance of the three-dimensional problem, Alice and Bob can then construct $d/3 - 1$ additional instances consisting of pairs of orthogonal vectors thereby embedding their single instance into a larger OR-instance.   Moreover, their instance can be inserted into any of the $d/3$ positions, ultimately leading to a $d/3$-factor information cost increase for the OR-instance.  To sample the additional pairs of orthogonal vectors from $\mathcal{D}_0$ however, they need shared information, namely one of the two vectors (then the other vector can be sampled privately). So we actually need to argue that a correct protocol has high external information cost on $\mathcal{D}_0$, when one of the two vectors is revealed.  This is simply the content of our previous lemma bounding the internal information cost on $\mathcal{D}_0$. 

\subsubsection{Proof of \cref{lem:two_player_inner_product_hardness}}
\begin{proof}[Proof of \cref{lem:two_player_inner_product_hardness}]
Our goal is to apply Theorem 5.6 of \cite{bar2004information}.  Actually, we will apply apply a version of this theorem for protocols with at most $r$ rounds for which the same proof applies.

We consider the following problem.  Alice and Bob are given $d/3$ three-dimension vectors $\vv_1$,$\vv_2$, $\ldots$, $\vv_{d/3}$ and $\vw_1, \vw_2, \ldots, \vw_{d/3}$.  They must output $0$ if (a) $\inner{\vv_i}{\vw_i} = 0$ for all $i$ and output $1$ if (b) $\abs{\inner{\vv_i}{\vw_i}}\geq \eps$ for exactly one $i.$  We show that this problem requires $\Omega(d\log\frac{1}{\eps})$ queries under the stated assumption on the number of rounds.  This will then immediately imply the lemma.  To see this, let $\vv$ and $\vw$ be the concatenations of the $\vv_i$'s and the $\vw_i$'s.  Then $\vv' := \frac{1}{\sqrt{d/3}}\vv$ and $\vw' := \frac{1}{\sqrt{d/3}}w$ are unit vectors.  In case (a), $\inner{\vv'}{\vw'}=0$ while in case (b), $\abs{\inner{\vv'}{\vw'}} \geq \frac{\eps}{d/3}\geq \frac{\eps}{d}.$ So a protocol to solve \cref{prob:two_player_inner_prod} could in particular distinguish between case (a) and case (b).

We use the notation of Theorem 5.6 in \cite{bar2004information}. Let $f$ be the boolean function above that Alice and Bob wish to compute.  Let $h:\R^3\times \R^3 \rightarrow \{0,1\}$ be the boolean function which is $0$ precisely when inputs are orthogonal.  In the language of 
\cite{bar2004information},  $f$ is OR-decomposable with primitive $h$ meaning that 
\[
f(\vv,\vw) = h(\vv_1,\vw_1) \vee \ldots \vee h(\vv_{d/3}, \vw_{d/3}).
\]

To apply Theorem 5.6 of \cite{bar2004information} we define a mixture of product distributions $\zeta$ as follows.  Let $D = (a,\vx)$ be a uniformly random sample from $\{0,1\}\times \sphere{2}.$  Roughly, $\vx$ reveals the vector for player $a.$  More formally, if $D=(0,\vx)$ for some $\vx$, then set $\mathbf{X}=\vx$ and $\mathbf{Y}$ uniform over unit vectors orthogonal to $\vx.$ Similarly if $D=(1,\vx)$, then set $\mathbf{Y}=\vx$ and set $\mathbf{X}$ to uniform over unit vectors orthogonal to $\vx.$  Now let $\zeta$ be the distribution for which $((\mathbf{X},\mathbf{Y}), D)$ is a sample from $\zeta.$

The distribution of inputs given by $\zeta^{d/3}$ is uniform over pairs of orthogonal vectors in each coordinate, and hence in the language of \cite{bar2004information}, this input distribution is a collapsing distribution for $f$\footnote{See \cite{bar2004information} for a detailed definition.  Roughly this means that if we replace the coordinate $i$ inputs from this distribution with a pair $(\vv_i',\vw_i')$, then $f(\vv,\vw) = h(\vv_i,\vw_i).$}.  It then follows from (the proof of\footnote{The only difference is that we impose a restriction on the number of rounds.  But \cite{bar2004information} proves their result by a simulation argument that preserves the number of rounds, so the proof is unchanged.}) Theorem 5.6 in \cite{bar2004information} that
\[
\CIC_{\zeta^{d/3},\delta,r}(f) \geq \frac{d}{3} \CIC_{\zeta, \delta,r}(h)
\]
where we define $CIC_{\zeta^{d/3},\delta,r}(f)$ to be the information complexity of $f$ on the input distribution $\eta$ for $r$-round protocols that succeed with probability at least $1-\delta$ on all valid inputs, given that $D$ is observed.  In other words $CIC_{\zeta^{d/3},\delta,r}(f)$ is the minimum of $I(\mathbf{X},\mathbf{Y}; \Pi | D)$ over all (private-coin) $r$-round protocols with success probability at least $1-\delta$, when $(D,(\mathbf{X},\mathbf{Y})) \sim \zeta.$

By our definition of $\zeta$, $D$ reveals one of the two vectors at random.  So 
\[
\CIC_{\zeta,\delta,r}(h) = \min_{\Pi} \left(\frac12 I(\mathbf{X};\Pi | \mathbf{Y}) + \frac12 I(\mathbf{Y};\Pi | \mathbf{X})\right),
\]
where the minimum is over $r$-round protocols that fail with probability at most $\delta$.  This latter quantity is simply half of the (internal) information complexity for $r$-round protocols that solve the $3$-dimensional version of \cref{prob:two_player_inner_prod}.   Any public coin protocol that solves this problem with at least $0.95$ probability can in particular distinguish between $\mathcal{D}_0$ and $\mathcal{D}_{\eps}'$ with $0.95$ probability, and thus by \cref{prop:distinguishing_D0_and_D_eps_prime}, has information complexity at least $\Omega(\log\frac{1}{\eps})$ on $\mathcal{D}_0.$ 

Then by \cref{prop:reverse_newman} any $r$-round private coin protocol with the same parameters requires at least $\Omega(\log\frac{1}{\eps})$ information cost when $r\leq c \log(1/\eps)/\log\log(1/\eps)$ for an absolute constant $c.$

We have therefore shown that $\CIC_{\zeta,\delta, r} \geq c\log\frac{1}{\eps}$. The lemma follows from noting that communication cost is at least information cost.
\end{proof}

\subsubsection{From Two Players to $s$ Players}

Our next task is to bootstrap our two-player lower bound to the $s$-player version.   We start with the following simple fact which follows from symmetrizing our hard distributions.  In this section we update our notation for the distributions to refer to distributions on $\R^d.$

\begin{proposition}
\label{prop:hard_distribution_for_two_player_game}
Let $\mathcal{D}_0$ be the rotationally invariant distribution over pairs of orthogonal vectors in $\R^d.$  There is another distribution $\mathcal{D}_1$ over pairs of vectors $(\vv,\vw)$ such that for $(\vv,\vw)\sim \mathcal{D}_1$ we have $\abs{\inner{\vv}{\vw}} \geq \frac{\eps}{\sqrt{d}}$ a.s. and distinguishing $\mathcal{D}_0$ and $\mathcal{D}_1$ with probability at least $0.9$ requires $\Omega(d\log\frac{1}{\eps})$ communication for any protocol with at most $C\log\frac{1}{\eps}/\log\log\frac{1}{\eps}$ rounds.
\end{proposition}

\begin{proof}
Let $\mathbf{U}$ be a (Haar-)random orthogonal transformation, and let $\tilde{\mathcal{D}}_0$ and $\tilde{\mathcal{D}}_1$ be the hard instance given by \cref{lem:two_player_inner_product_hardness}.  Then symmetrize by setting $\mathcal{D}_i = \mathbf{U} \tilde{\mathcal{D}}_i$ for $i=1,2.$  Note that symmetrizing does not make the problem easier, as the players could accomplish this on their own using shared randomness, and no communication.
\end{proof}

\begin{lemma}
    \label{lem:reduction_s_to_two_player}
    Suppose that there is a protocol for \cref{prob:s_player_inner_prod} that succeeds with probability at least $1-\delta$, uses at most $r$ rounds of communication between the coordinator and each server, and uses total communication at most $M$.  Then there is an $r$-round protocol to solve the two player \cref{prob:two_player_inner_prod} with probability at least $0.95 - \delta$ and total communication $20 M/s.$
\end{lemma}

\begin{proof}
Let $\mathcal{P}$ be a protocol that solves the $s$-player game with the parameters given above.  Now we define $s$ two-player protocols $\mathcal{P}_1, \ldots \mathcal{P}_s$ that distinguish between $\mathcal{D}_0$ and $\mathcal{D}_1.$

For protocol $\mathcal{P}_k$, Alice first samples $s-1$ i.i.d. vectors $\vw_1,\ldots, \vw_{k-1}, \vw_{k+1},\ldots, \vw_s$ (one for each index other than $k$) orthogonal to her vector $\vv$ from the distribution that is invariant under rotations fixing $\vv.$ Then Alice simulates $s-1$ servers indexed by $i\in [s]\setminus \{k\}$ with server $i$ holding $\vw_i.$  Finally server $k$ is simply taken to be Bob.  Then the $s$-player protocol $\mathcal{P}$ is run on these servers. Note that Alice can simulate communication with server $i\neq k$ using no communication, and can simulate communication with server $k$ by exchanging messages with Bob.  By correctness of $\mathcal{P}$ each protocol $\mathcal{P}_i$ correctly distinguishes between $\mathcal{D}_0$ and $\mathcal{D}_1$ with probability at least $1-\delta.$

Suppose that Alice and Bob's input comes from $\mathcal{D}_0.$ Conditioned on $\vv$, the random vectors $\vw_1, \vw_2, \ldots,\vw_s$ all have the same distribution.  Thus there must be a $j$ such that $\mathcal{P}_j$ uses at most $20 M/s$ communication with probability at least $0.95$.  To see this, let $E_i$ be the expected communication of protocol $\mathcal{P}$ between the coordinator and server $i$, when $v$ is uniform over the sphere, and $\vw_1,\ldots, \vw_s$ are drawn i.i.d. uniformly from $\vv^{\perp}.$ By linearity of expectation, $E_1 + \cdots + E_s \leq M$, so $E_j \leq M/s$ for some $j$.  Then by Markov's inequality $\mathcal{P}_j$ uses at most $20M/s$ communication with probability at least $0.9.$

We have constructed a two-player protocol $\mathcal{P}_j$ that is correct and uses at most $10M/s$ communication on input distribution $\mathcal{D}_0$ with probability at least $0.9$. We use this to construct a protocol $\mathcal{P}_j'$ by terminating $\mathcal{P}_j$ early if necessary.  To run $\mathcal{P}_j'$, Alice and Bob simply run $\mathcal{P}_j$, while keeping track of the total communication used.  If sending the next message of $\mathcal{P}_j$ would put the total communication above $10M/s$, then that player simply terminates the protocol and outputs $1$. Now we argue that $\mathcal{P}_j'$ is correct.  On the input distribution $\mathcal{D}_0$, $\mathcal{P}_j'$ fails if either $\mathcal{P}_j$ is terminated early, or $\mathcal{P}_j$ is incorrect.  By a union bound $\mathcal{P}_j$ fails on $\mathcal{D}_0$ with probability at most $0.05 + \delta.$   On $\mathcal{D}_1$, $\mathcal{P}_j'$ only fails if $\mathcal{P}_j$ fails which occurs with probability at most $\delta.$  So $\mathcal{P}_j'$ distinguishes $\mathcal{D}_0$ from $\mathcal{D}_1$ with probability at least $0.95 - \delta.$

\end{proof}

\begin{remark}
\label{remark:constant_d_lower_bound}
For constant $d$, this argument gives a communication lower bound of $\Omega(s\log\frac{1}{\eps})$ for solving \cref{prob:s_player_inner_prod} with constant probability.  In this case we do not need to restrict the number of rounds since \cref{prop:discrepancy_bound} immediately implies an $\Omega(\log\frac{1}{\eps})$ lower bound for \cref{prob:two_player_inner_prod} when $d=3$ (or in fact when $d$ is any fixed integer greater than $2$).
\end{remark}

\noindent\textbf{Proof of \cref{thm:s_player_inner_product_hardness}.}  For two-player protocols with at most $r=C\log\frac{1}{\eps}/\log\log\frac{1}{\eps}$ rounds we have a communication lower bound of $\Omega(d\log\frac{1}{\eps}$ from \cref{prop:hard_distribution_for_two_player_game} for distinguishing $mathcal{D}_0$ and $\mathcal{D}_1.$  Thus in the setup of \cref{lem:reduction_s_to_two_player} we must have $20M/s \geq \Omega(d\log \frac{1}{\eps})$ which implies that $M\geq \Omega(sd\log\frac{1}{\eps})$ as desired.

\subsection{High-Precision Lower Bound}

In this section we show the follow result for obtaining a solution to high additive error precision.
\thmHighAccuracyLowerBound*

Our main observation is that a Gaussian least-squares problem is somewhat sensitive to each individual row.

\begin{proposition}
    \label{prop:estimate_point_on_sphere}
    Consider a communication game between Alice and Bob, where Alice is given a uniformly random vector $\vv$ on the sphere $\sphere{d-1}.$  Bob would like to produce a vector $\hat{\vv}$ with $\norm{\vv-\hat{\vv}}{} \leq \eps.$

    Any protocol that succeeds at this game with probability at least $0.1$ must use at least $\Omega(d\log\frac{1}{\eps})$ communication.
\end{proposition}

\begin{proof}
It is a standard fact that there exists a $3\eps$ packing of $\sphere{d-1}$ of size $N:=(c/\eps)^{d-1}$ for an absolute constant $c$.
This follows from a volumetric argument for example (see e.g. \cite{vershynin2018high}).

Now consider the following game.  Alice is given a random integer $i$ in $\{1,\ldots N\}$ which Bob must learn with at least $0.1$ probability.  This game clearly requires $\Omega(\log N)$ communication.  To see this, note that it suffices to consider deterministic protocols by Yao's principle.  Any deterministic protocol sending $m$ bits, forces Bob into choosing one of at most $2^m$ outputs, so $2^m/N \geq 0.1$ which means that $m\geq \log(0.1 N).$

On the other this game reduces to the game stated in the proposition.  Alice identifies the numbers $1,\ldots, N$ with points $\vx_1, \ldots, \vx_N$ in the packing.  Given an index $i$, Alice chooses the point $\vx_i.$  Using shared randomness Alice and Bob choose a Haar-random orthogonal transformation $U$.  Then $\mathbf{U}\vx_i$ is uniform over $\sphere{d-1}$ so Alice and Bob may run the a protocol for the above problem, allowing Bob to find $\hat{\vv}$ with $\norm{\hat{\vv} - Ux_i}{} \leq \eps.$  Then Bob computes $\hat{x}_i = \mathbf{U}^{-1}\hat{\vv}$ which satisfies $\norm{\vx_i - \hat{\vx}_i}{} \leq \eps.$  Since $\vx_1, \ldots, \vx_N$ was a $3\eps$ packing, this allows Bob to recover $\vx_i$ and hence $i$ as desired.  The proposition follows.
\end{proof}

We will show the following technical communication result and then apply it.

\begin{lemma}
\label{lem:estimate_F_communication_game}
Consider a communication game between Alice and Bob where  Alice is given a (rotationally-invariant) random vector $\vv\in\R^d$ with $\norm{\vv}{}=\alpha \leq 1$, and both players see a matrix $\mm\in\R^{d\times d}$ and a vector $\vw\in \R^d.$  Suppose further that $\mathbf{I} \preceq \mm \preceq c\mathbf{I}$ for an absolute constant $c$ and that $\norm{\vw}{}\geq 1$. Define 
\[F(\vx) = F_{M,\vw}(\vx) = \frac{\vx^\top \mm \vw}{1 + \vx^\top \mm \vx} \mm \vx.\]
Alice would like to send a message to Bob from which Bob can produce $\vy\in\R^d$ with $\norm{F(\vv) - \vy}{} \leq \eps.$

Any communication protocol that succeeds with probability at least $0.9$ requires at least $\Omega(d\log\frac{\alpha}{\eps d})$ communication.
\end{lemma}

\begin{proof}
Let $S$ be the set of points $\vx$ in $\alpha \sphere{d-1}$ with $\vx^\top \mm \vw \geq \frac{\alpha}{10\sqrt{d}}$. Note that 
\[
\Pr(\vx\in S) 
= \frac12 \left(1 -\Pr\left(\abs{\vx^\top \mm \vw}\leq \frac{\alpha}{10\sqrt{d}}\right)\right) \geq \frac{1}{4}.
\]
To see this, note that $\norm{\mm \vw}{}\geq 1.$  Then by rotational invariance
\[
\Pr\left(\abs{\vs^\top \mm \vw}\leq \frac{\alpha}{10\sqrt{d}}\right) 
\leq \Pr\left(\abs{\vx^\top \ve_1} \leq \frac{\alpha}{10\sqrt{d}}\right)
= \Pr\left(\frac{1}{\alpha^2}\abs{\vx^\top \ve_1}^2 \leq \frac{1}{100 d}\right).
\]
The quantity $\frac{1}{\alpha^2}\abs{\vx^\top \ve_1}^2$ is distributed as $g_1^2 / (g_1^2 + \ldots + g_d^2)$ where each $g_i$ is a standard normal.  Note that $g_1^2 > 1/10$ with probability at least $3/4.$  Also $g_1^2 + \ldots + g_d^2 \leq 2d$ with probability at least $1-\exp(-cd)$ by Bernstein's inequality for example (see \cite{vershynin2018high} for example).  So when $d$ is a large enough constant, the probability above is at most $1/2.$ (The constants here are of course not close to optimal.)

For all $\vx\in S$ we have
\[
\frac{\vx^\top \mm \vw}{1 + \vx^\top \mm \vx} 
\geq 
\frac{\alpha}{20\sqrt{d}(1 + 2\alpha^2)}.
\]
Thus for $\vx_1,\vx_2\in S$ we have
\begin{align*}
\norm{F(\vx_1) - F(\vx_2)}{}
&= \norm{\frac{\vx_1^\top \mm \vw}{1 + \vx_1^\top \mm \vx_1}\mm \vx_1 - \frac{\vx_2^\top \mm \vw}{1 + \vx_2^\top \mm \vx_2}\mm \vx_2}{}\\
&\geq \min\left(\frac{\vx_1^\top \mm \vw}{1 + \vx_1^\top \mm \vx_1} , \frac{\vx_2^\top \mm \vw}{1 + \vx_2^\top \mm \vx_2}\right)\norm{\mm \vx_1 - \mm \vx_2}{}\\
&\geq \frac{\alpha}{10\sqrt{d}(1 + C\alpha^2)}\sigma_{\min}(M)\norm{\vx_1 - \vx_2}{}\\
&\geq \frac{\alpha}{10\sqrt{d}(1 + C\alpha^2)}\norm{\vx_1 - \vx_2}{}.\\
&:= \beta \norm{\vx_1-\vx_2}{}
\end{align*}

Condition on $\vv$ lying in $S$ which happens with at least $1/4$ probability. A protocol that solves the above communication problem allows Bob to produce a vector $\vy$ with $\norm{F(\vv) - \vy}{}\leq \eps.$  If $\hat{\vv}$ is another vector satisfying $\norm{F(\hat{\vv}) - \vy}{}\leq \eps$, then 
\[
2\eps \geq \norm{F(\vv) - F(\hat{\vv})}{}\geq \beta \norm{\vv-\hat{\vv}}{},
\]
so $\norm{\vv-\hat{\vv}}{} \leq \frac{2\eps}{\beta}.$

Thus with failure probability at most $\frac{3}{4} + \frac{1}{10} = 0.85$ Bob can produce a $\frac{2\eps}{\beta}$ additive approximation $\hat{\vv}$ to Alice's vector $\vv$.  However this latter communication problem requires at least $\Omega(d\log\frac{\beta}{\eps}) = \Omega(d\log\frac{\alpha}{\eps d})$ communication by \cref{prop:estimate_point_on_sphere}.

\end{proof}

\begin{proposition}
\label{prop:two_player_additive_lower_bound}
Consider the following communication game.  Alice holds row $\ma_1$ of $\ma$ which is uniform over $\alpha \sphere{d-1}$ with $\alpha=\sqrt{\frac{d}{s+d}}$.  Bob holds $s-1$ independent rows $\ma_2,\ldots, \ma_s$, each with the same distribution as Alice's.  Also $\ma_{s+k} = \ve_k$ for $k=1,\ldots,d.$  Set $\vb=\ve_{s+1}$ so that $\ma^\top \vb = \ve_1.$  Alice and Bob would like to compute $\hat{\vx}$ with $\norm{\ma\hat{\vx} - \vb}{} \leq \eps + \norm{\ma \vx_* - \vb}{},$ with probability at least $0.95$.
This communication game requires at least $\Omega\left(d\log\frac{1}{\eps(s+d)}\right)$ communication.
 \end{proposition}

\begin{proof}

First, observe that $\ma$ is well-conditioned with high probability. Indeed, with probability at least $\exp(-d)$ (which is at most $0.05$ for $d\geq 3$), we have $I \leq \ma^\top \ma \leq cI$ for an absolute constant $c.$ The first inequality always holds by construction of $\ma.$ The latter holds by standard concentration results for the top singular value of a matrix with subgaussian rows (see, for example, \cite[Section 4]{vershynin2018high}). Note that $\vx_* = (\ma^\top \ma)^{-1}\ma^\top \vb = (\ma^\top \ma)^{-1}\mathbf{e}_1.$  

Note that $\hat{\x}$ satisfies $\norm{\widehat{\vx}-\vx_*}{} \leq 2\sqrt{\eps},$ since
\[
\norm{\ma \hat{\vx} - \vb}{}^2 = \norm{\ma(\hat{\vx} - \vx_*)}{}^2 + \norm{\ma\vx_* - \vb}{}^2
\geq \norm{\hat{\vx} - \vx_*}{}^2 + \norm{\ma\vx_* - \vb}{}^2,
\]
and so under the stated conditions
\[
\norm{\hat{\vx} - \vx_*}{}^2 \leq \norm{\ma \hat{\vx} - \vb}{}^2 - \norm{\ma\vx_* - \vb}{}^2
\leq (1+\eps)^2 - 1^2 \leq 3\eps,
\]
where the inequality above uses $\norm{\ma\vx_* - \vb}{}^2 \leq \norm{\vb}{}^2 = 1.$

Let $\ma'$ denote $\ma$ with row $1$ removed.  Let $\ma_1$ denote the first row of $\ma$ as a column vector.  To simplify notation, let $\mm = (\ma')^\top \ma'.$
By the Sherman-Morrison formula, we have

\begin{align*}
    x_* = (\ma^\top \ma)^{-1} \mathbf{e}_1
    &= (\mm + \ma_1 \ma_1^\top)^{-1} \mathbf{e}_1\\
    &= \mm^{-1}\mathbf{e}_1 - \frac{(\mm^{-1}\ma_1)(\mm^{-1}\ma_1)^\top}{1 + \ma_1^\top \mm^{-1} \ma_1 }\mathbf{e}_1\\
    &= \mm^{-1}\mathbf{e}_1 + F_{\mm^{-1},\mathbf{e}_1}(\ma_1),
\end{align*}
in the notation of \cref{lem:estimate_F_communication_game}. 

Note that Bob can calculate $\mm^{-1}\mathbf{e}_1$ directly.  So a protocol that computes $\vx^\ast$ to within $2\sqrt{\eps}$ additive error, would yield a $2\sqrt{\eps}$ additive approximation to $F_{\mm^{-1},\mathbf{e}_1}(\ma_1).$ By \cref{lem:estimate_F_communication_game}, this requires at least $\Omega(d\log\frac{\alpha}{\sqrt{\eps} d}) = \Omega(d\log\frac{1}{\eps(s+d)})$ communication.

\end{proof}

\noindent Now we are ready to prove \cref{thm:high_precision_regression_lower_bound}.
\begin{proof}[Proof of \cref{thm:high_precision_regression_lower_bound}]
Similar to \cref{prop:two_player_additive_lower_bound} above, our hard instance is as follows.  The matrix $\ma$ has $s$ rows each uniform over $\sqrt{\frac{d}{s+d}} \sphere{d-1}$ each held on a different server.  We also include rows $\ve_1,\ldots, \ve_d$ which are known to all servers and set $\vb = \ve_{s+1}$, also known to all servers.

Note the input to each server is i.i.d. and so we can use the symmetrization argument of \cite{phillips2012lower}. We sketch the idea here and refer the reader to \cite{phillips2012lower} for more details.  Let $\mathcal{P}$ be a protocol for our $s$-player game with communication $C(\mathcal{P})$.  Then, construct a two player game $\mathcal{P}'$ by having Alice choose a uniformly random player $1,\ldots, s$ which she simulates and having Bob simulate the remaining players.  The protocol $\mathcal{P}'$ uses at most $100 C(\mathcal{P})/s$ communication with probability at least $0.98$ since the expected communication of  $\mathcal{P}'$ is at most $2C(\mathcal{P})/s$.  However, Alice and Bob now have a protocol that solves \cref{prop:two_player_additive_lower_bound} with at least $0.95$ probability, and therefore 
$
100 C(\mathcal{P})/s \geq \Omega\left(d\log\frac{1}{\eps (s+d)}\right),
$
or \[C(\mathcal{P}) \geq \Omega\left( sd \log\frac{1}{\eps(s+d)}\right).\] 

Finally note that since the matrix $\ma \preceq 2\mathbf{I}$, rounding all entries to $L$ bits of precision changes the solution error by at most $O(\sqrt{d} 2^{-L})$ on vectors of norm at most $1$ (which is all that we must consider since $\vb$ has norm $1$).
\end{proof}

\section{Conclusion and Future Directions}

In this section, we discuss a few open problems and possible directions for future research regarding the communication complexity of the convex optimization problems we discussed in the paper.

\paragraph{Linear regression on matrices with special structure.} Many optimization algorithms use linear regression as a subprocedure. One such example is IPMs, which are used to solve linear programs, which we discussed in this paper. In some scenarios, the linear regression problem has a special structure, e.g., the alternating least squares algorithm in tensor decomposition \cite{diao2019optimal,fahrbach2022subquadratic} and least squares with non-negative data~\cite{diakonikolas2022fast}, which appears in many real-world problems. It would be interesting to investigate the communication complexity of solving linear regression problems that have matrices with special structures.

\paragraph{Inverse maintenance.} Many recent improvements for the running time of convex optimization problems, such as linear programming and semi-definite programming, have relied on the use of inverse maintenance. We also used this to improve the communication complexity of solving linear programs when we use IPMs in a distributed setting. There has recently been some progress on analyzing inverse maintenance for general matrix formulas in an attempt to unify the analysis of many algorithms \cite{van2021unifying,anand2024bit}. An interesting direction is to analyze the communication complexity of inverse maintenance for such general matrix formulas.

\section*{Acknowledgements}
We are very grateful to the anonymous reviewers of STOC'24 for their detailed and constructive suggestions. We thank Krishna Pillutla for helpful references to related works in distributed optimization. Research of Yin Tat Lee  was supported by NSF awards CCF-1749609, DMS-1839116, DMS-2023166, CCF-2105772, a Microsoft Research Faculty
Fellowship, a Sloan Research Fellowship, and a Packard Fellowship. Research of  Swati Padmanabhan was supported by NSF awards CCF-1749609, DMS-1839116, DMS-2023166, and CCF-2105772. Research of Guanghao Ye was supported by NSF awards CCF-1955217 and DMS-2022448. Research of William Swartworth and David P. Woodruff were supported by a Simons Investigator Award. Part of this work was done while visiting the Simons Insitute for the Theory of Computing and Google Research. 

\newpage
\printbibliography

\newpage
\appendix
\section{Some Useful Technical Results}\label{sec:appendixNotation}
In this section, we present some definitions and properties from matrix analysis and convex
analysis that we use. These results are standard
and may be found in, for example, \cite{rockafellar1970convex,boyd2004convex}. 

\begin{fact}[\cite{horn2012matrix}]\label{fact:psdOrderingFacts}
 Given a positive definite matrix $\ma$, if the inequality $\ma \preceq \mathbf{I}$ holds,  where $\mathbf{I}$ is the appropriate-sized identity matrix, then we have  $\ma^2 \preceq \mathbf{I}$
 and $\ma^{-1} \succeq \mathbf{I}$.  
 
 Further, if the inequality chain $\mathbf{0}\prec\ma\preceq \mb$ holds, then the ordering  is preserved on premultiplying and postmultiplying both sides by the same positive definite matrix, which implies $\ma^{-1}\succeq \mb^{-1}$.  
\end{fact}

\begin{definition}
\label[defn]{defn:Conjugate} Let $f:\R^{n}\rightarrow\R.$ Then the function
$f^{\ast}:\R^n\rightarrow\R$ defined as 
\[
f^{\ast}(\by)=\sup_{\bx\in\textrm{dom}(f)}\left[\inprod{\bx}{\by}-f(\bx)\right]
\]
is called the Fenchel conjugate of the function $f.$ An immediate consequence
of the definition (and by applying the appropriate convexity-preserving
property) is that $f^{\ast}$ is convex, regardless of the convexity
of $f.$ We use the superscript $\ast$ on functions to denote their
conjugates. 
\end{definition}
\begin{remark}
    Using the above definition of conjugate, for a given function $\psi$, we use the shorthand notation $-\psi^*(-t \bc)$ to express   $\min_{\bx}\left[ t\cdot \bc^\top \bx + \psi(\bx)\right]$. 
\end{remark}

\begin{fact}[Biconjugacy]
\label[fact]{lem:fastast}  For a closed, convex function $f,$ we have
$f=f^{\ast\ast}.$ 
\end{fact}

\begin{fact}[\cite{rockafellar1970convex}]
\label[fact]{lem:GradConjugate} For a closed, convex differentiable function
$f,$ we have \[\by=\nabla f(\bx) \text{ if and only if } \bx=\nabla f^{\ast}(\by).\]  
\end{fact}

\begin{fact}[\cite{rockafellar1970convex}]
\label[fact]{lem:HessianConjugate} A strictly convex, twice-differentiable
function $f$ has $\nabla^{2}f^{\ast}(\nabla f(\bx))=(\nabla^{2}f(\bx))^{-1}.$ 
\end{fact}

\ifdefined\isneurips
\defPolar*
\lemPolarIntersection*
\else 
\begin{restatable}[Polar of a Set \cite{rockafellar2009variational}]{definition}{defPolar}
\label{def:Polar} Given a set $\mathcal{S}\subseteq\R^n,$ its polar
is defined as 
\[
\mathcal{S}^{\circ}\defeq\left\{\by\in\R^n:\inprod{\by}{\bx}\leq 1,\text{ }\forall\bx\in\mathcal{S}\right\}.
\]
\end{restatable}

\begin{restatable}[\cite{rockafellar1970convex}]{fact}{factPolarReversal}\label[fact]{fact:polarReversal}
Let $\mathcal{A}\subseteq\mathcal{B}\subseteq \R^n$ be closed, compact, convex sets. Then,  $\mathcal{A}^\circ \supseteq \mathcal{B}^\circ$. 
\end{restatable}

\begin{restatable}[\cite{rockafellar1970convex}]{fact}{lemPolarIntersection}
\label[fact]{lem:polarAfterAddingPoint} 
Let $\mathcal{S}\subseteq\R^n$ be a closed, compact, convex
set, and let $\by$ be a point. 
Then $(\conv\left\{\mathcal{S},\by\right\} )^{\circ}\subseteq\mathcal{S}^{\circ}\cap\calH$,
where $\calH$ is the halfspace defined by $\calH=\left\{ \bz\in\R^{n}:\inprod{\bz}{\by}\leq 1\right\}$.
\end{restatable}
\fi 

\begin{lemma}[Theorem 2 of \cite{zong2022short}]
    \label[lem]{lem:two-sided-ineq}Given a convex set  $\Omega$ with a $\nu$-self-concordant
    barrier $\barr_{\Omega}$ and inner radius $r$. Let $\vx_{t}=\arg\min_{\vx}t\cdot\vc^{\top}\vx+\barr_{\Omega}(\vx)$.
    Then, for any $t>0$, 
    \[
    \min\left\{\frac{1}{2t},\frac{r \|c\|_{2}}{4\nu + 4 \sqrt{\nu}}\right\}\leq\vc^{\top}\vx_{t}-\vc^{\top}\vx_{\infty}\leq\frac{\nu}{t}.
    \] While the theorem in \cite{zong2022short} is stated  only for polytopes,  their proof works for general convex sets.
\end{lemma}

\subsection{Background on Interior-Point Methods}

Our work draws heavily upon geometric properties of self-concordant
functions, which underpin the rich theory of interior-point methods. 
We list below the formal results needed for our analysis,
and refer the reader to \cite{nesterov1994interior, renegar2001mathematical} for a detailed exposition of this function
class. 
We begin with the definitions of self-concordant functions and self-concordant barriers:
\ifdefined\isneurips
\defScBarr*
\else 
\begin{restatable}[Self-concordance \cite{nesterov1994interior}]{definition}{defScBarr}
\label{def:SelfConcordanceBarr}
A function $F:Q\mapsto\R$ is a self-concordant function on a convex set $Q$ if for any $\bx\in Q$ and any direction $\bh$, 
\[
\lvert D^3 F(\bx)[\bh, \bh, \bh]\rvert \leq 2(D^2 F(\bx)[\bh, \bh])^{3/2}, 
\] 
where $D^k F(\bx)[\bh_1, \dotsc, \bh_k]$ is the $k$-th derivative of $F$ at $\bx$ along the directions $\bh_1,\dotsc,\bh_k$. We say $F$ is a $\nu$-self-concordant barrier if it further satisfies $\nabla F(\bx)^\top (\nabla^2 F(\bx))^{-1} \nabla F(\bx) \leq \nu$ for any $\bx\in Q$.
\end{restatable}
\fi 

\begin{theorem}[{\cite[Theorem 2.3.3]{renegar2001mathematical}}]
\label{thm:sc1} If $f$ is a self-concordant barrier, then for all
$\bx$ and $\by\in\textrm{dom}(f)$, we have $\inprod{\nabla f(\bx)}{\by-\bx}\leq\nu,$ where
$\nu$ is the self-concordance of $f$.
\end{theorem}
 
\begin{theorem}[{\cite[Theorem 4.2.5]{nesterov1998introductory}}]
\label{thm:sc2} If $f$ is a $\nu$-self-concordant barrier such that $\bx,\by\in\operatorname{dom}(f)$
satisfy $\inprod{\nabla f(\bx)}{\by-\bx}\geq0$, then $\by\in\calE(\bx,\nu+2\sqrt{\nu}).$
\end{theorem}

\begin{theorem}[{\cite[Theorem 4.2.6]{nesterov1998introductory}}]
\label{thm:dikin-radius} If $f$ is a $\nu$-self-concordant barrier for convex set $\calK$, $\bz\in \operatorname{dom}(f)$, and $\bz^\star$ is the \emph{analytic center} for $f$, then 
\[
\calE(\bz,1) \subseteq \calK \subseteq \calE(\bz^\star,\nu+2\sqrt{\nu
}) 
\] 
\end{theorem}

\begin{corollary}\label{cor:hessian-lower-bound}
    If $f$ is a self-concordant barrier for a given convex set $\mathcal{K}$ that satisfies $K\subset \mathcal{B}(0,R)$, then $\nabla^2 f(\bx) \succeq \frac1{4R^2}I$ for any $\bx \in \calK$.
\end{corollary}
\begin{proof}
    For the sake of contradiction, suppose $\nabla^2 f \not\succeq \frac1{4R^2}I$. This implies $(2R\bu)^\top(\nabla^2 f(\bx))(2R\bu)< 1$ for some unit vector $\bu$ and for a point $\bx\in \mathcal{K}$. In other words, $\bx+2R\bu\in \calE(\bx,1) \subseteq \calK$.
    We note that since $\calK \subset \mathcal{B}(0,R)$, it is not possible to have $\bx+2R\bu\in \calK$ for any $\bx\in\calK$ and unit vector $\bu$. 
\end{proof}

\begin{theorem}\label{thm:self-concordant-with-linear-term}
    If $f$ is a $\nu$-self-concordant barrier for a given convex set $\mathcal{K}$, then $g(\bx) = \bc^\top \bx + f(\bx)$ is also self-concordant. Moreover, if $\mathcal{K}\subseteq \mathcal{B}(0,R)$, then 
   $g$ is at most $3\nu + 12R^2\|\bc\|_2^2$-self-concordant.
\end{theorem}
\begin{proof}
    For the first part, it suffices to note that $\nabla^2 g = \nabla^2 f$. For the second part, since $\mathcal{K}\subseteq\mathcal{B}(0, R)$,  \cref{cor:hessian-lower-bound} applies to give $\nabla^2 g = \nabla^2 f \succeq \frac{1}{4R^2}I$. Then, we have
    \[
    \|\bc\|_{(\nabla^2 g(x))^{-1}}^2 \leq {4R^2} \|\bc\|_2^2. 
    \]
    Hence, we have
    $\|\nabla g\|_{(\nabla^2 g(x))^{-1}}^2 \leq 3\|\nabla f\|_{(\nabla^2 f(x))^{-1}}^2 + 3\|\bc\|_{(\nabla^2 g(x))^{-1}}^2  \leq 3\nu + 12{R^2}\|\bc\|_2^2.$ 
\end{proof}

\noindent We now state the following result from self-concordance calculus. 
\begin{theorem}[Theorem 3.3.1 of \cite{renegar2001mathematical}]
\label{thm:ConjugateSC} If $f$ is a (strongly nondegenerate) self-concordant
function, then so is its Fenchel conjugate $f^{\ast}.$ 
\end{theorem}

\begin{fact}[\cite{renegar2001mathematical}]\label{fact:SCofBarrierRestrictedToLinSubspace}
    Given a self-concordant barrier $f$ with self-concordant parameter $\nu$,  the function $f$ restricted to an affine subspace $\mathcal{S}$, also has self-concordance parameter $\nu$. 
\end{fact}

\begin{fact}[Theorem $2.3.8$ in \cite{renegar2001mathematical}]\label{fact:FirstOrderApproxOfScb}
    If $f$ is a $\nu$-self-concordant barrier, with $\bx$ and $\by$ both in $\textrm{dom}(f)$, then for $0\leq s\leq 1$, we have $f(\bx+s(\by-\bx)) \leq f(\bx) - \nu\log(s)$.
\end{fact}

The following result bounds the quadratic approximation of a function  
with the distance between two points measured in the local norm. 
\begin{theorem}[Theorem 2.2.2 of \cite{renegar2001mathematical}]
\label{thm:QuadApproxErr} Let $f$ be a self-concordant function, $\bx\in\textrm{dom}(f)$, and $\by\in\calB_{\bx}(\bx,1)$, and define the Dikin ellipsoid $\|\by-\bx\|_{\bx}^{2}\defeq\inprod{\by-\bx}{\nabla^{2}f(\bx)\cdot(\by-\bx)}$. Then, the following bound holds:
\[
f(\by)\leq f(\bx)+\inprod{\nabla f(\bx)}{\by-\bx}+\frac{1}{2}\|\by-\bx\|_{\bx}^{2}+\frac{\|\by-\bx\|_{\bx}^{3}}{3(1-\|\by-\bx\|_{\bx})}.
\]
\end{theorem}

\noindent Finally, we need the following definition of the universal barrier. 

\begin{restatable}[\cite{DBLP:books/daglib/0071613,lee2021universal}]{definition}{defUniBarrier}
\label[defn]{def:UniversalBarr} Given a convex body $\calK\subseteq\R^{n}$, let $(\calK-\bx)^\circ=\{\by\in \R^n:\by^\top(\bz-\bx)\leq 1,\forall \bz\in \calK\}$ be the polar  of $\calK$ with respect to $\bx$. Then the \emph{universal barrier} of $\calK$ is defined as $\psi:\operatorname{int}(\calK)\to \R$ by 
\[
\psi(\bx) = \log\vol((\calK-\bx)^\circ).
\]It was shown in \cite{lee2021universal} that the universal barrier is $n$-self-concordant.
\end{restatable}

\subsection{Facts from Convex Geometry}
\looseness=-1Since our analysis is contingent on the change in the volume of convex
bodies when points are added to them or when intersected
with halfspaces, we invoke  Gr\"unbaum's result several
times. 

\begin{restatable}[\cite{grunbaum1960partitions, bubeck2020chasing}]{theorem}{thmGrunBaumMain}
\label{thm:GrunbaumGeneral} Let $f$ be a log-concave distribution
on $\R^{d}$ with centroid $\bc_{f}.$ Let $\mathcal{H}=\left\{ \bu\in\R^{d}:{\bu}^\top{\bv}\geq q\right\} $
be a halfspace defined by a normal vector $\bv\in\R^{d}$. Then, $
\int_{\mathcal{H}}f(\bz)d\bz\geq\frac{1}{e}-t^{+},$ 
where $t=\frac{q-{\bc_{f}}^\top{\bv}}{\sqrt{\mathbb{E}_{\by\sim f}({\bv}^\top{\by-\bc_{f}})^{2}}}$
is the distance of the centroid to the halfspace scaled by the standard
deviation along the normal vector $\bv$ and $t^+ \defeq \max\{0, t\}$.  
\end{restatable}

\begin{remark}\label[rem]{rem:GrunbaumSimpleCase}
A crucial special case of \cref{thm:GrunbaumGeneral} is that cutting a convex set through its centroid yields two parts, the smaller of which has volume at least $1/e$ times the original volume and the larger of which is at most $1-1/e$ times the original total volume. 
Formally, let $\calK$ be a convex set with centroid
$\mu$ and covariance matrix $\Sigma.$ Then, for any point $\bx$
satisfying $\|\bx-\mu\|_{\Sigma^{-1}}\leq\eta$ and a halfspace $\calH$
such that $\bx\in\calH$, we have $\vol(\calK\cap\calH)\geq\vol(\calK)\cdot(1/e-\eta).$
\end{remark}

\begin{lemma}[Section 3 in \cite{bubeck2015entropic}; Section 3 of \cite{klartag2006convex}]\label[lem]{prop:EquivalenceOfCentroidAndMinimizer}
 Let $\theta\in\R^{n}$, and let $p_{\theta}$ be defined
as $p_{\theta}(\vx)\propto\exp({\theta}^\top{\vx}-f(\theta))$,
where $f(\theta)\defeq\log\left[\int_{\kcal}\exp({\theta}^\top{\vu})d\vu\right]$.
Then, 
\[
\mathbb{E}_{\vx \sim p_{\theta}}[\vx]=\arg\min_{\vx\in\mathrm{int}(\kcal)}\left\{ f^{\ast}(\vx) - {\theta}^\top{\vx}\right\}.
\]
\end{lemma} 

\end{document}